\documentclass[a4paper,twoside]{ociamthesis}

\usepackage{diagrams} 
\usepackage{amsthm} 
\usepackage{soul} 
\usepackage{bussproofs} 
\usepackage{cmll} 
\usepackage{wasysym} 
\usepackage{subcaption} 
\usepackage{pifont} 

\usepackage[style=numeric-comp, sorting=none, backend=biber, doi=false, isbn=false]{biblatex}
\newcommand*{\bibtitle}{References}

\title{In Search of Effectful Dependent Types}
\author{Matthijs V\'ak\'ar}
\college{Magdalen College}

\degree{Doctor of Philosophy}
\degreedate{Trinity 2017}

\newcommand{\vect}[1]{\boldsymbol{#1}}

\newcommand{\txt}[1]{\quad\textnormal{#1}\quad}

\newcommand{\p}[0]{\mathcal{P}}

\newcommand{\R}[0]{\mathbb{R}}

\newcommand{\N}[0]{\mathbb{N}}
\newcommand{\B}[0]{\mathbb{B}}
\newcommand{\funim}[0]{\mathsf{im}}

\newcommand{\lbi}[3]{\mathsf{let}\;#1\;\mathsf{be}\;#2\;\mathsf{in}\;#3}
\newcommand{\nf}[1]{#1_{\mathsf{nf}}}
\newcommand{\nnf}[1]{#1_{!\mathsf{nf}}}

\long\def\symbolfootnote[#1]#2{\begingroup%
\def\thefootnote{\fnsymbol{footnote}}\footnote[#1]{#2}\endgroup}
\newcommand{\exeq}[0]{\stackrel{!}{=}}
\newcommand{\ra}[1]{\stackrel{#1}{\longrightarrow}}

\newarrow{Squiggle}{dash}{dash}{dash}{dash}{O}
\newcommand{\sem}[2][M\!,g]{ [\![ #2 ]\!]^{}}

\newcommand{\ct}[1]{\underline{#1}}
\newcommand{\return}[0]{\mathsf{return}\;}
\newcommand{\sipm}[4]{\mathsf{pm}\;#1\;\mathsf{as}\;\langle #2,#3\rangle\;\mathsf{in}\;#4}
\newcommand{\upm}[2]{\mathsf{pm}\;#1\;\mathsf{as}\;\langle \rangle\;\mathsf{in}\;#2}
\newcommand{\force}[0]{\mathsf{force}\;}
\newcommand{\thunk}[0]{\mathsf{thunk}\;}
\newcommand{\toin}[3]{#1\;\mathsf{to}\;#2\;\mathsf{in}\;#3}

\newcommand{\Bcat}[0]{\mathcal{B}}
\newcommand{\Ccat}[0]{\mathcal{C}}
\newcommand{\Dcat}[0]{\mathcal{D}}
\newcommand{\proj}[2]{\mathbf{p}_{#1,#2}}
\newcommand{\diag}[2]{\mathsf{diag}_{#1,#2}}
\newcommand{\qu}[2]{\mathbf{q}_{#1,#2}}
\newcommand{\Set}[0]{\mathsf{Set}}
\newcommand{\diagv}[2]{\mathbf{v}_{#1,#2}}
\newcommand{\id}[0]{\mathsf{id}}
\newcommand{\der}[0]{\mathsf{der}}
\newcommand{\substsf}[0]{\mathsf{subst}}
\newcommand{\Sub}[0]{\mathsf{Sub}}
\newcommand{\odd}[0]{\mathsf{odd}}
\newcommand{\even}[0]{\mathsf{even}}
\newcommand{\just}[0]{\mathsf{just}}
\newcommand{\Ty}[0]{\mathsf{Ty}}
\newcommand{\Tm}[0]{\mathsf{Tm}}

\newcommand{\Cat}[0]{\mathsf{Cat}}
\newcommand{\SMCat}[0]{\mathsf{SMCat}}
\newcommand{\Id}[0]{\mathsf{Id}}
\newcommand{\self}[0]{\mathsf{self}}
\newcommand{\frob}[1]{\mathsf{frob}(#1)}
\newcommand{\fst}[0]{\mathsf{fst}\;}
\newcommand{\snd}[0]{\mathsf{snd}\;}
\newcommand{\inl}[0]{\mathsf{inl}}
\newcommand{\inr}[0]{\mathsf{inr}}
\newcommand{\refl}[1]{\mathsf{refl}(#1)}
\newcommand{\elsetext}[0]{\mathsf{else}}
\newcommand{\case}[0]{\mathsf{case}}
\newcommand{\type}[0]{\;\;\mathsf{type}}
\newcommand{\vtype}[0]{\;\;\mathsf{vtype}}
\newcommand{\ctype}[0]{\;\;\mathsf{ctype}}
\newcommand{\idpm}[3]{\mathsf{pm}\;#1\;\mathsf{as}\;(\refl #2)\; \mathsf{in}\;#3}
\newcommand{\Fam}[0]{\mathsf{Fam}}
\newcommand{\nil}[0]{\mathsf{nil\;}}
\newcommand{\tr}[0]{\mathsf{tr}\;}
\newcommand{\ctxt}[0]{\;\;\mathsf{ctxt}}
\newcommand{\Ctxt}[0]{\mathsf{Ctxt}}
\newcommand{\fin}[0]{\mathsf{fin}}
\newcommand{\print}[1]{\mathsf{print}\; #1\;.\; }
\newcommand{\diverge}[0]{\mathsf{diverge}\;}
\newcommand{\nondet}[2]{\mathsf{choose}_{#1}(#2)}
\newcommand{\readcell}[2]{\mathsf{readto}_{#1}(#2)}
\newcommand{\writecell}[1]{\mathsf{write}\;#1\;.\;}
\newcommand{\error}[1]{\mathsf{error\;#1\;}}

\newcommand{\csigma}[2]{\Sigma_{F(#1)}^\otimes #2}
\newcommand{\cpi}[2]{\Pi_{F(#1)}^\multimap #2}
\newcommand{\homtype}[0]{\stackrel{U}{\multimap}}
\newcommand{\functype}[0]{{\;}_F\hspace{-3pt}\multimap}

\newcommand{\BInd}[0]{{\mathbb{B}_*}}
\newcommand{\NInd}[0]{{\mathbb{N}_*}}
\newcommand{\ttt}[0]{\mathsf{tt}}
\newcommand{\fff}[0]{\mathsf{ff}}
\newcommand{\xmark}{\ding{55}}
\newcommand{\cmark}{\ding{51}}
\newcommand{\Gamecat}{\mathsf{Game}}
\newcommand{\DGame}{\mathsf{DGame}}
\newcommand{\str}{\mathsf{str}}

\newcommand{\ob}{\mathsf{ob}}

\renewcommand{\emph}[1]{\textbf{#1}} 
\newcommand{\Osat}[0]{\textnormal{\textsf{O}-\textsf{sat}}}
\newcommand{\tot}[1]{\smiley (#1)}
\newcommand{\win}[1]{W_{#1}}
\newcommand{\lltensor}{\otimes}
\newcommand{\linimpl}{\multimap}
\newcommand{\llwith}{\&}

\newcommand{\Inf}[1]{P_{#1}^{\infty}}

\newcommand{\allwin}{{\tt AllWin}}
\newcommand{\nowin}{{\tt NoWin}}
\newcommand{\HoTT}{\textrm{HoTT}}
\newcommand{\PCF}{\textrm{PCF}}
\newcommand{\DTT}{\textrm{DTT}}
\newcommand{\STT}{\textrm{STT}}
\newcommand{\DTTGame}{\mathsf{DTT}_{\mathsf{CBN}}}
\newcommand{\STTGame}{\mathsf{STT}_{\mathsf{CBN}}}
\newcommand{\op}[1]{\mathsf{op}(#1)}
\newcommand{\Coh}[0]{\mathsf{Coh}}
\newcommand{\Stable}[0]{\mathsf{Stable}}
\newcommand{\cliques}[0]{\mathsf{cliques}}
\newcommand{\fun}[0]{\mathsf{fun}}
\newcommand{\trace}[0]{\mathsf{trace}}
\newcommand{\fincliques}[0]{\mathsf{cliques}_{\mathsf{fin}}}
\newcommand{\Ida}{\mathsf{Id}^{\&}}
\newcommand{\Idm}{\mathsf{Id}^{\otimes}}
\newcommand{\Sa}{\Sigma^{\&}}
\newcommand{\Sm}{\Sigma^{\otimes}}
\newcommand{\J}{\mathsf{j}}
\newcommand{\reflind}{\resizebox{0.85\width}{!}{$\mathtt{refl}$}}
\newcommand{\Ref}{\mathsf{Ref}}
\newcommand{\deref}{\mathsf{deref}}
\newcommand{\assign}{\mathsf{assign}}
\newcommand{\New}{\mathsf{new}}
\newcommand{\done}{\mathsf{done}}
\newcommand{\Write}{\mathsf{write}}
\newcommand{\Read}{\mathsf{read}}
\newcommand{\Or}{\mathsf{or}}
\newcommand{\cell}{\mathsf{cell}}
\newcommand{\oracle}{\mathsf{oracle}}
\newarrow{Partial}----{harpoon}

\theoremstyle{plain}
\newtheorem{theorem}{Theorem}[section]
\newtheorem{definition}[theorem]{Definition}

\newtheorem{remark}[theorem]{Remark}
\newtheorem*{claim*}{Claim}
\newtheorem{corollary}[theorem]{Corollary}

\newtheorem{lemma}[theorem]{Lemma}

\begin{document}

\setlength{\textbaselineskip}{22pt plus2pt}

\setlength{\frontmatterbaselineskip}{17pt plus1pt minus1pt}

\setlength{\baselineskip}{\textbaselineskip}


\setcounter{secnumdepth}{3}
\setcounter{tocdepth}{2}


\begin{romanpages}

\maketitle

\begin{dedication}
To my parents,\\
\textit{Hilde and L\'aszl\'o},\\
\hspace{10pt}\;\\
and grandparents,\\
\textit{G\'e, Dick, Imi and L\'aszl\'o},\\
\hspace{10pt}\;\\
for their extraordinary selflessness.\\
\begin{savequote}[8cm]
Time's fun when you're having flies.
  \qauthor{--- Kermit the Frog}
\end{savequote}
\end{dedication}

\begin{acknowledgements}
 	\subsection*{Personal}
This thesis is the result of many years of academic and social support from far more people than I could list here. 

I would like to thank Samson Abramsky for giving me the fantastic and unique opportunity of almost limitless freedom to pursue my academic interests under his guidance and encouragement for the past years and for sharing his experience and wisdom on how to navigate academic life whenever I needed advice. In addition to Samson, I've had the joy of having Radha Jagadeesan as a collaborator and, effectively but unofficially, cosupervisor. His optimism and excitement about our project as well as his kindness as a person were a huge source of support for me. I would like to point out that I (acting as primary author) was lucky to produce the work on coherence space and game semantics for dependent types included in this thesis in collaboration with both Radha and Samson.

I am very thankful to Nick Benton for taking me on as an intern at Microsoft Research Cambridge, despite my theoretical background, and for patiently and in a fun way teaching me so much about practical computer science and software engineering. My thanks go out also to my examiners, Aleks Kissinger, Kobi Kremnitzer, {Guy McCusker}, Luke Ong and Sam Staton\mccorrect{,} as well as the numerous anonymous conference and journal referees who read my work and provided impressively precise and useful feedback. I learned a lot from visiting Bath, Birmingham, Bristol and Paris, which I owe to Fanny He, Neel Krishnaswami, James Ladyman and Alexis Saurin. I would further like to thank Hongseok Yang and Paul Levy for the discussions on programming language theory and Urs Schreiber for explaining to me his thoughts on linear dependent type theory. Further, it's been lots of fun and a great learning experience to get to collaborate (and lift weights!) with my good friend Neil Dhir.

More broadly, I am thankful to the Departments of Computer Science -- particularly the Quantum Group --, of Statistics and of Engineering at the University of Oxford for providing such a fascinating and friendly academic environment. Particularly, my experience in Oxford would have been much less enjoyable and educational if it hadn't been for the many discussions of logic over coffee with Alex Kavvos, Kohei Kishida and Norihiro Yamada. I am grateful to Destiny Chen and Julie Sheppard for their spectacular administrative support and \mccorrect{for} making sure I never failed to leave their offices with a smile. I am highly indebted to my friendly colleagues at MSR Cambridge, like Jonathan Balkind, Tony Hoare and Claudio Russo, who made my time there a real treat.

Before my time in Oxford, I was very fortunate to be inspired and supported in my ambition to pursue a doctorate, by many of the excellent professors I've been lucky enough to have been taught by. In particular, I am very thankful to Heinz Han{\ss}mann, Jan Hogendijk, Peter Johnstone, Corry Samson and Paul Ziche for being such marvellous academic r\^ole models. Similarly, I am highly indebted to my friends Dejan Gajic and Joost Nuiten, whose intelligence and work ethic I always hugely admired during our time together in Amersfoort and Utrecht.

It's been an especially incredible gift to get to spend this period of academic and personal development in an inspiring environment like Oxford. I am very thankful for the stunning physical environment, all the interesting academic events, but most of all for the totally extraordinary people I've had the privilege to meet there. In particular, I feel I've been very blessed to have wonderful friends here like Carlos, Christoph, Claudia, Jerome, Karine, Marieke, Molly, Paul and Santhy, as well as the support of the Oxford Thich Nhat Hanh sangha, the Clarendon Scholars community and my friendly housemates at 20 Tyndale Road. During the end of my time in Oxford, it was so meaningful and uplifting to write up together with Jenna and to be inspired by her optimism and empathy.

Finally, I cannot imagine what my life would have looked like without my caring family and my long-time friends Bart, Carien, Ewout, Hambo, Julius, Meike, Pieter, Temple, Victor and my Descartes College chums. I am very grateful to Gina for her love during many of the past years. Thanks to all of you for putting up with me! I can only hope to have given you as much joy and support as you have given me over the years.

\subsection*{Institutional}
I am enormously grateful to the EPSRC, the Clarendon Fund and the Department of Computer Science at the University of Oxford for funding this endeavour. Many diagrams in this thesis were produced using Paul Taylor's commutative diagrams package.

\begin{savequote}[8cm]
There are only two kinds of programming languages: those people always bitch about and those nobody uses.
  \qauthor{--- Bjarne Stroustrup}
\end{savequote}
\end{acknowledgements}

\begin{abstract}
	Real world programming languages crucially depend on the availability of computational effects to achieve programming convenience and expressive power as well as program efficiency. Logical frameworks rely on predicates, or dependent types, to express detailed logical properties about entities. According to the Curry-Howard correspondence, programming languages and logical frameworks should be very closely related. However, a language that has both good support for real programming and serious proving is still missing from the programming languages zoo. We believe this is due to a fundamental lack of understanding of how dependent types should interact with computational effects. In this thesis, we make a contribution towards such an understanding, with a focus on semantic methods.

Our first line of work concerns a dependently typed version of linear logic (which can be seen as a calculus for commutative effects). We develop a dependently typed dual intuitionistic linear logic as well as a sound and complete categorical semantics using certain indexed monoidal categories satisfying a comprehension axiom. We present a range of models, based on monoidal families, commutative effects, a double gluing construction, domains and strict functions and coherence spaces.

Our second line of work  develops a game semantics for dependent type theory, which had so far been missing altogether. We show that, if we work with deterministic well-bracketed history-free winning strategies, the semantics satisfies a full and faithful completeness result with respect to call-by-name dependent type theory for a hierarchy of types built from certain finite inductive families. We show that by relaxing the notion of strategy, we can further model various effects rather than the pure type theory.

Our final line of work explores a generalisation of Levy's call-by-push-value (CBPV) to encompass dependent types. We show that the syntax of CBPV naturally extends to a calculus we call dCBPV- in which types are allowed to depend on values but not computations. We show it has an elegant categorical semantics and a well-behaved operational semantics and that it admits a wide range of models arising from indexed monads on models of pure dependent type theory and from models of linear dependent type theory.  By contrast with the simply typed situation, however, it does not suffice to encode call-by-value and call-by-name versions of dependent type theories with unrestricted effects. To obtain those, we need a richer calculus dCBPV+ with a Kleisli extension principle for dependent functions, which turns out to be less well-behaved from a semantic point of view.

\begin{savequote}[8cm]
\textlatin{(...) Livet maa forstaaes baglaends. Men (...) maa leves forlaends.}

Life can only be understood backwards; but it must be lived forwards.
  \qauthor{--- Søren Kierkegaard}
\end{savequote}
\end{abstract}

\dominitoc 

\flushbottom

\tableofcontents




\end{romanpages}

\flushbottom
\begin{savequote}[8cm]
Logic, like whiskey, loses its beneficial effect when taken in too large quantities.
	\qauthor{--- Lord Dunsany}
\end{savequote}

\chapter{\label{ch:1}Introduction} 

\section{Motivation}
\subsection{The Limits of Logic and the Conception of Computers}
\mccorrect{Logic} and computer science have been intimately related since the latter's early days \cite{kleene1981origins,gandy1995confluence,soare1999history}. Indeed, the precise modern concept of computability\footnote{Following centuries of more informal descriptions of special cases of algorithms and computing machines, dating back at least to Euclid.} was rapidly formalised in the early 1930s by a group of logicians, motivated, at least in part, by questions in foundations of mathematics like Hilbert's Entscheidungsproblem. Particularly notable is that a wide range of formalisations of the concept of computability were proposed in short succession, many of which were proven to be equivalent in the so-called Church-Turing thesis. This ``confluence of notions'' of computation included but was by no means limited to
\begin{itemize}
\item Herbrand-G\"odel computable functions (or general \emph{recursive} functions), a scheme for axiomatising effectively computable functions, introduced in the aftermath of G\"odel's study of his incompleteness theorems, which demonstrated the limits of axiomatic systems to formalise mathematics;
\item Church's $\lambda$-calculus, a formal language for defining functions that can now be seen as a failed attempt at providing a foundation of mathematics: it turned out to be inconsistent as a \emph{logic}; in hindsight, one could argue that this was one of the first real \emph{programming languages}, however, for writing \emph{algorithms} or \emph{programs} rather than \emph{proofs}; modern functional languages still closely resemble it;
\item Kleene's $\mu$-recursive functions, which clearly show how the expressive power of general computable functions can be obtained from the weaker previously studied scheme of primitive recursive functions: by adding a minimisation operator, closely related to the fixpoint combinators definable in \mccorrect{the (untyped)} $\lambda$-calculus;
\item Turing machines, giving a universal notion of \emph{hardware} on which computation can be performed.
\end{itemize}
Of course, it would still take more than a decade of clever engineering to transform this theoretical groundwork into a working practical computer. For excellent accounts of this fascinating history, we refer the reader to  \cite{kleene1981origins,gandy1995confluence,soare1999history}.

\subsection{Terms and Types} 
To restore the logical consistency of his system, Church introduced devices called \emph{types}\footnote{The inconsistency is caused by self-application which allows us to construct Russell's paradox in the untyped $\lambda$-calculus. Note that types had been previously introduced by Russell already to circumvent the same paradox in Cantor's naive set theory.}, to classify the \emph{terms} of his $\lambda$-calculus (the \emph{programs}, if we view the calculus as a programming language). From a modern point of view, we can think of types as providing guarantees about a term (algorithm), for instance by putting certain (\emph{extensional}) restrictions on the \emph{inputs} it takes and the \emph{outputs} it produces or (\emph{intensional}) restrictions on \emph{the manner} in which the outputs are computed from the inputs.

Types were originally introduced by Church for foundational  reasons to restore the consistency of the $\lambda$-calculus as a logic. We must remember that the untyped $\lambda$-calculus was already a fine (albeit primitive) programming language! It is perhaps surprising, therefore, that types have turned out to be of huge practical value in software development, the main reason being that simple type annotations happen to catch many of the most common bugs introduced by programmers before the program is run. Moreover, types provide a useful abstraction of programs that helps programmers think about their code and certainly make it much easier to read code written by others. It is not a coincidence that the top four programming languages in the TIOBE Index of popular programming languages (Java, C, C++ and C\#) all have a strongly enforced type system \cite{tiobe2016}.

The motivations for types outlined in the previous paragraph are rather pragmatic in nature. A more principled motivation for types comes from the \emph{Curry-Howard correspondence}, which suggests that some \emph{type theories}, the simply typed $\lambda$-calculus being the prime example, can be interpreted both as a programming language and as an (intuitionistic\footnote{In the sense that the principle of double negation elimination does not hold: not not $A$ does not imply $A$. Note that such a formalism is strictly more expressive than classical logic as the latter is precisely the fragment of the former consisting of the doubly negated propositions.}) formal logic.

Informally, a type theory is a calculus for constructing terms (the \emph{programs} or \emph{proofs}) in a compositional way, starting from certain basic building blocks, subject to the restrictions  on their inputs and outputs (or assumptions and conclusions) imposed by the types (or \emph{propositions}). It can also be used to reason about the equality and conversion (execution/evaluation behaviour or proof normalization) of these terms. This dual reading of a type theory as a programming language and a logic is very roughly summarised in figure \ref{fig:curryhoward}.

\begin{figure}[!tb]\resizebox{\linewidth}{!}{
\begin{tabular}{l||l|l}
\textbf{Type theory} & \textbf{Programming} & \textbf{Intuitionistic Logic}\\
\hline
Type $A$ & (Data) Type $A$ & Proposition $A$\\
Term $b:B$ & Program $b$ with output of type $B$& Proof $b$ with conclusion $B$\\
Typing context $x_1:A_1,\ldots, x_n:A_n$ & Inputs of type $A_1,\ldots, A_n$ & Assumptions $A_1,\ldots, A_n$\\
Conversion $b\leadsto b':B$ & Execution $b\leadsto b':B$ & Proof normalization $b\leadsto b':B$\\
Product type $A\times B$ & Type of pairs of type $A$ and $B$ & Conjunction $A\wedge B$\\
Sum type $A+B$ & Disjoint union of types $A$ and $B$ & Disjunction $A\vee B$ \\
Function type $A\Rightarrow B$&  Type of (first class) functions from $A$ to $B$ & Implication $A\Rightarrow B$\\
Singleton type $1$ & \texttt{void} (Type of returning commands) & True\\
Empty type $0$ & \texttt{error} (Type of non-returning commands) & False \\
Parametric polymorphism $\Pi_{A}$ & Generics & $2^{nd}$-order quantification $\forall_A$
\end{tabular}}
\caption{\label{fig:curryhoward} An informal sketch of some instances of the Curry-Howard correspondence.}
\end{figure}
In particular, simple type theory (or the simply-typed $\lambda$-calculus) can not only be viewed as a primitive programming language, but also as a formalism for writing (natural deduction style) proofs for intuitionistic (implicational\footnote{Of course, we can add product and sum types to the simple $\lambda$-calculus to get a correspondence with full intuitionistic propositional logic.}) propositional logic.

\subsection{Programming Requires More Terms: Effects}
The simply typed $\lambda$-calculus is a rather unexpressive programming language, even when enriched with\mccorrect{ }ground types for booleans and natural numbers (the so-called G\"odel system T). Indeed, it can only define primitive recursive functionals (a generalization of the class of primitive recursive functions to a system with higher types), in particular functions that always terminate, and we do not have the power of general recursion available: it is not Turing complete. In fact, just as important in practice as mere \emph{expressive power}\footnote{For instance, it is well-known that such a language of total functions cannot define its own interpreter. \cite{mcbridehaskelllist}} is the \emph{practical convenience} that general recursion schemes provide for programmers: some primitive recursive functionals can be defined more conveniently using general recursion, as we do not have the burden of proving termination\footnote{For instance, writing a program that computes the same function as the simplex algorithm for linear programming would be possible in a language without general recursion by looping over the (finite) number of vertices in the \mccorrect{polytope}. However, such a solution would be more effort to implement and less efficient than the usual general recursive solution using a while loop as it would involve at the very least computing from the problem specification the number of vertices in the simplex.}.

The reason the untyped $\lambda$-calculus was inconsistent as a logic turned out to be that the absence of types made so-called fixpoint combinators definable. We now know (as was already foreshadowed by Kleene's $\mu$-recursive functions) that these are the crucial ingredient on top of primitive recursion in defining general computable functions.

One can explicitly add such fixpoint combinators to the syntax of a simply typed $\lambda$-calculus to have the expressive power of general recursion in a typed setting. However, the resulting language, known as {\PCF}  if we start from system T, is again inconsistent as a logic, as programs involving a fixpoint combinator do not correspond to acceptable proofs. Such programs which do not correspond to logical proofs are often called \emph{effectful} and they are of crucial importance in real world software development. By contrast, programs corresponding to proofs are called \emph{(purely) functional}. As purely functional code tends to be less error-prone and easier to reason about, a lot can be said in its favour.

However, in software engineering practice, pure functionality is often too much of a restriction, for reasons of efficiency, expressivity or mere programmer convenience. We have already seen the example of general recursion, which is an important feature in real programming languages. Another example of a class of effects are those that are introduced to give the programmer the option of more explicit low-level manipulation of \mccorrect{the} way the program is executed on the available hardware. A prime concrete example is the explicit manipulation of memory (or \emph{state}). This can lead to a reduced (time and space) resource consumption. It can also make a certain algorithm easier to understand and implement for the programmer\footnote{
An example is given by matrix multiplication. Of course, one can give a purely functional implementation, representing matrices as lists of lists on which we recurse, but it would be complicated and inefficient compared to the obvious imperative definition.
}. In other cases, effectful behaviour is an essential part of the specification of a program: for instance, we may want a program to generate a random number, to process keyboard input provided by the user, to generate output to a display or we may want to write a program that never terminates, like an operating system or a server, or to implement a counter.

We conclude that type theories require various extra terms, called (computational) effects, in order to constitute a practical programming language. Since these extra terms do not correspond to logical proofs, this renders the type theory of a practical programming language inconsistent as a logic.

\subsection{Logic Requires More Types: Dependent Types}
At the same time, we may observe that while many real world programming languages implementing such \emph{effectful type theories} have a type system implementing the equivalent of the logical connectives of propositional logic (so-called simple types) and even the equivalents of some higher-order predicates and quantifiers (so-called parametric polymorphism or generics), the equivalents of first-order predicates and quantifiers (so-called dependent types) are missing. However, these first-order quantifiers are of crucial importance in logical frameworks that are sufficiently expressive to be useful to formalise mathematics. From a programming point of view, such dependent types allow us to assign more precise types to existing programs of the simple $\lambda$-calculus, expressing detailed and useful program properties which a program that type checks is guaranteed to satisfy.

This means we are faced with a choice, at the moment: either we choose a language with many programs (an effectful programming language) while accepting a type system missing dependent types or we choose to have many types (a dependently typed language) and accept the lack of effects. All practical programming languages are in the former camp (e.g. Java, C++, Python, OCaml, Go) while the languages in this camp are inconsistent as a logic. The languages in the latter camp (e.g. Coq, Agda) can be useful as a logic or proof-assistant, but the lack of effects usually renders them impractical as programming languages.

\subsection{Effects as Proofs of Modal Propositions}\label{sec:effmodalformulae}
A first important issue to address if one wants to close the gap between programming languages and logics is the logical inconsistency introduced by effects. Effects need to be excluded from proofs, in order to retain their logical consistency (otherwise, using for instance general recursion, we could trivially construct a proof of any proposition), but not from programs. One possible way of solving this issue is to introduce new types of which the possibly effectful programs will be inhabitants, while keeping the inhabitants of other types pure. In such a formalism, not all types are propositions, just the ones whose terms are pure computations, not involving effects. In addition to restoring the logical consistency of the type theory, such a typing discipline makes it easier to reason about programs, as it is immediately clear from the type system which effects may occur in terms.

A particularly pleasing such formalism is given by \emph{(strong) monads}, which can be used to encapsulate effects  \cite{moggi1988computational,wadler1995monads}. The idea is that code is by default pure, unless specified otherwise by the type system. For example, a program of type $\N\Rightarrow \N$ is a primitive recursive functional from natural numbers to natural numbers, but a program of type $\N\Rightarrow T_{\mathsf{rec}}\N$ may be a general recursive function, where $T_{\mathsf{rec}}$ is a (strong) monad that makes fixpoint combinators available.  

Particularly pleasing is that such (strong) monads $T$ can in fact be given a logical interpretation. Under the Curry-Howard correspondence, they correspond to certain \emph{diamond modalities} $\Diamond$ on the level of logic (sometimes called lax modalities and written $\bigcirc$) \cite{benton1998computational}. This means that we can interpret the effectful programs of type $A\Rightarrow T B$ as all the extra proofs of $A\Rightarrow \Diamond B$ that do not arise from proofs of $A\Rightarrow B$, if you will all the derivations starting from $A$ of ``possibly $B$'' that aren't also proofs of $B$.

A point that is often elaborated on is that many such modalities may already be definable in a pure type theory. For instance, global state can be emulated with a modality $S\Rightarrow ((-)\times S)$, errors with a modality $(-)+E$, non-determinism with a powerset modality $\mathcal{P}(-)$ (if our pure type theory is a higher-order logic) and printing with a modality $(-)\times M$ where $M$ is some internal monoid in the type theory. Another concrete example would be the double negation modality $\neg \neg (-)$ or, more generally, continuation modalities $((-)\Rightarrow R)\Rightarrow R$, which make the classical principle of Peirce's law (equivalent to double negation elimination) available from a logical point of view and the (universal) control operator $\mathtt{call/cc}$ from the point of view of programming languages \cite{griffin1989formulae}. In this sense, (constructive\footnote{Constructive in the sense that double negation elimination is not an isomorphism of types. This means that we avoid equating all terms of the same type, which would otherwise happen according to Joyal's lemma \cite{lambek1988introduction}.}) classical propositional logic is at the same time a simply typed programming language as well: not a purely functional one, but one enriched with constructs for non-local control flow. Another interesting modality  $\mathsf{Dist}(-)$ is that of probability distributions, which, when definable in a pure type theory, lets us emulate probabilistic choice. It is interesting to note that one would define $\mathsf{Dist}(X)$\mccorrect{, for a discrete type $X$,} as the type $\Sigma_{f:X\Rightarrow\R^+}\Id_{\R^+}(\int f,1)$ of pairs of a positive real valued function $f:X\Rightarrow \R^+$ and a proof $p:\Id_{\R^+}(\int f,1)$ that $f$ sums (or integrates) to $1$. We see that in order to define such modalities, we need, in particular, dependent type formers $\Sigma$ and $\Id$ corresponding to existential quantification and identity predicates. 

Such definable modalities allows us to \emph{emulate} certain computational effects in a pure language.  We would like to stress, however, that to treat effects \emph{natively} with their intended custom operational semantics, rather than the emulation inherited from the conversions of the pure type theory, modalities should explicitly \mccorrect{be} added to the type system as new type formers and effects as new term formers which inhabit these types.

\subsection{CBV, CBN and Half-Modalities}
Recall that in pure functional languages the choice of an \emph{evaluation strategy} does not effect the result of computations, merely their efficiency. This is why we call these languages declarative: they specify what should be computed, not how it should be computed. The user does not need to know about the how; this is left to the discretion of the compiler.

By contrast, the same is not generally true for languages with effects. Effects tend to bring us into the realm of imperative languages: the evaluation strategy (the order in which we evaluate the various parts of the program) can have a significant impact on the result of computations, so we need to think about language constructs not only in terms of what they compute but also in terms of how they compute it in time. Here, it is important for the user to know which strategy is being used.

Two strategies are particularly studied from a theoretical point of view: \emph{call-by-value (CBV)} and \emph{call-by-name (CBN)} evaluation. An important distinction between the two is that in CBN function arguments are only evaluated when they are needed, while in CBV they are always evaluated eagerly whether they are needed or not. CBN evaluation can sometimes be preferable from a performance or correctness point of view. On the other hand, in the presence of some effects, it can be difficult to reason about, which is why  CBV evaluation is often preferred as the default option in software engineering practice, with CBN being reserved for special situations.

An idea that we believe to be underemphasized in literature is the perspective that it is instructive to further decompose a monad $T$ into an \emph{adjunction} $F\dashv U$ (or the corresponding modality $\Diamond$ into a pair of ``\emph{half-modalities}'') between two type theories of \emph{values} on the one hand and \emph{computations} (and more generally \emph{stacks}) on the other, whose types we shall write $A,A',\ldots$ and $\ct{B},\ct{B}',\ldots$ respectively. This is the point of view taken by Levy's \emph{call-by-push-value (CBPV)}. The advantage is that we obtain an elegant language (simpler than a monadic language, in many ways) for proofs and effectful computations with a single intuitive canonical evaluation strategy that is expressive enough to encode both CBV and CBN and many things in between.

In particular, if we define the {monad} $T:=UF$ and \emph{comonad}\footnote{This corresponds to a certain \emph{box modality} if we try to give the type theory for computations a logical interpretation. As we shall see, it can be understood to define a certain generalization of linear logic, hence the notation $!$ for the comonad.} $!:=FU$, we recover (thunks of) CBV computations as the terms of type $x:A\vdash a : TA'$ and the CBN computations as those of type $x:!\ct{B} \vdash b : \ct{B}'$. We see that the type system now also provides guarantees about the evaluation strategy of programs. Meanwhile, proofs can still be interpreted as general terms $x: A\vdash a: A'$ (including the proofs of modal propositions which can also be read as thunked call-by-value computations).

\subsection{The Relationship Between Proving and Programming}
It is clear that a blunt statement of the Curry-Howard correspondence like ``a programming language is the same as a logic'' is far from the truth. In fact, even the weaker statement, which may be closer to the truth, that ``every logic extends to a programming language'' does not accurately reflect the reality of programming languages research at the moment, although it may be a possible future that the field is trying to realise.

The relationship between programming languages and logics may be more accurately summarised by figure \ref{fig:logicsandpls}, where we refer to pure languages without dependent types as \emph{pidgins} (e.g. the pure polymorphic $\lambda$-calculus) as they can be seen as simplified languages that can be interpreted both as a programming language and as a logic but are not entirely satisfactory as either. So far, however, it is not yet clear if there exists a Promised Land of genuine \emph{programming logics}, languages that can serve as both a useful programming language and logic by combining dependent types and effects.

What is clear is that there is an inherent tension between the extensions of a pidgin with effects and with dependent types. The extra terms introduced by the former allow for wilder kinds of program behaviour while the extra types introduced by the latter serve to tame the behaviour of a program.

\begin{figure}[!tb]
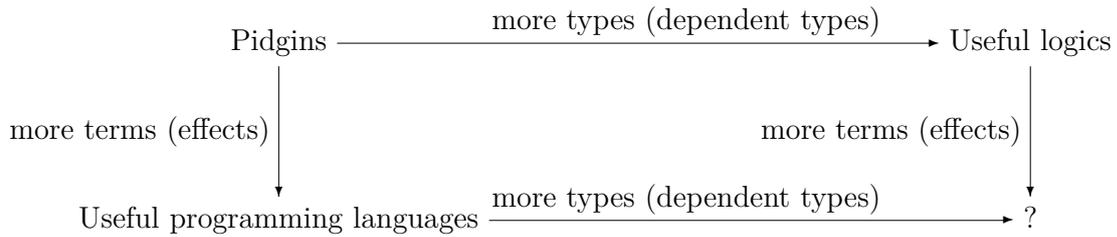
\resizebox{\linewidth}{!}{
\begin{diagram}
\textnormal{Pidgins} & \rTo^{\textnormal{more types (dependent types)}} & \textnormal{Useful logics}\\
\dTo^{\textnormal{more terms (effects)}} & & \dTo^{\textnormal{more terms (effects)}} \\
\textnormal{Useful programming languages} & \rTo^{\textnormal{more types (dependent types)}} & ?
\end{diagram}
}
\caption{\label{fig:logicsandpls} The present relationship between type theories that can serve as a logic and as a programming language: it is not clear what sort of type theory would be satisfactory as both.}
\end{figure}

\subsection{Why Unify Proving and Programming?}
One might wonder why we should be looking for such a Promised Land at all. In fact, who promised such a land in the first place?

Firstly, it is of fundamental importance to both the disciplines of mathematical logic and programming language research to be very precise about the relationship between mathematical proofs and computer programs. The promise of a unification of proving and programming has been repeatedly made either implicitly or explicitly, given how intertwined both disciplines have been historically, how much cross fertilisation has taken place and how many parallels have been sketched (often under the name of a Curry-Howard correspondence). It is stunning that no precise result exists yet to either show how to unify the notions of logic (using dependent types) and practical programming (using effects) in a single language or to show that this cannot be done satisfactorily. Such a marriage or the demonstration of its impossibility would provide conceptual clarity about the foundations about two important disciplines that have been flirting with each other for eighty years.

Secondly, a combined system with dependent types and effects could provide a very useful practical framework for writing verified software. It may give us a single language to both write real world programs (making use of effects) and prove their correctness (using the expressive logic embedded in the type system, using dependent types). Currently, this often needs to happen in two separate systems: a programming language and a proof assistant or another verification tool. We are required to build a model of the program in the verification tool in order to prove its properties. This transcription results in an overhead of work as well as in an extra source of potential bugs\footnote{In fact, it turns out that the construction of such a model can in some cases be automated using game semantics \cite{abramsky2004applying}.}. It is both safer and more efficient to directly prove properties of the production code.

\section{Goals of This Thesis}
The distal goal that this thesis can be understood to be pursuing is an understanding of the precise relationship between logic and programming. The main motivating questions for this line of work are the following three.
\begin{itemize}
\item How should we understand the relationship between logic and programming?
\item Can we design languages that are simultaneously satisfactory as a programming language and as a logic and in which both aspects of the language interact in a meaningful way?
\item  Can we use such a language for certified real world programming?
\end{itemize} 

The desire to answer these questions leads us to the more proximal goal of understanding how dependent types (from the realm of logic) can be combined with computational effects (which define real world programming languages):
\begin{itemize}
\item Can we combine dependent types and computational effects in an elegant and meaningful way?
\end{itemize} 

We believe the goals and questions we are pursuing are of tremendous importance both from a fundamental academic point of view and from the concrete point of view of software engineering. If these hugely ambitious questions had a straightforward answer, the community would have found it a long time ago. This thesis only claims to make a small contribution to solving this difficult puzzle, while hoping to illustrate both its relevance and complexity.

Concretely, this thesis describes three closely related lines of work:
\begin{enumerate}
\item Studying a dependently typed version of \emph{linear logic}, in the sense of a dependent type theory in which terms cannot be copied or discarded freely;
\item Providing a \emph{game semantics} for dependent type theory, interpreting types as games and their terms as strategies on these games;
\item Studying a dependently typed version of Levy's \emph{call-by-push-value} in the presence of various effects.
\end{enumerate}
These are closely related to the goals of thesis. Indeed, firstly, CBPV is a very useful paradigm for understanding effectful languages and their relationship to logic as it gives us a fine-grained way of controlling where effects are allowed to occur (and in what order they should be evaluated) and what parts of a program should be pure.

Secondly, effectful computations and the stacks used to evaluate these behave linearly in some sense. To be precise, they cannot be discarded in the syntax of CBPV\footnote{This corresponds to the structural rule of weakening not being valid. We generally, for non-commutative effects, only consider contexts of at most one identifier of computation/stack type $\ct{B}$, meaning that the rule of contraction does not have any meaning. We shall later see, in theorems \ref{thm:lliscomm} and \ref{thm:commembedsinll}, that we can conservatively extend the syntax for stacks with a multiplicative conjunction $\otimes$, or, equivalently, with linear contexts of longer length, if we are dealing with only commutative effects.}. On a more conceptual level, we like to point out that effectful computations can be \emph{dynamic}, in the sense that their reductions generally break equality, for instance for a non-deterministic choice we can have a reduction $\nondet{}{\return\ttt,\return\fff}\leadsto \return \fff$ where the result, after the choice has been made, should clearly not be considered 'equal' to the initial computation before the program makes a non-deterministic choice. This should be contrasted with the \emph{static} nature of values or pure computations (whose normalization does not break equations). Dynamic objects, in particular, cannot be copied in the usual sense, as both copies might later cease to be equal, and are in that sense linear. In fact, we shall argue that linear logic can be seen as a type system for commutative effectful computations.

Thirdly, game semantics has been perhaps the most successful paradigm for providing a unified intuitive semantics for many effectful programming languages and pure logics. We can hope to gain useful semantic intuitions for the problem of how to relate effects to dependent types, here. Moreover, game semantics naturally arises from a model of linear logic. Recall that game semantics is naturally effectful in the sense that computational effects like state, non-termination and non-local control have to be explicitly excluded by putting conditions on the strategies we consider on games. Therefore, we believe, the current absence of a game semantics for dependent types reflects the same lack of fundamental understanding of how to relate logic to programming, particularly the question of how type dependency should interact with effectful computations. 

We encounter similar possibilities and obstacles in all three lines of work -- for instance, we need to decide if it makes sense to have types depend on dynamic or linear objects -- and the simultaneous study of these three topics has hugely helped us to see a bigger picture emerge of what the Promised Land of genuine programming logics might look like. We hope it will do the same for the reader.

\section{Key Contributions}
This thesis makes the following key contributions:
\begin{itemize}
\item Explaining the difficulty of combining dependent types with linear types, game semantics and effects  and presenting a way of still doing so in the following sense;
\item Developing a syntax for dependently typed dual intuitionistic linear logic (dDILL);
\item Developing a categorical semantics for dDILL and the dependently typed linear/non-linear (dLNL) calculus;
\item Developing a range of concrete models for dDILL and dLNL, including a coherence space semantics;
\item Explaining the relationship with commutative effects;
\item Presenting a game semantics for dependent type theory;
\item Showing it has strong (full and faithful) completeness properties with respect to CBN dependent type theory;
\item Examining effectful game models of dependent types;
\item Presenting a dependently typed call-by-push-value (dCBPV-) calculus;
\item Developing its categorical semantics;
\item Showing that dLNL models give models of dCBPV-, as do algebras for indexed monads on models of \mccorrect{pure} dependent type theory;
\item Showing that the operational semantics of dCBPV- is well-behaved;
\item Showing that we need to extend dCBPV- to dCBPV+ with dependent Kleisli extensions if we want CBV and CBN translations;
\item Showing that dCBPV+ is less straightforward than dCBPV- from the point of view of operational semantics and concrete categorical models...;
\item ... and that the same goes for dCBPV- extended with dependent projection products (additive $\Sigma$-types);
\item As an alternative, presenting a dependently typed enriched effect calculus (dEEC) and showing it to be very well-behaved.
\end{itemize}
In the course of his doctoral studies, the author has communicated the majority of the material included in this thesis in \cite{
vakar2014syntax,
vakar2015syntax,
abramsky2015games,
vakar2016gamsem,
vakar2015framework,
vakar2016effectful} and in various oral presentations.

\section{Thesis Outline}
We have chosen to present this work in the order in which the research was conducted, to best convey to the reader the author's motivations for studying the various topics. Chapter \ref{ch:2} provides background material on CBPV, linear logic and game semantics, all in simply typed form, as well as on\mccorrect{ }(cartesian) dependent type theory. Most material in these sections is not original, but we present it in a novel way, in order to ensure a smooth transition to the rest of this thesis. Our first pillar of original work is presented in chapter \ref{ch:4}, which discusses a dependently typed version of linear logic. This naturally leads us to chapter \ref{ch:5}, where we discuss our second line of work: a game semantics for dependent type theory. Our third and last topic, a discussion of dependently typed CBPV, can be found in chapter \ref{ch:3}. We end on a discussion of our conclusions and future work in chapter \ref{ch:6}.

We have tried to keep chapters \ref{ch:4}, \ref{ch:5} and \ref{ch:3} as self-contained as possible (apart from their dependence on the appropriate sections of chapter \ref{ch:2}). Historically, our three lines of work roughly relate to each other as follows. Following a question by Samson Abramsky, we set off to construct a game semantics for dependent type theory or to understand why none existed yet. As categories of games originate from models of linear logic (categories of cofree $!$-coalgebras), this pushed us to investigate the relationship between linearity and dependent types first. Later, we came to understand the tension between game semantics and dependent types as arising from the natural effectful character of unrestricted strategies. This understanding made clear to us our bias in studying type dependency only in CBN game semantics and generally focussing on Girard's first (CBN) translation into linear logic. This finally led to our study of dependently typed CBPV, which we now understand, after realising that linear logic can be read as a calculus for commutative effects, as giving a generalization of our work on linear dependent type theory to non-commutative effects, providing, additionally, an account of operational semantics.
\begin{savequote}[8cm]
Knowledge is knowing that a tomato is a fruit; wisdom is not putting it in a fruit salad.
	\qauthor{--- Miles Kington} 
\end{savequote}

\chapter{\label{ch:2}Preliminaries} 
In this chapter, we present our views on simply typed linear logic, game semantics and call-by-push-value, as well as on their relation to each other, in order to easily be able to extend all three with a notion of type dependency in later chapters. We start with a discussion of\mccorrect{ }cartesian type theory, however. \mccorrect{The material in this chapter mostly consists of definitions and results published by other authors as well as folklore results. In most cases, however, it is reformulated in a non-trivial way in order to make developments in further chapters go through as smoothly as possible. We hope that this novel presentation of known results can be of value in its own right.}

\section{Cartesian Type Theory}
\mccorrect{We} briefly recall the syntax and semantics of\mccorrect{ }simple (\STT) and dependent type theory (\DTT). We describe a general syntactic and semantic framework for both, in the context of which we can consider many theories (in the case of syntax) or models (in the case of semantics). We discuss, in particular, two CBN type theories {$\STTGame$} and {$\DTTGame$}, with respect to which the game theoretic models we discuss in section \ref{sec:backgame} and chapter \ref{ch:5} have full and faithful completeness properties\footnote{These are CBN type theories in the sense that we only demand a limited $\eta$-rule with some (but not all) commutative conversions for ground types. As we shall see in section \ref{sec:backcbpv}, the full $\eta$-law is typically broken in effectful settings under CBN evaluation. If we were to demand the fully general $\eta$-rule (which would automatically imply all commutative conversions), we would be modelling pure type theory. We briefly note that  the usual set theoretic semantics is fully and faithfully complete for this pure type theory.}.

\subsection{Syntax of Type Theories}\label{sec:DTT}
\subsubsection{Dependently Typed Equational Logic}
\mccorrect{In} this section, we briefly recall the framework of dependently typed equational logic (sometimes called generalised algebraic theories \cite{cartmell1986generalised}), which will serve as the structural core type theory and  on top of which we  later consider two theories in particular: a flavour of simple type theory ({$\STTGame$}) and a flavour of dependent type theory ({$\DTTGame$}). This framework puts both flavours of type theory on an equal footing and allows us to better study their relationship. We go into this level of precision in our specification of the syntax we are modelling, in order to accurately state the appropriate completeness results in sections \ref{sec:backgame} and \ref{sec:compl}. Although much more informal, our treatment is close in spirit to those of \cite{pitts1995categorical} and \cite{hofmann1997syntax}, to which we refer the interested reader for more background and where the reader can find details on delicate topics like pre-syntax, $\alpha$-conversion, identifier binding and capture-avoiding substitution.

The key feature of a dependent type system is that we allow types to refer to free identifiers from the context. The reader may want to keep in mind the analogy that dependent types are to predicates what non-dependent types are to propositions. One consequence is that order in the context becomes important as all free identifiers in a type need to be declared in the context to its left. As types can depend on terms in a dependently typed system and equations of terms can lead to equations of types which can lead to new typing judgements, we define all judgements together in one big inductive definition.

\subsubsection*{Judgements}
Figure \ref{fig:dttjudgements} presents the various kinds of \emph{judgements} of dependently typed equational logic and their intended meaning.
\begin{figure}[!tb]
\centering
\fbox{\resizebox{\textwidth}{!}{\footnotesize
\begin{tabular}{ll}
\textbf{Judgement} & \textbf{Intended meaning}\vspace{2pt}\\
$\vdash \Gamma \;\mathsf{ctxt}$ & $\Gamma$ is a valid context\\
$\Gamma \vdash A\;\mathsf{type}$ &  $A$ is a type in  context $\Gamma$\\
$\Gamma\vdash a:A$ & $a$ is a term of type $A$ in context $\Gamma$\\
$\vdash \Gamma = \Gamma'$\hspace{60pt} & $\Gamma$ and $\Gamma'$ are judgementally equal contexts\\
$\Gamma\vdash A= A'$ & $A$ and $A'$ are judgementally equal types in context $\Gamma$\\
$\Gamma\vdash a= a':A$ & $a$ and $a'$ are judgementally equal terms of type $A$ in context $\Gamma$
\end{tabular}\hspace{60pt}\;}}
\normalsize
\caption{\label{fig:dttjudgements} Judgements of dependently typed equational logic.}
\end{figure}
Here, $\Gamma, \Gamma', A,, A', a$ and $a'$ are all symbolic expressions from a set $\mathsf{Expr}$, built from an alphabet $\mathsf{Sym}$, in which we have countably infinite designated subsets $\mathsf{Idf}$ of identifiers and $\mathsf{Cons}$ of constants. As usual, we distinguish between the free and bound identifiers occurring in an expression $\mathcal{J}$ and we consider expressions $\mathcal{J}$ up to $\alpha$-equivalence, or up to permutations of $\mathsf{Idf}$ fixing the free identifiers of $\mathcal{J}$. We denote the syntactic metaoperation of capture-avoiding substitution of an expression $a$ for all occurrences of a free identifier $x$ in an expression $\mathcal{J}$ by $\mathcal{J}[a/x]$.

\subsubsection*{Structural Rules and Theories}Dependently typed equational logic has the structural rules presented in figure \ref{fig:dttstructural}, which will be shared in particular by {$\STTGame$} and {$\DTTGame$}.
\begin{figure}[!tb]
\begin{subfigure}[!t]{\textwidth}
\centering
\fbox{\resizebox{\textwidth}{!}{
\begin{tabular}{lll}
\AxiomC{$\Gamma,\Gamma'\vdash\mathcal{J}$}
\AxiomC{$\vdash \Gamma,x:A,\Gamma'\;\mathsf{ctxt}$}
\RightLabel{\textsf{Weak}}
\BinaryInfC{$\Gamma,x:A,\Gamma'\vdash \mathcal{J}$}
\DisplayProof\hspace{50pt}\;
&
\AxiomC{${\Gamma},x:A,\Gamma' \vdash \mathcal{J}$}
\AxiomC{$\Gamma \vdash a:A$}
\RightLabel{\textsf{Subst}}
\BinaryInfC{${\Gamma},\Gamma'[{a}/x] \vdash \mathcal{J}[{a}/x]$}
\DisplayProof\hspace{50pt}\;
\end{tabular}\hspace{25pt}\;}}
\caption{\label{fig:dttweak} Weakening and substitution rules. Here, $\mathcal{J}$ represents a statement of the form $B\;\mathsf{type}$, $B= B'$, $b:B$, or $b= b':B$. Note that these rules, additionally, make contraction and exchange rules derivable.}
\end{subfigure}
\begin{subfigure}[!t]{\textwidth}
\centering
\fbox{\resizebox{\textwidth}{!}{
\begin{tabular}{lcr}
\AxiomC{}
\RightLabel{\textsf{C-Emp}}
\UnaryInfC{$\vdash\cdot\;\mathsf{ctxt}$}
\DisplayProof\hspace{-20pt}\;
& 
\AxiomC{$\vdash\Gamma\; \mathsf{ctxt}$}
\AxiomC{$\Gamma \vdash A\;\mathsf{type}$}
\AxiomC{$x$ is fresh for $\Gamma$ and $A$}
\RightLabel{\textsf{C-Ext}}
\TrinaryInfC{$\vdash \Gamma,x:A \;\mathsf{ctxt}$}
\DisplayProof\hspace{-20pt}\;

&
\AxiomC{$\vdash \Gamma,x:A,\Gamma'\;\mathsf{ctxt}$}
\RightLabel{\textsf{Idf}}
\UnaryInfC{$\Gamma,x:A,\Gamma'\vdash x:A$}
\DisplayProof\\
& &
\\
&
\AxiomC{$\Gamma=\Gamma'\;\mathsf{ctxt}$}
\AxiomC{$\Gamma \vdash A= B$}
\AxiomC{$\vdash\Gamma,x:A\;\mathsf{ctxt}$}
\AxiomC{$\vdash\Gamma',x:B\;\mathsf{ctxt}$}
\RightLabel{\textsf{C-Ext-Eq}}
\QuaternaryInfC{$\vdash \Gamma,x:A=\Gamma',x:B$}
\DisplayProof&
\end{tabular}
}}
\normalsize
\caption{\label{fig:dttctxt} Context formation and identifier declaration rules.}
\end{subfigure}
\begin{subfigure}[!t]{\textwidth}
\centering
\fbox{\resizebox{\textwidth}{!}{
\begin{tabular}{l}
\begin{tabular}{lll}
\AxiomC{$\vdash \Gamma\;\mathsf{ctxt}$}
\RightLabel{\textsf{C-Eq-R}}
\UnaryInfC{$\vdash \Gamma= \Gamma$}
\DisplayProof\hspace{55pt}
&
\AxiomC{$\vdash \Gamma= \Gamma'$}
\RightLabel{\textsf{C-Eq-S}}
\UnaryInfC{$\vdash \Gamma'= \Gamma$}
\DisplayProof\hspace{55pt}
&
\AxiomC{$\vdash \Gamma= \Gamma'$}
\AxiomC{$\vdash \Gamma'= \Gamma''\;\mathsf{ctxt}$}
\RightLabel{\textsf{C-Eq-T}}
\BinaryInfC{$\vdash \Gamma= \Gamma''$}
\DisplayProof \\
&\\
\AxiomC{$\Gamma\vdash A\;\mathsf{type}$}
\RightLabel{\textsf{Ty-Eq-R}}
\UnaryInfC{$\Gamma\vdash A= A$}
\DisplayProof
&
\AxiomC{$\Gamma\vdash A= A'$}
\RightLabel{\textsf{Ty-Eq-S}}
\UnaryInfC{$\Gamma\vdash A'= A$}
\DisplayProof
&
\AxiomC{$\Gamma\vdash A= A'$}
\AxiomC{$\Gamma\vdash A'= A''$}
\RightLabel{\textsf{Ty-Eq-T}}
\BinaryInfC{$\Gamma\vdash A= A''$}
\DisplayProof
\\
&\\
\AxiomC{$\Gamma\vdash a:A$}
\RightLabel{\textsf{Tm-Eq-R}}
\UnaryInfC{$\Gamma\vdash a= a: A$}
\DisplayProof
&
\AxiomC{$\Gamma\vdash a= a':A$}
\RightLabel{\textsf{Tm-Eq-S}}
\UnaryInfC{$\Gamma\vdash a'= a: A$}
\DisplayProof

&
\AxiomC{$\Gamma\vdash a= a':A$}
\AxiomC{$\Gamma\vdash a'= a'':A$}
\RightLabel{\textsf{Tm-Eq-T}}
\BinaryInfC{$\Gamma\vdash a= a'': A$}
\DisplayProof
\end{tabular}
\\
\\
\begin{tabular}{ll}
\AxiomC{$\Gamma\vdash A\;\mathsf{type}$}
\AxiomC{$\vdash \Gamma= \Gamma'\;\mathsf{ctxt}$}
\RightLabel{\textsf{Ty-Conv}}
\BinaryInfC{$\Gamma'\vdash A\;\mathsf{type}$}
\DisplayProof\hspace{80pt}
&
\AxiomC{$\Gamma\vdash a:A$}
\AxiomC{\mccorrect{$\vdash \Gamma= \Gamma'\;\mathsf{ctxt}$}}
\AxiomC{$\Gamma;\cdot \vdash A= A'\;\mathsf{type}$}
\RightLabel{\textsf{Tm-Conv}}
\TrinaryInfC{$\Gamma'\vdash a:A'$}
\DisplayProof\\
&\\

\AxiomC{$\Gamma\vdash a=a':A$}
\AxiomC{$\Gamma,x:A,\Gamma'\vdash B\;\mathsf{type}$}
\RightLabel{\mccorrect{\textsf{Ty-Cong}}}
\BinaryInfC{$\Gamma,\Gamma'[a/x]\vdash B[a/x]=B[a'/x]\;\mathsf{type}$}
\DisplayProof
&
\AxiomC{$\Gamma\vdash a=a':A$}
\AxiomC{$\Gamma,x:A,\Gamma'\vdash b:B$}
\RightLabel{\mccorrect{\textsf{Tm-Cong}}}
\BinaryInfC{$\Gamma,\Gamma'[a/x]\vdash b[a/x]=b[a'/x]:B$}
\DisplayProof
\end{tabular}
\end{tabular}
\normalsize}}
\caption{\label{fig:dtteq} Rules for judgemental equality, making it a congruence relation, compatible with typing.}
\end{subfigure}
\caption{\label{fig:dttstructural} The structural rules of dependently typed equational logic.}
\end{figure}
We can use our framework to talk about various type theories. By a \emph{theory}, we mean a set $\mathbb{T}$ of judgements which is closed under the structural rules above, in the sense that their conclusions (written under   the horizontal line) are in $\mathbb{T}$ if their hypotheses (written above) are. Usually, we specify a theory by a set of \emph{axioms}, a set of judgements which can be inductively closed under the structural rules to obtain a theory. For our purposes, we only consider theories with no context symbols. That is, all our contexts consist of lists of type declarations for identifiers \mccorrect{and context equalities consist of type equalities and identifier equalities}.

\subsubsection{A Simple Type Theory, {$\STTGame$}}\label{sec:sttgame}
The simple type theory we use is a variant {$\STTGame$} of the simply typed $\lambda$-calculus with finite product types and finite inductive\footnote{\mccorrect{We use this terminology as we see them as a specific instance of general inductive types, to which one might want to generalise in future work.}} types $\{a_i\;|\; i\}$ for any finite set of \emph{distinct} constants $a_1,\ldots, a_n$, with $\beta$- and $\eta'$-rules\footnote{Note that we are using a restricted form of the $\eta$-rule for inductive types which we call $\eta'$. This is why we are left to impose certain commutative conversions, which (among other things) would be implied by the general $\eta$-rule $ \mathsf{case}_{\{a_i\;|\;i\},\{a_i\;|\;i\}}(x,\{b[a_i/x]\}_i)= b$. More discussion of the matter of commutative conversions and $\eta$-rules can be found in \cite{Ghani1995}. Our equational theory is easily seen to precisely correspond precisely to observational equivalence if we extend the syntax with some sufficiently evil computational effect (in fact, it can be shown using an embedding into CBPV that no other effect can weaken the equational theory of pure type theory further) like printing or state and use CBN evaluation. We present this equational theory as it will precisely correspond to equality in CBN game semantics.  As a rule of thumb, we would like to note that in the presence of effects (which can be modelled in game semantics) the general $\eta$-laws fail for positive connectives in CBN and for negative connectives in CBV. This is one of the mysteries that CBPV addresses.} and certain commutative conversions for the corresponding $\mathsf{case}$-constructs -- essentially the {\PCF} commutative conversions \cite{abramsky2000full} \mccorrect{(section 3.2)}. We are considering a total finitary \PCF, if you will. Specifically, with {$\STTGame$}, we are referring to the theory in dependently typed equational logic generated by the rules of figure \ref{fig:stt}\mccorrect{.} \mccorrect{Note that the rule that $\Gamma\vdash t=u:A$ implies that $\Gamma\vdash t:A$ is admissible.}
\begin{figure}[!tb]
\begin{subfigure}[!t]{\linewidth}
\centering
\fbox{
\resizebox{\textwidth}{!}{
\begin{tabular}{lccr}
&&&\\
\AxiomC{$\vdash\Gamma\;\mathsf{ctxt}$}
\RightLabel{\mccorrect{\textsf{$1$-F}}}
\UnaryInfC{$\Gamma\vdash 1\;\mathsf{type}$}
\DisplayProof
&
\AxiomC{$\vdash\Gamma\;\mathsf{ctxt} $}
\RightLabel{\textsf{$1$-I}}
\UnaryInfC{$\Gamma\vdash \langle\rangle:1$}
\DisplayProof
&
\AxiomC{$\Gamma\vdash t:1 $}
\RightLabel{\textsf{$1$-$\eta$}}
\UnaryInfC{$\Gamma\vdash t= \langle\rangle:1$}
\DisplayProof
&\\
&&&\\
&&&\\
\AxiomC{$\Gamma\vdash B\;\mathsf{type}$}
\AxiomC{$\Gamma \vdash C\;\mathsf{type}$}
\RightLabel{\textsf{$\times$-F}}
\BinaryInfC{$\Gamma \vdash B\times C\;\mathsf{type}$}
\DisplayProof
&
\AxiomC{$ \Gamma\vdash b:B$ }
\AxiomC{$\Gamma\vdash c:C$}
\RightLabel{\textsf{$\times$-I}}
\BinaryInfC{$\Gamma\vdash \langle b,c\rangle :B\times C$}
\DisplayProof &

\AxiomC{$\Gamma\vdash d:B\times C$}
\RightLabel{\textsf{$\times$-E1}}
\UnaryInfC{$\Gamma\vdash \mathsf{fst}(d):B$}
\DisplayProof
&
\AxiomC{$\Gamma\vdash d:B\times C$}
\RightLabel{\textsf{$\times$-E2}}
\UnaryInfC{$\Gamma\vdash \mathsf{snd}(d):C$}
\DisplayProof \\
\quad & & &\\
\quad & & &\\

\AxiomC{$\Gamma\vdash  \mathsf{fst}(\langle b,c\rangle):B$}
\RightLabel{\textsf{$\times$-$\beta1$}}
\UnaryInfC{$\Gamma\vdash \mathsf{fst}(\langle b,c\rangle)= b:B$}
\DisplayProof
&
\AxiomC{$\Gamma\vdash  \mathsf{snd}(\langle b,c\rangle):C$}
\RightLabel{\textsf{$\times$-$\beta2$}}
\UnaryInfC{$\Gamma\vdash  \mathsf{snd}(\langle b,c\rangle)= c:C$}
\DisplayProof
&
\AxiomC{$\Gamma\vdash  \langle \mathsf{fst}(d),\mathsf{snd}(d)\rangle:B\times C$}
\RightLabel{\textsf{$\times$-$\eta$}}
\UnaryInfC{$\Gamma\vdash  \langle \mathsf{fst}(d),\mathsf{snd}(d)\rangle= d:B\times C$}
\DisplayProof &
\\
&&&\\
&&&\\
\AxiomC{$\Gamma\vdash B\;\mathsf{type}$}
\AxiomC{$\Gamma\vdash C\;\mathsf{type}$}
\RightLabel{\textsf{$\Rightarrow$-F}}
\BinaryInfC{$\Gamma\vdash B\Rightarrow C\;\mathsf{type}$}
\DisplayProof
 &
\AxiomC{$\Gamma,x:B\vdash c:C$}
\RightLabel{\textsf{$\Rightarrow$-I}}
\UnaryInfC{$\Gamma\vdash \lambda_{x:B}c:B\Rightarrow C$}\DisplayProof
 & 
\AxiomC{$\Gamma\vdash f:B\Rightarrow C$}
\AxiomC{$\Gamma\vdash b:B$}
\RightLabel{\textsf{$\Rightarrow$-E}}
\BinaryInfC{$\Gamma\vdash f(b):C$}
\DisplayProof & \\
&& &\\
&&&\\
\AxiomC{$\Gamma\vdash (\lambda_{x:B}c)(b):C$}
\RightLabel{\textsf{$\Rightarrow$-$\beta$}}
\UnaryInfC{$\Gamma\vdash (\lambda_{x:B}c)(b)= c[b/x]:C$}
\DisplayProof
& \AxiomC{$\Gamma\vdash \lambda_{x:B} f(x):B\Rightarrow C$}
\RightLabel{\textsf{$\Rightarrow$-$\eta$}}
\UnaryInfC{$\Gamma\vdash \lambda_{x:B} f(x)= f:B\Rightarrow C$}
\DisplayProof & &\\
\end{tabular}}}
\caption{Formation (F), introduction (I), elimination (E) and $\beta$- and $\eta$-conversion rules for the usual connectives of simple type theory. For $\Rightarrow-\eta$, we demand the usual side condition that $x$ not free in $f$.}
\end{subfigure}
\begin{subfigure}[!t]{\linewidth}
\fbox{
\resizebox{\linewidth}{!}{
\begin{tabular}{l}
\\
\begin{tabular}{ll}
\AxiomC{$\vdash\Gamma\;\mathsf{ctxt}$}
\noLine
\UnaryInfC{$a_i, \quad 1\leq i\leq n,\quad\textnormal{distinct constants}$}
\RightLabel{\mccorrect{$\{a_i\;|\; i\}$-\textsf{F}}}
\UnaryInfC{$\Gamma\vdash \{a_i\;|\; i\}\;\mathsf{type}$}
\DisplayProof &\\
&\\
&\\
\AxiomC{$\vdash \Gamma\;\mathsf{ctxt}$}
\RightLabel{\mccorrect{$\{a_i\;|\;i\}$-\textsf{I}$_j$}}
\UnaryInfC{$\Gamma \vdash a_j:\{a_i\;|\; i\}$}
\DisplayProof
&
\AxiomC{$\{\Gamma\vdash c_i:C\}_{1\leq i\leq n}$}
\AxiomC{$\Gamma\vdash a:\{a_i\;|\; i\}$}
\RightLabel{\mccorrect{$\{a_i\;|\;i\}$-\textsf{E}}}
\BinaryInfC{$\Gamma\vdash \mathsf{case}_{\{a_i\;|\; i\},C}(a,\{c_i\}_i):C$}
\DisplayProof
\\
&\\
&\\
\AxiomC{$\Gamma\vdash \mathsf{case}_{\{a_i\;|\; i\},C}(a_j,\{c_i\}_i):C$}
\RightLabel{\mccorrect{$\{a_i\;|\;i\}$-$\beta_j$}}
\UnaryInfC{$\Gamma\vdash \mathsf{case}_{\{a_i\;|\; i\},C}(a_j,\{c_i\}_i)= c_j:C$}
\DisplayProof
&
\AxiomC{$\Gamma,x:\{a_i\;|\; i\}\vdash \mathsf{case}_{\{a_i\;|\;i\},\{a_i\;|\;i\}}(x,\{a_i\}_i):\{a_i\;|\; i\}$}
\RightLabel{\mccorrect{$\{a_i\;|\;i\}$-$\eta'$}}
\UnaryInfC{$\Gamma,x:\{a_i\;|\; i\}\vdash \mathsf{case}_{\{a_i\;|\;i\},\{a_i\;|\;i\}}(x,\{a_i\}_i)= x:\{a_i\;|\; i\}$}
\DisplayProof
\end{tabular}
\\
\\
\\
\AxiomC{$\Gamma\vdash \mathsf{case}_{\{a_i\;|\;i\},B\times C}(x,\{d_i\}_i):B\times C$}
\RightLabel{$\{a_i\;|\;i\}$-\textsf{Comm}-$\langle -,-\rangle$}
\UnaryInfC{
$\Gamma\vdash \mathsf{case}_{\{a_i\;|\;i\},B\times C}(x,\{d_i\}_i)=\langle \mathsf{case}_{\{a_i\;|\;i\},B}(x,\{\mathsf{fst}(d_i)\}_i), \mathsf{case}_{\{a_i\;|\;i\},C}(x,\{\mathsf{snd}(d_i)\}_i)\rangle:B\times C$
}
\DisplayProof
\\
\\
\\
\AxiomC{$\Gamma\vdash \mathsf{case}_{\{a_i\;|\;i\},B\Rightarrow C}(x,\{f_{i}\}_i):B\Rightarrow C$}
\RightLabel{$\{a_i\;|\;i\}$-\textsf{Comm}-$\lambda$}
\UnaryInfC{
$\Gamma\vdash \mathsf{case}_{\{a_i\;|\;i\},B\Rightarrow C}(x,\{f_{i}\}_i)= \lambda_{y:B}\mathsf{case}_{\{a_i\;|\;i\},C}(x,\{f_{i}(y)\}_i):B\Rightarrow C$
}
\DisplayProof\\
\\
\\
\AxiomC{$\Gamma\vdash \mathsf{case}_{\{b_j\;|\;j\},C}(\mathsf{case}_{\{a_i\;|\;i\},\{b_j\;|\;j\}}(x,\{b'_i\}),\{c_j\}_j):C$}
\RightLabel{$\{a_i\;|\;i\}$-\textsf{Comm}-$\mathsf{case}$}
\UnaryInfC{$\Gamma\vdash \mathsf{case}_{\{b_j\;|\;j\},C}(\mathsf{case}_{\{a_i\;|\;i\},\{b_j\;|\;j\}}(x,\{b'_i\}),\{c_j\}_j)=\mathsf{case}_{\{a_i\;|\;i\},\{b_j\;|\;j\}}(x,\{\mathsf{case}_{\{b_j\;|\;j\},C}(b'_i,c_j)\}_j):C $}
\DisplayProof
\end{tabular}}}
\caption{The rules for a notion of ground types for simple type theory: finite inductive types.}
\end{subfigure}
\caption{\label{fig:stt} The rules generating the axioms for {$\STTGame$}\mccorrect{.}}
\end{figure}

\subsubsection{A Dependent Type Theory, {$\DTTGame$}$-$}
Similarly, we can present our preferred variant {$\DTTGame$} of dependent type theory as a theory in dependently typed equational logic. First, we present a smaller theory {$\DTTGame$}$-$, which does not yet include the $\beta$- and $\eta$-rules and commutative conversions of {$\DTTGame$}, but rather only consists of its $F$-, $I$- and $E$-rules. Later, {$\DTTGame$} is obtained by adding to {$\DTTGame$}$-$ the equational theory that results from that of {$\STTGame$}, under a syntactic translation to {$\STTGame$}.

{$\DTTGame$}$-$ consists of the rules of figure \ref{fig:dtt-}.
\begin{figure}
\begin{subfigure}[!tb]{\linewidth}
\centering
\fbox{
\resizebox{\textwidth}{!}{
\begin{tabular}{llll}
\AxiomC{$\vdash\Gamma\;\mathsf{ctxt} $}
\RightLabel{\textsf{\mccorrect{$1$-F}}}
\UnaryInfC{$\vdash 1\;\mathsf{type}$}
\DisplayProof
&
\AxiomC{$\vdash\Gamma\;\mathsf{ctxt} $}
\RightLabel{\textsf{$1$-I}}
\UnaryInfC{$\Gamma\vdash \langle\rangle:1$}
\DisplayProof
&
&\\
&&&\\
\AxiomC{$\Gamma,x:A\vdash B\;\mathsf{type}$}
\RightLabel{\textsf{$\Sigma$-F}}
\UnaryInfC{$\Gamma\vdash \Sigma_{x:{A}}B\;\mathsf{type}$}
\DisplayProof &

\AxiomC{$\Gamma \vdash a:A$}
\AxiomC{$\Gamma \vdash b:B[{a}/x]$}
\RightLabel{\textsf{$\Sigma$-I}}
\BinaryInfC{$\Gamma \vdash  \langle {a} ,b\rangle :\Sigma_{x:{A}}B $}
\DisplayProof
&
\AxiomC{$\Gamma \vdash t:\Sigma_{x:{A}}B$}
\RightLabel{\textsf{$\Sigma$-E1}}
\UnaryInfC{$\Gamma \vdash \mathsf{fst}(t):A$}
\DisplayProof &
\hspace{-50pt}
\AxiomC{$\Gamma \vdash t:\Sigma_{x:{A}}B$}
\RightLabel{\textsf{$\Sigma$-E2}}
\UnaryInfC{$\Gamma \vdash \mathsf{snd}(t):B[\mathsf{fst}(t)/x]$}
\DisplayProof\\
&&&\\
&&&\\
\AxiomC{$\Gamma,x:A \vdash B\;\mathsf{type}$}
\RightLabel{\textsf{$\Pi$-F}}
\UnaryInfC{$\Gamma\vdash\Pi_{x:{A}}B\;\mathsf{type}$}
\DisplayProof
&
\AxiomC{$\Gamma,x:A\vdash b:B$}
\RightLabel{\textsf{$\Pi$-I}}
\UnaryInfC{$\Gamma\vdash \lambda_{x:{A}}b:\Pi_{x:{A}}B$}
\DisplayProof
&
\AxiomC{$\Gamma \vdash a:A$}
\AxiomC{$\Gamma\vdash f:\Pi_{x:{A}}B$}
\RightLabel{\textsf{$\Pi$-E}}
\BinaryInfC{$\Gamma\vdash f({a}):B[{a}/x]$}
\DisplayProof &\\
&&&\\
\AxiomC{$\Gamma \vdash a:A$}
\AxiomC{$\Gamma \vdash a':A$}
\RightLabel{\textsf{$\Id$-F}}
\BinaryInfC{$\Gamma \vdash \Id_{A}(a,a')\;\mathsf{type}$}
\DisplayProof
&
\AxiomC{$\Gamma \vdash a:A$}
\RightLabel{\textsf{$\Id$-I}}
\UnaryInfC{$\Gamma\vdash \refl{a}:\Id_{A}(a,a)$}
\DisplayProof
&\hspace{-5pt}\begin{tabular}{l}
\AxiomC{\begin{tabular}{ll}
$\Gamma \vdash a:A$ & \\
$\Gamma \vdash a':A$ &$\Gamma,x:A,x':A ,y:\Id_{A}(x,x')\vdash D\;\mathsf{type}$\\
$\Gamma  \vdash p:\Id_{A}(a,a')$ &  
$\Gamma,z:A \vdash d:D[{z}/x,{z}/x',\refl{z}/y]$
\end{tabular}
}
\RightLabel{\textsf{$\Id$-E}}
\UnaryInfC{$\Gamma \vdash \mathsf{let}\; p\;\mathsf{be}\;\refl{z}\;\mathsf{in}\; d:D[{a}/x,{a'}/x',p/y]$}
\DisplayProof
\end{tabular}\hspace{-130pt}\;

\end{tabular}
}}
\caption{Rules for $1$-, $\Sigma$-, $\Pi$-, and $\Id$-types. In case $x$ is not free in $B$, we sometimes write $A\Rightarrow B$ for $\Pi_{x:A}B$ and $A\times B$ for $\Sigma_{x:A}B$.}
\end{subfigure}

\begin{subfigure}[!tb]{\linewidth}
\centering
\fbox{\resizebox{\textwidth}{!}{
\begin{tabular}{l}
\begin{tabular}{ccc}\AxiomC{$\vdash \Gamma\;\mathsf{ctxt}\qquad \vdash a_1:A$}
\AxiomC{$\ldots$}
\AxiomC{$\vdash a_n:A$}
\noLine
\TrinaryInfC{$b_{i,j},\;\;\; 1\leq i\leq n,\; 1\leq j\leq m_i,\;$ distinct constants}
\RightLabel{\mccorrect{$(a_i\mapsto_i\{b_{i,j}\;|\;j\})(x)$-\textsf{F}}}
\UnaryInfC{$\Gamma,x:A\vdash (a_i\mapsto_i\{b_{i,j}\;|\;j\})(x)\;\mathsf{type}$}
\DisplayProof\hspace{-30pt}\;
&\qquad\qquad &
\AxiomC{$\vdash \Gamma\;\mathsf{ctxt}$}
\RightLabel{\mccorrect{$(a_i\mapsto_i\{b_{i,j}\;|\;j\})(x)$-\textsf{I}$_{i,j}$}}
\UnaryInfC{$\Gamma\vdash b_{i,j}:(a_i\mapsto_i\{b_{i,j}\;|\;j\})(a_i)$}
\DisplayProof\end{tabular}\\
\\
\AxiomC{
\begin{tabular}{ll}
$\Gamma\vdash a:A$ & 
$
x:A,y:(a_i\mapsto_i\{b_{i,j}\;|\;j\})(x)\vdash C\;\mathsf{type}$\\ 
$\Gamma\vdash b:(a_i\mapsto_i\{b_{i,j}\;|\;j\})(a)$
&
$\{\Gamma\vdash c_{i,j}:C[a_i/x,b_{i,j}/y]\}_{i,j}$
\end{tabular}
}
\RightLabel{\mccorrect{$(a_i\mapsto_i\{b_{i,j}\;|\;j\})(x)$-\textsf{E}}}
\UnaryInfC{
$\Gamma\vdash\mathsf{case}_{(a_i\mapsto_i\{b_{i,j}\;|\;j\})(a),C}(b,\{c_{i,j}\}_{i,j}):C[a/x,b/y]$}
\DisplayProof
\\
\\
\AxiomC{\begin{tabular}{ll}
$\Gamma\vdash a:A$ &$\Gamma,y:(a_i\mapsto_i\{b_{i,j}\;|\;j\})(a)\vdash C\;\mathsf{type}$ \\
$\Gamma\vdash b:(a_i\mapsto_i\{b_{i,j}\;|\;j\})(a)$ & $\{\Gamma,p_{i,j}:\Id_A(a_i,a),q_{i,j}:\Id_{(a_i\mapsto_i\{b_{i,j}\;|\;j\})(a)}(\mathsf{subst}(p,b_{i,j}),b)\vdash c_{ij}:C[b/y]\}_{i,j}$ 
\end{tabular}
}
\RightLabel{$(a_i\mapsto_i\{b_{i,j}\;|\;j\})(x)$-\textsf{E'}}
\UnaryInfC{$\Gamma\vdash \mathsf{case}^{p,q}_{(a_i\mapsto_i\{b_{i,j}\;|\;j\})(a),C}(b,\{c_{i,j}\}_{i,j}):C[b/y]$}
\DisplayProof
\end{tabular}
}}
\caption{\label{fig:case1} Rules for a finite inductive type family $x:A\vdash (a_i\mapsto_i\{b_{i,j}\;|\;j\})(x)\;\mathsf{type}$, generated by $\vdash b_{i,1},\ldots,b_{i,m_i}:(a_i\mapsto_i\{b_{i,j}\;|\;j\})(x)[a_{i}/x]$ for $\vdash a_1,\ldots,a_n:A$.}
\end{subfigure}
\caption{\label{fig:dtt-} The rules generating the axioms for {$\DTTGame$}$-$\mccorrect{.}}
\end{figure}
In addition to term and type formation rules for $\Sigma$-, $\Pi$- and $\Id$-types,  we have a mechanism for forming finite inductive type families, which play the r\^ole of ground types. Let $\vdash A\type$\footnote{Perhaps, it would be more elegant to allow the specification of an inductive type family depending on an arbitrary context $\vdash x_1:A_1,\ldots,x_n:A_n\ctxt$ rather than a single type. However, given that we consider a system with strong $\Sigma$-types, the two are equivalent and only letting inductive families depend on a single types allows us to keep notation more lightweight.}. Then, we can give a finite inductive definition of a type family $x:A\vdash (a_i\mapsto_i\{b_{i,j}\;|\;j\})(x)\;\mathsf{type}$ by specifying finitely many closed terms $a_1,\ldots, a_n :A$ and distinct symbols $b_{i,j}$, $1\leq i\leq n$, $1\leq j\leq m_i$. The idea is that $B=(a_i\mapsto_i\{b_{i,j}\;|\;j\})(x)$ is a type family, such that $(a_i\mapsto_i\{b_{i,j}\;|\;j\})(a_i)$ contains precisely the distinct closed terms $b_{i,1},\ldots,b_{i,m_i}$. These type families are more limited than general inductive definitions as they are freely generated by \emph{closed} terms, while one would allow open terms in the general case \cite{dybjer1994inductive}.  This means that we precisely get the inductive type families with finitely many non-empty fibres which are all finite types. An example the reader may want to keep in mind is given by calendars in format dd-mm (for the year 1984, for instance): here $A=\mathsf{mm}:=\{01,\ldots,12\}$ and $B=\mathsf{dd-mm}:=(i\mapsto_i \{01\textnormal{-}i,\ldots,N_i\textnormal{-}i\})(x)$, where $N_i$ is $29$, $30$, or $31$, depending on the number of days the month in question has.

We interpret such a definition as specifying $F$-, $I$- and $E$-rules for $(a_i\mapsto_i\{b_{i,j}\;|\;j\})(x)$.  In fact, instead of $(a_i\mapsto_i\{b_{i,j}\;|\;j\})(x)$-$E$, we may equivalently specify an alternative elimination rule $(a_i\mapsto_i\{b_{i,j}\;|\;j\})(x)$-$E'$. While the former is the usual elimination rule for finite inductive type families, the latter is closer, in a sense, to the intuition of our model and arises naturally in the completeness proofs in chapter \ref{ch:5}. Here, we write $\mathsf{subst}$ for the following principle of indiscernability of identicals.\\
\\
\resizebox{\linewidth}{!}{\begin{tabular}{l}
$$\footnotesize
\AxiomC{$\Gamma,x:A\vdash B\;\mathsf{type}$}
\UnaryInfC{$\Gamma,x:A\vdash \lambda_{y:{B}}y:\Pi_{y:B}B$}
\AxiomC{}
\UnaryInfC{$\Gamma,x,x':A,p:\Id_{A}(x,x')\vdash x,x':A$}
\AxiomC{}
\UnaryInfC{$\Gamma,x,x':A,p:\Id_{A}(x,x')\vdash p:\Id_{A}(x,x')$}
\RightLabel{\textsf{$\Id$-$E$}}
\TrinaryInfC{$\Gamma,x,x':A,p:\Id_{A}(x,x')\vdash \mathsf{subst}(p,-):\Pi_{B}B[x'/x]$}
\DisplayProof$$\end{tabular}\hspace{55pt}\;}\quad\\
\\
More generally, for a context $\Gamma,x:A,y_1:B_1,\ldots,y_n:B_n$, we can inductively define\\
\resizebox{\linewidth}{!}{$
\Gamma,x,x':A,p:\Id_A(x,x'), y_1:B_1,\ldots , y_{n-1}:B_{n-1}\vdash \mathsf{subst}(p,-): \Pi_{y_n:B_n}B_n[x'/x,\ldots,\mathsf{subst}(p,y_i)/y_i,\ldots].$}\\
We note that $(a_i\mapsto_i\{b_{i,j}\;|\;j\})(x)$-$E$ and $(a_i\mapsto_i\{b_{i,j}\;|\;j\})(x)$-$E'$ really are equivalent in a precise sense.

\begin{theorem}
We have translations between $(a_i\mapsto_i\{b_{i,j}\;|\;j\})(x)$-$E$ and $(a_i\mapsto_i\{b_{i,j}\;|\;j\})(x)$-$E'$. These become mutually inverse in the equational theory of {$\DTTGame$}.
\end{theorem}
\begin{proof} Let us write $B$ for $(a_i\mapsto_i\{b_{i,j}\;|\;j\})(x)$. In the presence of $B$-$E'$, we can define $B$-$E$ by noting that\\
\\
\resizebox{\linewidth}{!}{
$$
\footnotesize
\AxiomC{$x:A,y:B,z_{1,1}:C[a_1/x,b_{1,1}/y],\ldots,z_{n,m_n}:C[a_n/x,b_{n,m_n}/y]\vdash z_{i,j}:C[a_i/x,b_{i,j}/y]$}
\UnaryInfC{$x:A,y:B,z_{1,1}:C[a_1/x,b_{1,1}/y],\ldots,z_{n,m_n}:C[a_n/x,b_{n,m_n}/y],p_{i,j}:\Id_A(a_i,x),q:\Id_{B}(\mathsf{subst}(p_{i,j},b_{i,j}),y)\vdash  z_{i,j}:C[a_i/x,b_{i,j}/y]$}
\UnaryInfC{$x:A,y:B,z_{1,1}:C[a_1/x,b_{1,1}/y],\ldots,z_{n,m_n}:C[a_n/x,b_{n,m_n}/y],p_{i,j}:\Id_A(a_i,x),q_{i,j}:\Id_{B}(\mathsf{subst}(p_{i,j},b_{i,j}),y)\vdash  \mathsf{subst}(q_{i,j},\mathsf{subst}(p_{i,j},z_{i,j})):C$}
\DisplayProof
$$}\\
\\
and applying $B$-$E'$ with\\$A'=\Sigma_{x:A}\Sigma_{y:B}\Sigma_{z_{1,1}:C[a_1/x,b_{1,1}/y]}\cdots\Sigma_{z_{n,m_n-1}:C[a_n/x,b_{n,m_n-1}/y]}C[a_n/x,b_{n,m_n}/y]$, $a=x$ (the projection to $A$), $b=y$ (the projection to $B$) and the from $z_{i,j}$ derived expression above for $c_{i,j}$ to derive $B$-$E'$.\\
\\
Conversely, in the presence of $B$-$E$, we derive $B$-$E'$:\\
\\
\resizebox{\linewidth}{!}{
\AxiomC{$\{x':A',p_{i,j}:\Id_A(a_i,a),q_{i,j}:\Id_{B[a/x]}(\mathsf{subst}(p_{i,j},b_{i,j}),b)\vdash c_{i,j}:C[b/y]\}_{i,j}$}
\UnaryInfC{$\{\vdash  \lambda_{x':A'}\lambda_{p_{i,j}:\Id_A(a_i,a)}\lambda_{q_{i,j}:\Id_{B[a/x]}(\mathsf{subst}(p_{i,j},b_{i,j}),b)}c_{i,j}:\Pi_{x':A'}\Pi_{p_{i,j}:\Id_A(a_i,a)}\Pi_{q_{i,j}:\Id_{B[a/x]}(\mathsf{subst}(p_{i,j},b_{i,j}),b)}C[b/y]\}_{i,j}$}
\UnaryInfC{$x:A,y':B\vdash \mathsf{case}_{B,\Pi_{x':A'}\Pi_{p:\Id_A(x,a)}\Pi_{q:\Id_{B[a/x]}(\mathsf{subst}(p,y'),b)}C[b/y]}(y',\{\lambda_{x':A'}\lambda_{p_{i,j}:\Id_A(a_i,a)}\lambda_{q_{i,j}:\Id_{B[a/x]}(\mathsf{subst}(p_{i,j},b_{i,j}),b)}c_{i,j}\}):\Pi_{x':A'}\Pi_{p:\Id_A(x,a)}\Pi_{q:\Id_{B[a/x]}(\mathsf{subst}(p,y'),b)}C[b/y]$}
\UnaryInfC{$x':A'\vdash \mathsf{case}_{B,\Pi_{x':A'}\Pi_{p:\Id_A(x,a)}\Pi_{q:\Id_{B[a/x]}(\mathsf{subst}(p,y'),b)}C[b/y]}(y',\{\lambda_{x':A'}\lambda_{p_{i,j}:\Id_A(a_i,a)}\lambda_{q_{i,j}:\Id_{B[a/x]}(\mathsf{subst}(p_{i,j},b_{i,j}),b)}c_{i,j}\})[a/x,b/y'](x',\refl{a},\refl{b}):C[b/y]$}
\DisplayProof
}
\quad\\
These translations are easily seen to be mutually inverse in their translation to {$\STTGame$}, which we define in the next section, due to the $\{b_{i,j}\;|\; i,j\}$-$\mathrm{Comm}$-$\lambda$-rule. Therefore, they are mutually inverse in {$\DTTGame$}.
\end{proof}

We conclude that $\mathsf{case}$ and $\mathsf{case}^{p,q}$ are equivalent. We prefer to use the latter as the default, as it naturally arises in the completeness proofs in chapter \ref{ch:5}. For the purposes of proof theory, however, the former may be the preferred choice, as the metatheory of the resulting system is known to be well-behaved (at least in absence of the commutative conversions).

\subsubsection{A Syntactic Translation from {$\DTTGame$} to {$\STTGame$}}
\label{sec:trans} Morally, {$\DTTGame$} should describe the same algorithms as {$\STTGame$} (at least at the type hierarchy over finite types), possibly assigning them a more precise type. Formally, this idea is captured by the existence of a syntactic translation from \mccorrect{{$\DTTGame$}$-$} into {$\STTGame$}. By noting that it is compositional and faithful on all term constructors, we note that we can add to {$\DTTGame$}$-$ the equational theory of {$\STTGame$} under this translation. We refer to the theory we obtain as {$\DTTGame$}. Some examples of equations this implies are $\beta$- and $\eta$-laws for $1$-, $\Sigma$- and $\Pi$-types and finite inductive type  families and commutative conversions for the $\mathsf{case}$-constructs, analogous to those for their simply typed equivalents, as well as $\beta$-laws for $\Id$-types which state that $\lbi{\refl{z}}{\refl{z}}{d}=d$. We note that we then have a faithful translation $(-)^T$ from {$\DTTGame$} to {$\STTGame$}.

The translation $(-)^T$ is inductively defined on types and terms through the schema of figure \ref{fig:translation}. This translation will later suggest our game theoretic interpretation of dependent type theory, by demanding that it agrees with (a total, finitary equivalent of) the usual \PCF{}\mccorrect{ }game semantics \cite{abramsky2000full} after translating the syntax. A semantically inclined reader may want to think about the translation we define as a faithful non-full functor $(-)^T$ from the syntactic category (or, category of contexts) of {$\DTTGame$} to the syntactic category of {$\STTGame$}.

\begin{figure}[!tb]
\fbox{
\resizebox{\linewidth}{!}{$
\begin{array}{lll}
\vdash b_{i,j}:(a_i\mapsto_i\{b_{i,j}\; |\; j\})(a_i)  & \mapsto & \vdash b_{i,j}:\{b_{i,j}\;|\; i,j\} \\ x:A,y:B,z_{1,1}:C[a_1/x,b_{1,1}/y],\ldots, & \mapsto & x:A^T,y:B^T,z_{1,1}:C^T,\ldots, z_{n,m_n}:C^T\\
z_{n,m_n}:C[a_n/x,b_{n,m_n}/y]\vdash \mathsf{case}_{B,C}(y,\{z_{i,j}\}_{i,j}):C  & & \vdash \mathsf{case}_{B^T,C^T}(y,\{z_{i,j}\}_{i,j}):C^T\\
 x':A'\vdash \mathsf{case}^{p,q}_{B[a/x],C}(b,\{c_{i,j}\}_{i,j}) :C[b/y]&\mapsto &x':(A')^T\vdash \mathsf{case}_{B^T,C^T}(b^T,\{c_{i,j}^T[\reflind/p_{i,j},\reflind/q_{i,j}]\}_{i,j}):C^T\\
x:A\vdash \langle\rangle:1 &\mapsto & x:A^T \vdash \langle\rangle:1\\
x:A\vdash \langle b,c\rangle :\Sigma_{y:B}C & \mapsto & x:A^T\vdash \langle b^T,c^T\rangle :B^T\times C^T\\
x:A\vdash \mathsf{fst}(d):B & \mapsto & x:A^T\vdash \mathsf{fst}(d^T):B^T\\
x:A\vdash \mathsf{snd}(d):C[\mathsf{fst}(d)/y] & \mapsto & x:A^T\vdash \mathsf{snd}(d^T):C^T\\
x:A\vdash \lambda_{y:B}c :\Pi_{y:B}C & \mapsto & x:A^T\vdash \lambda_{y:B^T}c^T :B^T\Rightarrow C^T\\
x:A\vdash f(b) :C[b/y] & \mapsto & x:A^T\vdash f^T(b^T) :C^T\\
x:A\vdash \refl{b}:\Id_B(b,b) & \mapsto & x:A^T\vdash \reflind:\{\reflind\}\\
x:A\vdash \mathsf{let}\; p\;\mathsf{be}\;\refl{z}\;\mathsf{in}\; d :D[b/y,b'/y',p/w] &\mapsto & x:A^T \vdash \mathsf{case}_{\{\reflind\},D^T}(p^T,\{d^T[b^T/z]\}):D^T\\
\Gamma, x:A,\Delta\vdash x:A &\mapsto & \Gamma^T,x:A^T,\Delta^T\vdash x:A^T
\end{array}$
}}
\caption{\label{fig:translation} A syntactic translation \mccorrect{on terms and types} from {$\DTTGame$}$-$ into {$\STTGame$}. Note that it is functorial in the sense that it respects identifiers and substitutions.}
\end{figure}

Finally, we define {$\DTTGame$} as the theory generated by the rules of {$\DTTGame$}$-$ together with the final rule, $\DTTGame$\textsf{-Eq}, of figure \ref{fig:dtteqfinal}, which says that {$\DTTGame$} inherits the judgemental equalities of {$\STTGame$}\mccorrect{.} We note that $\DTTGame$-\textsf{Eq} gives us a concise way of equipping {$\DTTGame$} with the appropriate $\beta$- and $\eta$-rules for its type formers as well as all necessary commutative conversions. We sometimes also consider the rule \textsf{Ty-Ext}, which expresses that types are extensional from the point of view of their sections\footnote{It remains to be verified if type checking remains decidable in the presence of this rule.}.

\begin{figure}[!tb]
\fbox{\resizebox{\linewidth}{!}{
\parbox{1.2\linewidth}{\begin{tabular}{l}
\AxiomC{$\Gamma\vdash_{\DTTGame} a:A$}
\AxiomC{$\Gamma\vdash_{\DTTGame}  b:A$}
\AxiomC{$\Gamma^T\vdash_{\STTGame} a^T= b^T:A^T$}
\RightLabel{{$\DTTGame$}-\textsf{Eq}}
\TrinaryInfC{$\Gamma\vdash_{\DTTGame} a= b:A$}
\DisplayProof\\
\\
\AxiomC{
\begin{tabular}{l}
$\forall_{1\leq i\leq 2}\;x_1:A_1,\ldots,x_n:A_n\vdash B_i\type$\\
$\vdash B_1^T=B_2^T$\\
$\forall_{\vdash t_1:A_1}\ldots \forall_{\vdash t_n:A_n[t_1/x_1,\ldots,t_{n-1}/x_{n-1}]}\vdash B_1[t_1/x_1,\ldots,t_n/x_n]=B_2[t_1/x_1,\ldots,t_n/x_n]$
\end{tabular}}
\RightLabel{\textsf{Ty-Ext}}
\UnaryInfC{$\vdash \Pi_{x_1:A_1}\cdots \Pi_{x_n:A_n} B_1=\Pi_{x_1:A_1}\cdots \Pi_{x_n:A_n} B_2$}
\DisplayProof
\end{tabular}
}\hspace{80pt}\;}
}
\caption{\label{fig:dtteqfinal} The final rule, $\DTTGame$-\textsf{Eq}, which {$\DTTGame$} has on top of {$\DTTGame$}$-$, letting it inherit the equational theory of {$\STTGame$}, as well as the type extensionality rule \textsf{Ty-Ext} which we sometimes consider.}
\end{figure}
We note that, by induction, $(-)^T$ respects the judgemental equalities introduced by the rule above, meaning that $(-)^T$ defines a translation from {$\DTTGame$} to {$\STTGame$}. This lets us conclude the following.

\begin{corollary}\label{thm:trans}
The translation $(-)^T$ defined above defines a faithful translation from {$\DTTGame$} to {$\STTGame$}.
\end{corollary}
We observe that we have defined a flavour of intensional type theory.
\begin{remark}
The reader might wonder how the equational theory of {$\DTTGame$} compares to the usual ones we use for dependent type theories. We note that it implies all the usual $\beta$- and $\eta$-rules \mccorrect{(weak $\eta'$ for inductive families)} for the type formers we consider (as well as some $\PCF$-like commutative conversions), with the exception of the $\eta$-rule $\lbi{z}{\refl{x}}{c[\refl{x}/z]}=c$ for $\Id$-types.

Indeed, we easily see that {$\DTTGame$} refutes one of the notorious consequences of $\Id-\eta$, the principle of equality reflection,\\
\\
\resizebox{\linewidth}{!}{\parbox{1.2\linewidth}{
$$\AxiomC{$\Gamma\vdash p:\Id_A(f,g)$}
\RightLabel{$\mathsf{Reflection}$}
\UnaryInfC{$\Gamma\vdash f= g:A$}
\DisplayProof,
$$}}\\
\\
as, for instance, for $\Gamma = x:\{a\}$ and $A= \{a\}$, $f=\mathsf{case}_{\{a\},\{a\}}(x,a)$ and $g=a$, we do not have that $x:\{a\}\vdash f= g:\{a\}$, while we do have $x:\{a\}\vdash \mathsf{case}_{\{a\},\Id_{\{a\}}(f,g)}(x,\refl{a}):\Id_{\{a\}}(f,g)$.
\end{remark}

\subsection{Categorical Semantics}
In this section, we briefly discuss a notion of categorical semantics for dependently typed equational logic.

It is clear that a type theory with dependent types should be modelled by some indexed category $\Bcat^{op}\ra{\Ccat}\Cat$. (This would be even more obvious if we represented context morphisms as first class objects in the syntax, in the style of Pitts \cite{pitts1995categorical}.) Indeed, we have the category $\Bcat$ which interprets the contexts and the morphisms between them. We have a type theory in each context $\Gamma$, which is modelled by some category $\Ccat(\sem{\Gamma})$ with structure to interpret the appropriate connectives. And, whenever two contexts are related by some context morphism $\Gamma' \vdash \gamma:\Gamma$, we have substitution operations going from the type theory in context $\Gamma$ to that in context $\Gamma'$, which are modelled by structure preserving functors $\Ccat(\sem{\Gamma})\ra{\Ccat(\sem{\gamma})}\Ccat(\sem{\Gamma'})$ (as substitution usually is compatible with all term and type formers). In this view, it is easily seen that existential quantifiers should get interpreted, \`a la Lawvere \cite{lawvere1970equality}, as left adjoints to these substitution functors while universal quantifiers are their right adjoints.

The missing ingredient is that, in dependent type theory, quantification is not external but internal: the entities (in $\Bcat$) we are quantifying over are of the same nature as the proofs of the predicates (in $\Ccat$) that we quantify over. The idea is that objects in $\Bcat$ can be built as lists of objects in the fibres of $\Ccat$ and that the morphisms in $\Bcat$ (the interpretation of context morphisms) then arise as corresponding lists of morphisms in the fibres of $\Ccat$ (the interpretation of terms). This intuition is formalised by the so-called comprehension axiom.

\begin{definition}[Comprehension Axiom]\label{def:comprehension} Let $\Bcat^{op}\ra{\Ccat}\Cat$ be a strict\footnote{For brevity, from now on we shall often drop the modifier ``strict'' for indexed structures. For instance, if we mention an indexed honey badger, we shall really mean a strict indexed honey badger.} indexed category (writing $\Cat$ for  the category of small categories and functors). Given $B'\ra{f} B$ in $\Bcat$, let us write $-\{f\}$ for the change of base functor \mccorrect{$\Ccat(f):\Ccat(B)\ra{}\Ccat(B')$}. Recall that $\Ccat$ is said to satisfy the \emph{comprehension axiom} if
\begin{itemize}
\item $\Bcat$ has a terminal object $\cdot$;
\item all fibres $\Ccat(B)$ have terminal objects $1_B$ which are stable under change of base (for which we just write $1$);
\item the presheaves (writing $\Set$ for  the category of small sets and functions)
\begin{diagram}
(\Bcat/B)^{op} & \rTo & \Set\\
(B'\ra{f}B) & \rMapsto  & \Ccat(B')(1,C\{f\})
 \end{diagram}
are representable. That is, we have representing objects $B.C\ra{\proj{B}{C}}B$ and natural bijections
\begin{diagram}
\Ccat(B')(1,C\{f\})& \rTo^{\cong} &\Bcat/B(f,\proj{B}{C})\\
c & \rMapsto & \langle f,c\rangle .
\end{diagram}
\end{itemize}
\end{definition}
\mccorrect{We} write $\diagv{B}{C}$ for the element of $\Ccat(B.C)(1,C\{\proj{B}{C}\})$ corresponding to $\id_{\proj{B}{C}}$ (the universal elements of the representation). We define the morphisms
\begin{diagram} B.C & \rTo^{\diag{B}{C}:=\langle\id_{B.C},\diagv{B}{C} \rangle} & B.C.C\{\proj{B}{C}\};\\
B'.C\{f\} & \rTo^{\qu{f}{C}:=\langle\proj{B'}{C\{f\}};f,\diagv{B'}{C\{f\}}\rangle} & B.C.\end{diagram}
We have maps (defining the \emph{comprehension functor})
\begin{diagram}
\Ccat(B)(C',C)& \rTo^{\proj{B}{-}} &\Bcat/B(\proj{B}{C'},\proj{B}{C})\\
c & \rMapsto &\proj{B}{c}:= \langle \proj{B}{C'},\diagv{B}{C'};c\{\proj{B}{C'}\}\rangle .
\end{diagram}
When these are full and faithful, we call the comprehension \emph{full and faithful}, respectively. When it induces an equivalence \mccorrect{$\Ccat(\cdot)\cong \Bcat/\cdot\cong \Bcat$}, we call the comprehension \emph{democratic}.

Note that the comprehension axiom\mccorrect{ }says that we build as lists of closed terms the morphisms into objects that arise as lists of types in our category of contexts $\Bcat$. Demanding the comprehension functor to be fully faithful means that also the\mccorrect{ }terms in $\Ccat(\Gamma)(A,B)$ correspond precisely with the\mccorrect{ }terms \mccorrect{in} $\Ccat(\Gamma.A)(1,B\{\proj{\Gamma}{A}\})$.  This is essential to get a precise fit with the syntax for cartesian dependent type theory. The notion of democracy corresponds to the syntactic condition that all contexts are formed from the empty context by adjoining types.
\begin{remark}[Correspondence with Comprehension Categories]
The definition of an indexed category with comprehension is easily seen to be equivalent to Jacobs' notion of a split comprehension category with unit \cite{jacobs1993comprehension}. We prefer this formulation in terms of indexed categories as strictness is important in computer science \mccorrect{(syntactic substitution is strict)}, in which case the fibrational perspective is needlessly abstract. Jacobs' notion of fullness of a comprehension category corresponds -- confusingly -- to our demand of the comprehension both being full and faithful. We believe it is useful to use this more fine-grained terminology.

\begin{mccorrection}
Let us make the correspondence a bit more precise -- the reader can find all details in \cite{jacobs1993comprehension}. There is a well-known correspondence between strict indexed categories and split fibrations:
\begin{itemize}
\item given a strict indexed category $\Bcat^{op}\ra{\Ccat}\Cat$, we can define a split fibration $\int\Ccat\ra{\mathbf{p}_\Ccat}\Bcat$ by using the Grothendieck construction: we take $\int\Ccat$ to have objects $\ob(\int\Ccat):=\Sigma_{B\in\ob(\Bcat)}\ob(\Ccat(B))$ and morphisms $$\left(\int\Ccat\right)(\langle B,C\rangle,\langle B' , C'\rangle):=\Sigma_{b\in \Bcat(B,B')}\Ccat(C,C'\{b\})$$
and we take $\mathbf{p}_\Ccat$ to be the projection to the first component;
\item given a split fibration $\mathcal{E}\ra{\mathbf{p}}\Bcat$, we can define a strict indexed category $\Bcat\ra{\Ccat}\Cat$ as $\Ccat(B):=\mathbf{p}^{-1}(B)$ where the functors $-\{f\}:\Ccat(B')\ra{}\Ccat(B)$ arise as the inverse image functors of the split fibration $\mathcal{E}$ along $B\ra{f}B'\in\Bcat$. 
\end{itemize}
Now, our formulation of the comprehension axiom in definition \ref{def:comprehension} for a strict indexed category $\Bcat^{op}\ra{\Ccat}\Cat$ then is easily seen to correspond to putting the following conditions on $\int\Ccat\ra{\mathbf{p}_\Ccat}\Bcat$:
\begin{itemize}
\item $\Bcat$ has a terminal object $\cdot$;
\item $\mathbf{p}_\Ccat\ra{\mathbf{p}_\Ccat}\id_\Bcat$ has a fibred right adjoint $\id_\Bcat\ra{1}\mathbf{p}_\Ccat$ (fibred terminal objects);
\item this functor $\Bcat\ra{1}\int\Ccat$ has a further right adjoint $\int\Ccat\ra{-.-}\Bcat; \langle B,C\rangle \mapsto B.C$.
\end{itemize}
Jacobs calls this structure a \emph{split comprehension category with unit}.
\end{mccorrection}
\end{remark}

Strict indexed categories with \mccorrect{ \emph{full and faithful} comprehension admit a more minimalistic presentation in the form of} Dybjer's notion of categories with families \mccorrect{with unit}. Recall that these categories are another standard notion of  model of dependently typed equational logic~\cite{dybjer1995internal,hofmann1997syntax}.

\begin{definition}[Category with Families] \label{def:cwf} A category with families (CwF) is a category $\Bcat$ with a terminal object $\cdot$, for all objects $\Gamma$ a set $\mathsf{Ty}(\Gamma)$, for all $A\in \mathsf{Ty}(\Gamma)$ a set $\mathsf{Tm}(\Gamma,A)$, for all $\Gamma'\ra{f}~\Gamma$ in $\Bcat$ functions $\mathsf{Ty}(\Gamma)\ra{-\{f\}}\mathsf{Ty}(\Gamma')$ and $\mathsf{Tm}(\Gamma,A)\ra{-\{f\}}\mathsf{Tm}(\Gamma',A\{f\})$, such that
$$\begin{array}{llll}
A\{\mathsf{id}_\Gamma\}=A & \textnormal{(Ty-Id)}\hspace{45pt}\; & A\{f;g\}=A\{g\}\{f\}\hspace{23pt} & \textnormal{(Ty-Comp)}\\
t\{\mathsf{id}_\Gamma\}=t \hspace{23pt} & \textnormal{(Tm-Id)}& t\{f;g\}=t\{g\}\{f\} & \textnormal{(Tm-Comp)},
\end{array}$$
for $A\in\mathsf{Ty}(\Gamma)$ a morphism $\Gamma.A\ra{\proj{\Gamma}{A}}\Gamma$ of $\Bcat$ and $\diagv{\Gamma}{A}\in\mathsf{Tm}(\Gamma.A,A\{\proj{\Gamma}{A}\})$ and, finally, for all $t\in \mathsf{Tm}(\Gamma',A\{f\})$ a morphism $\Gamma'\ra{\langle f,t\rangle}\Gamma.A$ such that
$$
\begin{array}{llll}
 \langle f,t\rangle;\proj{\Gamma}{A}=f & \textnormal{(Cons-L)}\hspace{45pt} \;& 
\diagv{\Gamma}{A}\{\langle f,t\rangle\}=t &\textnormal{(Cons-R)}  \\
\langle \proj{\Gamma}{A},\diagv{\Gamma}{A}\rangle=\mathsf{id}_{\Gamma.A} \hspace{7pt} &\textnormal{(Cons-Id)} &
 g;\langle f,t\rangle=\langle g;f,t\{g\}\rangle \hspace{7pt} &\textnormal{(Cons-Nat)}.
\end{array}
$$
\mccorrect{A CwF is said to have a \emph{unit} if we have $1\in \Ty(\Gamma)$, for all $\Gamma\in\ob(\Bcat)$, such that $\Tm(\Gamma,1)\cong \{*\}$ and $1\{f\}=1$ for all $f\in\Bcat$.}
\end{definition}

\begin{mccorrection}
The correspondence works as follows. Every strict indexed category with comprehension is easily seen to define a CwF with unit, if we define $\Ty(\Gamma)$ and $\Tm(\Gamma, A)$ as $\ob(\Ccat(\Gamma))$ and $\Ccat(\Gamma)(1,A)$, respectively. Conversely, we can define a strict indexed category with comprehension from a CwF with unit by defining $\ob(\Ccat(\Gamma)):=\Ty(\Gamma)$ and $\Ccat(\Gamma)(A,B):=\Tm(\Gamma.A,B\{\proj{\Gamma}{A}\})$. We see that the resulting comprehension is full and faithful. Starting from a CwF with unit, defining the corresponding strict indexed category with comprehension and then defining the CwF from that again gives us back the CwF with unit we started with (up to equivalence) as $\Tm(\Gamma.1,B\{\proj{\Gamma}{1}\})\cong\Tm(\Gamma,B)$. Starting from a strict indexed category $\Ccat$ with comprehension, defining the CwF with unit and from that again a strict indexed category $\Ccat'$, gives us the strict indexed category with full and faithful comprehension where we have redefined $\Ccat'(\Gamma)(A,B):=\Ccat(\Gamma.A)(1,B\{\proj{\Gamma}{A}\})$. We see that this restricts to a bijective correspondence between CwFs with unit and strict indexed categories with full and faithful comprehension (up to equivalence).
\end{mccorrection}

The advantage we see for formulating models as indexed categories is that various connectives get an elegant interpretation. We state some of these interpretations below, where we make use of the usual equational theory for extensional dependent type theory, using $\beta$- and $\eta$-rules for all type formers (including $\Id$-types, unlike in section \ref{sec:trans}).

\begin{theorem}[Pure DTT Semantics, \cite{jacobs1993comprehension,hofmann1997syntax}] 
\mccorrect{We} have a sound interpretation of pure dependent type theory with $1$-types in any  indexed category $\Bcat^{op}\ra{\Ccat}\Cat$ with full and faithful comprehension. We list necessary and sufficient conditions for the model to support various type formers:
\begin{itemize}
\item strong $\Sigma$-types\footnote{That is, $\Sigma$-types with a dependent elimination rule. We call $\Sigma$-types $\Sigma_{x:A}B$ weak if the type we are eliminating into is not allowed to depend on $\Sigma_{x:A}B$ in the elimination rule. We use a similar terminology for other positive connectives. We note that as soon as we have strong $\Sigma$-types, the strong and weak elimination rules for other positive connectives become equivalent.} -- objects $\Sigma_CD$ of $\Ccat(B)$ such that $\proj{B}{\Sigma_CD}=\proj{B.C}{D};\proj{B}{C}$;
\item weak $\Sigma$-types -- left adjoint functors $\Sigma_C\dashv -\{\proj{B}{C}\}$ satisfying the left Beck-Chevalley condition\footnote{Remember that the (left) Beck-Chevalley condition for a left adjoint functor $f_!$ to $f^*:=\Ccat(f)$ for a pullback square \begin{diagram}
A & \rTo^h & B\\
\dTo^f & & \dTo_k\\
C & \rTo_g & D
\end{diagram}
corresponds to the statement that the obvious morphism (from commuting of pullback square, unit, and counit) $f_!h^*\ra{}f_!h^*k^*k_!\ra{\cong}f_!f^*g^*k_!\ra{}g^*k_!$ is an isomorphism. Similarly, by the (right) Beck-Chevalley condition for a right adjoint $f_*$ to $f^*$ we mean that the obvious morphism $g^*k_*\ra{}g^*k_*h_*h^*\ra{\cong}g^*g_*f_*h^*\ra{}f_*h^*$ is an iso. The reader is encouraged to think of this condition as the equivalent for $\Sigma$-, $\Pi$- and $\Id$-types of the condition on the \mccorrect{substitution} functors preserving the appropriate categorical structure for other type formers. It says that, in a sense, $\Sigma$-, $\Pi$- and $\Id$-types are preserved under substitution.} for pullback squares in $\Bcat$ of the following form, which we shall later refer to as $\mathbf{p}$-squares,
\begin{diagram}
B'.{C\{f\}} & \rTo^{\qu{f}{{C}}} & B.{C}\\
\dTo^{\proj{B'}{{C\{f\}}}} &  & \dTo_{\proj{B}{{C}}}\\
B' & \rTo_f & B;
\end{diagram}
\item strong extensional\footnote{These are identity types, which, in addition to the $\beta$-rule $\lbi{\refl{x}}{\refl{x}}{d}=d$ also satisfy the $\eta$-rule $d=\lbi{x}{\refl{x}}{d[x/x',\refl{x}/p]}$. This $\eta$-rule is known to make type checking undecidable in the presence of the strong elimination rule, hence it is often omitted \cite{streicher1993investigations}. \mccorrect{We include it to obtain a more elegant categorical semantics; we could also easily omit it.}} $\Id$-types -- objects $\Id_C$ of $\Ccat(B.C.C)$ such that $\proj{B.C.C}{\Id_C}=\diag{B}{C}$;
\item weak extensional $\Id$-types -- left adjoints $\Id_C\dashv -\{\diag{B}{C}\}$ satisfying the left Beck-Chevalley condition for pullback squares in $\Bcat$ of the following form, which we shall later refer to as $\mathsf{diag}$-squares,
\begin{diagram}
B'.{C\{f\}} & \rTo^{\qu{f}{{C}}} & B.{C}\\
\dTo^{\diag{B'}{C\{f\}}} &  & \dTo_{\diag{B}{C}}\\
B'.C\{f\}.C\{f\}\{\proj{B'}{C\{f\}}\} & \rTo_{\qu{\qu{f}{C}}{C\{\proj{B}{C}\}}} & B.C.C\{\proj{B}{C}\};
\end{diagram}
\item weak $0,+$-types\footnote{The syntactic rules for $0,+$-types can be found in \cite{jacobs1999categorical}. We are mostly interested in stating the semantic condition here, as we shall need it to describe the semantics for sum types in linear and effectful settings, where we shall also treat the syntax.} -- finite indexed coproducts (i.e. finite coproducts in all fibres that are stable under change of base);
\item strong $0,+$-types -- if additionally the following canonical morphisms are bijections
$$\Ccat(C.\Sigma_{1\leq i\leq n}C_i)(C',C'')\ra{}\Pi_{1\leq i\leq n}\Ccat(C.C_i)(C'\{\proj{C}{\langle i,\id_{C_i}\rangle }\},C''\{\proj{C}{\langle i,\id_{C_i}\rangle }\});
$$
\item $\Pi$-types -- right adjoint functors $ -\{\proj{B}{C}\}\dashv \Pi_C$ satisfying the right Beck-Chevalley condition for $\mathbf{p}$-squares.
\end{itemize}
In fact, the interpretation in such categories is complete in the sense that an equality holds in all interpretations iff it is provable in the syntax of dependent type theory where we use both $\beta$- and $\eta$-equality rules for all type formers.
\end{theorem}
\begin{remark}
Note that (weak) $\Sigma$-types and $\Pi$-types in particular allow us to interpret $\times$-types and $\Rightarrow$-types as their existence makes $\Ccat$ into an indexed cartesian closed category (that is, equips the fibres of $\Ccat$ with a cartesian closed structure that is stable under change of base).
\end{remark}

In particular, we can use such categories to model pure simple type theory with $0,+,1,\times,\Rightarrow$-types as a special case, rather than using the usual notion of model of a bicartesian closed category $\Ccat$. Indeed, starting from such a bicartesian closed category $\Ccat$, we can produce an indexed category $\Ccat^{op}\ra{\self(\Ccat)}\Cat$ where $\self(\Ccat)(A)$ has the same objects as $\Ccat$ and $\self(\Ccat)(A)(B,C)=\Ccat(A\times B,C)$ with the obvious identities and composition and with the change of base functors defined to be the identity on objects and to act on morphisms in the obvious way through precomposition. We see that every model of simple type theory gives, in particular, rise to a (rather degenerate) model of dependent type theory. 

\begin{theorem}
For a bicartesian closed category $\Ccat$, $\self(\Ccat)$ is an indexed category with full and faithful democratic comprehension, which supports $0$-, $+$-, $\Sigma$- and $\Pi$-types. In this case, $\proj{A}{B}$ is the usual product projection from $A\times B\ra{}A$. It does not usually support extensional $\Id$-types as these correspond to objects $1/A$ such that $1/A\times A\cong 1$.
\end{theorem}

Inductive types and type families, in particular finite ones, can be given a pretty categorical semantics as initial algebras for certain endofunctors. We refer the interested reader to \cite{hermida1998structural}, as we shall not need those details in our development.

\section{Call-By-Push-Value and Effectful Simple Type Theory}\label{sec:backcbpv}
We believe Levy's call-by-push-value (CBPV) is an excellent setting for studying effectful type theories \cite{levy2012call}. It unifies the CBV and CBN paradigms as follows.

Recall that one origin of the CBV-CBN-distinction is the fact that, in an effectful type theory, we cannot usually have both coproduct types and function types with their general $\eta$-laws: the $\eta$-law has to fail either for the former type formers, leading to CBN, or for the latter, leading to CBV. For a particular instantiation of this idea, we would \mccorrect{like} to remind the reader of the folklore theorem that a cartesian closed category with coproducts -- just an initial object is enough, in fact -- degenerates to the trivial category if it has fixpoints \cite{huwig1990note}. CBPV unifies CBV and CBN type theories by having two distinct classes of types: those for which we have connectives like coproduct types and those for which we have ones like function types. This allows us to retain the general $\eta$-laws for all connectives and lets us encode traditional CBV and CBN type theories.

In this section, we present a slight reformulation and simplification of CBPV's simply typed version, with the purpose of extending it with dependent types later. We start by discussing a syntax in section \ref{sec:scbpvsyn} which is almost identical to Levy's CBPV except that computations are treated as special stacks/homomorphisms. In section \ref{sec:scbpvsem}, we discuss a modified but equivalent (in the categorical sense) presentation of Levy's categorical semantics of simple CBPV that makes the transition to dependent types more natural, after which we give a few examples of models in section \ref{sec:simplmod}. Next, we briefly discuss the small-step operational semantics for CBPV in section \ref{sec:simplop}. Finally, in section \ref{sec:simpleff}, we sketch how one proceeds to add effects to the pure CBPV calculus, which is, after all, the point of our  endeavour. For this, we mostly take an operational point of view.

\subsection{Syntax} \label{sec:scbpvsyn}
We encourage the reader to look at the syntax of call-by-push-value (CBPV) in the following slightly simplified way: as providing an adjunction decomposition of Moggi's monadic metalanguage \cite{moggi1991notions}, similar (dual) to the one that Benton's linear/non-linear (LNL) calculus \cite{benton1995mixed} gives of (the comonadic) dual intuitionistic linear logic (DILL) \cite{barber1996dual}, but in the more general setting of possibly non-commutative effects. Roughly, CBPV consists of two type theories, related by an adjunction $F\dashv U$: one for defining \emph{values and their types}, to be thought of as \emph{static} objects which behave like a pure cartesian type theory, and one for defining effectful \emph{computations/stacks and their types}, to be thought of as \emph{dynamic} objects which behave linearly.

Therefore, CBPV distinguishes between two classes of types: \emph{value types} and \emph{computation types}. These can be similarly read as, respectively, positive and negative types or as types of data and codata. The linear types of the LNL calculus should be thought of as analogous to computation types, while its cartesian types correspond to value types. The idea is that in natural deduction, for some connectives, the positive/value connectives, the introduction rule involves a choice, while the elimination rule is invertible (works through pattern matching) and for others, the negative/computation connectives, the opposite is true. As a rule of thumb, connectives that operate on \mccorrect{value} types arise as left adjoint functors in the categorical semantics, while connectives that operate on \mccorrect{computation} types are right adjoint functors.

Call-by-push-value (and polarised logic) chooses to keep the classes of types formed from both classes of connectives separate and adds two extra connectives $F$, which turns a value type into a type of computations that \emph{return} a result of the original value type, and $U$, which turns a computation type into a value type of \emph{thunks} of computations of the original computation type. This allows us to use the full $\beta\eta$-equational theory for all connectives, even in the presence of effects. Importantly, we have CBV and CBN embeddings of (effectful) type theory into (effectful) CBPV, that give rise to the usual equational theories.

CBPV has two classes of types (we sometimes underline types to emphasize that we mean a computation type):
\begin{align*}&\textnormal{value types \quad}A \qquad\qquad\qquad\textnormal{computation types \quad}\ct{B}.\end{align*}
In this thesis, simple value and computation types are formed using the connectives of figure \ref{fig:vctypes}, excluding general inductive and coinductive types. Here, $1,\times$ will denote pattern-matching products, while $\top,\&$ are projection products\footnote{Note that these correspond to the two ways of defining products in the categorical semantics: as left adjoints to the internal hom or as right adjoints to the diagonal functor, as positive and negative connectives, respectively.}. More generally, following Levy, we include primitives $\Pi_{1\leq i\leq n}\ct{B_i}$ for $n$-ary projection products and $\Sigma_{1\leq i\leq n}A_i$ for $n$-ary sum (we write nullary and binary sum as $0$ and $A+A'$). We do this to emphasize their similarity to $\cpi{-}{}$- and $\Sigma$-types in the dependently typed version of CBPV. We write $A \functype \ct{B}$ for the type of computations that take an input of type $A$ and return a computation of type $\ct{B}$. (These are conventionally written $A\Rightarrow \ct{B}$. We choose our notation to be reminiscent of the LNL calculus expression $F(A)\multimap\ct{B}$\mccorrect{, which it should generalise}.)

\begin{figure}[!tb]
\begin{tabular}{c|c}
\textbf{value/positive types} $A$& \textbf{computation/negative types} $\ct{B}$\\
\hline
$0$, $A + A'$, $\Sigma_{1\leq i\leq n}A_i$ & $A \functype \ct{B}$\\
$1$, $A \times A'$ & $\top$, $\ct{B} \& \ct{B'}$, $\Pi_{1\leq i\leq n}\ct{B_i}$\\
$U\ct{B}$ & $FA$\\
(inductive types) & (coinductive types)
\end{tabular}
\caption{\label{fig:vctypes} An overview of the simple value and computation types we consider with exception of general inductive and coinductive types which we shall not attempt to incorporate.}
\end{figure}

Similarly, CBPV has separate typing judgements for terms representing values and computations, respectively,
\begin{align*}
\Gamma\vdash^v a:A\qquad\qquad\qquad
\Gamma\vdash^c b:\ct{B}.
\end{align*}
Here, $\Gamma$ is a context, or list $x_1:A_1,\ldots,x_n:A_n$ of declarations of distinct identifiers $x_i$ of value type $A_i$. Additionally, Levy considers stacks (sometimes called homomorphisms, as many effects equip computation types with an algebraic structure which stacks preserve), which are represented as typed terms
$$
\Gamma ; {\nil}: \ct{B} \vdash^k c:\ct{C},
$$
where $\Gamma$, as before, is a context of identifier declarations of value type and ${\nil}$ is an identifier of computation type $\ct{B}$. For notational convenience, and unlike Levy, we unify the computation and stack judgements as a single judgement $$\Gamma;\Delta\vdash b:\ct{B},$$ where $\Gamma$ is as before and $\Delta$ is a context of identifier declarations of computation type. For now, $\Delta$ will have at most length 1 and in that case is often referred to as a \emph{stoup}. The case that $\Delta$ has length 0 corresponds to Levy's computation judgement and the case of length 1 to his stack judgement. To keep the notation light, we also omit the annotation $v$ on the sequent in the value judgement and simply write
$$\Gamma\vdash a:A.$$ We encourage the reader to think of the dual context $\Gamma;\Delta$ to consist of cartesian region $\Gamma$ in which the usual structural rules of weakening and contraction are valid and of $\Delta$ as a linear region in which they are not. These typing judgements are defined through the rules of figure   \ref{fig:vcterms} and the obvious (admissible) two substitution rules and weakening rule for identifiers of value type.

As usual, we distinguish between free and bound (i.e. non-free) identifiers and consider terms up to $\alpha$-equivalence, or permutation of their bound identifiers. The rules of the type theory force the free identifiers of a well-typed term to be declared in the context. For notational convenience, we treat indices $i$ of (terms of) a sum $\Sigma_{1\leq i\leq n}A_i$ or product $\Pi_{1\leq i\leq n}\ct{B_i}$ similarly to bound identifiers. A proper formal treatment would involve including the indices and their range in the context, to distinguish between bound and free indices and to consider freshness of the appropriate indices in various $\eta$-rules. We prefer to avoid this extra formality and keep their treatment informal as we are convinced\mccorrect{ that} the intended meaning will be clear to the reader and \mccorrect{that} anyone so inclined can fill in the technical details. 

\begin{figure}[!tb]
\centering
\fbox{\resizebox{\linewidth}{!}{
\begin{tabular}{ll}
\AxiomC{}
\UnaryInfC{$\Gamma,x:A,\Gamma'\vdash x:A$}
\DisplayProof\hspace{50pt} & \AxiomC{$\Gamma\vdash V:A$}
\AxiomC{$\Gamma,x:A,\Gamma'\vdash^{} W:{A'}$}
\BinaryInfC{$\Gamma,\Gamma'\vdash^{} \lbi{x}{V}{W} :{A'}$}
\DisplayProof\\
&\\

& \AxiomC{$\Gamma\vdash V:A$}
\AxiomC{$\Gamma,x:A,\Gamma';\Delta\vdash^{} K:\ct{B}$}
\BinaryInfC{$\Gamma,\Gamma';\Delta\vdash^{} \lbi{x}{V}{K} :\ct{B}$}
\DisplayProof
\\
&\\

\AxiomC{}
\UnaryInfC{$\Gamma;{\nil}:\ct{B}\vdash {\nil}:\ct{B}$}
\DisplayProof &
\AxiomC{$\Gamma;\Delta\vdash^{} K:\ct{B}$}
\AxiomC{$\Gamma;{\nil}:\ct{B}\vdash L:\ct{C}$}
\BinaryInfC{$\Gamma;\Delta\vdash^{} \lbi{{\nil}}{K}{L}:\ct{B}$}
\DisplayProof\\
&\\
\AxiomC{$\Gamma\vdash V:A$}
\UnaryInfC{$\Gamma;\cdot\vdash \return\; V:FA$}
\DisplayProof &
\AxiomC{$\Gamma;\Delta\vdash K:FA$}
\AxiomC{$\Gamma,x:A,\Gamma';\cdot\vdash N:\ct{B}$}
\BinaryInfC{$\Gamma ,\Gamma';\Delta\vdash \toin{K}{x}{N}:\ct{B}$}
\DisplayProof\\
&\\
\AxiomC{$\Gamma;\cdot\vdash M:\ct{B}$}
\UnaryInfC{$\Gamma\vdash \thunk M:U\ct{B}$}
\DisplayProof
&
\AxiomC{$\Gamma\vdash V: U\ct{B}$}
\UnaryInfC{$\Gamma;\cdot\vdash \force V: \ct{B}$}
\DisplayProof\\
&\\
\AxiomC{$\Gamma\vdash V_i: A_i$}
\UnaryInfC{$\Gamma\vdash \langle i,V_i\rangle : \Sigma_{1\leq i\leq n}A_i$}
\DisplayProof
&
\AxiomC{$\Gamma\vdash V: \Sigma_{1\leq i\leq n}A_i$}
\AxiomC{$\{\Gamma,x:A_i\vdash W_i : A'\}_{1\leq i\leq n}$}
\BinaryInfC{$\Gamma\vdash \sipm{V}{i}{x}{W_i} : A'$}
\DisplayProof\\
&\\

 &
\AxiomC{$\Gamma\vdash V: \Sigma_{1\leq i\leq n}A_i$}
\AxiomC{$\{\Gamma,x:A_i;\Delta\vdash K_i : \ct{B}\}_{1\leq i\leq n}$}
\BinaryInfC{$\Gamma;\Delta\vdash \sipm{V}{i}{x}{K_i} : \ct{B}$}
\DisplayProof
\\
&\\
\AxiomC{}
\UnaryInfC{$\Gamma\vdash\langle\rangle :1$}
\DisplayProof&
\AxiomC{$\Gamma\vdash V:1$}
\AxiomC{$\Gamma\vdash W:A'$}
\BinaryInfC{$\Gamma\vdash \upm{V}{W}:A'$}
\DisplayProof
\\
&\\
&\AxiomC{$\Gamma\vdash V:1$}
\AxiomC{$\Gamma;\Delta\vdash K:\ct{B}$}
\BinaryInfC{$\Gamma;\Delta\vdash \upm{V}{K}:\ct{B}$}
\DisplayProof\\
&\\
\AxiomC{$\Gamma\vdash V_1:A_1$}
\AxiomC{$\Gamma\vdash V_2:A_2$}
\BinaryInfC{$\Gamma\vdash \langle V_1,V_2\rangle :A_1\times A_2$}
\DisplayProof\hspace{30pt}\;
&
\AxiomC{$\Gamma\vdash V: A_1\times A_2$}
\AxiomC{$\Gamma,x:A_1,y:A_2\vdash W:A'$}
\BinaryInfC{$\Gamma\vdash \sipm{V}{x}{y}{W}:A'$}
\DisplayProof\\
&\\
&
\AxiomC{$\Gamma\vdash V: A_1\times A_2$}
\AxiomC{$\Gamma,x:A_1,y:A_2;\Delta\vdash K:\ct{B}$}
\BinaryInfC{$\Gamma;\Delta\vdash \sipm{V}{x}{y}{K}:\ct{B}$}
\DisplayProof\hspace{30pt}\;
\\
&\\
\AxiomC{$\{\Gamma;\Delta\vdash K_i :\ct{B_i}\}_{1\leq i\leq n}$}
\UnaryInfC{$\Gamma;\Delta\vdash \lambda_i K_i : \Pi_{1\leq i\leq n}\ct{B_i}$}
\DisplayProof
&
\AxiomC{$\Gamma;\Delta\vdash K: \Pi_{1\leq i\leq n}\ct{B_i}$}
\UnaryInfC{$\Gamma;\Delta\vdash i\textquoteleft K : \ct{B_i}$}
\DisplayProof\\
&\\
\AxiomC{$\Gamma,x:A;\Delta\vdash K:\ct{B}$}
\UnaryInfC{$\Gamma;\Delta\vdash \lambda_xK:A\functype\ct{B}$}
\DisplayProof
&
\AxiomC{$\Gamma\vdash V:A$}
\AxiomC{$\Gamma;\Delta\vdash K:A\functype\ct{B}$}
\BinaryInfC{$\Gamma;\Delta\vdash V\textquoteleft K : \ct{B}$}
\DisplayProof
\end{tabular}
}
}
\caption{\label{fig:vcterms} Values, computations and stacks of simple CBPV.}
\end{figure}
We can consider these terms up to $\alpha$-equivalence and, as such, define an operational semantics for them in section \ref{sec:simplop}. We frequently also consider the terms up to the \mccorrect{additional} equational theory of figure \ref{fig:vceqs} together with the rules which state that all term formers respect equality and that equality is an equivalence relation, where we write $M[V/x]$ for the syntactic metaoperation of capture avoiding substitution of $V$ for $x$ in $M$. We shall see that this equational theory naturally arises from the categorical semantics of simple CBPV.
\begin{figure}[!tb]
\fbox{
\resizebox{\linewidth}{!}
{
\begin{tabular}{ll}
$\lbi{w}{S}{R}=R[S/w]$ &\\
$\toin{(\return V)}{x}{M} = M[V/x]$ & $L[K/ {\nil}] \stackrel{\# x}{=}\toin{K}{x}{L[\return x/{\nil}]}$\\
$\force\thunk M=M$ & $V=\thunk\force V$\\
$\sipm{\langle i, V\rangle }{i}{x}{R_i}=R_i[V/x]$ & $R[V/z]\stackrel{\#i,x}{=}\sipm{V}{i}{x}{R[\langle i,x\rangle/z]}$\\
$\upm{\langle\rangle}{R}=R$ & $R[V/z]=\upm{V}{R[\langle\rangle/z]}$\\
$\sipm{\langle V,V'\rangle }{x}{y}{R}=R[V/x,V'/y]$ \hspace{30pt}\;& $R[V/z]\stackrel{\#x,y}{=}\sipm{V}{x}{y}{R[\langle x,y\rangle/z]}$\hspace{30pt}\;\\
$i\textquoteleft \lambda_j K_j=K_i$ & $K\stackrel{\# i}{=}\lambda_i i\textquoteleft K$\\
$V\textquoteleft \lambda_x K=K[V/x]$ & $K\stackrel{\#x}{=}\lambda_x x\textquoteleft K$
\end{tabular}
}
}
\caption{\label{fig:vceqs} Equations of simple CBPV. These should be read as equations of typed terms: we impose them if we can derive that both sides of the equation are well-typed terms of the same type in the same context. We write $\stackrel{\#x_1,\ldots,x_n}{=}$ to indicate that for the equation to hold, the identifiers or indices $x_1,\ldots,x_n$ should, in both terms being equated, be replaced by fresh ones, in order to avoid unwanted identifier bindings. Note that in the first equation, $w$ might either be an identifier of value type or of computation type.}
\end{figure}

\begin{figure}[!tb]
\fbox{
\resizebox{\linewidth}{!}{
\begin{tabular}{l|l||l|l}
\textbf{CBV type}  & \textbf{CBPV type} & \textbf{CBV term } & \textbf{CBPV term}\\
\hline
$A$ & $A^v$ & $x_1:A_1,\ldots,x_m:A_m\vdash M:A$ & $x_1:A_1^v,\ldots,x_m:A_m^v;\cdot\vdash M^v:F(A^v)$\\
 && $x$& $\return x$\\
  &&$\lbi{x}{M}{N}$ & $\toin{M^v}{x}{N^v}$ \\
$\Sigma_{1\leq i\leq n }A_i$ & $\Sigma_{1\leq i\leq n }A_i^v$& $\langle i,M\rangle $&$\toin{M^v}{x}{\return \langle i, x\rangle }$ \\
 &&$\sipm{M}{i}{x}{N_i}$& $\toin{M^v}{z}{(\sipm{z}{i}{x}{N_i^v})}$\\
$\Pi_{1\leq i\leq n}A_i $ & $U\Pi_{1\leq i \leq n} FA_i^v$ &$\lambda_iM_i$ &$\return \thunk (\lambda_i  M_i^v)$\\
&&$i\textquoteleft  N $&$\toin{N^v}{z}{(i\textquoteleft \force z)}$\\
$A\Rightarrow A'$ & $U(A^v \functype F A'^v)$ & $\lambda_x M$&$\return \thunk \lambda_x M^v$\\
&&$M\textquoteleft N$ &$\toin{M^v}{x}{(\toin{N^v}{z}{(x\textquoteleft \force z)})}$\\
$1$ & $1$ & $\langle\rangle$ & $ \return \langle\rangle$  \\
&&$\upm{M}{N}$&$\toin{M^v}{z}{(\upm{z}{N^v})}$\\
$A \times A'$ & $A^v \times A'^v$ & $ \langle M, N\rangle $  & $\toin{M^v}{x}{(\toin{N^v}{y}{\return \langle x,y\rangle})}$\\
&&$\sipm{M}{x}{y}{N}$&$\toin{M^v}{z}{(\sipm{z}{x}{y}{N^v})}$\\
\end{tabular}
}
}
\caption{\label{fig:cbvtrans} A CBV translation of a simple $\lambda$-calculus into CBPV.}
\end{figure}
\begin{figure}[!tb]
\fbox{
\resizebox{\linewidth}{!}{
\begin{tabular}{l|l||l|l}
\textbf{CBN type}  & \textbf{CBPV type} & \textbf{CBN term } & \textbf{CBPV term }\\
\hline
$\ct{B}$ &  $\ct{B}^n$& $x_1:\ct{B}_1,\ldots,x_m:\ct{B}_m\vdash M:\ct{B}$& $x_1:U\ct{B}_1^n,\ldots,x_m:U\ct{B}_m^n;\cdot\vdash M^n:\ct{B}^n$ \\
 && $x$& $\force x$\\
  & & $\lbi{x}{M}{N}$ & $\lbi{x}{(\thunk M^n)}{N^n}$ \\
$\Sigma_{1\leq i\leq n }\ct{B}_i$ & $F\Sigma_{1\leq i\leq n }U\ct{B}_i^n$& $\langle i,M\rangle $&$\return \langle i,\thunk M^n\rangle $ \\
 & &$\sipm{M}{i}{x}{N_i}$&$\toin{M^n}{z}{(\sipm{z}{i}{x}{N_i^n})}$ \\
$\Pi_{1\leq i\leq n}\ct{B}_i $ & $\Pi_{1\leq i \leq n} \ct{B}_i^n$ & $\lambda_iM_i$& $\lambda_iM_i^n$\\
 && $i\textquoteleft M$ & $i\textquoteleft M^n$\\
$\ct{B}\Rightarrow \ct{B'}$ & $(U\ct{B}^n)\functype \ct{B'}^n$ & $\lambda_x M $& $\lambda_xM^n$\\
 &&$N\textquoteleft M$ & $(\thunk N^n) \textquoteleft M^n$ \\
$1$ & $F1$ & $\langle\rangle$ & $\return \langle\rangle$  \\
&&$\upm{M}{N}$&$\toin{M^n}{z}{(\upm{z}{N^n})}$\\
$\ct{B} \times \ct{B'}$ & $F(U\ct{B}^n \times U\ct{B'}^n)$ & $\langle M, N\rangle $  & $\return \langle \thunk M^n,\thunk N^n\rangle$\\
&& $\sipm{M}{x}{y}{N}$& $\toin{M^n}{z}{(\sipm{z}{x}{y}{N^n})}$ \\
\end{tabular}
}
}
\caption{\label{fig:cbntrans} A CBN translation of a simple $\lambda$-calculus into CBPV.}
\end{figure}

Recall that a call-by-value (CBV) and call-by-name (CBN) evaluation strategy on the $\lambda$-calculus generally give rise to different equational theories (in the presence of effects) \cite{plotkin1975call}. For instance, the $\eta$-rule for function types typically fails in the former and that for sum types in the latter. CBPV gives rise to both of these equational theories by embedding an (impure) $\lambda$-calculus either with a CBV or with a CBN translation. 

In the presence of effects, the usual pure connectives of products, coproducts and function types bifurcate into many variants due to the distinction of versions of different arities and the distinction between projection and pattern matching products. These are nicely and uniformly treated in Levy's Jumbo $\lambda$-calculus \cite{levy2006jumbo}.

There are fully faithful translations $(-)^v$ and $(-)^n$, respectively, from CBV and CBN versions of this whole calculus into CBPV \cite{levy2012call} and, in fact, the same is true if we consider arbitrary theories rather than the pure calculi. To convey the intuition without getting stuck on technicalities, we present some special cases of the translations in figures \ref{fig:cbvtrans} and \ref{fig:cbntrans}.

\subsection{Categorical Semantics}
\label{sec:scbpvsem}
CBPV admits a simple notion of a categorical model. We present a variation of that of \cite{levy2005adjunction} to allow a smooth transition to dependent types. The philosophy is to add to a model $\self(\Ccat)$ of pure simple type theory an extra (locally) indexed category $\Dcat$ to model computations and stacks separately from values and to demand all appropriate negative (right adjoint) connectives in $\Dcat$ and all positive (left adjoint) ones in $\self(\Ccat)$. The idea will be that values $\Gamma\vdash V:A$ denote elements of $\self(\Ccat)(\sem{\Gamma})(1,\sem{A})$ and that computations and stacks $\Gamma;\Delta\vdash M:\ct{B}$  denote elements of $\Dcat(\sem{\Gamma})(\sem{\Gamma;\Delta},\sem{\ct{B}})$.

\begin{definition}[Simple CBPV Model] By a categorical model of simple CBPV, we shall mean the following data.
\begin{itemize}
\item A cartesian category $(\Ccat,1,\times)$ of \emph{values};
\item a locally indexed category $\Ccat^{op}\ra{\Dcat}\Cat$ of \emph{stacks} (and \emph{computations}, in particular), \mccorrect{that is, an indexed category such that the change of base functors are identity on objects};
\item $0,+$-types in $\self(\Ccat)$\footnote{This amounts to having distributive finite coproducts in $\Ccat$.} such that, additionally, the following obvious induced maps are bijections: $$\Dcat(C.\Sigma_{1\leq i\leq n} C_i )(\ct{D},\ct{D'})\ra{}\Pi_{1\leq i \leq n}\Dcat(C.C_i)(\ct{D},\ct{D'});$$
\item an indexed adjunction\footnote{As Plotkin pointed out at the time of Moggi's original work on the monadic metalanguage, this gives a strong monad $T=UF$ on $\Ccat$ \cite{moggi1988computational}.}\mbox{
\begin{diagram}
\Dcat & \pile{\lTo^F\\\bot\\\rTo_U} & \self(\Ccat);
\end{diagram}}
\item $\cpi{-}{}$-types in $\Dcat$ in the sense of having right adjoint functors $-\{\proj{A}{B}\}\dashv \cpi{B}{}:\Dcat(A)\ra{}\Dcat(A.B)$ satisfying the right Beck-Chevalley condition for $\mathbf{p}$-squares;
\item Finite indexed products $(\top,\&)$ in $\Dcat$ (finite products, stable under change of base);
\end{itemize}
Note that $\self(\Ccat)$ automatically has $1$- and $\Sigma$-types.
\end{definition}

\begin{theorem}[Simple CBPV Semantics] We have a sound interpretation of CBPV in a CBPV model:\\
\\
\resizebox{\linewidth}{!}{
$
\begin{array}{ll}
\sem{\cdot}  = 1 & \sem{\Gamma;\cdot} = F1\\
\sem{\Gamma,x:A} = \sem{\Gamma}.\sem{A} &\sem{\Gamma;{\nil}:\ct{B}} = \sem{\ct{B}} \\
\sem{\Gamma\vdash A}=\self(\Ccat)(\sem{\Gamma})(1,\sem{A}) & \sem{\Gamma;\Delta\vdash \ct{B}} = \Dcat(\sem{\Gamma})(\sem{\Gamma;\Delta},\sem{\ct{B}})\\
\sem{U\ct{B}} = U\sem{\ct{B}} & \sem{FA}=F\sem{A} \\
\sem{\Sigma_{1\leq i\leq n}A_i}=(\cdot(\sem{A_1}+\sem{A_2})+\cdots)+\sem{A_n}) & \sem{\Pi_{1\leq i\leq n}\ct{B}_i}=(\cdot(\sem{\ct{B}_1}\&\sem{\ct{B}_2})\&\cdots)\&\sem{\ct{B}_n})\\
\sem{A\times A'} = \sem{A}\times \sem{A'}\cong\Sigma_{\sem{A}}\sem{A'}\{\proj{\sem{\Gamma}}{\sem{A}}\}& \sem{A\functype \ct{B}}  = \cpi{\sem{A}}{\sem{\ct{B}}\{\proj{\sem{\Gamma}}{\sem{A}}\}}\\
\sem{1}=1,\\
\end{array}$}\\
\\
together with the obvious interpretation of terms. The interpretation in such categories is complete in the sense that an equality of values or computations holds in all interpretations iff it is provable in the syntax of CBPV. In fact, we have a 1-1 relationship between models and theories which satisfy mutual soundness and completeness results.
\end{theorem}
Let us write $T$ for the indexed monad $UF$ on $\self(\Ccat)$ and $!$ for the indexed comonad $FU$ on $\Dcat$. We note that the translations from CBV and CBN into CBPV correspond to interpreting CBV and CBN in the Kleisli and co-Kleisli categories for $T$ and $!$ respectively. More generally, we can note that the translations 
of figures \ref{fig:cbvtrans} and \ref{fig:cbntrans} can be transformed into semantic translations which means that any CBPV model gives rise to models of the CBV and CBN\mccorrect{ }$\lambda$-calculus.

\begin{theorem}[Simple CBV Semantics] We obtain a sound interpretation of the CBV\mccorrect{ }$\lambda$-calculus with $1,\times,\Rightarrow,\Sigma_{1\leq i\leq n},\Pi_{1\leq i\leq n}$-types in the Kleisli category for $T$:\\
\resizebox{\linewidth}{!}{\parbox{\linewidth}{
\begin{align*}
\sem{A_1,\cdots,A_n\vdash A}=\Dcat(\sem{A_1}. \cdots.\sem{A_n})(F1,F\sem{A})&\cong\self(\Ccat)_T(\sem{A_1}. \cdots.\sem{A_n})(1,\sem{A})\\
&\cong \self(\Ccat)(\cdot)_T(\sem{A_1}\times \cdots\times\sem{A_n},\sem{A}).
\end{align*}}}\\
The interpretation is complete with respect to this class of models.
\end{theorem}

\begin{theorem}[Simple CBN Semantics] We obtain a sound interpretation of the CBN\mccorrect{ }$\lambda$-calculus with $1,\times,\Rightarrow,\Sigma_{1\leq i\leq n},\Pi_{1\leq i\leq n}$-types in the co-Kleisli category for $!$:\\
\resizebox{\linewidth}{!}{\parbox{\linewidth}{
\begin{align*}
\sem{\ct{B_1},\cdots,\ct{B_n}\vdash \ct{B}}=\Dcat(U\sem{\ct{B_1}}.\cdots.U\sem{\ct{B_n}})(F1,\sem{\ct{B}})&\cong\Dcat_{!}(U\sem{\ct{B_1}}.\cdots.U\sem{\ct{B_n}})(\top,\sem{\ct{B}})\\
&\cong \Dcat(\cdot)_!(\sem{\ct{B_1}}\&\cdots\&\sem{\ct{B_n}},\sem{\ct{B}}).
\end{align*}}}\\
The interpretation is complete with respect to this class of models.
\end{theorem}
Again, both of these results could be strengthened to the statement that we have a 1-1 relationship between models and theories which satisfy mutual soundness and completeness results.

\subsection{A Few Words about Models}\label{sec:simplmod}
An extensive discussion of particular models as well as comparisons between CBPV models and other notion\mccorrect{s} of categorical models of effects can be found in \cite{levy2012call}. Here, we shall be very brief and just recall the following two results and provide some context for the relationship between effects and linear logic.

\begin{theorem}\label{thm:lliscomm}
Let $\Dcat'$ be a model of intuitionistic exponential additive multiplicative linear logic (see section \ref{sec:backll}) in the sense of a symmetric monoidal closed category $(\Dcat',I,\otimes)$ with finite products $(\top,\&)$, finite coproducts $(0,\oplus)$ that distribute over $\otimes$ and a comonad $!$ that is induced by some adjunction
\begin{diagram}
\Ccat' & \pile{\rTo^{F'}\\\bot\\\lTo_{U'}} & \Dcat'
\end{diagram}
 to a cartesian monoidal category $(\Ccat',1,\times)$ with strong monoidal left adjoint $F'$. In that case, $\Dcat'$ gives rise to a canonical model $F\dashv U:\self(\Ccat)\leftrightarrows \Dcat$ of simple CBPV where $UF$ is a commutative monad \cite{benton1996linear}.
\end{theorem}
\begin{proof}[Proof] First note that we can replace $\Ccat'$ with its completion $\Ccat$ under finite distributive coproducts. (Think of this as the completion under the notion of finite coproducts in the $2$-category of symmetric monoidal categories and lax symmetric monoidal functors. Similarly, we should think of an adjunction with strong monoidal left adjoint as corresponding with an adjunction in this $2$-category and a commutative monad as a monad in this $2$-category.) Indeed, we have a full and faithful embedding $\iota:\Ccat'\hookrightarrow \Ccat$ and because $\Dcat$ has finite distributive coproducts and $F'$ is strong monoidal, we can extend $F'\dashv U'$ to an adjunction $F\dashv U:\Ccat\leftrightarrows \Dcat'$ with strong monoidal $F$, where we define $F(\Sigma_{1\leq i\leq n} \iota(C_i)):=\Sigma_{1\leq i\leq n}F'C_i$ and $UD:=\iota(U'D)$:
\begin{align*}
\Ccat(\Sigma_{1\leq i\leq n}\iota (C_i),UD) )& :=
\Ccat(\Sigma_{1\leq i\leq n}\iota (C_i),\iota(U'(D)) ) \\
&\cong \Pi_{1\leq i\leq n} \Ccat(\iota(C_i),\iota(U'(D)))\\
&\cong \Pi_{1\leq i\leq n}\Ccat'(C_i,U'D)\\
&\cong \Pi_{1\leq i\leq n}\Dcat'(F'C_i,D)\\
&\cong \Dcat'(\Sigma_{1\leq i\leq n} F'C_i,D)\\
&=: \Dcat'(F \Sigma_{1\leq i\leq n} \iota(C_i),D).
\end{align*}

We define the indexed category $\Ccat^{op}\ra{\Dcat}\Cat$ as having the same objects as $\Dcat'$ in each fibre and morphisms $\Dcat(A)(B,C):=\Dcat'(B,FA\multimap  C)$. To see that the monad is commutative, we note that a commutative monad is the same as a lax symmetric monoidal monad \cite{kock1972strong}.
\end{proof}
In this way, we can see that linear logic describes certain commutative effects. CBPV models for possibly non-commutative effects can be obtained from any monad model \cite{moggi1991notions} of the monadic metalanguage \cite{levy2012call}.
\begin{theorem}Any bicartesian closed category $\Ccat$ with a strong monad $T$ gives rise to a  CBPV model.
\end{theorem}
\begin{proof} \mccorrect{It is well-known that the forgetful functor $\Ccat^T\to \Ccat{}$ creates finite products (limits).} Recall that in this setting the Eilenberg-Moore category $\Ccat^T$ has Kleisli exponentials, in the sense of algebras $A\functype k$ of homomorphisms from free algebras $\mu_A$ to general algebras $k$ ($A\Rightarrow Uk$ inherits a $T$-algebra structure from $k$) \cite{moggi1991notions}. We define the indexed category $\Ccat^{op}\ra{\Dcat}\Cat$ to have the same objects as the Eilenberg-Moore category $\Ccat^T$ in each fibre and morphisms $\Dcat(A)(k,l):=\Ccat^T( k,A\functype l)$. $F\dashv U$ is interpreted by the usual Eilenberg-Moore adjunction.
\end{proof}

\subsection{Operational Semantics}\label{sec:simplop}
Importantly, CBPV admits a natural operational semantics that, for terms of ground type, reproduces the usual operational semantics of CBV and CBN under the specified translations into CBPV \cite{levy2012call} and that can easily be extended to incorporate various effects that we may choose to add to pure CBPV. We very briefly discuss this.

First, we note that Levy chooses to only provide an operational semantics for computations without \emph{complex values}. Complex values are defined to be values containing $\mathsf{pm}\;\;\mathsf{as}\;\;\mathsf{in}\;\;$- and $\lbi{}{}{}$-constructs. He does this as complex values introduce arbitrary choices into the operational semantics, as we need to decide  when to evaluate them. As the normalization of  values does not produce effects (in particular, values are equal to their normal form; they are static), all reasonable evaluation strategies for them are observationally indistinguishable and we could choose our favourite.

While excluding complex values from computations is not a terrible restriction (one can show that any computation is judgementally equal to one not having any complex values as subterms and the CBV and CBN translations do not produce any complex values), we do not see the need to introduce this restriction. Indeed, complex values will turn out to be useful in a dependently typed CBPV, when we want to substitute them in dependent types. For instance, we might want to define a dependent type through a case distinction.

For that purpose, let us point out that the $\beta$-reductions for complex values are directed versions (left-to-right) of their equations in the left hand column of figure \ref{fig:vceqs}. We use the parallel nested closure of $\beta$-reductions as our notion of reduction for values. Following the usual argument of logical relations \cite{tait1967intensional}, this gives us a strong normalization result for values. Let us write $\nf{V}$ for the normal form of a value $V$. We write $\nnf{V}$ to indicate a value  which is not in normal form.

We present a small-step operational semantics for CBPV computations in terms of a simple abstract machine that Levy calls the CK-machine. The configuration of such a machine consists of a pair $M,K$ where $\Gamma;\cdot\vdash M:\ct{B}$ is a computation and $\Gamma;{\nil}:\ct{B}\vdash K:\ct{C}$ is a compatible {stack}. We call $\ct{C}$ the type of the configuration. The idea is that transitions are defined on a pair of a computation and a stack, rather than simply on computations, to be able to correctly model the operational behaviour of sequencing and function application: we push parts of a computation to the stack if other parts need to be executed first before we can pop the stack and resume.

The initial configurations, transitions (which embody left-to-right-directed versions of the $\beta$-rules of our equational theory) and terminal configurations in the evaluation of a computation $\Gamma;\cdot\vdash M:\ct{C}$ on the CK-machine are specified by figure \ref{fig:ckmachine} where we use the following abbreviations for stacks\\
\resizebox{\linewidth}{!}{\parbox{1.1\linewidth}{ \begin{align*}
V::K &:= \lbi{{\nil}}{V\textquoteleft {\nil}}{K}\\ 
j::K &:= \lbi{{\nil}}{j\textquoteleft {\nil}}{K}\\
\toin{[\cdot]}{x}{M}::K &:= \lbi{{\nil_1}}{(\toin{{\nil_2}}{x}{M})}{K}.
\end{align*}}}
We recall the following basic results about this operational semantics from \cite{levy2012call,levy2006call}.
\begin{theorem}[Determinism, Strong Normalization and Subject Reduction] For every configuration of the CK-machine, at most one transition applies. No transition applies precisely when the configuration is terminal. Every configuration of type $\ct{C}$ reduces, in a finite number of transitions, to a unique terminal configuration of type $\ct{C}$.
\end{theorem}
\begin{proof}
The only real modification from \cite{levy2012call,levy2006call} is that our terms include complex values. It is well-known that in a pure simple type theory with projection products and coproducts, the reductions are strongly normalizing and satisfy subject reduction. This shows that the transitions for complex values do not break strong normalization or subject reduction. The distinction between values in normal form and those not in normal form ensures that determinism still applies.
\end{proof}
\begin{figure}[!tb]
\fbox{
\resizebox{\linewidth}{!}{
\parbox{1.2\linewidth}{
\textbf{Initial Configuration}\\
\begin{tabular}{lll}
$M$ &,& ${\nil}$
\end{tabular}\\
\\
\textbf{Transitions}\\
\begin{tabular}{lllllll}
$\lbi{\nnf{V}}{x}{M}$ &,& $K$& $\leadsto\hspace{20pt}$ & $\lbi{\nf{V}}{x}{M}$ &,& $K$ \\
$\lbi{\nf{V}}{x}{M}$ &,& $K$& $\leadsto\hspace{20pt}$ & $M[\nf{V}/x]$ &,& $K$ \\
$\lbi{M}{{\nil}}{L}$&,& $K$& $\leadsto\hspace{20pt}$ & $L[M/{\nil}]$ &,& $K$ \\
$ \toin{M}{x}{N} $ &,& $K$& $\leadsto$ & $ M$ &,& $\toin{[\cdot]}{x}{N}::K$ \hspace{10pt}\;\\
$   \return \nnf{V}     $ \hspace{20pt}&,& $K$& $\leadsto$ & $   \return \nf{V}       $ &,& $K$ \\
$   \return \nf{V}     $ \hspace{20pt}&,& $\toin{[\cdot]}{x}{N}::K$& $\leadsto$ & $   N[\nf{V}/x]       $ &,& $K$ \\
$\force\nnf{V} $ & , & $K$ & $\leadsto$ & $\force \nf{V}$ & , & $K$\\
$   \force\thunk M     $ &,& $K$& $\leadsto$ & $   M       $ &,& $K$ \\
$ \sipm{\nnf{V}}{i}{x}{M_i}      $ &,& $K$& $\leadsto$ & $  \sipm{\nf{V}}{i}{x}{M_i}       $ &,& $K$ \\
$ \sipm{\langle j,V\rangle}{i}{x}{M_i}      $ &,& $K$& $\leadsto$ & $   M_j[\nf{V}/x]       $ &,& $K$ \\
$ \upm{\nnf{V}}{M}       $ &,& $K$& $\leadsto$ & $\upm{\nf{V}}{M}         $ &,& $K$ \\
$ \upm{\langle\rangle}{M}       $ &,& $K$& $\leadsto$ & $ M         $ &,& $K$ \\
$  \sipm{\nnf{V}}{x}{y}{M}      $ &,& $K$& $\leadsto$ & $ \sipm{\nf{V}}{x}{y}{M}          $ &,& $K$ \\
$  \sipm{\langle V,W\rangle }{x}{y}{M}      $ &,& $K$& $\leadsto$ & $  M[V/x,W/y]        $ &,& $K$ \\
$ j\textquoteleft M       $ &,& $K$& $\leadsto$ & $       M   $ &,& $j::K$ \\
$ \lambda_i M_i       $ &,& $j::K$& $\leadsto$ & $    M_j      $ &,& $K$ \\
$ \nnf{V}\textquoteleft M       $ &,& $K$& $\leadsto$ & $     \nf{V}\textquoteleft M    $ &,& $K$ \\
$ \nf{V}\textquoteleft M       $ &,& $K$& $\leadsto$ & $      M   $ &,& $\nf{V}::K$ \\
$ \lambda_x M       $ &,& $V::K$& $\leadsto$ & $   M[V/x]       $ &,& $K$ 
\end{tabular}\\
\\
\textbf{Terminal Configurations}\\
\begin{tabular}{lll}
$\return \nf{V}$ &,& ${\nil}$\\
$\lambda_i M_i$ &,& ${\nil}$\\
$\lambda_x M$ &,& ${\nil}$\\
$\force\nf{V}^{x'}$ &,& $ K$\\
$\sipm{\nf{V}^{x'}}{i}{x}{M_i}$ &,& $K$\\
$\upm{\nf{V}^{x'}}{M}$ &,& $K$\\
$\sipm{\nf{V}^{x'}}{x}{y}{M}$ &,& $K$
\end{tabular}
}
}
}
\caption{\label{fig:ckmachine} The behaviour of the CK-machine in the evaluation of a computation $\Gamma;\cdot\vdash M:\ct{C}$. We write $\nf{V}^{x'}$ for a non-canonical normal form of a value which has at least one free identifier $x'$. Every time we encounter a computation term former taking a value as an argument, we first normalize the value before proceeding to the corresponding transition for the term former. We leave out type annotations.}
\end{figure}

\subsection{Adding Effects}\label{sec:simpleff}
So far, we have considered pure CBPV computations. Next, we add effects to them, making them into real dynamic objects in the sense that their reductions might not respect equality. We recall by example how one adds effects to CBPV. Figure \ref{fig:effects} gives some examples of effects one could consider, from left to right, top to bottom:  divergence, recursion, printing an element $m$ of some monoid $\mathcal{M}$, erratic choice from finitely many alternatives, errors $e$ from some set $E$, writing a global state $s\in S$ and reading a global state to $s$. We note that the framework fits many more examples like probabilistic erratic choice, local references and control operators \cite{levy2012call}.
\begin{figure}[!tb]
\fbox{
\resizebox{\linewidth}{!}{
\begin{tabular}{llll}
\AxiomC{}
\UnaryInfC{$\Gamma;\cdot\vdash \diverge :\ct{B}$}
\DisplayProof
&\AxiomC{$\Gamma,z:U\ct{B};\cdot\vdash M : \ct{B}$}
\UnaryInfC{$\Gamma;\cdot\vdash \mu_z M : \ct{B}$}
\DisplayProof &

\AxiomC{$\Gamma;\cdot\vdash M:\ct{B}$}
\UnaryInfC{$\Gamma;\cdot\vdash \print{m}M:\ct{B}$}
\DisplayProof
&
\AxiomC{$\{\Gamma;\cdot\vdash M_i:\ct{B}\}_{1\leq i\leq n}$}
\UnaryInfC{$\Gamma;\cdot\vdash \nondet{i}{M_i}:\ct{B}$}
\DisplayProof  \\
&
&
&\\
&&&\\
\AxiomC{}
\UnaryInfC{$\Gamma;\cdot\vdash\error{e} :\ct{B}$}
\DisplayProof
& &
\AxiomC{$\Gamma;\cdot\vdash M:\ct{B}$}
\UnaryInfC{$\Gamma;\cdot\vdash\writecell{s}M:\ct{B}$}
\DisplayProof
&
\AxiomC{$\{\Gamma;\cdot\vdash M_s:\ct{B}\}_{s\in S}$}
\UnaryInfC{$\Gamma;\cdot\vdash \readcell{s}{M_s}:\ct{B}$}
\DisplayProof

\end{tabular}
}
}
\caption{\label{fig:effects} Some examples of effects we could add to CBPV. $\mu_z$ is a name binding operation that binds the identifier $z$ and $\nondet{i}{}$ and $\readcell{s}{}$ bind the indices $i$ and $s$ respectively.}
\end{figure}

\begin{figure}[!tb]
\fbox{
\resizebox{\linewidth}{!}{
\begin{tabular}{l}
\textbf{Transitions}\\
\begin{tabular}{lllllll}
$\diverge$ &,& $K$& $\leadsto\hspace{20pt}$ & $\diverge$ &,& $K$\\
$\mu_z M $ & , & $K$ & $ \leadsto\hspace{20pt}$ & $M[\thunk \mu_z M /z]$ & , & $K$\\
$\nondet{i}{M_i}$ & , & $K$ & $ \leadsto\hspace{20pt}$ & $M_j$ & , & $K$\\
\end{tabular}\\
\\
\textbf{Terminal Configurations}\\
\begin{tabular}{lll}
$\error e$ & , & $K$ 
\end{tabular}
\end{tabular}\hspace{220pt}\;
}
}
\caption{\label{fig:opsemdivs} The operational semantics for divergence, recursion, erratic choice and errors.}
\end{figure}
The small-step semantics of divergence, recursion, erratic choice and errors can easily be explained on our CK-machine as it is. This is summed up in figure \ref{fig:opsemdivs}. For the operational semantics of printing and state, we need to add some hardware to our machine. For that purpose, a configuration of our machine will now consist of a quadruple $M,K,m,s$ where $M,K$ are as before, $m$ is an element of our printing monoid $(\mathcal{M},\epsilon,*)$ which models some channel for output and $s$ is an element of our finite pointed set of states $(S,s_0)$ which is the current value of our storage cell. We lift the operational semantics of all existing language constructs to this setting by specifying that they do not modify $m$ and $s$, that terminal configurations can have any value of $m$ and $s$ and that initial configurations always have value $m=\epsilon$ and $s=s_0$ for the fixed initial state $s_0$. Printing and writing and reading the state can now be given the operational semantics of figure \ref{fig:opsemprint}.
 
\begin{figure}[!tb]
\fbox{
\resizebox{\linewidth}{!}{
\begin{tabular}{l}
\textbf{Transitions}\\
\begin{tabular}{lllllllllllllll}
$\print n M$ &,& $K$& ,& $m$& , & $s$ & $\leadsto\hspace{20pt}$ & $M$ &,& $K$ & ,& $m*n$& , & $s$\\
$\writecell {s'} M$ &,& $K$& ,& $m$& , & $s$ & $\leadsto\hspace{20pt}$ & $M$ &,& $K$ & ,& $m$& , & $s'$\\
$\readcell {s'} {M_{s'}}$ &,& $K$& ,& $m$& , & $s$ & $\leadsto\hspace{20pt}$ & $M_s$ &,& $K$ & ,& $m$& , & $s$\\
\end{tabular}
\end{tabular}\hspace{80pt}\;
}
}
\caption{\label{fig:opsemprint} The operational semantics for printing and writing and reading global state.}
\end{figure} 
We can try to extend the results of the previous section to this effectful setting and indicate when they break \cite{levy2012call}.
\begin{theorem}[Determinism, Strong Normalization and Subject Reduction] Every transition respects the type of the configuration. No transition occurs precisely if we are in a terminal configuration. In absence of erratic choice, at most one transition applies to each configuration. In absence of divergence and recursion, every configuration reduces to a terminal configuration in a finite number of steps.
\end{theorem}

We can again translate effectful CBV and CBN$\lambda$-calculi into CBPV with the appropriate effects as is indicated in figure \ref{fig:transeff}.

\begin{figure}[!tb]
\fbox{
\resizebox{\linewidth}{!}{
\begin{tabular}{l|ll||ll|l}
\textbf{CBV Term} $M$ & \textbf{CBPV Term} $M^v$&\qquad\qquad\qquad && \textbf{CBN Term} $M$ & \textbf{CBPV Term} $M^n$\\
\hline
$\op{M} $& $\op{M^v}$ &&& $\op{M} $& $\op{M^n}$\\
$\mu_x M $ & $\mu_z (\toin{\force z}{x}{M^v})$ &&&$\mu_z M $ & $\mu_z M^n$ 
\end{tabular}}}
\caption{\label{fig:transeff} The CBV and CBN translations for effectful terms. $z$ is assumed to be fresh in the CBV translation $\mu_x M$. For our examples, $\op{-}$ ranges over ${\diverge}$, $\error{e}$, $\nondet{i}{-}$, $\print{m}{(-)}$, $\readcell{s}{-}$ and $\writecell{s}{(-)}$.}
\end{figure}

Let us write $M\Downarrow N,m,s$ for a closed term $\cdot;\cdot\vdash M:\ct{B}$ if $M,{\nil},\epsilon,s_0$ reduces to the terminal configuration $N,{\nil},m,s$. We call this the \emph{big-step semantics} of CBPV. Recall that, at least for terms of ground type, CBPV induces the usual operational semantics via the CBV and CBN translations \cite{levy2006call}.
\begin{theorem}
The big-step semantics for CBPV induces the usual CBV and CBN big-step semantics for terms of ground  type, via the respective translations.\end{theorem}

We list the basic equations we would typically demand for the effects we consider in figure \ref{fig:effeqn}. In addition to these general equations, we could include the usual specific equations from the algebraic theory for $\op{-}$ (like the lookup-update algebra equations for global state of Plotkin and Power \cite{plotkin2002notions}). In a dependently typed setting, we have to decide which effect specific equations to include as judgemental equalities, such that the type checker has to be able to decide them, and which to include as propositional equalities for manual reasoning by the user.

Although one could write down an equational theory for these effects and a corresponding categorical semantics, in which case one would obtain soundness and completeness properties for the CBV and CBN translations, we  choose not to do so here for reasons of space. For this, we refer the reader, for instance, to \cite{plotkin2002notions,levy2012call}. The important thing to note is that the CBV and CBN translations for effectful CBPV typically result in a broken $\eta$-law for function types and sum types respectively as is well-known from traditional CBV and CBN semantics of effectful type theories.

\begin{figure}
[!tb]
\fbox{
\resizebox{\linewidth}{!}{
\begin{tabular}{ll}\parbox{.55\linewidth}{
$K[\op{M}/{\nil}]=\op{K[M/{\nil}]}$} &
\parbox{.55\linewidth}{
$\mu_z M = M[\thunk\mu_zM/z]$
}
\end{tabular}}}
\caption{\label{fig:effeqn} For effects, we demand the basic equation defining the fixpoint combinator $\mu_z$ as well as algebraicity equations for all effects $\op{-}$ (in addition to the usual equational theory for the specific operations $\op{-}$, like the Plotkin-Power equations for global state). These algebraicity equations state that a stack $K$ is a homomorphism of the algebra defined by the operations $\op{-}$. For our examples, $\op{-}$ ranges over ${\diverge}$, $\error{e}$, $\nondet{i}{-}$, $\print{m}{(-)}$, $\readcell{s}{-}$ and $\writecell{s}{(-)}$.}
\end{figure}

\section{Linear Types}\label{sec:backll}
Linear logic was introduced by Girard in \cite{girard1987linear} as a resource sensitive refinement of intuitionistic logic, which was inspired by the structure present in certain models for system F. From a modern perspective, we can see the essence of linear logic, or rather that of its proof term calculus, the linear $\lambda$-calculus, to already be present in \cite{eilenberg1966closed}. Put simply, the linear $\lambda$-calculus provides an internal language for symmetric monoidal closed categories in the same way that the ordinary (simply-typed) $\lambda$-calculus does for cartesian closed categories. The system is resource sensitive in the sense that a possibly non-cartesian monoidal structure does not generally admit copying and deleting morphisms. This means that, in the corresponding logic or $\lambda$-calculus, we lose the structural rules of contraction and weakening. This results in an exposure of the frequency with which assumptions are used in proofs in logic and gives us a better grip on complexity in the $\lambda$-calculus.

To be precise, the logic that is arises from this linear $\lambda$-calculus via a Curry-Howard correspondence is referred to as \emph{(multiplicative) intuitionistic linear logic}. This system is strictly more general than the \emph{(multiplicative-additive-exponential) classical linear logic} studied by Girard. This latter system differs from the former in three significant ways.
\begin{enumerate}
\item It admits a \emph{classical} duality in the sense that there is a dualising object\footnote{That is, an object $\bot$ such that the canonical evaluation morphism $A\ra{}(A\Rightarrow \bot)\Rightarrow \bot$ is an isomorphism for all objects $A$.} $\bot$ for the implication $\multimap$. At the same time it still admits a non-trivial term calculus. This is one of the historically surprising aspects of the system, in the light of the Joyal lemma (see e.g. \cite{abramsky2009no}), which states that a cartesian closed category with a dualising object is a preorder.
\item It comes equipped with a comodality (that is, a $\Box$-modality) $!$, called the \emph{exponential}, which recovers the structural rules.
\item It comes equipped with an additional notion of conjunction, called the \emph{additive conjunction}, written $\&$, to be contrasted with the multiplicative conjunction $\otimes$ from multiplicative intuitionistic linear logic. It represents an internal choice, rather than a simultaneous occurence of resources. Classical duality also gives us an \emph{additive disjunction} $\oplus$, which represents an external choice, and which, in absence of classical duality, we might choose to include in our linear logic as a primitive.
\end{enumerate}
It will be the level of greater generality of (multiplicative) intuitionistic linear logic, including the more specific cases of systems \`a la Girard, that we think of when we refer to \emph{linear logic}.

\subsection{Categorical Semantics}\label{sec:backllsem}
Intuitionistic linear logic admits a relatively simple, though not historically uncontroversial, sound and complete categorical semantics which we  describe here briefly. Our principal reference will be \cite{barber1996dual}. Some more background is provided in \cite{mellies2009categorical} and \cite{bierman1994intuitionistic}.

There are several notions of model in use that are equivalent, which only differ in their interpretation of $!$. It is clear that $!$ should be interpreted as a comonad, but the exact properties of those comonad can be stated in several equivalent ways. For our purposes, the notion of a linear/non-linear model of \cite{benton1995mixed} is the best fit.

\begin{definition}[Linear/Non-Linear Adjunction] By a linear/non-linear adjunction, we shall mean a lax symmetric monoidal adjunction (i.e. an adjunction in the 2-category of symmetric monoidal categories and lax symmetric monoidal functors)
\begin{diagram}
(\mathcal{C},1,\times) & \pile{\rTo^F \\ \bot \\ \lTo_U} & (\mathcal{D},I,\otimes)
\end{diagram}
from a symmetric monoidal category $\Dcat$ to a cartesian monoidal category $\mathcal{C}$. An equivalent condition for the adjunction $F\dashv U$ to be lax symmetric monoidal is for the functor $F$ to be strongly symmetric monoidal, in which case the symmetric oplax structure on $F$ transfers along the adjunction to a symmetric lax structure on $U$. 
\end{definition}

\begin{definition}[Model of Linear Logic] \label{def:llmodel}
A model of \mccorrect{intuitionistic} linear logic \mccorrect{with $I$- and $\otimes$-types} consists of a symmetric monoidal category $\Dcat$. The model supports...
\begin{itemize}
\item $\multimap$- types iff $\Dcat$ is closed (as a multicategory or as a symmetric monoidal category in case we have $I$- and $\otimes$-types);
\item $\top$- and $\&$-types iff $\Dcat$ has finite products;
\item $0$- and $\oplus$-types iff $\Dcat$ has distributive coproducts (or coproducts in the multicategorical sense);
\item $!$-types iff $\Dcat$ is equipped with a comonad $!$ which arises as $FU$ for a linear/non-linear adjunction $F\dashv U: (\Ccat,1,\times)\leftrightarrows (\Dcat,I,\otimes)$.
\end{itemize}
\end{definition}

\subsection{Syntax}
As is the case for its categorical semantics, there are many different roughly equivalent syntactic proof calculi for linear logic. In order to allow a natural generalization to dependent types (which most naturally come in a natural deduction formulation), we choose a calculus in natural deduction style, rather than a sequent calculus. Of the natural deduction formalisms for linear logic, the two most mature options are Barber and Plotkin's \emph{dual intuitionistic linear logic (DILL)} \cite{barber1996dual} and Benton's \emph{linear/non-linear (LNL) calculus} \cite{benton1995mixed}.

DILL chooses to work with a single typing judgement $\Gamma;\Delta\vdash b:B$ and is closer to Girard's original formulations of linear logic. It uses a dual context $\Gamma;\Delta$, however, which consist of a cartesian region $\Gamma$, in which the structural rules of weakening and contraction are valid, and a linear region $\Delta$, in which they are not. This separation of context should be seen as a metaoperation internalising $!$, which was missing from Girard's formulations, just as $\otimes$ internalises the comma in the context and $\multimap$ internalises the turnstyle $\vdash$. We should see DILL as an internal language for $\Dcat$ of definition \ref{def:llmodel}.

The LNL calculus, by contrast, should be seen as providing an internal language for both $\Ccat$ and $\Dcat$ (including their relationship through $F\dashv U$) of definition \ref{def:llmodel}. It adds on top of the linear typing judgement $\Gamma;\Delta\vdash b:B$ of DILL (which models the morphisms $\Dcat$) a cartesian typing jugement $\Gamma\vdash a:A$ to model the morphisms of $\Ccat$. This means we -- in particular -- have two kinds of types: linear types $B$ and cartesian types $A$. Here, $\Gamma$ consists of cartesian types and $\Delta$ of linear types\footnote{We can read a DILL context $\Gamma;\Delta$ as the LNL context $U\Gamma;\Delta$, where we apply $U$ to all types in $\Gamma$.}. In a model of the LNL calculus, $\Ccat$ is considered part of the structure of the model, while, for a model of DILL, we only demand the existence of a linear/non-linear adjunction to some cartesian category $\Ccat$.

We have chosen to work with DILL in this thesis and generalise it with a notion of type dependency (see chapter \ref{ch:4}), mostly because it is closer to what we believe most people understand to be the essence of linear logic and \mccorrect{because it} seems to be more widely used. However, the LNL calculus can be useful to keep in mind in order to see better how CBPV generalises linear logic to non-commutative effects in a sense. We choose not to elaborate on the syntax of linear logic, here, as  \cite{barber1996dual,benton1995mixed} are excellent references for the syntax of DILL and the LNL calculus, respectively.

\subsection{Girard Translations}\label{sec:girardtrans}
An important aspect of linear type theory is that we have two translations of cartesian type theory \mccorrect{(with commutative effects)} into it, called the Girard translations. It turns out that, from the  point of view of CBPV, these are simply the CBN and CBV translations, in disguise. Indeed, the point we like to stress in this thesis is that the LNL calculus is essentially the same as CBPV for commutative effects with the extra connectives of $\otimes$ and $\multimap$. It might seem as if DILL is a serious restriction in expressive power compared to the LNL calculus. In particular, from the point of view of CBPV, it seems as if DILL only allows the CBN translation of cartesian type theory (known as the first Girard translation in the context of linear logic) as we are missing  value types. However, the linear connectives $\otimes$ and $\multimap$ allow us to express the CBV translation purely in terms of computation types. Altough this was already known to Girard, he thought this second translation was ``not of much interest'' and stressed the importance of his first (CBN) translation \cite{benton1996linear}. It is precisely the absence of the connectives $\otimes $ and $\multimap$ for non-commutative effects (particularly the latter) that forces CBPV to consider two separate typing judgements, while linear logic can be formulated equally well using only one.

For completeness sake, figure \ref{fig:girardtranslations} shows the CBN and CBV translations of cartesian type theory into DILL, at least at the level of types. \mccorrect{Note that as we might be working with an effectful cartesian type theory, in general, we distinguish between product types with a pattern matching eliminator, which we denote $A_1\times \cdots \times A_n$, and product types with a projection eliminator, which we denote $\Pi_{1\leq i\leq n}A_i$.} We trust that the reader can fill in the definitions at the level of terms, from the corresponding translations for CBPV. While $(-)^f$ corresponds precisely to the CBN translation $(-)^n$ of CBPV, $(-)^s$ can be read as the adjoint of $(-)^v$ in a sense. Indeed, while $(-)^v$ sends a term $x_1:A_1,\ldots,x_n:A_n\vdash M:A$ to a term \mccorrect{$x_1:A_1^v,\ldots,x_n:A_n^v;\cdot \vdash M^v:FA^v$}, $(-)^s$ sends it to the equivalent term \mccorrect{$\cdot;x_1:A_1^s,\ldots,x_n:A_n^s\vdash   M^s:A^s$, where $A^s=FA^v$}. This equivalence follows because $F(A_1\times \cdots \times A_n)\cong FA_1\otimes \cdots \otimes FA_n$.

We would like to point out that the conventional Girard translations choose to use projection products for the CBN translation and pattern matching products for the CBV translation, to make sure that $\eta$\mccorrect{ }survives in either case. Note that there is no analogous way of salvaging the $\eta$-law for sum types in CBN. We also note that we only need additive conjunctions for the CBN translation of projection products.

Note that the usual practice of constructing models of cartesian type theory out of models of linear type theory $\Dcat$ by taking the co-Kleisli category $\Dcat_!$ for $!$ precisely is the semantic equivalent of Girard's first translation. As far as we are aware, there is no well-known categorical construction corresponding to Girard's second translatio\mccorrect{n}, perhaps because it relies on the specific properties of $!$ as a comonad (its compatibility with the monoidal structure on $\Dcat$).

\begin{figure}[!tb]
\fbox{
\resizebox{\linewidth}{!}{
\begin{tabular}{l||l|l}
\textbf{Cartesian Type} $A$ \hspace{20pt}\;& \textbf{CBN Translation} $A^f$ \hspace{20pt}\;& \textbf{CBV Translation} $A^s$\hspace{20pt}\;\\
\hline
$1$ & $I$ & $I$\\
$A_1\times A_2$ & $!A_1^f\otimes !A_2^f$ & $A_1^s\otimes A_2^s$\\ 
$\Pi_{1\leq i \leq n}A_i$& $A_1^f\&\ldots \& A_n^f$ & \mccorrect{$!A_1^s\otimes \ldots\otimes !A_n^s$}\\
$A_1\Rightarrow A_2$ & $ !A_1^f\multimap A_2^f$ & $!(A_1^s\multimap A_2^s)$\\
$0$ & $0$ & $0$\\
$A_1+A_2 $ & $!A_1^f\oplus !A_2^f$ & $A_1^s \oplus A_2^s$
\end{tabular}
}
}
\caption{\label{fig:girardtranslations} The definitions of the CBN and CBV translations of cartesian type theory into DILL, also known as the first and second Girard translation, respectively.}
\end{figure}

\subsection{Concrete Models}
\subsubsection{Commutative Computational Effects}In section \ref{sec:simplmod}, we saw that linear logic gives rise to certain models for CBPV for commutative effects. In fact, a following partial converse result can be obtained. A similar result was stated without proof or attribution\footnote{We believe the result on closure should be attributed to \cite{kock1971closed} while the construction of the symmetric monoidal structure might be inspired by the results of \cite{linton1969coequalizers}.} in \cite{benton1996linear}, but we provide a construction here, in order to generalize it later.

\begin{theorem}\label{thm:commtolinear}Let $\Ccat$ be a cartesian closed category with a commutative monad $T$, where $\Ccat$ additionally has equalisers and the Eilenberg-Moore category $\Ccat^T$ has reflexive coequalisers\footnote{In fact, \cite{keigher1978symmetric} provides an alternative construction for $\otimes$-types for which we demand instead all coequalisers in $\Ccat$.}. Then, $\Ccat^T$ is symmetric monoidal closed and has finite products to interpret additive conjunctions and the Eilenberg-Moore adjunction $F\dashv U$ defines a linear/non-linear adjunction.
\end{theorem}
\begin{proof}The statement about additive conjunctions follows from the well-known result that the forget functor from the Eilenberg-Moore category creates limits.

For two algebras $k,l\in\Ccat^T$, we define an object $k\homtype l$ of $\Ccat$ as the equaliser (which represents the subobject of morphisms satisfying the homomorphism equations)
\begin{diagram}
k\homtype l & \rInto^m & Uk\Rightarrow Ul & \pile{\rTo^{\lambda_{f:Uk\Rightarrow Ul} Tf;l} \\ \rTo_{\lambda_{f:Uk\Rightarrow Ul} k;f}} & TUk \Rightarrow Ul.
\end{diagram}
We note that we can equip $k\homtype l$ with a $T$-algebra structure if $T$ is a commutative monad. Indeed, using the morphism $T(A\Rightarrow B)  \ra{\phi_{A,B}} A\Rightarrow TB $ which exists for every strong monad, we can define the map 
\begin{diagram}
T(k\homtype l) & \rTo^{Tm} & T(Uk \Rightarrow Ul) & \rTo^{\phi_{Uk,Ul}} & Uk\Rightarrow TUl & \rTo^{\lambda_ff;l} & Uk\Rightarrow Ul.
\end{diagram}
Commutativity of the monad gives us that this map is equalising, so we obtain a unique factorisation of the map over $k\homtype l$, which gives us our algebra structure $k\multimap l$.

For $\otimes$-types note that we have the following reflexive fork in $\Ccat^T$, where we write $TA\times TB\ra{t_{A,B}}T(A\times B)$ for the left or right pairing for the strong monad $T$ (it doesn't matter which, as $T$ is commutative):
\begin{diagram}
F(Uk\times Ul) & \rTo^{F\langle \eta_{Uk},\eta_{Ul}\rangle} & F(UFUk\times UFUl) & \pile{\rTo^{F\langle k,l\rangle}\\\rTo_{F(t_{Uk,Ul});\epsilon_{F(Uk\times Ul)}} } & F(Uk\times Ul).
\end{diagram}
Taking the coequaliser of this fork gives us our interpretation of $k\otimes l$. Given morphisms $k\ra{\phi}k'$ and $l\ra{\psi}l'$, we easily see that we get natural transformations between the respective coequaliser diagrams (using the homomorphism laws and the naturality of $\epsilon$ and $t$), which, therefore, give us morphisms $k\otimes l\ra{\phi\otimes \psi}k'\otimes l'$. This is easily seen to make $\otimes$ a functor in each argument. Using the commutativity of $T$, we can show that it is, in fact, a bifunctor. Using the strength and pairing, we can always define a cocone on the coequaliser diagram which gives rise to an associator $k\otimes (l\otimes m)\ra{}(k\otimes l)\otimes m$. Next, observe that $FA\otimes FB\cong F(A\times B)$. To see this, note that for $k=FA$ and $l=FB$, we have that $F(UFA\times UFB)\ra{F(t_{A,B}); \epsilon_{F(A\times B)}}F(A\times B)$ is a cocone for the diagram above. Moreover, it is easily seen to have a section $F\langle \eta_A,\eta_B\rangle$, making it into a split epi. Given another cocone $\psi$ for the diagram, we can now define a factorisation over $F(t_{A,B});\epsilon_{F(A\times B)}$ by $F(\langle \eta_A,\eta_B\rangle);\psi$, which is unique as our cocone is an epi. Similarly, we can see that $F1\otimes k\cong k\cong k\otimes F1$, showing that $I:=F1$ makes $\otimes$ into a monoidal structure on $\Ccat^T$. We can further note that commutativity of the monad means that the braiding of $\times$ gives us a cocone on the coequaliser diagram which, gives rise to a braiding for $\otimes$, which inherits to property of being involutive from the braiding of $\times$. The triangle, pentagon and hexagon identities all follow from the universal property of the coequaliser defining $\otimes$. We conclude that $\otimes$ is a symmetric monoidal structure.

Using the universal property of the coequaliser defining $\otimes$ as well as the naturality of $t$ and $\epsilon$, the definition of a $T$-homomorphism and the universal property of the equaliser defining $\multimap$, it is a straightforward calculation to establish that $B \otimes (-)\dashv B \multimap (-)$.
\end{proof}

The condition that $\Ccat^T $ has reflexive coequalisers can, of course, be reduced to $\Ccat$ having such coequalisers and $T$ preserving them. This happens, for instance, when $T$ is induced by a finitary algebraic theory, as finite powers in a cartesian closed category preserve \mccorrect{reflexive coequalisers}  \cite{johnstone2002sketches} \mccorrect{(section D5.3)}.

\begin{remark}[A Linear Logic for Non-Commutative Effects?] In the light of theorem \ref{thm:commtolinear}, it is tempting to wonder if we can define a similar, perhaps non-commutative, linear logic to describe non-commutative computational effects. It is clear that theorem \ref{thm:commtolinear} would not straightforwardly generalise to a non-commutative setting, however, as it has been shown in \cite{foltz1980algebraic} that none of the categories of magmas, monoids, groups and rings admit a monoidal biclosed structure. At the same time, they arise as Eilenberg-Moore categories for a (strong) monad on a complete cocomplete cartesian closed category ($\Set$). In fact, \cite{kock2012commutative} shows that for a strong monad $T$, $k\homtype l$ above is a subalgebra of $Uk \Rightarrow Ul$ (which is the carrier of the Kleisli exponential $Uk\functype l$, which is well-known to exist as a $T$-algebra) for all $T$-algebras $k$ and $l$ if and only if $T$ is commutative. The construction above, inspired by \cite{linton1969coequalizers}, however, does yield a suitable (not necessarily symmetric, non-biclosed) premonoidal  structure (see \cite{power1997premonoidal}) on on categories of algebras (with reflexive coequalisers) for arbitrary strong monads.
\end{remark}

In fact, if we do not have the appropriate limits and colimits, we can always extend our model to incorporate them.
\begin{theorem}\label{thm:commembedsinll}
Every cartesian closed category $\Ccat$ with a commutative monad $T$ embeds fully and faithfully into a model of intuitionistic exponential additive multiplicative linear logic inducing the monad $T$.
\end{theorem}
\begin{proof}[Proof]Let $\widehat{\Ccat}$ be the category of presheaves on $\Ccat$ (its cocompletion). We note that the Yoneda embedding defines a strict 2-functor from the 2-category of categories to the 2-category of cocomplete categories (computing its action on morphisms by taking Yoneda extensions). In fact, using the Day convolution \cite{day1970closed}, it defines a strict 2-functor from the 2-category of symmetric monoidal categories (with lax symmetric monoidal functors and symmetric monoidal natural transformations) \mccorrect{$\mathfrak{SMCat}$} to the (sub-) 2-category of cocomplete symmetric monoidal categories \mccorrect{$\mathfrak{cCSMCat}$}. Noting that a commutative monad is precisely the same as a monad in \mccorrect{$\mathfrak{SMCat}$} \mccorrect{(Proposition 20 in} \cite{benton1995mixed}) and that 2-functors preserve monads, we get a \mccorrect{monad in $\mathfrak{cCSMCat}$ which is a} cocontinuous commutative monad $\widehat{T}$ on $\widehat{\Ccat}$, which restricts to $T$ on $\Ccat$. Note that $\widehat{\Ccat}$ is bicartesian closed with equalisers and coequalisers (a topos even) and that $\widehat{T}$ preserves colimits (in particular, reflexive coequalisers), so $\widehat{\Ccat}^{\widehat{T}}$ has reflexive coequalisers. Therefore, we can apply theorem \ref{thm:commtolinear} for the result that $\widehat{\Ccat}^{\widehat{T}}\leftrightarrows \widehat{\Ccat}$ defines a model of intuitionistic exponential additive multiplicative linear logic.\end{proof}

\begin{remark}We see that intuitionistic linear logic almost precisely describes all commutative effects. Still, this point of view does not seem to be widely held. Perhaps this is due to the fact that in the (initial) syntactic model of linear logic, a so-called principle of uniformity of threads (called such because it implies the usual principle $\Dcat(!A,B)\cong \Dcat(!A,!B)$) holds \cite{Abramsky00axiomsfor}: the unit of the adjunction $F\dashv U$ is an isomorphism $\id_\Ccat\cong UF=:T$. In this sense, the free linear logic model does not describe any effects from the monadic point of view. All its interesting information is contained in the comonad $!:=FU$ of the adjunction\mccorrect{.}
\end{remark}

\subsubsection{Scott Domains and Strict Functions}\label{sec:scottdom}
A simple model of linear type theory (with recursion, in the sense of a model of CBPV with recursion) can be built from Scott domains and strict functions. Here, $\Dcat$ has as objects Scott domains (i.e. bounded complete, directed complete, algebraic cpos) and as morphisms strict (preserving the bottom element $\bot$) continuous (preserving directed colimits) functions between them. If we define $\Ccat$ to be the category of Scott predomains (i.e. countable disjoint unions of Scott domains) and continuous functions, we can note that the inclusion $U$ of $\Dcat$ into $\Ccat$ has a left adjoint $F$ which adjoins a new bottom element. $\Dcat$\mccorrect{ }is easily seen to have a terminal object $\top$ (the one-element domain) and binary products $A\& B$ (the set-theoretic product, equipped with the product order $\langle a,b\rangle\leq \langle a',b\rangle := a\leq a \;\wedge\; b\leq b'$). The same is true for $\Ccat$ where we write the cartesian structure $(1,\times)$. $\Dcat$ also supports $\multimap$-types, where $A\multimap B$ is the set of strict continuous function from $A$ to $B$ under the pointwise order ($f\leq g:= \forall_{x\in A} f(x)\leq g(x)$). We can note that $A\multimap -$ has a left adjoint $A\otimes -$ which gives rise to a symmetric monoidal closed structure on $\Dcat$: $A\otimes B$ is defined as the smash product $\{\langle a,b\rangle \in UA \& UB \;|\;a\neq \bot\;\wedge \;b\neq \bot\}\cup\{\bot\}$. This has a unit $I:=\{\bot\leq \top\}$.  Note that this monoidal structure makes $F$ into a strong symmetric monoidal functor. We see that we have a model of linear logic with $\top,\&,I,\otimes,\multimap$ and $!$-types.

\subsubsection{Coherence Spaces}
One of the most canonical kinds of denotational semantics of linear logic -- and, in fact, the original motivation for Girard to introduce linear connectives -- is found in stable domain theory and its linear decomposition through coherence spaces. Imposing the property of stability on top of continuity can be seen as taking a step closer to (what is definable in) the syntax of a functional language with recursion.

We briefly recall some of the definitions in Girard's coherence space model of (classical) linear logic. The model is given by the category $\Coh$ of coherence spaces and cliques. Its objects are coherence spaces $(A,\coh_A)$ (or, undirected graphs): a set $A$ of tokens with a reflexive relation $\coh_A$, called the coherence relation. We write $\scoh_A$ for the irreflexive part of $\coh_A$, $\sincoh_A$ for the complement of $\coh_A$, and $\incoh_A$ for the complement of $\scoh_A$. Before we define the morphisms of the category, we describe a few operations on objects.

Given a coherence space $A$, we define its {linear negation} $A^\bot$ as the space with the same underlying set $A$ of tokens and coherence relation $\incoh_A$:
$$a\coh_{A^\bot}a':=\neg(a\scoh_A a').$$

Given coherence spaces $A$ and $B$, we define their {multiplicative conjunction} $A\otimes B$ as having underlying set the product $A\times B$ of the underlying sets of $A$ and $B$ and coherence relation
$$(a,b)\coh_{A\otimes B}(a',b'):=a\coh_A a'\wedge b\coh_B b'.$$
We can then define their {multiplicative disjunction} $A\parr B$, through De Morgan duality,
$$A\parr B:=(A^\bot\otimes B^\bot)^\bot, $$
and their {(multiplicative) linear implication} $A\multimap B$,
$$A\multimap B:= A^\bot\parr B.$$
(Explicitly, $(a,b)\scoh_{A\parr B}(a',b'):=a\scoh_A a'\vee b\scoh_B b'$ and $(a,b)\scoh_{A\multimap B}(a',b'):=a\coh_A a'\Rightarrow b\scoh_B b'$.)

We can define their {additive disjunction} $A\oplus B$ as the disjoint union of coherence spaces, where never $a\coh_{A\oplus B} b$ if $a\in A$ and $b\in B$ and $\coh_{A\oplus B}$ restricts to $\coh_A$ and $\coh_B$, and we can define their {additive conjunction} $A\& B$, through De Morgan duality, as
$$A\& B:=(A^\bot\oplus B^\bot)^\bot.$$
(Explicitly, always $a\coh_{A\&B}b$ for $a\in A$, $b\in B$ and $\coh_{A\& B}$ restricts to $\coh_A$ and $\coh_B$.)

The operations $\otimes$, $\parr$, $\oplus$, and $\&$ also have neutral elements which we shall denote by $I$, $\bot$, $0$, and $\top$, respectively. Indeed, $I=\bot=\{*\}$ and $0=\top=\emptyset$ are easily seen to do the trick. (These identities between the units can be seen as degeneracies of this model of linear type theory, as they do not follow from the syntax of classical linear logic.)

We now define the morphisms: $$\Coh(A,B):=\cliques(A\multimap B),$$
where a clique $\sigma$ in $A$ is a subset $\sigma\subseteq A$ such that $a,a'\in\sigma\Rightarrow a\coh_A a'$. We compose cliques as relations, which gives us the identity relations (which are cliques!) as the identities of our category.

We note that $(I,\otimes,\multimap)$ make $\Coh$ into a symmetric monoidal closed category, that $\top$ and $\&$ are our nullary and binary products, and that $0$ and $\oplus$ are our nullary and binary (distributive) coproducts. We note that we have obtained a model of linear type theory with $I$-, $\otimes$-, $\multimap$-, $\top$-, $\&$-, $0$-, and $\oplus$-types. In fact, as $((-)^{\bot})^\bot\cong \id_{\Coh}$, we even have a model of classical linear type theory. \cite{mellies2009categorical}

We have a linear/non-linear adjunction between the category $\Stable$ of Scott predomains with pullbacks\mccorrect{ }and continuous stable functions\footnote{Recall that a function $D'\ra{f}D$ is called continuous if it preserves directed suprema and that it is called stable if it preserves all pullbacks: $d_0,d_1\leq d_\top$ and $d_0\wedge d_1$ exists implies that $f(d_0\wedge d_1)=f(d_0)\wedge f(d_1)$.} and the category $\Coh$ of coherence spaces, $F\dashv U$:
\begin{diagram}(\Stable,1,\times) & \pile{\rTo^F\\\bot\\ \lTo_U} & (\Coh,I,\otimes).
\end{diagram}
$U$ takes the domain of cliques on objects and sends a clique $\sigma$ in $A\multimap B$ to the continuous stable function $d\mapsto \{b\;|\; \exists_{a\in d}(a,b)\in\sigma\}$. $F$ sends a predomain $D$ to the coherence space with set of tokens the compact elements of $D$ and coherence relation $s\coh_{FD}t:=\exists_{u\in D}(s\leq u)\wedge (t\leq u)$ and sends a continuous stable function $D'\ra{f}D$ to the clique $\{(x,y)\;| \; y\leq f(x)\wedge \forall_{x'\leq x}y\leq f(x')\Rightarrow x=x'\}$. (Note that $F$ is a strong monoidal functor.) We have the following bijection of homsets, which demonstrates the adjunction,
\begin{diagram}
\sigma & \rMapsto & d\mapsto  \{c \; |\exists_{d'\leq d}(d',c)\in\sigma  \;\} \\
\Coh(FD,C) &\pile{\rTo^{\fun}\\ \cong \\ \lTo_{\trace}} & \Stable(D,UC)\\
\{(d,c)\; |\; c\in f(d)\wedge \forall_{d'\leq d}c\in f(d')\Rightarrow d'=d\} & \lMapsto & f.
\end{diagram}
This induces a comonad $!:=FU$ on $\Coh$. Explicitly, $!A$ has set of tokens $\fincliques(A)$ and coherence relation
$$s\coh_{!A} s':= (s\cup s'\in !A).$$
This shows that that our model $\Coh$ of linear logic additionally supports $!$-types.

\subsubsection{AJM-Games}
In section \ref{sec:backgame}, we introduce another important class of models of linear logic: categories of games and strategies. We encourage the reader to think of game semantics as giving a further decomposition of coherence space semantics, which itself gave a decomposition of domain semantics by interpreting domain elements as sets of tokens. Indeed, it replaces tokens with (even length) plays in a game which are built up as a sequence of moves. Strategies (certain sets of plays) will then play the r\^ole of cliques. This can be formalised as the statement that there is a (faithful) forgetful functor from the category of CBN AJM-games\footnote{\begin{mccorrection}
In this thesis, we work with AJM-games because, in order to interpret dependent types, we need the extra option of explicitly restricting plays by defining the set $P_A$, rather than being forced to work with all legal positions, as we would be in HO-games. Indeed, \cite{levy2012call} has shown that all HO-games can be defined from a simple type system with product and coproduct types as well as type level recursion. This shows that HO-games do not suffice to interpret more expressive types like dependent function types. 
\end{mccorrection}} to the category of coherence spaces, which sends a game to the coherence spac\mccorrect{e} with even length plays as tokens, where tokens are called coherent if they agree on $P$-moves. Strategies $\sigma:A\multimap B$ are then interpreted as the clique $\{(s\upharpoonright_A,s\upharpoonright_B)\;|\; s\in \sigma\}$. In \cite{calderon2010understanding}, it is shown that this functor can be made full if we work with a suitable category of coherence spaces with a partial order on the tokens (representing the idea that some plays extend others in time). In that sense, game semantics \mccorrect{is} coherence space semantics extended in time. In turns out this more fine-grained description provided by game semantics is enough to precisely pin down the functions that are definable in a functional language with recursion (and with various other effects).

\section{AJM Game Semantics}\label{sec:backgame}
The idea behind game semantics is to model a computation by an alternating sequence of interactions (the play) between a program (Player) and its environment (Opponent), following some rules specified by its (data)type (the game). In this translation, programs become Player strategies, while termination corresponds to a strategy being winning or beating all Opponents. The charm of this interpretation is that it not only fully captures the intensional aspects of a program but that it combines this with the structural clarity of a categorical model, thus interpolating between traditional operational and denotational semantics.

If we view a type theory as a logic rather than as a programming language, its game semantics formalises the idea of Socratic dialogues. The interpretation of a proposition can be thought of as the game of all formal debates about its validity, where Player argues in its favour and Opponent argues against it. In this view, a proof of a proposition gets interpreted as a winning strategy for Player. We see that proofs get interpreted by winning strategies, when giving a game semantics of a logic, while partial strategies are of interest too, for the game semantics of a programming language, as these model programs that do not always terminate.

We assume the reader has some familiarity with the basics of categories of AJM-games (contrasted with the other style of HO game semantics \cite{hyland2000full}) and ($\approx$-saturated\footnote{Note that this is a mild technical difference from the formalism of \cite{abramsky2000full}, where strategies are what we call skeletons, here, which are considered up to a partial equivalence relation induced by $\approx$. Both formalisms are equivalent as a class of skeletons up to this partial equivalence relation can precisely be identified with the unique strategy obtained by closing the plays of the skeleton under $\approx$.}) strategies, as described in \cite{abramsky2009game}, and will only briefly recall the definitions. We define a category $\Gamecat$ which has as objects AJM-games.

\begin{mccorrection}Let us fix some universal set of moves $\mathcal{M}$ with injective functions
\begin{align*}\mathcal{M}+ \mathcal{M}&\ra{+}\mathcal{M}\\
\mathcal{M}\times \mathcal{M}&\ra{\times}\mathcal{M}\\
\mathcal{M}\times \mathbb{N}&\ra{\times}\mathcal{M},
\end{align*}
say the set of ASCII strings with $[x\mapsto ``\mathtt{inl}(\textrm{''}++x++``)\textrm{''}, x\mapsto ``\mathtt{inr}(\textrm{''}++x++``)\textrm{''}]$, $\langle x,y\rangle \mapsto ``\langle\textrm{''}++x++``,\textrm{''}++y++``\rangle\textrm{''}$ and $\langle x,y\rangle \mapsto ``\langle\textrm{''}++x++``,\textrm{''}++y++``\rangle\textrm{''}$, to make sure that AJM-games (and later games with dependency) form a set.
\end{mccorrection}
\begin{definition}[Game] A \emph{game} $A$ is a tuple $(M_A,\lambda_A,\J_A,P_A,\approx_A,W_A)$, where
\begin{itemize}
\item \mccorrect{$M_A\subseteq \mathcal{M}$ is set of \emph{moves}};
\item \mbox{\begin{diagram}M_A & \rTo^{\textnormal{\scriptsize$\lambda_A=\langle\lambda_A^{OP},\lambda_A^{QA}\rangle$}} & \{O,P\}\times \{Q,A\}\end{diagram}} is a function which indicates if a move is made by \emph{Opponent ($O$) or Player ($P$)} and if it is a \emph{Question} ($Q$) or an \emph{Answer} ($A$), for which we write $\overline{O}=P$, $\overline{P}=O$ and $M_A^O:={\lambda_A^{OP}}^{-1}(O)$, $M_A^P:={\lambda_A^{OP}}^{-1}(P)$, $M_A^Q:={\lambda_A^{QA}}^{-1}(Q)$ and $M_A^A:={\lambda_A^{QA}}^{-1}(A)$;
\item  \mbox{\begin{diagram}M_A & \rPartial^{\J_A} & M_A\end{diagram}} is a partial function which indicates the \emph{justifier} of a move, with the properties
\begin{itemize}
\item[] (Well-Foundedness): $\J_A$ defines a well-founded forest in the sense that for each move $m\in M_A$ there is some number $k$ such that $\J_A^k(m)$ is undefined; such a move with an undefined justifier is called an \emph{initial move};
\item[] (Compatibility with $\lambda_A$): $P$-moves are justified by $O$-moves and vice-versa; answers $m$ are justified by questions $n$ (but not necessarily vice-versa); in this case, we say that $m$ \emph{answers} $n$.
\end{itemize}
 $\J_{A}$ will be used to enforce \emph{stack discipline} in strategies.
\item  $P_A\subseteq M_A^\circledast$ is a non-empty prefix-closed set of \emph{plays}, where $M_A^\circledast$ is the set of finite sequences of moves, with the properties
\begin{itemize}
\item[] (Initial Move): Opponent moves first;
\item[] (Alternation): Player and Opponent alternate in making a move;
\item[] (Linearity): Every move occurs at most once in a play;
\item[] (Justification): A move can only be played after its justifier.
\end{itemize}
\item $\approx_A$ is an equivalence relation on $P_A$, satisfying
\begin{enumerate}
\item[] (Compatibility with $\lambda_A$): $s\approx_A t\Rightarrow \lambda^*_A(s)=\lambda^*_A(t)$;
\item[] (Prefix-Closure): $s\approx_A t\; \wedge \;s'\leq s\; \wedge \;t'\leq t \; \wedge \; |s'|=|t'|\; \Rightarrow s'\approx_A t'$;
\item[] (Completeness): $s\approx_A t\wedge sa\in P_A\Rightarrow \exists_bsa\approx_A tb$.
\end{enumerate}
Here, $\lambda_A^*$ is the extension of $\lambda_A$ to sequences. The intuition is that $\approx_A$-equivalent plays represent the same computation performed using different threads. 
\item $W_A\subseteq P_A^\infty$ is a set of \emph{winning plays}, where $P_A^\infty$ is the set of infinite plays, i.e. infinite sequences of moves such that all their finite prefixes are in $P_A$, such that $W_A$ is closed under $\approx_A$ in the sense that
$$\left( s\in W_A\wedge t\notin W_A \right)\Rightarrow \exists_{s_0\leq s,t_0\leq t}|s_0|=|t_0|\wedge s_0\not\approx_A t_0. $$
The intuition is that Opponent is the one who caused interactions in $W_A$  to be infinite.
\end{itemize}
\end{definition}

Our notion of morphism will be defined in terms of strategies on games.

\begin{definition}[Strategy] A (Player) \emph{strategy on $A$} is a non-empty subset $\sigma\subseteq P_A^{\mathsf{even}}$ satisfying
\begin{itemize}
\item[] (Causal Consistency): $sab\in \sigma\Rightarrow s\in \sigma$;
\item[] (Representation Independence): $ s\in\sigma\;\wedge\; s\approx_A t\Rightarrow t\in\sigma$.
\end{itemize}
\end{definition}
We sometimes identify $\sigma$ with the subset of $P_A$ that is obtained as its prefix closure. Generally, we impose some more conditions on strategies.
\begin{definition}[Conditions on Strategies]
We call a strategy $\sigma$ on $A$ \emph{deterministic} if it satisfies
\begin{itemize}
\item[] (Determinacy): $sab,ta'b'\in \sigma\;\wedge\; sa\approx_A ta'\Rightarrow sab\approx_A ta'b'$.
\end{itemize}
We call it \emph{well-bracketed} if it satisfies
\begin{itemize}
\item[] (Well-Bracketing): If an answer is played, it is in response to (i.e. justified by) the pending question (i.e. the last unanswered question).
\end{itemize}
We call $\sigma$ \emph{history-free}, if there exists a non-empty causally consistent subset $\phi\subseteq \sigma$ (called a \emph{history-free skeleton}) such that
\begin{itemize}
\item[] (Uniformization): $\forall_{sab\in\sigma}s\in\phi\Rightarrow\exists !_{b'}sab'\in\phi$;
\item[] (History-Freeness 1): $sab,tac\in\phi\Rightarrow b=c$;
\item[] (History-Freeness 2): $\left(sab,t\in\phi\;\wedge\;ta\in P_A\right)\Rightarrow tab\in\phi$.
\end{itemize}
We call $\sigma$ \emph{winning} if it satisfies
\begin{itemize}
\item[] (Finite Wins): If $s$ is $\leq$-maximal in $\sigma$, then $s$ is $\leq$-maximal in $P_A$.
\item[] (Infinite Wins): If $s_0\leq s_1\leq \ldots$ is an infinite chain in $\sigma$, then $\bigcup_i s_i\in W_A$.
\end{itemize}
\end{definition}
The idea is that game semantics naturally models various effects (and indeed does so very precisely in the sense that full-abstraction results can be obtained): non-determinism \cite{harmer1999fully}, non-local control flow\footnote{If we drop the well-bracketing condition altogether, Laird showed that we allow for very wild kinds of control flow, which are customary in CBV but not CBN. To obtain a precise correspondence with a CBN language with the control operator $\mathsf{call/cc}$ we would still impose the weaker condition on strategy of being \emph{weakly well-bracketed}: we are allowed to answer any open question, where by open question we mean a question for which no more recent question has been answered already. This corresponds to a stack discipline in which we can not just pop the top element, but we can pop an element that is deeper in the stack with the rule that we have to discard all elements on top of it.} \cite{laird1997full,laird1999semantic}, local references of ground type\footnote{While in HO-games naturally model general references, it is not clear to the author if these can be modelled in AJM-style. Indeed, strategies on AJM-games (at least the simply typed hierarchy) automatically satisfy the so-called visibility condition as a consequence of the restrictions on valid plays (specifically the switching conditions). Visibility is known to be the semantic condition which corresponds to the exclusion of higher-order references. \cite{abramsky1998fully}} \cite{abramsky1996linearity} and recursion/non-termination \cite{abramsky2000full}. These four conditions on strategies respectively serve to exclude these four classes of effects. This idea has been dubbed the ``semantic cube'' by Abramsky \cite{abramsky1997game}, where the axes of the (hyper)cube correspond to various conditions one could impose on strategies.

We write $\str(A)$ for the cpo of strategies (satisfying our favourite selection of the four conditions above) on $A$ ordered under inclusion and write $\bot_A$ or simply $\bot$ for the strategy $\{\epsilon\}$. In the rest of this thesis, we assume strategies to satisfy all four conditions, unless specified otherwise explicitly. (However, all constructions and results, with the exception of completeness results, go through for any of these classes of strategies.)

We note that a history-free skeleton $\phi$ for a strategy $\sigma$ is induced by a partial function on moves and that it satisfies $\sigma=\{t \;|\; \exists_{s\in\phi}t\approx_A s\}$. A winning strategy is the semantic equivalent of a normalising or total term. It always has a response to any valid $O$-move. Furthermore, if the result of the interaction between a winning strategy and any (possibly history-sensitive) Opponent is an infinite play, then this is a member of the set of winning plays, capturing the idea that the infinite interaction is Opponent's fault.

Next, we define some constructions on games, starting with their symmetric monoidal closed structure.

\begin{definition}[Tensor Unit] We define the game $I:=(\emptyset,\emptyset,\emptyset,\{\epsilon\},\{(\epsilon,\epsilon)\},\emptyset)$.\end{definition}

\begin{definition}[Tensor] Given games $A$ and $B$, we define the game $A\otimes B$ by
\begin{itemize}
\item $M_{A\otimes B}:=M_A+M_B$;
\item $\lambda_{A\otimes B}:=[\lambda_A,\lambda_B]$;
\item $\J_{A\otimes B}:= \J_A + \J_B$
\item $P_{A\otimes B}:=\{s\in M_{A\otimes B}^\circledast\;|\;  s\upharpoonright_A \in P_A \wedge s\upharpoonright_B \in P_B  \}$;
\item $s\approx_{A\otimes B} t:= s\upharpoonright_A \approx_A t\upharpoonright_A \;\wedge \; s\upharpoonright_B\approx_B t\upharpoonright_B\;\wedge\; \forall_{1\leq i\leq |s|} (s_i\in M_B\Leftrightarrow t_i\in M_B)$;
\item $W_{A\otimes B} :=\left\{s\in P_{A\otimes B}^\infty\;|\;\left(s\upharpoonright_A\in P_A^\infty\Rightarrow s\upharpoonright_A\in W_A\right)\wedge\left( s\upharpoonright_B\in P_B^\infty \Rightarrow s\upharpoonright_B\in W_B\right)\right\}$.
\end{itemize}
\end{definition}

\begin{definition}[Linear Implication] Given games $A$ and $B$, and writing $\mathsf{init}_B\subseteq M_B$ for the set where $\J_B$ is undefined, we define the game $A\multimap B$ by
\begin{itemize}
\item $M_{A\multimap B}:=M_A\times \mathsf{init}_B + M_B$; we write $s\upharpoonright_A$ for the subsequence of moves in $M_A\times\mathsf{init}_B$ of the play $s$ where we further project to a sequence in $M_A$; 
\item $\lambda_{A\multimap B}:=[\overline{\lambda_A},\lambda_B]$, where $\overline{\lambda_A}(m,n):=\overline{\lambda_A(m)}$;
\item $\begin{array}{ll}\J_{A\multimap B}(m,n):=n & \mathsf{if}\;m\in \mathsf{init}_A\\
\J_{A\multimap B}(m,n):=\J_A(m) & \mathsf{if}\;m\in M_A\setminus \mathsf{init}_A\\ 
\J_{A\multimap B}(n):=\J_B(n) & \mathsf{if}\; n\in M_B 
\end{array}$;
\item $P_{A\multimap B}:=\{s\in M_{A\multimap B}^\circledast\;|\; s\upharpoonright_A \in P_A \wedge s\upharpoonright_B \in P_B \}$;
\item $s\approx_{A\multimap B} t:= s\upharpoonright_A \approx_A t\upharpoonright_A \;\wedge \; s\upharpoonright_B\approx_B t\upharpoonright_B\;\wedge\; \forall_{1\leq i\leq |s|} (s_i\in M_B\Leftrightarrow t_i\in M_B)$;
\item $W_{A\multimap B}:=\{s\in P_{A\multimap B}^\infty\;|\; s\upharpoonright_A\in W_A \Rightarrow s\upharpoonright_B\in W_B\}$.
\end{itemize}
\end{definition}

$I$ is the unique game whose only play has length $0$. Both $A\otimes B$ and $A\multimap B$ are obtained by playing $A$ and $B$ in parallel by interleaving. Note that the definitions of $P_{-}$ and $\lambda_{-}$ imply that in $A\otimes B$ only Opponent can switch between $A$ and $B$, while in $A\multimap B$ only Player can, and that in $A\otimes B$ Opponent can start the play in either $A$ or $B$, while in $A\multimap B$ the play must commence in $B$. In both cases, a question is answered in the game where it was asked.

These definitions on objects extend to strategies, e.g. for  strategies $\sigma\in\str(A), \tau\in\str(B)$, we can define a  strategy $\sigma\otimes \tau=\{s\in P_{A\otimes B}^\mathsf{even}\;|\; s\upharpoonright_A\in \sigma \;\wedge\; s\upharpoonright_B\in \tau\}\in\str(A\otimes B)$. This gives us a model of multiplicative intuitionistic linear logic, with all structural morphisms consisting of appropriate variants of copycat strategies, which are introduced next.

\begin{theorem}[Linear Category of Games] We define a category $\Gamecat$ by
\begin{itemize}
\item $\mathsf{ob}(\Gamecat):=\{A\;|\; A\;\textnormal{ is an AJM-game}\}$;
\item $\Gamecat(A,B):=\str(A\multimap B)$;
\item $\mathsf{id}_A:=\{s\in P_{A\multimap A}^\mathsf{even}\;|\; \forall_{s' \in P_{A\multimap A}^\mathsf{even}}s'\leq s \Rightarrow s'\upharpoonright_{A^{(1)}}\approx_A s'\upharpoonright_{A^{(2)}}\}$, the \emph{copycat strategy} on $A$;
\item for $A\ra{\sigma}B\ra{\tau}C$, the composition (or \emph{interaction}) $A\ra{\sigma;\tau}C$ is defined from parallel composition $\sigma||\tau:=\{s\in M_{(A\multimap B)\multimap C}^\circledast \;|\; s\upharpoonright_{A,B}\;\in\sigma\;\wedge \; s\upharpoonright_{B,C}\;\in\tau\}$ plus hiding: $\sigma;\tau:=\{s\upharpoonright_{A,C}\;|\; s\in \sigma||\tau\}$.
\end{itemize}
Then, $(\Gamecat,I,\otimes,\multimap)$ is, in fact, a symmetric monoidal closed category.
\end{theorem}

To make this into a model of intuitionistic logic, a cartesian closed category (ccc), through the first Girard translation (CBN translations), we need two more constructions on games, to interpret the additive conjunction and exponential, respectively.

\begin{definition}[With] Given games $A$ and $B$, we define the game $A\& B$ by
\begin{itemize}
\item $M_{A\& B}:=M_A+M_B$;
\item $\lambda_{A\& B}:=[\lambda_A,\lambda_B]$;
\item $\J_{A\& B}:= \J_A + \J_B$;
\item $P_{A\& B}:=P_A+P_B$;
\item $\approx_{A\& B} :=\approx_A+\approx_B$;
\item $W_{A\& B} :=W_A + W_B$.
\end{itemize}
\end{definition}

\begin{definition}[Bang] Given a game $A$, we define the game $!A$ by
\begin{itemize}
\item $M_{!A}:=\mathbb{N}\times M_A$;
\item $\lambda_{!A}(i,a):=\lambda_A(a)$;
\item $\J_{!A}(i,a):=(i,\J_A(a))$;
\item $P_{!A}:=\{s\in M_{!A}^\circledast\;|\; \forall_{i\in\mathbb{N}}s\upharpoonright_i\in P_A\}$; 
\item $s\approx_{!A} t:=\exists_{\pi\in S(\mathbb{N})}\forall_{i\in \mathbb{N}}s\upharpoonright_i\approx_A t\upharpoonright_{\pi(i)}\; \wedge \; ( \mathsf{fst};\pi)^*(s)=\mathsf{fst}^*(t)$, writing $S(\mathbb{N})$ for the set of permutations of $\mathbb{N}$;
\item $W_{!A}:= \{s\in P_{!A}^\infty\;|\; \forall_i s\upharpoonright_i\in P_A^\infty \Rightarrow s\upharpoonright_i\in W_A\}$.
\end{itemize}
\end{definition}

A play in $A\& B$ consists of either a play in $A$ or in $B$, where Opponent chooses which as by our convention Opponent always makes the initial move. A play in $!A$ consists of any number of interleaved \emph{threads} of plays in $A$. Because of the definition of $\approx_A$, $!A$ behaves as a countably infinite symmetric $\otimes$-product of $A$ with itself. As before, the definition of $\lambda_{!A}$ assures only Opponent can switch games, while the definition of justification ensures that a question is answered in the thread in which it was asked.

Next, we note that $!$ can be made into a comonad by defining, for $A\ra{\sigma}B$, 
$$!\sigma:=\{s\in P_{!A\multimap !B}^\mathsf{even}\;|\; \exists_{\pi\in S(\mathbb{N})} \forall_{i\in \mathbb{N}} s\upharpoonright_{(\pi(i),A), (i,B)}\in \sigma\},$$
and natural transformations $!A\ra{\mathsf{der}_A}A$ and $!A\ra{\delta_A}!!A$ where $$\mathsf{der}_A:=\{s\in P_{!A\multimap A}^\mathsf{even}\;|\; \forall_{s' \in P_{!A\multimap A}^\mathsf{even}}s'\leq s \Rightarrow \exists_{i\in\mathbb{N}} s'\upharpoonright_{!A}\upharpoonright_i\approx_A s'\upharpoonright_A\}\txt{and}$$ $$\delta_A:=\{s\in P_{!A\multimap !!A}^\mathsf{even}\;|\; \forall_{s' \in P_{!A\multimap !! A}^\mathsf{even}}s'\leq s \Rightarrow \exists_{p:\mathbb{N}\times\mathbb{N}\hookrightarrow\mathbb{N}} \forall_{i,j\in \mathbb{N}} s'\upharpoonright_{!A}\upharpoonright_{p(i,j)}\approx_A s'\upharpoonright_{!!A}\upharpoonright_i\upharpoonright_j \}.$$

This allows us to define the co-Kleisli category $\Gamecat_!$ of CBN games, which has the same objects as $\Gamecat$, while $\Gamecat_!(A,B):=\Gamecat(!A,B)$. Let us write $\mathsf{dom}(f)$ for the domain of a morphism $f$. We have a composition $(f,g)\mapsto f^\dagger;g$, where we write $f^\dagger:=\delta_{\mathsf{dom}(f)};!(f)$, for which the strategies $\mathsf{der}_A$ serve as identities. We can define finite products in $\Gamecat_!$ by $I$ and $\&$ and write 
\begin{align*}
\mathsf{diag}_A:=\Big\{s\in P_{!A\multimap (A\& A)}^\mathsf{even}\;|\; &\forall_{s' \in P_{!A\multimap (A\& A)}^\mathsf{even}}s'\leq s \Rightarrow \exists_{i\in\mathbb{N}}(s'=\epsilon)\; \vee \\ & (s'\upharpoonright_{!A}\upharpoonright_i\approx_A s'\upharpoonright_{A^{(1)}}\neq \epsilon)\vee (s'\upharpoonright_{!A}\upharpoonright_i\approx_A s'\upharpoonright_{A^{(2)}}\neq \epsilon)\Big\}
\end{align*}
for the diagonal $A\ra{\mathsf{diag}_A} A\& A$ in $\Gamecat_!$. Moreover, we have Seely-isomophisms $!I\cong I$ and $!(A\&B)\cong !A\otimes !B,$ so we obtain a linear/non-linear adjunction $\Gamecat\leftrightarrows\Gamecat_!$, hence a model of multiplicative exponential intuitionistic linear logic. In particular, by defining $A\Rightarrow B:=!A\multimap B$, we make $\Gamecat_!$ into a ccc. We write $\mathsf{comp}_{A,B,C}$ for the internal composition $((A\Rightarrow B) \;\&\; (B\Rightarrow 
C))\ra{} A\Rightarrow C$ in $\Gamecat_!$.

\begin{theorem}[Intuitionist Category of Games]$(\Gamecat_!,I,\&,\Rightarrow)$ is a ccc.\end{theorem}

\begin{remark}Note that for the hierarchy of cartesian types $A$ that are formed by operations $I$, $\&$ and $\Rightarrow$ from finite games (games $A$ with finite $P_A$),  winning strategies are the total strategies -- strategies which respond to any $O$-move -- for which infinite plays can only occur because Opponent opens infinitely many threads of the same game.
\end{remark}

So far, we have shown that $\Gamecat_!$ provides a model of {$\STTGame$} without finite inductive types. It is not clear that inductive types and (weak) coproducts are  supported in our large world $\Gamecat_!$ of all games if we work with history-free strategies\footnote{In fact, the condition of history-freeness is replaced with innocence  in \cite{mccusker1996games} which mends  this defect. \mccorrect{Note that history-free and innocent strategies coincide on simple types. This is no longer true in the presence of coproducts, as we can no longer encode the $P$-view in a thread index. In this case, innocent strategies give the right notion to obtain definability and full abstraction results} \cite{mccusker2012games}.}. To support these -- to be precise, in order to define the appropriate eliminators --, we restrict the games we consider.

We can in fact construct a model of all of {$\STTGame$} in a full subcategory $\Gamecat_!^{\mathrm{fin}1\times\Rightarrow}$ of $\Gamecat_!$ by giving a suitable interpretation to finite inductive types, which serve as the ground types for a type hierarchy built with $1$, $\times$ and $\Rightarrow$. For a set $X$, let us define ${X_*}$ to be a so-called \emph{flat game} with Player moves $X$ which are all justified by a single initial Opponent move $*$, $P_{{X_*}}=\{\epsilon,*\}\cup\{ *x\;| x\in X\}$ and $\approx_X=\{(s,s)\in P_X\times P_X\}$ where the initial move $*$ is a question and the moves from $X$ are answers. Let us interpret a finite inductive type $\{a_i\;|\; i\}$ as a finite flat game ${\{a_i\;|\; i\}_*}$. Let us write $\Gamecat_!{}^{\mathrm{fin}1\times\Rightarrow}$ for the full subcategory of $\Gamecat_!$ on the objects formed from finite flat games by $1$, $\times$ and $\Rightarrow$. Then, $\Gamecat_!^{\mathrm{fin}1\times\Rightarrow}$ is a model of {$\STTGame$}. Indeed, the interpretation of the introduction rule for $a_i$ is the strategy which answers $a_i$ to $*$, while the $\mathsf{case}$-eliminators for finite inductive types are interpreted inductively on the structure of the type $C$ we are eliminating into: the cases where $C$ is not a finite inductive type are defined through equations $1-\eta$, $\{a_i\;|\;i\}-\mathsf{Comm}-\langle -,-\rangle$ and $\{a_i\;|\;i\}-\mathsf{Comm}-\lambda$ of figure \ref{fig:stt}. In case $C$ is a finite inductive type $\{c_j\;|\;j\}$, we interpret $\mathsf{case}_{\{a_i\;|\;i\},C}$ as the winning history-free strategy on
 $ {\{a_i\;|\;i\}_*}\Rightarrow { {\{c_j\;|\;j\}_*}{}^{(1)}}\Rightarrow \ldots\Rightarrow {  {\{c_j\;|\;j\}_*}{}^{(n)}}\Rightarrow { {\{c_j\;|\;j\}_*}}$ which is given \mccorrect{by} the $\approx$-closure of the set of traces defined by the following partial function $f$ on moves:
$$*^{({ C})}\mapsto (0,*)^{(!{{\{a_i\;|\;i\}_*}})}\quad (0,a_i)\mapsto (0,*)^{(!{ C^{(i)}})}\quad (0,c_j)^{(!{ C^{(i)}})}\mapsto c_j^{({ C})}.
$$

One of the interesting aspects of game semantics are the strong correspondences that can often be established with the syntax we are modelling. In this case, we have the following very strong completeness result. Note that fullness of the interpretation is strictly stronger than completeness: it is a notion of completeness with respect to proofs rather than mere provability.
\begin{theorem}\label{thm:sttcompl}
The interpretation functor ${\STTGame}\ra{\sem{-}}\Gamecat_!$ is full and faithful and hence is an equivalence of categories to $\Gamecat_!^{\mathrm{fin}1\times\Rightarrow}$\mccorrect{, where we write $\STTGame$ for the syntactic category corresponding to the eponymous theory of section} \ref{sec:sttgame}.
\end{theorem}
\begin{proof}
This is a straightforward finitary total variation on the results of \cite{abramsky2000full} that can, in particular, be obtained as a special case of the results in \cite{Abramsky00axiomsfor}. We note that winning strategies are total and therefore the case of $\bot$ in the decomposition lemma does not occur. Moreover, we note that the iterated decomposition terminates if we start with a winning strategy, for reasons outlined in the proof of lemma \ref{lem:norm} (essentially because infinite plays are always Opponent's responsibility, we can assign a finite size to a strategy which shrinks under the decomposition).
\end{proof}

\begin{savequote}[8cm]
``Well! I've often seen a cat without a grin,'' thought Alice; ``but a grin
without a cat! It's the most curious thing I ever saw in all my life.''. 
  \qauthor{--- Lewis Carroll}
\end{savequote}

\chapter{\label{ch:4}Linear Dependent Type Theory}

Starting from Church's simply typed $\lambda$-calculus (or cartesian propositional type theory), two extensions \mccorrect{depart} in perpendicular directions: \vspace{-5pt}
\begin{itemize}
\item[$\bullet$] following the Curry-Howard propositions-as-types interpretation\mccorrect{,} \emph{dependent type theory} (DTT)  \cite{martin1998intuitionistic} extends the simply typed $\lambda$-calculus from a proof calculus of intuitionistic propositional logic to one for predicate logic;
\item[$\bullet$] \emph{linear logic} \cite{girard1987linear} gives a more detailed resource sensitive analysis, exposing precisely how many times each assumption is used in proofs.
\end{itemize}
\vspace{-5pt}

A combined \emph{linear dependent type theory} is one of the interesting directions to explore to gain a more fine-grained understanding of dependent type theory from a computer science point of view, explaining its flow of information. Indeed, many of the usual settings for computational semantics are naturally linear in character, either because they arise from a model of linear logic as $!$-co-Kleisli categories (coherence space and game semantics) or for more fundamental reasons (quantum computation). Relatedly, as we have seen, linear types naturally arise in the semantics of commutative effects.

Combining dependent types and linear types is a non-trivial task, however, and despite some work by various authors that we shall discuss, the precise relationship between the two systems remains poorly understood. The discrepancy between linear and dependent types is the following.\vspace{-5pt}
\begin{itemize}
\item[$\bullet$] The lack of structural rules in \emph{linear type theory} forces us to refer to each identifier precisely once -- for a sequent $x:A\vdash t:B$, $x$ occurs uniquely in $t$.
\item[$\bullet$] In \emph{dependent type theory}, types can have free identifiers -- $x:A\vdash B\;\type$, where $x$ is free in $B$. Crucially, if $x:A\vdash t:B$, $x$ may also be free in $t$.
\end{itemize}
\vspace{-7pt}\nopagebreak
What does it mean for $x$ to occur uniquely in $t$ in a dependent setting? Do we not count its occurrence in $B$? This point of view seems incompatible with universes, which play an important r\^ole in dependent type theory. If we do, however, the language seems to lose much of its expressive power. In particular, it prevents us from talking about constant types, it seems. 

The usual way out, which we shall follow too, is to restrict type dependency on cartesian terms, which can be copied and deleted freely. Although this seems very limiting -- for instance, we do not obtain full equivalents of the Girard translations, embedding DTT in the resulting system --, it is not clear that there is a reasonable alternative. Moreover, as even this limited scenario has not been studied extensively, we hope that a semantic analysis, which was so far missing entirely, may shed new light on the old mystery of linear type dependency.

Historically, Girard's early work in linear logic already makes movements to extend a linear analysis to predicate logic. Although it talks about first-order quantifiers, the analysis appears to have stayed rather superficial, omitting the identity predicates which, in a way, are what make first-order logic tick. Closely related is that an account of internal quantification, or a linear variant of Martin-L\"of's type theory, was missing, let alone a Curry-Howard correspondence.

Later, linear types and dependent types were first combined in a Linear Logical Framework \cite{cervesato1996linear}, where a syntax was presented that extends a Logical Framework with linear types (that depend on terms of cartesian types). This has given rise to a line of work in the computer science community \cite{dal2011linear,petit2012linear,gaboardi2013linear}. All the work seems to be syntactic in nature, however, and seems to be mostly restricted to the asynchronous fragment in which we only have $\multimap$-, $\Pi_!^\multimap$-, $\top$-, and $\&$-types. An exception is the Concurrent Logical Framework \cite{watkins2003concurrent}, which treats synchronous connectives resembling our $I$-, $\otimes$-, $\Sigma_!^\otimes$-, and $!$-types. An account of additive disjunctions and identity types is missing entirely.

On the other hand, similar ideas, this time at the level of categorical semantics and specific models (from homotopy theory, algebra, and physics), have emerged in the mathematical community  \cite{may2006parametrized,shulman2013enriched,ponto2012duality,schreiber2014quantization}. In these models, as with Girard, a notion of comprehension was missing and, with that, a notion of identity type. Although in the last while some suggestions have been made on the nLab and nForum of possible connections between the syntactic and semantic work, no account of the correspondence was published.

The point of this chapter is to close this gap between syntax and semantics and to pave the way for a proper semantic analysis of linear type dependency, treating a range of type formers including $\Id_!^\otimes$-types. Firstly, in section \ref{sec:syn}, we present a syntax, dependently typed dual intuitionistic linear logic (dDILL), a natural blend of the dual intuitionistic linear logic (DILL) \cite{barber1996dual} and dependent type theory (DTT) \cite{hofmann1997syntax} which generalises both. Secondly, in section \ref{sec:sem}, we present a complete categorical semantics, an obvious combination of linear/non-linear adjunctions \cite{barber1996dual} and comprehension categories  \cite{jacobs1993comprehension}. Thirdly, we discuss how our semantics applies to the dependently typed LNL calculus \cite{krishnaswami2015integrating} in section \ref{sec:deplnl} and discuss dependently typed Girard translations in section \ref{sec:depgirard}. Finally, in sections \ref{sec:dismod}, \ref{sec:ddillcommeffects}, \ref{sec:doublegluing}, \ref{sec:lindepscott} and \ref{sec:depcohsp} we present various concrete models, including a class of models arising from commutative effects and a coherence space semantics.

\begin{remark}[Related Publications] This chapter is largely based on \cite{vakar2015syntax,vakar2015framework}. The material on many (the exception being the monoidal families) of the concrete models is new as is the material on dependent Girard translations and the dependent LNL calculus semantics. Around the same time that the author published his study \cite{vakar2014syntax,
vakar2015syntax} of dDILL, \cite{krishnaswami2015integrating} independently developed a syntax (but not a denotational semantics) and applications for dLNL.
\end{remark}

\section{Syntax of dDILL}
\label{sec:syn}
We next present the formal syntax of dDILL\mccorrect{, c.f. section} \ref{sec:DTT}. We start with a presentation of its judgements and then discuss its rules of inference: first its structural core, then the logical rules for a series of optional type formers. We conclude this section with a few basic results about the syntax.

\subsubsection*{Judgements}
We adopt a notation $\Gamma;\Delta$ for contexts, where $\Gamma$ is `a cartesian region' and $\Delta$ is `a linear region', similarly to \cite{barber1996dual}. The idea will be that we have an empty context and can extend an existing context $\Gamma;\Delta$ with both cartesian and linear types that are allowed to depend on $\Gamma$. Our language will express judgements of the six forms of figure \ref{fig:ildttjudgements}.

\subsubsection*{Structural Rules}
We use the structural rules of figures \ref{fig:ildttstruct1}, \ref{fig:ildttstruct2} and \ref{fig:ildttstruct3}, which are essentially the structural rules of dependent type theory where some rules appear in both a cartesian and a linear form. We present the rules per group, with their names, from left-to-right, top-to-bottom.

\subsubsection*{Logical Rules}We introduce some basic (optional) type and term formers, for which we give type formation (denoted -F), term introduction (-I), term elimination (-E), term computation rules (-$\beta$), and (judgemental) term uniqueness principles (-$\eta$), in figure \ref{fig:ildttlogical1}, \ref{fig:ildttlogical2} and \ref{fig:ildtteqns}\mccorrect{.} Moreover, $\Sigma^\otimes_{!(x:A)}$, $\Pi^\multimap_{!(x:A)}$, $\lambda_{!(x:A)}$, and $\lambda_{x:A}$ are name binding operators, binding free occurences of $x$ within their scope. Preempting some theorems of the calculus, we overload some of the notation for -I and -E rules of various type formers, in order to avoid unnecessary syntactic clutter. Needless to say, uniqueness of typing can easily be restored by carrying around enough type information on the term formers corresponding to the various -I and -E rules.

Note that we are working with weak (non-dependent) elimination rules for positive connectives. This is forced on us by the requirement that types do not depend on linear assumptions. As an alternative, we could demand strong elimination rules, but only for terms without linear assumptions.

\begin{figure}[tb]
\centering
\fbox{\resizebox{\linewidth}{!}{
\begin{tabular}{ll}
\textbf{dDILL judgement} & \textbf{Intended meaning}\vspace{2pt}\\
$\vdash \Gamma;\Delta \ctxt$ & $\Gamma;\Delta$ is a valid context\\
$\Gamma;\cdot \vdash A\type$ &  $A$ is a type in (cartesian) context $\Gamma$\\
$\Gamma;\Delta\vdash a:A$ & $a$ is a term of type $A$ in context $\Gamma;\Delta$\\
$\vdash \Gamma;\Delta = \Gamma';\Delta'$\hspace{40pt} & $\Gamma;\Delta$ and $\Gamma';\Delta'$ are judgementally equal contexts\\
$\Gamma;\cdot\vdash A= A'$ & $A$ and $A'$ are judgementally equal types in (cartesian) context $\Gamma$\\
$\Gamma;\Delta\vdash a= a':A$ & $a$ and $a'$ are judgementally equal terms of type $A$ in context $\Gamma;\Delta$
\end{tabular}\hspace{40pt}\;}}
\caption{\label{fig:ildttjudgements}Judgements of dDILL.}
\end{figure}
\begin{figure}[!tb]
\centering
\fbox{\resizebox{\linewidth}{!}{
\begin{tabular}{ll}
&\\
\AxiomC{}
\RightLabel{\textsf{C-Emp}}
\UnaryInfC{$\cdot;\cdot \ctxt$}
\DisplayProof
& \\
& \\
\AxiomC{$\vdash\Gamma;\cdot\; \ctxt$}
\AxiomC{$\Gamma;\cdot \vdash A\type$}
\RightLabel{\textsf{Cart-C-Ext}}
\BinaryInfC{$\vdash \Gamma,x:A;\cdot\ctxt$}
\DisplayProof\hspace{50pt}\;

&
\AxiomC{$\Gamma;\Delta=\Gamma';\Delta'$}
\AxiomC{$\Gamma;\cdot \vdash A= B$}
\RightLabel{\textsf{Cart-C-Ext-Eq}}
\BinaryInfC{$\vdash \Gamma,x:A;\Delta=\Gamma',y:B;\Delta'$}
\DisplayProof\hspace{50pt}\;\\
& \\
\AxiomC{$\vdash \Gamma;\Delta\ctxt$}
\AxiomC{$\Gamma;\cdot \vdash A\type$}
\RightLabel{\textsf{Lin-C-Ext}}
\BinaryInfC{$\vdash \Gamma;\Delta,x:A\ctxt$}
\DisplayProof
&

\AxiomC{$\Gamma;\Delta=\Gamma';\Delta'$}
\AxiomC{$\Gamma;\cdot\vdash A= B$}
\RightLabel{\textsf{Lin-C-Ext-Eq}}
\BinaryInfC{$\vdash \Gamma;\Delta,x:A=\Gamma';\Delta',y:B$}
\DisplayProof
\\
&\\
&\\
\AxiomC{$\Gamma,x:A,\Gamma';\cdot\ctxt$}
\RightLabel{\textsf{Cart-Idf}}
\UnaryInfC{$\Gamma,x:A,\Gamma';\cdot\vdash x:A$}
\DisplayProof
&
\AxiomC{$\Gamma;x:A\ctxt$}
\RightLabel{\textsf{Lin-Idf}}
\UnaryInfC{$\Gamma;x:A\vdash x:A$}
\DisplayProof
\end{tabular}}}
\caption{\label{fig:ildttstruct1} Context formation and identifier declaration rules.}
\end{figure}

\begin{figure}[!tb]
\centering
\fbox{\resizebox{\linewidth}{!}{
\begin{tabular}{ll}
&\\
\AxiomC{$\vdash \Gamma;\Delta\ctxt$}
\RightLabel{\textsf{C-Eq-R}}
\UnaryInfC{$\vdash \Gamma;\Delta= \Gamma;\Delta$}
\DisplayProof
&
\AxiomC{$\vdash \Gamma;\Delta= \Gamma';\Delta'$}
\RightLabel{\textsf{C-Eq-S}}
\UnaryInfC{$\vdash \Gamma';\Delta'= \Gamma;\Delta$}
\DisplayProof\\
&\\
\AxiomC{$\vdash \Gamma;\Delta= \Gamma';\Delta'$}
\AxiomC{$\vdash \Gamma';\Delta'= \Gamma'';\Delta''$}
\RightLabel{\textsf{C-Eq-T}}
\BinaryInfC{$\vdash \Gamma;\Delta= \Gamma'';\Delta''$}
\DisplayProof &\\
&\\
\AxiomC{$\Gamma;\cdot\vdash A\type$}
\RightLabel{\textsf{Ty-Eq-R}}
\UnaryInfC{$\Gamma;\cdot\vdash A= A$}
\DisplayProof
&
\AxiomC{$\Gamma;\cdot\vdash A= A'$}
\RightLabel{\textsf{Ty-Eq-S}}
\UnaryInfC{$\Gamma;\cdot\vdash A'= A$}
\DisplayProof\hspace{50pt}\;\\
&\\
\AxiomC{$\Gamma;\cdot\vdash A= A'$}
\AxiomC{$\Gamma;\cdot\vdash A'= A''$}
\RightLabel{\textsf{Ty-Eq-T}}
\BinaryInfC{$\Gamma;\cdot\vdash A= A''$}
\DisplayProof
&\\
&\\
\AxiomC{$\Gamma;\Delta\vdash a:A$}
\RightLabel{\textsf{Tm-Eq-R}}
\UnaryInfC{$\Gamma;\Delta\vdash a= a: A$}
\DisplayProof
&
\AxiomC{$\Gamma;\Delta\vdash a= a':A$}
\RightLabel{\textsf{Tm-Eq-S}}
\UnaryInfC{$\Gamma;\Delta\vdash a'= a: A$}
\DisplayProof
\\
&\\
\AxiomC{$\Gamma;\Delta\vdash a= a':A$}
\AxiomC{$\Gamma;\Delta\vdash a'= a'':A$}
\RightLabel{\textsf{Tm-Eq-T}}
\BinaryInfC{$\Gamma;\Delta\vdash a= a'': A$}
\DisplayProof
&\\
&\\
\AxiomC{$\Gamma;\Delta\vdash a:A$}
\AxiomC{$\vdash \Gamma;\Delta= \Gamma;\Delta'$}
\AxiomC{$\Gamma;\cdot \vdash A= A'$}
\RightLabel{\textsf{Tm-Conv}}
\TrinaryInfC{$\Gamma';\Delta'\vdash a:A'$}
\DisplayProof &
\AxiomC{$\Gamma';\cdot\vdash A\type$}
\AxiomC{$\vdash \Gamma;\cdot= \Gamma';\cdot$}
\RightLabel{\textsf{Ty-Conv}}
\BinaryInfC{$\Gamma';\cdot\vdash A\type$}
\DisplayProof
\end{tabular}\hspace{35pt}\;}}
\caption{\label{fig:ildttstruct2} A few standard rules for judgemental equality, saying that it is an equivalence relation and is compatible with typing.}
\end{figure}
\begin{figure}[!tb]
\centering
\fbox{\resizebox{\linewidth}{!}{
\begin{tabular}{ll}
&\\
\AxiomC{$\Gamma,\Gamma';\Delta\vdash\mathcal{J}$}
\AxiomC{$\Gamma;\cdot\vdash A\type$}
\RightLabel{\textsf{Cart-Weak}}
\BinaryInfC{$\Gamma,x:A,\Gamma';\Delta\vdash \mathcal{J}$}
\DisplayProof
&\\
&\\
&\\
\AxiomC{${\Gamma},x:A,\Gamma';\cdot \vdash B\type$}
\AxiomC{$\Gamma;\cdot \vdash a:A$}
\RightLabel{\textsf{Cart-Ty-Subst}}
\BinaryInfC{${\Gamma},\Gamma'[{a}/x];\cdot \vdash B[{a}/x]\type$}
\DisplayProof\hspace{-5pt}
&
\AxiomC{${\Gamma},x:A,\Gamma';\cdot \vdash B= B'$}
\AxiomC{$\Gamma;\cdot \vdash a:A$}
\RightLabel{\textsf{Cart-Ty-Subst-Eq}}
\BinaryInfC{${\Gamma},\Gamma'[{a}/x];\cdot \vdash B[{a}/x]= B'[{a}/x]$}
\DisplayProof\\
&\\
&\\
\AxiomC{${\Gamma},x:A,\Gamma';\Delta \vdash b:B$}
\AxiomC{$\Gamma;\cdot \vdash a:A$}
\RightLabel{\textsf{Cart-Tm-Subst}}
\BinaryInfC{${\Gamma},\Gamma'[{a}/x];\Delta[{a}/x] \vdash b[{a}/x]:B[{a}/x]$}
\DisplayProof
&
\AxiomC{${\Gamma},x:A,\Gamma';\Delta \vdash b= b':B$}
\AxiomC{$\Gamma;\cdot \vdash a:A$}
\RightLabel{\textsf{Cart-Tm-Subst-Eq}}
\BinaryInfC{${\Gamma},\Gamma'[{a}/x];\Delta \vdash b[{a}/x]= b'[{a}/x]:B[{a}/x]$}
\DisplayProof
\\
&\\
&\\
\AxiomC{$\Gamma;\Delta,x:A,\Delta'\vdash b:B$}
\AxiomC{$\Gamma;\Delta''\vdash a:A$}
\RightLabel{\textsf{Lin-Tm-Subst}}
\BinaryInfC{\mccorrect{$\Gamma;\Delta,\Delta',\Delta''\vdash b[a/x]:B$}}
\DisplayProof
&
\AxiomC{$\Gamma;\Delta,x:A,\Delta'\vdash b= b':B$}
\AxiomC{$\Gamma;\Delta''\vdash a:A$}
\RightLabel{\textsf{Lin-Tm-Subst-Eq}}
\BinaryInfC{\mccorrect{$\Gamma;\Delta,\Delta',\Delta''\vdash b[a/x]= b'[a/x]:B$}}
\DisplayProof\\
&\\
&\\
\AxiomC{$\Gamma;\cdot\vdash a=a':A$}
\AxiomC{$\Gamma,x:A,\Gamma';\Delta\vdash b:B$}
\RightLabel{\mccorrect{\textsf{Cart-Tm-Cong}}}
\BinaryInfC{$\Gamma,\Gamma'[a/x];\Delta[a/x]\vdash b[a/x]=b[a'/x]:B$}
\DisplayProof
&
\AxiomC{$\Gamma;\cdot\vdash a=a':A$}
\AxiomC{$\Gamma,x:A,\Gamma';\Delta\vdash B\;\mathsf{type}$}
\RightLabel{\mccorrect{\textsf{Cart-Ty-Cong}}}
\BinaryInfC{$\Gamma,\Gamma'[a/x];\Delta[a/x]\vdash B[a/x]=B[a'/x]$}
\DisplayProof\\
&\\
&\\
\AxiomC{$\Gamma;\Delta''\vdash a=a':A$}
\AxiomC{$\Gamma;\Delta,x:A,\Delta'\vdash b:B$}
\RightLabel{\mccorrect{\textsf{Lin-Tm-Cong}}}
\BinaryInfC{$\Gamma;\Delta,\Delta',\Delta''\vdash b[a/x]=b[a'/x]:B$}
\DisplayProof&\\
\end{tabular}
}}
\caption{\label{fig:ildttstruct3} Weakening, substitution and \mccorrect{congruence} rules. Here, $\mathcal{J}$ represents a statement of the form $B\type$, $B= B'$, $b:B$, or $b= b':B$, such that all judgements are well-formed. \mccorrect{Note that these imply exchange rules for both linear and cartesian identifiers as well as a contraction rule for cartesian identifiers.}}
\end{figure}

\begin{figure}[!tb]
\centering
\fbox{
\resizebox{\linewidth}{!}{
\begin{tabular}{ll}
&\\
\AxiomC{}
\RightLabel{\textsf{$I$-F}}
\UnaryInfC{$\Gamma;\cdot\vdash I\type$}
\DisplayProof
&\\
&\\
\AxiomC{$\Gamma;\cdot \vdash A\type$}
\AxiomC{$\Gamma;\cdot \vdash B\type$}
\RightLabel{\textsf{$\otimes$-F}}
\BinaryInfC{$\Gamma;\cdot \vdash A\otimes B\type$}
\DisplayProof \hspace{72pt}\;
&
\AxiomC{$\Gamma;\cdot \vdash A\type$}
\AxiomC{$\Gamma;\cdot \vdash B\type$}
\RightLabel{\textsf{$\multimap$-F}}
\BinaryInfC{$\Gamma;\cdot \vdash A\multimap B\type$}
\DisplayProof\hspace{72pt}\;\\
&\\
\AxiomC{$\Gamma,x:A;\cdot\vdash B\type$}
\RightLabel{\textsf{$\Sigma_!^\otimes$-F}}
\UnaryInfC{$\Gamma;\cdot\vdash \Sigma^\otimes_{!(x:A)}B\type$}
\DisplayProof
&
\AxiomC{$\Gamma,x:A;\cdot \vdash B\type$}
\RightLabel{\textsf{$\Pi_!^\multimap$-F}}
\UnaryInfC{$\Gamma;\cdot\vdash\Pi^\multimap_{!(x:A)}B\type$}
\DisplayProof\\
&\\
\AxiomC{}
\RightLabel{\textsf{$\top$-F}}
\UnaryInfC{$\Gamma;\cdot\vdash \top\type$}
\DisplayProof
&
\AxiomC{$\Gamma;\cdot \vdash A\type$}
\AxiomC{$\Gamma;\cdot \vdash B\type$}
\RightLabel{\textsf{$\&$-F}}
\BinaryInfC{$\Gamma;\cdot \vdash A\& B\type$}
\DisplayProof\\
&\\
\AxiomC{}
\RightLabel{\textsf{$0$-F}}
\UnaryInfC{$\Gamma ; \cdot\vdash 0\type$}
\DisplayProof
&
\AxiomC{$\Gamma;\cdot \vdash A\type$}
\AxiomC{$\Gamma;\cdot \vdash B\type$}
\RightLabel{\textsf{$\oplus$-F}}
\BinaryInfC{$\Gamma;\cdot \vdash A\oplus B\type$}
\DisplayProof\\
&\\
\AxiomC{$\Gamma;\cdot\vdash A\type$}
\RightLabel{\textsf{$!$-F}}
\UnaryInfC{$\Gamma;\cdot \vdash !A\type$}
\DisplayProof 
&
\AxiomC{$\Gamma;\cdot \vdash a:A$}
\AxiomC{$\Gamma;\cdot \vdash a':A$}
\RightLabel{\textsf{$\Id_!^\otimes$-F}}
\BinaryInfC{$\Gamma;\cdot \vdash \Id^\otimes_{!A}(a,a')\type$}
\DisplayProof
\end{tabular}
}}
\caption{\label{fig:ildttlogical1} Type formation rules for the various connectives.}
\end{figure}

\begin{figure}[!tb]
\centering
\fbox{\resizebox{\linewidth}{!}{
\begin{tabular}{ll}
&\\
\AxiomC{}
\RightLabel{\textsf{$I$-I}}
\UnaryInfC{$\Gamma;\cdot\vdash *:I$}
\DisplayProof
&
\AxiomC{$\Gamma;\Delta'\vdash t:I$}
\AxiomC{$\Gamma;\Delta\vdash a:A$}
\RightLabel{\textsf{$I$-E}}
\BinaryInfC{$\Gamma;\Delta,\Delta'\vdash \mathsf{let}\;t\;\mathsf{be}\;*\;\mathsf{in}\;a:A$}
\DisplayProof 
 \\
&\\
\AxiomC{$\Gamma;\Delta\vdash a:A$}
\AxiomC{$\Gamma;\Delta'\vdash b:B$}
\RightLabel{\textsf{$\otimes$-I}}
\BinaryInfC{$\Gamma;\Delta,\Delta'\vdash a\otimes b:A\otimes B$}
\DisplayProof
&
\AxiomC{$\Gamma;\Delta\vdash t:A\otimes B$}
\AxiomC{$\Gamma;\Delta',x:A,y:B\vdash c:C$}
\RightLabel{\textsf{$\otimes$-E}}
\BinaryInfC{$\Gamma;\Delta,\Delta'\vdash \mathsf{let}\; t\;\mathsf{be}\;x\otimes y\;\mathsf{in} \; c:C$}
\DisplayProof
\\
&\\
\AxiomC{$\Gamma;\Delta,x:A\vdash b:B$}
\RightLabel{\textsf{$\multimap$-I}}
\UnaryInfC{$\Gamma;\Delta\vdash \lambda_{x:A}b:A\multimap B$}
\DisplayProof
&
\AxiomC{$\Gamma;\Delta\vdash f:A\multimap B$}
\AxiomC{$\Gamma;\Delta'\vdash a:A$}
\RightLabel{\textsf{$\multimap$-E}}
\BinaryInfC{$\Gamma;\Delta,\Delta'\vdash f(a):B$}
\DisplayProof\\
&\\
\AxiomC{$\Gamma;\cdot \vdash a:A$}
\AxiomC{$\Gamma;\Delta \vdash b:B[{a}/x]$}
\RightLabel{\textsf{$\Sigma_!^\otimes$-I}}
\BinaryInfC{$\Gamma ; \Delta \vdash ! {a} \otimes b:\Sigma^\otimes_{!(x:A)}B $}
\DisplayProof
&
\AxiomC{\begin{tabular}{l}
$\Gamma;\cdot \vdash C\type$\\
$\Gamma;\Delta \vdash t:\Sigma^\otimes_{!(x:A)}B$\\
$\Gamma,x:A;\Delta',y:B\vdash c:C$
\end{tabular}}
\RightLabel{\textsf{$\Sigma_!^\otimes$-E}}
\UnaryInfC{$\Gamma;\Delta,\Delta' \vdash \mathsf{let}\;t\;\mathsf{be}\;  !{x} \otimes y \;\mathsf{in}\;c:C$}
\DisplayProof\\
&\\
\AxiomC{$\vdash \Gamma;\Delta\ctxt$}
\AxiomC{$\Gamma,x:A;\Delta\vdash b:B$}
\RightLabel{\textsf{$\Pi_!^\multimap$-I}}
\BinaryInfC{$\Gamma;\Delta\vdash \lambda_{!(x:A)}b:\Pi^\multimap_{!(x:A)}B$}
\DisplayProof
&
\AxiomC{$\Gamma;\cdot \vdash a:A$}
\AxiomC{$\Gamma;\Delta\vdash f:\Pi^\multimap_{!(x:A)}B$}
\RightLabel{\textsf{$\Pi_!^\multimap$-E}}
\BinaryInfC{$\Gamma;\Delta\vdash f(!{a}):B[{a}/x]$}
\DisplayProof\\
&\\
\AxiomC{$\vdash\Gamma;\Delta\ctxt$}
\RightLabel{\textsf{$\top$-I}}
\UnaryInfC{$\Gamma;\Delta\vdash \langle\rangle:\top$}
\DisplayProof &\\
&\\
\AxiomC{$\Gamma;\Delta\vdash a:A$}
\AxiomC{$\Gamma;\Delta\vdash b:B$}
\RightLabel{\textsf{$\&$-I}}
\BinaryInfC{$\Gamma;\Delta\vdash \langle a, b\rangle:A\& B$}
\DisplayProof
&
\begin{tabular}{ll}
\AxiomC{$\Gamma;\Delta\vdash t:A\& B$}
\RightLabel{\textsf{$\&$-E1}}
\UnaryInfC{$\Gamma;\Delta\vdash \mathsf{fst}(t):A$}
\DisplayProof
&
\AxiomC{$\Gamma;\Delta\vdash t:A\& B$}
\RightLabel{\textsf{$\&$-E2}}
\UnaryInfC{$\Gamma;\Delta\vdash \mathsf{snd}(t):B$}
\DisplayProof
\end{tabular}\\
&\\
&\AxiomC{$\Gamma;\Delta\vdash t:0$}
\RightLabel{\textsf{$0$-E}}
\UnaryInfC{$\Gamma;\Delta,\Delta'\vdash \mathsf{false}(t) :B$}
\DisplayProof
\\
&\\
\begin{tabular}{l}
\AxiomC{$\Gamma ;\Delta\vdash a: A$}
\RightLabel{\textsf{$\oplus$-I1}}
\UnaryInfC{$\Gamma;\Delta\vdash \mathsf{inl}(a): A\oplus B$}
\DisplayProof
\\
\\
\AxiomC{$\Gamma ;\Delta\vdash b: B$}
\RightLabel{\textsf{$\oplus$-I2}}
\UnaryInfC{$\Gamma;\Delta\vdash \mathsf{inr}(b): A\oplus B$}
\DisplayProof\end{tabular}
&\AxiomC{$\Gamma ;\Delta,x:A\vdash c: C$}
\AxiomC{$\Gamma ;\Delta,y:B\vdash d: C$}
\AxiomC{$\Gamma ;\Delta'\vdash t:A\oplus B$}
\RightLabel{\textsf{$\oplus$-E}}
\TrinaryInfC{$\Gamma;\Delta,\Delta'\vdash \mathsf{case}\; t\;\mathsf{of}\;\mathsf{inl}(x)\rightarrow c\;||\;\mathsf{inr}(y)\rightarrow d :C$}
\DisplayProof\\
&\\
\AxiomC{$\Gamma;\cdot \vdash a:A$}
\RightLabel{\textsf{$!$-I}}
\UnaryInfC{$\Gamma;\cdot\vdash !a:!A$}
\DisplayProof
&
\AxiomC{$\Gamma;\Delta\vdash t:!A$}
\AxiomC{$\Gamma,x:A;\Delta'\vdash b:B$}
\RightLabel{\textsf{$!$-E}}
\BinaryInfC{$\Gamma;\Delta,\Delta'\vdash \mathsf{let}\; t\;\mathsf{be}\;!x\; \mathsf{in}\; b:B$}
\DisplayProof\\
&\\
\AxiomC{$\Gamma;\cdot \vdash a:A$}
\RightLabel{\textsf{$\Id_!^\otimes$-I}}
\UnaryInfC{$\Gamma;\cdot \vdash \refl{!a}:\Id^\otimes_{!A}(a,a)$}
\DisplayProof
&
\AxiomC{
\begin{tabular}{ll}
$\Gamma ;\cdot \vdash a:A$ &\\
$\Gamma;\cdot \vdash a':A$ & $\Gamma,x:A,x':A;\cdot \vdash D\type$\\
$\Gamma ;\Delta' \vdash p:\Id^\otimes_{!A}(a,a')$ & $\Gamma,z:A;\Delta \vdash d:D[{z}/x,{z}/x']$
\end{tabular}}
\RightLabel{\textsf{$\Id_!^\otimes$-E}}
\UnaryInfC{$\Gamma;\Delta[a/z],\Delta' \vdash \mathsf{let}\; (a,a',p)\;\mathsf{be}\;(z,z,\refl{!z})\;\mathsf{in}\; d:D[{a}/x,{a'}/x']$}
\DisplayProof
\end{tabular}}}
\caption{\label{fig:ildttlogical2} Term introduction and elimination rules for the various connectives.}
\end{figure}

\begin{figure}
\fbox{
\resizebox{\linewidth}{!}{
\begin{tabular}{ll}
$\mathsf{let}\;*\;\mathsf{be}\;*\;\mathsf{in}\;a= a$ & $c[d/z]=\mathsf{let}\;d\;\mathsf{be}\;*\;\mathsf{in}\;c[*/z]$\\
$\mathsf{let}\; a\otimes b \;\mathsf{be}\;x\otimes y\;\mathsf{in} \; c = c[a/x,b/y]$&$c[d/z]\stackrel{\#x,y}{=}\mathsf{let}\; d \;\mathsf{be}\;x\otimes y\;\mathsf{in} \; c[x\otimes y/z]$\\
$(\lambda_{x:A}b)(a)= b[a/x]$& $f\stackrel{\#x}{=}\lambda_{x:A}f(x)$ \\
$\mathsf{let}\; !{a} \otimes b\;\mathsf{be}\; ! {x} \otimes y \;\mathsf{in}\;c= c[{a}/x,b/y]$&$c[d/z]\stackrel{\#x,y}{=} \mathsf{let}\;d\;\mathsf{be}\;  !{x} \otimes y \;\mathsf{in}\; c[!{x} \otimes y/z]$\\
$(\lambda_{!(x:A)}b)(!{a})= b[{a}/x]$& $f\stackrel{\# x}{=} \lambda_{!(x:A)}f(!x)$ \\
 $c=\langle\rangle$ &\\
$\mathsf{fst}(\langle a,b\rangle)= a$&$c=\langle\mathsf{fst}(c),\mathsf{snd}(c) \rangle$\\
$\mathsf{snd}(\langle a,b\rangle)= b$&\\
& $c[d/z]=\mathsf{false}(d)$\\
$\mathsf{case}\; \mathsf{inl}(a)\;\mathsf{of}\;\mathsf{inl}(x)\rightarrow c\;||\;\mathsf{inr}(y)\rightarrow d = c[a/x]$& $c[d/z]\stackrel{\#x,y}{=}\mathsf{case}\; d\;\mathsf{of}\;\mathsf{inl}(x)\rightarrow c[\mathsf{inl}(x)/z]\;||\;\mathsf{inr}(y)\rightarrow c[\mathsf{inr}(y)/z]$\\
$\mathsf{case}\; \mathsf{inr}(b)\;\mathsf{of}\;\mathsf{inl}(x)\rightarrow c\;||\;\mathsf{inr}(y)\rightarrow d = d[b/y]$ & \\
$\mathsf{let}\; !a\;\mathsf{be}\;!x\; \mathsf{in}\; b= b[{a}/x]$ & $c[d/z]\stackrel{\# x}{=}\mathsf{let}\;d\;\mathsf{be}\;!x\; \mathsf{in}\; c[!x/z] $\\
$ \mathsf{let}\; (a,a,\refl{!a})\;\mathsf{be}\;(z,z,\refl{!z})\;\mathsf{in}\; d= d[{a}/z]$ & $c[d/x,d'/y,e/z]\stackrel{\#w}{=}\mathsf{let}\;(d,d',e)\;\mathsf{be}\;(w,w,\refl{!w})\;\mathsf{in}\;c[w/x,w/y,\refl{!w}/z]$\\
\end{tabular}
}
}
\caption{\label{fig:ildtteqns} $\beta$- and $\eta$-equations for the various connectives. These should be read as equations of typed terms in context: we impose them if we can derive that both terms being equated are well-typed of equal type in equal context. We write $\stackrel{\#x_1,\ldots,x_n}{=}$ to indicate that for the equation to hold, the identifiers $x_1,\ldots, x_n$ should, in both terms being equated, be replaced by fresh ones, in order to avoid unwanted identifier bindings.}
\end{figure}

\begin{remark}Note that all type formers that are defined context-wise ($I$, $\otimes$, $\multimap$, $\top$, $\&$, $0$, $\oplus$, and $!$) are automatically preserved under the substitutions from Cart-Ty-Subst (up to canonical isomorphism\footnote{By an isomorphism of types $\Gamma;\cdot\vdash A\type$ and $\Gamma;\cdot\vdash B\type$ in context $\Gamma$, we here mean a pair of terms $\Gamma;x:A\vdash f:B$ and $\Gamma;y:B\vdash g:A$ together with a pair of judgemental equalities $\Gamma;x:A\vdash g[f/y]= x:A$ and $\Gamma;y:B\vdash f[g/x]= y:B$.}), in the sense that $F(A_1,\ldots, A_n)[{a}/x]$ is isomorphic to $F(A_1[{a}/x],\ldots,A_n[{a}/x])$ for an $n$-ary type former $F$. Similarly, for \mccorrect{$T=\Sigma^\otimes$ or $\Pi^\multimap$}, we have that $(T_{!(y:B)}C)[{a}/x]$ is isomorphic to $T_{!(y:B[{a}/x])}C[{a}/x]$ and $(\Id_{!B}(b,b'))[a/x]$ is isomorphic to $\Id_{!B[a/x]}(b[a/x],b'[a/x])$. (This gives us Beck-Chevalley conditions in the categorical semantics.\mccorrect{)}\end{remark}

\begin{remark}
The reader can note that the usual formulation of universes for DTT transfers very naturally to dDILL, giving us a notion of universes for linear types, where terms of the universes without linear assumptions code for types. This allows us to write rules for forming types as rules for forming terms, as usual. We do not choose this approach and define the various type formers in the setting without universes, as this will give a cleaner categorical semantics. \mccorrect{As we shall argue in remark} \ref{rmk:universes}\mccorrect{, it is more natural to consider a universe as a cartesian type.}
\end{remark}
\subsubsection*{Some Basic Results} As the focus of this chapter is the syntax-semantics correspondence, we only briefly mention some syntactic results. For some metatheoretic properties for the $\multimap,\Pi_!^\multimap,\top,\&$-fragment of our syntax, like confluence, Church-Rosser, subject reduction and strong normalisation for the (parallel nested, transitive closure of) $\beta$-reductions, we refer the reader to \cite{cervesato1996linear}. Standard techniques \cite{martin1998intuitionistic} and some small adaptations of the system should be enough to extend the results to all of dDILL. \mccorrect{As we discuss a wide range of non-trivial models in section} \ref{sec:concretemodels}\mccorrect{, consistency of dDILL follows immediately, both in the sense that not all terms are equated and in the sense that not all types are inhabited.}
\begin{theorem}[Consistency] dDILL with all its type formers is consistent.\end{theorem}

To give the reader some intuition for the novel connectives $\Pi_!^\multimap$- and $\Sigma_!^\otimes$, we suggest the following two interpretations.

\begin{theorem}[$\Pi_!^\multimap$ and $\Sigma_!^\otimes$ as Dependent $!(-)\multimap(-)$ and $!(-)\otimes(-)$] Suppose we have $!$-types. Let $\Gamma,x:A;\cdot \vdash B\type$, where $x$ does not occur freely in $B$. Then, for the purposes of the type theory,
\begin{enumerate} 
\item $\Pi^\multimap_{!(x:A)}B$ is isomorphic to $!A\multimap B$, if we have $\Pi_!^\multimap$-types and $\multimap$-types;
\item $\Sigma^\otimes_{!(x:A)}B$ is isomorphic to $!A\otimes B$, if we have $\Sigma_!^\otimes$-types and $\otimes$-types.
\end{enumerate}
\end{theorem}

\begin{proof}\begin{enumerate}
\item We construct terms
$$\Gamma;y:\Pi^\multimap_{!(x:A)}B\vdash f:!A\multimap B\txt{and} \Gamma;y':!A\multimap B\vdash g:\Pi^\multimap_{!(x:A)}B
$$
s.t.
\\
\resizebox{\linewidth}{!}{$\Gamma;y:\Pi^\multimap_{!(x:A)}B\vdash g[f/y']= y:\Pi^\multimap_{!(x:A)}B\txt{and}\Gamma;y':!A\multimap B\vdash f[g/y]= y':!A\multimap B.$}\\
First, we construct $f$.\\
\\
\resizebox{\linewidth}{!}{
\AxiomC{}
\RightLabel{\textsf{Cart-Idf}}
\UnaryInfC{$\Gamma,x:A;\cdot \vdash x:A$}
\AxiomC{}
\RightLabel{\textsf{Lin-Idf}}
\UnaryInfC{$\Gamma,x:A;y:\Pi^\multimap_{!(x:A)}B\vdash y:\Pi^\multimap_{!(x:A)}B$}
\RightLabel{\textsf{$\Pi_!^\multimap$-E}}
\BinaryInfC{$\Gamma,x:A;y:\Pi^\multimap_{!(x:A)}B\vdash y(!x):B$}
\AxiomC{}
\RightLabel{\textsf{Lin-Idf}}
\UnaryInfC{$\Gamma;x':!A\vdash x':!A$}
\RightLabel{\textsf{$!$-E}}
\BinaryInfC{$\Gamma;y:\Pi^\multimap_{!(x:A)}B,x':!A\vdash \mathsf{let}\; x'\;\mathsf{be}\;!x\;\mathsf{in}\;y(!x) :B$}
\RightLabel{\textsf{$\multimap$-I}}
\UnaryInfC{$\Gamma;y:\Pi^\multimap_{!(x:A)}B\vdash f:!A\multimap B$}
\DisplayProof\hspace{50pt}\;}
\quad\\
\\
Then, we construct $g$.\\
\\
\resizebox{\linewidth}{!}{
\AxiomC{}
\RightLabel{\textsf{Cart-Idf}}
\UnaryInfC{$\Gamma,x:A;\cdot \vdash x:A$}
\RightLabel{\textsf{$!$-I}}
\UnaryInfC{$\Gamma,x:A;\cdot \vdash !x:!A$}
\AxiomC{}
\RightLabel{\textsf{Lin-Idf}}
\UnaryInfC{$\Gamma,x:A;y':!A\multimap B\vdash y':!A\multimap B$}
\RightLabel{\textsf{$\multimap$-E}}
\BinaryInfC{$\Gamma,x:A;y':!A\multimap B\vdash y'(!x):B$}
\RightLabel{\textsf{$\Pi_!^\multimap$-I}}
\UnaryInfC{$\Gamma;y':!A\multimap B\vdash g:\Pi^\multimap_{!(x:A)}B$}
\DisplayProof\hspace{200pt}\;}
\quad\\
\\
It is easily verified that $\multimap$-$\beta$, $!$-$\beta$, and $\Pi_!^\multimap$-$\eta$ imply the first judgemental equality:\\
 $g[f/y']= \lambda_{!(x:A)}(\lambda_{x':!A}\mathsf{let}\;x'\;\mathsf{be}\;!x\;\mathsf{in}\;y(!x))(!x)= \lambda_{!(x:A)}\mathsf{let}\;!x\;\mathsf{be}\;!x\;\mathsf{in}\;y(!x)= \lambda_{!(x:A)} y(!x)= y$.\\
\\
Similarly, $\Pi_!^\multimap$-$\beta$, $!$-$\eta$, and $\multimap$-$\eta$ imply the second judgemental equality:\\ $f[g/y]= \lambda_{x':!A}\mathsf{let}\;x'\;\mathsf{be}\;!x\;\mathsf{in}\;(\lambda_{!(x:A)}y'(!x))(!x)= \lambda_{x':!A}\mathsf{let}\;x'\;\mathsf{be}\;!x\;\mathsf{in}\;y'(!x)= \lambda_{x':!A}y'(\mathsf{let}\;x'\;\mathsf{be}\;!x\;\mathsf{in}\;!x) = \lambda_{x':!A}y'(x')= y'$.

\item We construct terms
\\
$$\Gamma;y:\Sigma^\otimes_{!(x:A)}B\vdash f:!A\otimes B\txt{and} \Gamma;y':!A\otimes B\vdash g:\Sigma^\otimes_{!(x:A)}B
$$
s.t.
\\
\resizebox{\linewidth}{!}{$\Gamma;y:\Sigma^\otimes_{!(x:A)}B\vdash g[f/y']= y:\Sigma^\otimes_{!(x:A)}B\txt{and}\Gamma;y':!A\otimes B\vdash f[g/y]= y':!A\otimes B.$}\\
First, we construct $f$.\\
\\
\resizebox{\linewidth}{!}{
\AxiomC{}
\RightLabel{\textsf{Lin-Idf}}
\UnaryInfC{$\Gamma;y:\Sigma^\otimes_{!(x:A)}B\vdash y:\Sigma^\otimes_{!(x:A)}B$}
\AxiomC{}
\RightLabel{\textsf{Lin-Idf}}
\UnaryInfC{$\Gamma;x':!A\vdash x':!A$}
\AxiomC{}
\RightLabel{\textsf{Lin-Idf}}
\UnaryInfC{$\Gamma;z:B\vdash z:B$}
\RightLabel{\textsf{$\otimes$-I}}
\BinaryInfC{$\Gamma;x':!A,z:B\vdash x'\otimes z:!A\otimes B$}
\RightLabel{\textsf{Cart-Weak}}
\UnaryInfC{$\Gamma,x:A;x':!A,z:B\vdash x'\otimes z:!A\otimes B$}
\AxiomC{}
\RightLabel{\textsf{Cart-Idf}}
\UnaryInfC{$\Gamma,x:A;\cdot \vdash x:A$}
\RightLabel{\textsf{$!$-I}}
\UnaryInfC{$\Gamma,x:A;\cdot \vdash !x:!A$}
\RightLabel{\textsf{Lin-Subst}}
\BinaryInfC{$\Gamma,x:A;z:B\vdash !x\otimes z:!A\otimes B$}
\RightLabel{\textsf{$\Sigma_!^\otimes$-E}}
\BinaryInfC{$\Gamma;y:\Sigma^\otimes_{!(x:A)}B\vdash f:!A\otimes B$}
\DisplayProof}\\
\\
Then, we construct $g$.\\
\\
\resizebox{\linewidth}{!}{
\AxiomC{}
\RightLabel{\textsf{Lin-Idf}}
\UnaryInfC{$\Gamma;y':!A\otimes B\vdash y':!A\otimes B$}
\AxiomC{}
\RightLabel{\textsf{Cart-Idf}}
\UnaryInfC{$\Gamma,x:A;\cdot\vdash x:A$}
\AxiomC{}
\RightLabel{\textsf{Lin-Idf}}
\UnaryInfC{$\Gamma,x:A;y:B\vdash y:B$}
\RightLabel{\textsf{$\Sigma_!^\otimes$-I}}
\BinaryInfC{$\Gamma,x:A;y:B\vdash  !{x} \otimes y :\Sigma^\otimes_{!(x:A)}B$}
\AxiomC{}
\RightLabel{\textsf{Lin-Idf}}
\UnaryInfC{$\Gamma;x':!A\vdash x':!A$}
\RightLabel{\textsf{$!$-E}}
\BinaryInfC{$\Gamma;x':!A,y: B\vdash \mathsf{let}\;x'\;\mathsf{be}\;!x\;\mathsf{in}\; ! {x} \otimes y:\Sigma^\otimes_{!(x:A)}B$}
\RightLabel{\textsf{$\otimes$-E}}
\BinaryInfC{$\Gamma;y':!A\otimes B\vdash g:\Sigma^\otimes_{!(x:A)}B$}
\DisplayProof}
\\
\\
Here, the first judgemental equality follows from $\otimes$-$\beta$, $!$-$\beta$, and $\Sigma_!^\otimes$-$\eta$:\\
$g[f/y']=\lbi{(\lbi{y}{!{x} \otimes z}{!x\otimes z})}{x'\otimes y}{(\lbi{x'}{!x}{! {x} \otimes y})}= \lbi{y}{!x \otimes z}{\lbi{!x\otimes z}}{x'\otimes y}{(\lbi{x'}{!x}{ !{x} \otimes y})}=  \lbi{y}{!{x} \otimes z}{(\lbi{x'}{!x}{ !{x} \otimes y})[!x/x'][z/y]} = \lbi{y}{!{x} \otimes z}{(\lbi{!x}{!x}{ !{x} \otimes z})} 
= \lbi{y}{!{x} \otimes z}{ !{x} \otimes z} = y$.

The second judgemental equality follows from $\Sigma_!^\otimes$-$\beta$, $!$-$\eta$, and $\otimes$-$\eta$:\\
$f[g/y]= \lbi{(\lbi{y'}{x'\otimes y}{(\lbi{x'}{!x}{ !{x} \otimes y})})}{ !{x} \otimes z}{!x\otimes z}= \lbi{y'}{x'\otimes y}{\lbi{x'}{!x}{\lbi{ !{x} \otimes y}{ !{x} \otimes z}{!x\otimes z}}}= \lbi{y'}{x'\otimes y}{\lbi{x'}{!x}{(!x\otimes y)}}= \lbi{y'}{x'\otimes y}{(\lbi{x'}{!x}{!x})\otimes y}= \lbi{y'}{x'\otimes y}{x'\otimes y}= y'$.
\end{enumerate}
\end{proof}

In particular, we have the following stronger version of a special case.
\begin{theorem}[$!$ as $\Sigma I$]\label{thm:!fromsigma}
Suppose we have $\Sigma_!^\otimes$- and $I$-types. Let $\Gamma;\cdot \vdash A\type$. Then, $\Sigma^\otimes_{!(x:A)}I$ satisfies the rules for $!A$. Conversely, if we  have $!$- and $I$-types, then $!A$ satisfies the rules for $\Sigma^\otimes_{!(x:A)}I$. 
\end{theorem}
\begin{proof}We obtain the $!$-$\mathsf{I}$ rule as follows.\\
\\
\resizebox{\linewidth}{!}{\AxiomC{$\Gamma;\cdot\vdash a:A$}
\AxiomC{}
\RightLabel{\textsf{$I$-I}}
\UnaryInfC{$\Gamma,x:A; \cdot \vdash *:I$}
\RightLabel{\textsf{$\Sigma_!^\otimes$-I}}
\BinaryInfC{$\Gamma;\cdot\vdash  !{a} \otimes *:\Sigma^\otimes_{!(x:A)}I$}
\DisplayProof\hspace{380pt}\;}
\\
\\
We obtain the $!$-E rule as follows.\\
\\
\resizebox{\linewidth}{!}{\AxiomC{$\Gamma;\Delta\vdash t:\Sigma^\otimes_{!(x:A)}I$}
\AxiomC{$\Gamma,x:A;\Delta'\vdash c:C$}
\AxiomC{}
\RightLabel{\textsf{Lin-Idf}}
\UnaryInfC{$\Gamma;y:I\vdash y:I$}
\RightLabel{\textsf{$I$-E}}
\BinaryInfC{$\Gamma,x:A;\Delta',y:I\vdash \mathsf{let}\; y\;\mathsf{be}\; *\;\mathsf{in}\;c:C$}
\RightLabel{\textsf{$\Sigma_!^\otimes$-E}}
\BinaryInfC{$\Gamma;\Delta,\Delta'\vdash \mathsf{let}\; t\;\mathsf{be}\; !{x} \otimes y\;\mathsf{in}\;\mathsf{let}\; y\;\mathsf{be}\; *\;\mathsf{in}\;c:C$.}
\DisplayProof\hspace{250pt}\;}\\
\\
It is easily seen that $\Sigma_!^\otimes$-$\beta$ and $I$-$\beta$ imply $!$-$\beta$ ($\mathsf{let}\;!{a} \otimes *\;\mathsf{be}\;{!x} \otimes y\;\mathsf{in}\; \mathsf{let}\; y\;\mathsf{be}\; *\;\mathsf{in}\;c=(\mathsf{let}\; y\;\mathsf{be}\; *\;\mathsf{in}\;c)[{a}/x][*/y]= \mathsf{let}\; *\;\mathsf{be}\; *\;\mathsf{in}\;c[{a}/x]  = c[{a}/x]$) and that $I$-$\eta$ and $\Sigma_!^\otimes$-$\eta$ imply $!$-$\eta$ ($\mathsf{let}\;t\;\mathsf{be}\;!{x} \otimes y\;\mathsf{in}\;\mathsf{let}\; y\;\mathsf{be}\; *\;\mathsf{in}\;c[!{x} \otimes */z]= \mathsf{let}\;t\;\mathsf{be}\;!{x} \otimes y\;\mathsf{in}\;c[!{x} \otimes \mathsf{let}\; y\;\mathsf{be}\; *\;\mathsf{in}\;*/z] = \mathsf{let}\;t\;\mathsf{be}\;!{x} \otimes y\;\mathsf{in}\;c[!{x} \otimes y/z]=c[ t/z]$).\\
\\
The converse statement follows through a similarly trivial argument, noting that $I[{a}/x]$ is isomorphic to $I$.
\end{proof}

A second interpretation is that $\Pi_!^\multimap$ and $\Sigma_!^\otimes$ generalise $\&$ and $\oplus$. Indeed, the idea is that that (or their infinitary equivalents) is what they reduce to when taken over discrete types. The subtlety in this result will be the definition of a discrete type. The same phenomenon is observed in a different context in section \ref{sec:dismod}.

For our purposes, a discrete type is a strong sum of $I$ (a sum with a dependent -E-rule). Let us for simplicity limit ourselves to the binary case. For us, the discrete type with two elements will be $2=I\oplus I$, where $\oplus$ has a strong/dependent -E-rule (note that this is not our $\oplus$-E). Explicitly, $2$ is a type with the rules of figure \ref{fig:ildttdiscretetype}.
\begin{figure}[!tb]
\quad\\
\fbox{\resizebox{\linewidth}{!}{
\begin{tabular}{l}
\\
\begin{tabular}{lcr}
\AxiomC{}
\RightLabel{\textsf{$2$-F}}
\UnaryInfC{$\Gamma;\cdot\vdash 2\type$}
\DisplayProof
\hspace{110pt}
&
\AxiomC{}
\RightLabel{\textsf{$2$-I1}}
\UnaryInfC{$\Gamma;\cdot\vdash \mathsf{tt}:2$}
\DisplayProof
\hspace{110pt}
&
\AxiomC{}
\RightLabel{\textsf{$2$-I2}}
\UnaryInfC{$\Gamma;\cdot\vdash \mathsf{ff}:2$}
\DisplayProof
\end{tabular}
\\
\\
\\
\begin{tabular}{c}
\AxiomC{$\Gamma,x:2;\cdot \vdash A\type$}
\AxiomC{$\Gamma;\cdot\vdash t:2$}
\AxiomC{$\Gamma;\Delta[\ttt/x]\vdash a_{\mathsf{tt}}:A[{\mathsf{tt}}/x]$}
\AxiomC{$\Gamma;\Delta[\fff/x]\vdash a_{\mathsf{ff}}:A[{\mathsf{ff}}/x]$}
\RightLabel{\textsf{$2$-E}}
\QuaternaryInfC{$\Gamma;\Delta[t/x]\vdash \mathsf{if}\; t\;\mathsf{then}\; a_{\mathsf{tt}}\;\mathsf{else}\;a_{\mathsf{ff}}:A[{t}/x]$}
\DisplayProof
\end{tabular}
\\
\\
\\
\begin{tabular}{lr}
\AxiomC{$\Gamma;\Delta\vdash \mathsf{if}\; \mathsf{tt}\;\mathsf{then}\; a_{\mathsf{tt}}\;\mathsf{else}\;a_{\mathsf{ff}}:A[\mathsf{tt}/x]$}
\RightLabel{\textsf{$2$-$\beta$1}}
\UnaryInfC{$\Gamma;\Delta\vdash \mathsf{if}\; \mathsf{tt}\;\mathsf{then}\; a_{\mathsf{tt}}\;\mathsf{else}\;a_{\mathsf{ff}}= a_{\mathsf{tt}}:A[\mathsf{tt}/x]$}
\DisplayProof
&\hspace{68pt}
\AxiomC{$\Gamma;\Delta\vdash \mathsf{if}\; \mathsf{ff}\;\mathsf{then}\; a_{\mathsf{tt}}\;\mathsf{else}\;a_{\mathsf{ff}}:A[\mathsf{ff}/x]$}
\RightLabel{\textsf{$2$-$\beta$2}}
\UnaryInfC{$\Gamma;\Delta\vdash \mathsf{if}\; \mathsf{ff}\;\mathsf{then}\; a_{\mathsf{tt}}\;\mathsf{else}\;a_{\mathsf{ff}}= a_{\mathsf{ff}}:A[\mathsf{ff}/x]$}
\DisplayProof\\
&\\

\AxiomC{$\Gamma;\Delta\vdash \mathsf{if}\; t\;\mathsf{then}\; c[\mathsf{tt}/x]\;\mathsf{else}\;c[\mathsf{ff}/x]:C$}
\RightLabel{\textsf{$2$-$\eta$}}
\UnaryInfC{$\Gamma;\Delta\vdash c[t/x]= \mathsf{if}\; t\;\mathsf{then}\;  c[\mathsf{tt}/x]\;\mathsf{else}\;c[\mathsf{ff}/x]:C$}
\DisplayProof
\end{tabular}
\end{tabular}
}}
\caption{\label{fig:ildttdiscretetype} Rules for a discrete type $2$.}
\end{figure}
\begin{theorem}[$\Pi_!^\multimap$ and $\Sigma_!^\otimes$ as Infinitary Non-Discrete $\&$ and $\oplus$]\label{thm:pisigmainf} If we have a discrete type $2$ and a type family $\Gamma,x: 2;\cdot\vdash A$, then
\begin{enumerate}
\item $\Pi^\multimap_{!(x:{2})}A$ satisfies the rules for $A[{\mathsf{tt}}/x]\& A[{\mathsf{ff}}/x]$;
\item $\Sigma^\otimes_{!(x:{2})}A$ satisfies the rules for $A[{\mathsf{tt}}/x]\oplus A[{\mathsf{ff}}/x]$.
\end{enumerate}
\end{theorem}
\begin{proof}\begin{enumerate}
\item We obtain $\&$-I as follows.\\
\quad\\
\resizebox{\linewidth}{!}{
\AxiomC{$\Gamma,x:2;\Delta\vdash a:A[{\mathsf{tt}}/x]$}
\AxiomC{$\Gamma,x:2;\Delta\vdash b:A[{\mathsf{ff}}/x]$}
\AxiomC{}
\RightLabel{\textsf{Cart-Idf}}
\UnaryInfC{$\Gamma,x:2;\cdot\vdash x:2$}
\AxiomC{}
\RightLabel{\textsf{Assumption}}
\UnaryInfC{$\Gamma,x:2;\cdot \vdash A\type$}
\RightLabel{\textsf{2-E-dep}}
\QuaternaryInfC{$\Gamma,x:2;\Delta\vdash \mathsf{if}\; x\;\mathsf{then}\; a\;\mathsf{else}\; b:A$}
\RightLabel{\textsf{$\Pi_!^\multimap$-I}}
\UnaryInfC{$\Gamma;\Delta\vdash \lambda_{!(x:{2})}\mathsf{if}\; x\;\mathsf{then}\; a\;\mathsf{else}\; b:\Pi^\multimap_{!(x:{2})}A$}
\DisplayProof}\\
\\
Moreover, we obtain $\&$-E1 as follows (similarly, we obtain $\&$-E2).\\
\\
\resizebox{\linewidth}{!}{
\AxiomC{$\Gamma;\Delta\vdash t:\Pi^\multimap_{!(x:2)}A$}
\AxiomC{}
\RightLabel{\textsf{2-I1}}
\UnaryInfC{$\Gamma;\cdot \vdash \mathsf{tt}:2$}
\RightLabel{\textsf{$\Pi_!^\multimap$-E}}
\BinaryInfC{$t(!\mathsf{tt})$}
\DisplayProof\hspace{350pt}\;}\\
\\
The $\&$-$\beta$-rules follow from $\Pi_!^\multimap$-$\beta$ and $2$-$\beta$, e.g.\\ $\mathsf{fst}\langle a,b\rangle := (\lambda_{!(x:{2})}\mathsf{if}\; x\;\mathsf{then}\; a\;\mathsf{else}\; b)(!\mathsf{tt})
= \mathsf{if}\; \mathsf{tt}\;\mathsf{then}\; a\;\mathsf{else}\; b= a$.

The $\&$-$\eta$-rules follow from $\Pi_!^\multimap$-$\eta$ and $2$-$\eta$:\\ $\langle \mathsf{fst}(t),\mathsf{snd}(t)\rangle:= \lambda_{!(x:{2})}\mathsf{if}\; x\;\mathsf{then}\; t(!\mathsf{tt})\;\mathsf{else}\; t(!\mathsf{ff})= \lambda_{!(x:{2})}t(!x)= t$.

\item We obtain $\oplus$-I1 as follows (and similarly, we obtain $\oplus$-I2):\\
\\
\resizebox{\linewidth}{!}{
\AxiomC{}
\RightLabel{\textsf{2-I1}}
\UnaryInfC{$\Gamma;\cdot \vdash \mathsf{tt}:2$}
\AxiomC{$\Gamma;\Delta\vdash a:A[{\mathsf{tt}}/x]$}
\RightLabel{\textsf{$\Sigma_!^\otimes$-I}}
\BinaryInfC{$\Gamma;\Delta\vdash  {!\mathsf{tt}} \otimes a:\Sigma^\otimes_{!(x:{2})}A$}
\DisplayProof\hspace{330pt}\;}
\\
\\
Moreover, we obtain $\oplus$-E as follows.\\
\\
\resizebox{\linewidth}{!}{
\AxiomC{$\Gamma;\Delta'\vdash t:\Sigma^\otimes_{!(x:{2})}A$\hspace{-20pt}\;}
\AxiomC{$\Gamma;\Delta,z:A[\ttt/x]\vdash c:C$}
\AxiomC{$\Gamma;\Delta,w:A[\fff/x]\vdash d:C$}
\AxiomC{}
\RightLabel{\textsf{}}
\UnaryInfC{$\Gamma,x:2;\cdot\vdash x:2$}
\AxiomC{}
\UnaryInfC{$\Gamma,x:2;\cdot\vdash A\type$}
\LeftLabel{\textsf{2-E}}
\QuaternaryInfC{$\Gamma,x:2;\Delta,y:A\vdash \mathsf{if}\;x\;\mathsf{then}\;c[y/z]\;\mathsf{else}\;d[y/w]:C$}
\RightLabel{\textsf{$\Sigma_!^\otimes$-E}}
\BinaryInfC{$\Gamma;\Delta,\Delta'\vdash \mathsf{let}\;t \;\mathsf{be}\; !{x} \otimes y \;\mathsf{in}\; \mathsf{if}\;x\;\mathsf{then}\;c[y/z]\;\mathsf{else}\;d[y/w] :C$}
\DisplayProof}\\
\\
The $\oplus$-$\beta$-rules follow from $\Sigma_!^\otimes$-$\beta$ and 2-$\beta$, e.g.
\begin{align*} 
&\mathsf{case}\; \mathsf{inl}(a)\;\mathsf{of}\;\mathsf{inl}(z)\rightarrow c||\;\mathsf{inr}(w)\rightarrow d:  =\\
&\mathsf{let}\; {!\mathsf{tt}} \otimes a \;\mathsf{be}\; !{x} \otimes y \;\mathsf{in}\; \mathsf{if}\;x\;\mathsf{then}\;c[y/z]\;\mathsf{else}\;d[y/w]
=\\
&\mathsf{if}\;\mathsf{tt}\;\mathsf{then}\;c[a/z]\;\mathsf{else}\;d[a/w]=c[a/z].
\end{align*}

The $\oplus$-$\eta$-rules follow from $\Sigma_!^\otimes$-$\eta$ and 2-$\eta$:
\begin{align*}
&\mathsf{case}\; t\;\mathsf{of}\;\mathsf{inl}(z)\rightarrow c[\mathsf{inl}(z)/u]||\;\mathsf{inr}(w)\rightarrow c[\mathsf{inr}(w)/u]:  =\\
& \mathsf{let}\;t \;\mathsf{be}\; !{x} \otimes y \;\mathsf{in}\; \mathsf{if}\;x\;\mathsf{then}\;c[\mathsf{inl}(z)/u][y/z]\;\mathsf{else}\;c[\mathsf{inr}(w)/u][y/w]
=\\ 
&\mathsf{let}\;t \;\mathsf{be}\;! {x} \otimes y \;\mathsf{in}\; \mathsf{if}\;x\;\mathsf{then}\; c[{!\mathsf{tt}} \otimes z/u][y/z]\;\mathsf{else}\; c[{!\mathsf{ff}} \otimes w/u][y/w]
=\\ 
&\mathsf{let}\;t \;\mathsf{be}\; !{x} \otimes y \;\mathsf{in}\; \mathsf{if}\;x\;\mathsf{then}\; c[{!\mathsf{tt}} \otimes y/u]\;\mathsf{else}\; c[{!\mathsf{ff}} \otimes y/u]
=\\
& \mathsf{let}\;t \;\mathsf{be}\; !{x} \otimes y \;\mathsf{in}\; c[{!(\mathsf{if}\;x\;\mathsf{then}\; \mathsf{tt}\;\mathsf{else}\;\mathsf{ff})} \otimes y/u]
=\\
& \mathsf{let}\;t \;\mathsf{be}\;! {x} \otimes y \;\mathsf{in}\;  [!{x} \otimes y/u]
= c[t/u].
\end{align*}

\end{enumerate}
\end{proof}
We see that we can also view $\Pi_!^\multimap$ and $\Sigma_!^\otimes$ as generalisations of $\&$ and $\oplus$, respectively.

\section{Semantics of dDILL}
\label{sec:sem}
The idea behind the categorical semantics we present for the structural core of our syntax (with $I$- and $\otimes$-types) will be to take our suggested categorical semantics for the structural core of DTT (with $1$- and $\times$-types) and relax the assumption of the cartesian character of its fibres to them only being (possibly non-cartesian) symmetric \mccorrect{monoidal}. This entirely reflects the relation between the conventional semantics of non-dependent cartesian and linear type systems. The structure we obtain is that of a strict indexed symmetric monoidal\footnote{\mccorrect{It is plausible that w}e could obtain a sound and complete semantics for only the structural core, possibly without $I$- and $\otimes$-types, by considering strict indexed symmetric multicategories with comprehension.} category with comprehension.

The $\Sigma_!^\otimes$- and $\Pi_!^\multimap$-types arise as left and right adjoints of substitution functors along projections in the base-category and the $\Id_!^\otimes$-types arise as left adjoints to substitution along diagonals, all satisfying Beck-Chevalley (and Frobenius) conditions, as is the case in the semantics for DTT. The $!$-types boil down to having a left adjoint to the comprehension (which can be made a functor), giving a linear/non-linear adjunction as in the conventional semantics for linear logic. Finally, additive connectives arise as compatible cartesian and distributive cocartesian structures on the fibres, as would be expected from the semantics of linear logic.

\subsection{Models of dDILL (Tautologically)}
First, we translate the structural core of our syntax to the tautological notion of model. We shall later prove this to be equivalent to the more intuitive notion of categorical model we referred to above.
\begin{mccorrection}
\begin{definition}[Model of dDILL] By a model $\widetilde{\mathbb{T}}$ of dDILL, we shall mean the following data.\\
\\
\begin{tabular}{ll}
(\textsf{Contexts}) & A set $\mathsf{CCtxt}$;\\&\\
(\textsf{Types}) & A map $\mathsf{CCtxt}\ra{\mathsf{LType}} \mathsf{Set}$;\\&\\
\parbox{0.20\linewidth}{(\textsf{Terms},\\ \textsf{C-Emp1},\\ \textsf{Lin-C-Ext})}& \parbox{0.75\linewidth}{A map $\Sigma_{\Gamma\in\mathsf{CCtxt}}\mathsf{LCtxt}(\Gamma)\times\mathsf{LType}(\Gamma)\ra{\mathsf{LTerm}}\mathsf{Set}$, where we use the syntactic sugar $\mathsf{LCtxt}(\Gamma)$ for the free monoid on $\mathsf{LType}(\Gamma)$ whose unit and multiplication we shall write $\cdot$ and $-.-$;}\\&\\
(\textsf{C-Emp2})& An element $\cdot\in\mathsf{CCtxt}$;\\&\\
(\textsf{Cart-C-Ext}) & A map $\Sigma_{\Gamma\in\mathsf{CCtxt}}\mathsf{LType}(\Gamma)\ra{-.-}\mathsf{CCtxt}$;\\&\\
\parbox{0.20\linewidth}{(\textsf{Cart-Weak})\\
\\}& \parbox{0.75\linewidth}{Maps $\mathsf{LType}(\Gamma.\Gamma')\ra{\mathsf{weak}}\mathsf{LType}(\Gamma.A.\mathsf{weak}(\Gamma'))$ and $\mathsf{LTerm}(\Gamma.\Gamma',\Delta,B)\ra{\mathsf{weak}}\mathsf{LTerm}(\Gamma.A.\mathsf{weak}(\Gamma'),\mathsf{weak}(\Delta),\mathsf{weak}(B))$ (where we slightly abuse notation);}\\&\\
(\textsf{Cart-Idf})& Elements $\mathsf{der}\in \mathsf{LTerm}(\Gamma.A.\Gamma',\cdot, \mathsf{weak}(A))$;\\&\\
(\textsf{Lin-Idf})& Elements $\mathsf{id}\in \mathsf{LTerm}(\Gamma,A,A)$;\\&\\
\parbox{0.20\linewidth}{(\textsf{Cart-Ty-Subst})\\
}& \parbox{0.75\linewidth}{For $B\in\mathsf{LType}(\Gamma.{A}.\Gamma')$ and $a\in\mathsf{LTerm}(\Gamma,\cdot,A)$, we have  $B\{\Gamma.{a}.\Gamma'\}\in \mathsf{LType}(\Gamma.\Gamma'\{\Gamma.{a}\})$;}\\&\\
\parbox{0.20\linewidth}{(\textsf{Cart-Tm-Subst})\\}& \parbox{0.75\linewidth}{For $b\in\mathsf{LTerm}(\Gamma.{A}.\Gamma',\Delta,B)$ and $a\in\mathsf{LTerm}(\Gamma,\cdot,A)$, we have $b\{\Gamma.{a}.\Gamma'\}\in \mathsf{LTerm}(\Gamma.\Gamma'\{\Gamma.{a}\},\Delta\{\Gamma.{a}\},B\{\Gamma.{a}\})$;}\\&\\
\parbox{0.20\linewidth}{(\textsf{Lin-Tm-Subst})\\} & \parbox{0.75\linewidth}{For $b\in\mathsf{LTerm}(\Gamma,\Delta.A.\Delta',B)$ and $a\in\mathsf{LTerm}(\Gamma,\Delta'',A)$, we have $(\Delta.a.\Delta');b\in\mathsf{LTerm}(\Gamma,\Delta.\Delta'.\Delta'',B)$,}
\end{tabular}\\
\\
such that
\begin{itemize}
\item $\mathsf{weak}$ preserves $\mathsf{id}$, $\mathsf{der}$, $-;-$ and $-\{-\}$ in the obvious sense;
\item $-\{-\}$ commutes with $-;-$ in the obvious sense;
\item $-;-$ is associative;
\item $-;-$-substitutions in disjoint parts of the context commute: if $j<i$ (set $N=0$) or $j>i+m-1$ (set $N=m$) then\\
$(C_1.\ldots .C_{j-1}. a'.C_{j+1}.\ldots.C_{n+m-1});(A_1.\ldots . A_{i-1}. a.A_{i+1}.\ldots.A_n);b\\
=(C_1.\ldots . C_{j-1}. a.C_{j+1}.\ldots.C_{n+m-1});(A_1.\ldots . A_{j+N-1}. a'.A_{j+N+1}.\ldots.A_n);b$;
\item $-\{-\}$ on terms is associative;
\item $-\{-\}$-term substitutions in disjoint parts of the context commute (as for $-;-$ substitutions);
\item $(\Delta.\mathsf{id}.\Delta');b=b$ for all $b\in\mathsf{LTerm}(\Gamma,\Delta.A.\Delta',B)$;
\item the actions on both $\mathsf{LType}$ and $\mathsf{LTerm}$ of $-\{\Gamma.\mathsf{der}.\Gamma'\}$ (``substituting a diagonal'') and $\mathsf{weak}$ (``substituting a projection'') satisfy all equations induced by the theory of cartesian products (see \cite{jacobs1994semantics} for these precise equations).
\end{itemize}
\end{definition}
The equations we demand in this definition are all the standard equations that are implicit for syntactic substitution. The point of these laws is that we can form context morphisms as lists of compatible terms, which we can then substitute into terms (and types) in an associative way, using the operations $-\{-\}$ and $\mathsf{weak}$ in the case of cartesian contexts and using $-;-$ in the case of linear contexts. Note that commutativity of disjoint substitutions and the fact that $\mathsf{weak}$ preserves $-\{-\}$  imply that this parallel substitution is well-defined.

We interpret $\sem{\vdash \mathsf{ctxt}}:=\mathsf{Ctxt}$, $\sem{\Gamma\vdash \mathsf{type}}:=\mathsf{LType}(\Gamma)$ and $\sem{\Gamma;\Delta\vdash A}:=\mathsf{LTerm}(\Gamma,\Delta,A)$. We interpret judgemental equality of contexts, types and terms as the equality on the sets $\mathsf{Ctxt}$, $\mathsf{LType}(\Gamma)$ and $\mathsf{LTerm}(\Gamma,\Delta,A)$. Note that all rules for judgemental equality (the rules with \textsf{Eq}, \textsf{Conv} and \textsf{Cong} in their name) then automatically follow.
\end{mccorrection}
It is tautological that there is a one to one correspondence between theories $\mathbb{T}$ in dDILL and models $\widetilde{\mathbb{T}}$ of this sort.\\
\\
We now define what it means for the model to support various type formers.
\begin{definition}[Semantic $I$- and $\otimes$-types] We say a model $\widetilde{\mathbb{T}}$ supports $I$-types, if for all $\Gamma\in\mathsf{CCtxt}$, we have an $I\in\mathsf{LType}(\Gamma)$ and $*\in\mathsf{LTerm}(\Gamma,\cdot, I)$    and whenever $t\in\mathsf{LTerm}(\Gamma,\Delta,I)$ and $a\in\mathsf{LTerm}(\Gamma,\Delta',A)$, we have $\mathsf{let}\; t\;\mathsf{be}\;*\;\mathsf{in}\; a\in\mathsf{LTerm}(\Gamma,\Delta.\Delta',A)$, such that $\mathsf{let}\; *\;\mathsf{be}\;*\;\mathsf{in}\;a=a$ and $(\Delta'.a);t=\mathsf{let}\;a\;\mathsf{be}\;*\;\mathsf{in}\;((\Delta'.*);t)$.

Similarly, we say it admits $\otimes$-types, if for all $A,B\in\mathsf{LType}(\Gamma)$, we have a $A\otimes B\in\mathsf{LType}(\Gamma)$, for all $a\in \mathsf{LTerm}(\Gamma,\Delta,A), b\in\mathsf{LTerm}(\Gamma,\Delta',B)$, we have $a\otimes b\in\mathsf{LType}(\Gamma,\Delta.\Delta',A\otimes B)$, and if $t\in \mathsf{LTerm}(\Gamma,\Delta,A\otimes B)$ and $c\in\mathsf{LTerm}(\Gamma,\Delta'.A.B,C)$, we have $\mathsf{let}\; t\;\mathsf{be}\; \id_A\otimes \id_B\;\mathsf{in}\;c\in\mathsf{LTerm}(\Gamma,\Delta.\Delta',C)$, such that $\mathsf{let}\;a\otimes b\;\mathsf{be}\;\id_A\otimes \id_B\;\mathsf{in}\;c=c$ and $(\Delta'.d);t=\mathsf{let}\;t\;\mathsf{be}\;\id_A\otimes \id_B\;\mathsf{in}\;(\Delta'.(\id_A\otimes \id_B);t)$.

Note that this defines a function $\mathsf{LCtxt}(\Gamma)\ra{\bigotimes} \mathsf{LType}$. The $\beta$-rule precisely says that from the point of view of the (terms of the) type theory this map is an injection, while the $\eta$-rule says it is a surjection\footnote{The precise statement that we are alluding to here would be that the multicategory of linear contexts is equivalent to the (monoidal) multicategory of linear types. Really, $\bigotimes$ is only part of an equivalence of categories rather than an isomorphism, i.e. it is injective on objects up to isomorphism rather than on the nose.}. We conclude that in the presence of $I$- and $\otimes$-types, we can faithfully describe the type theory without mentioning linear contexts, replacing them by the linear type that is their $\otimes$-product.
\end{definition}
We shall henceforth assume that our type theory has $I$- and $\otimes$-types, as this simplifies the categorical semantics\footnote{To be precise, it allows us to give a categorical semantics in terms of monoidal categories rather than multicategories.} and is appropriate for the examples we are interested in.

For the other type formers, one can give a similar, almost tautological, translation from the syntax into a model. We leave this to the reader when we discuss the semantic equivalent of various type formers in the categorical semantics we present next.

\subsection{Categorical Semantics of dDILL}
\subsubsection*{Strict Indexed Symmetric Monoidal Categories with Comprehension}
We now introduce a notion of categorical model for which soundness and completeness results hold with respect to the syntax of dDILL in the presence of $I$- and $\otimes$-types\footnote{In case we are interested in the case without $I$- and $\otimes$-types, the semantics easily generalises to strict indexed symmetric multicategories with comprehension.}. This notion of model will prove to be particularly useful when thinking about various  type formers.

\begin{definition}By a \emph{strict indexed symmetric monoidal category with comprehension}, we mean the following data.
\begin{enumerate}
\item A category $\Bcat$ with a terminal object $\cdot$.
\item A strict indexed symmetric monoidal category $\Dcat$ over $\Bcat$, i.e. a contravariant functor $\Dcat$ into the category $\mathsf{SMCat}$ of (small) symmetric monoidal categories and \mccorrect{strict} monoidal functors $\Bcat^{op}\ra{\Dcat}\mathsf{SMCat}.$
We  also write $-\{f\}:=\Dcat(f)$ for the action of $\Dcat$ on a morphism $f$ of $\Bcat$.
\item A \emph{comprehension schema}, i.e. for each $\Gamma\in\mathsf{ob}(\Bcat)$ and $A\in\mathsf{ob}(\Dcat(\Gamma))$ a representation for the functor $$x\mapsto\Dcat(\mathsf{dom}(x))(I,A\{x\}):(\Bcat/\Gamma)^{op}\ra{}\Set.$$ We  write its representing object\footnote{Really, $\Gamma.UA\ra{\proj{\Gamma}{UA}}\Gamma$ would be a better notation, where we think of $F\dashv U$ as an adjunction inducing $!$, but it would be very verbose.} $\Gamma.{A}\ra{\mathbf{p}_{\Gamma,{A}}}\Gamma\in\mathsf{ob}(\Bcat/\Gamma)$ and universal element $\mathbf{v}_{\Gamma,{A}}\in\Dcat(\Gamma.{A})(I,A\{\mathbf{p}_{\Gamma,{A}}\})$. We  write $a  \mapsto  \langle f,a\rangle$ for the isomorphism $\Dcat(\Gamma')(I,A\{f\}) \cong  \Bcat/\Gamma(f,\mathbf{p}_{\Gamma,{A}})$.
\end{enumerate}
\end{definition}
Again, the comprehension schema means that the morphisms in our category of contexts $\Bcat$, into a context built by adjoining types, arise as lists of closed linear terms. Here, there is the crucial identification with cartesian terms of linear terms without linear assumptions: they can be freely copied and discarded.

We note that the definition of comprehension for an indexed symmetric monoidal category is almost identical to that of definition \ref{def:comprehension} for an indexed cartesian monoidal category. The only difference is that the tensor unit now plays the r\^ole of the terminal object. We again use the same definitions for $\mathsf{diag}$, $\mathbf{q}$ and the comprehension functors $\proj{\Gamma}{-}$.

\begin{theorem}[Comprehension functor]\label{thm:comprfunc} A comprehension schema $(\mathbf{p},\mathbf{v})$ on a strict indexed symmetric monoidal category $(\Bcat,\Dcat)$ defines a  morphism $\Dcat\ra{U}\Ccat$ of indexed symmetric monoidal categories, which lax-ly sends the monoidal structure of $\Dcat$ to products in $\Ccat$ (where they exist), where $\Ccat$ is the full subindexed\footnote{Here, we use the axiom of choice to make a choice of pullback and make $\Ccat$ really into a (non-strict) indexed category (or cloven fibration). Alternatively, we can avoid the axiom of choice and treat it as a more general fibration.} category of $\Bcat/-$ on the objects of the form $\mathbf{p}_{\Gamma,{A}}$.\end{theorem}
\begin{proof}First note that a morphism $U$ of indexed symmetric monoidal categories consists of lax monoidal functors $U_\Gamma$ in each context $\Gamma\in\Bcat$ such that
\begin{diagram}
\Dcat(\Gamma) & \rTo^{U_\Gamma} & \Ccat(\Gamma)\\
\dTo^{\Dcat(f)} &\quad \cong  & \dTo_{\Ccat(f)=\textnormal{``pullback along $f$''}}\\
\Dcat(\Gamma') & \rTo^{U_{\Gamma'}} & \Ccat(\Gamma').
\end{diagram}
We define
\begin{diagram}
U_\Gamma(A\ra{a}B):=\mathbf{p}_{\Gamma,{A}} & \rTo^{\langle\mathbf{p}_{\Gamma,{A}},\mathbf{v}_{\Gamma,{A}};a\{\mathbf{p}_{\Gamma,{A}}\} \rangle} & \mathbf{p}_{\Gamma,{B}}.
\end{diagram}
Functoriality follows from the uniqueness property of \mccorrect{$\langle\mathbf{p}_{\Gamma,A},\mathbf{v}_{\Gamma,A};a\{\mathbf{p}_{\Gamma,A}\}\rangle$}.

We define the lax monoidal structure
\begin{diagram}
\id_\Gamma & \rTo^{m_\Gamma^I} & U_\Gamma(I)=\mathbf{p}_{\Gamma,{I}}\\
\mathbf{p}_{\Gamma.{A},{B\{\proj{\Gamma}{A}\}}};\mathbf{p}_{\Gamma,{A}} =U_\Gamma(A)\times U_\Gamma(B)& \rTo^{m_\Gamma^{A,B}} & U_\Gamma(A\otimes B)=\mathbf{p}_{\Gamma,{A\otimes B}},
\end{diagram}
where $m_\Gamma^{A,B}:=\langle \mathbf{p}_{\Gamma.{A}, B\{\proj{\Gamma}{A}\}};\mathbf{p}_{\Gamma,{A}},\mathbf{v}_{\Gamma,{A}}\{\mathbf{p}_{\Gamma.{A}, B\{\proj{\Gamma}{A}\}}\}\otimes\mathbf{v}_{\Gamma.{A},{B\{\proj{\Gamma}{A}\}}} \rangle$ and $m_\Gamma^I:=\langle \id_\Gamma,\id_{I} \rangle$.

Finally, we verify that $\Ccat(f)U_\Gamma=U_{\Gamma'}\Dcat(f)$. This follows directly from the fact that the following square is a pullback square:
\begin{diagram}
\Gamma'.{A\{f\}} & \rTo^{\mathbf{q}_{f,{A}}} & \Gamma.{A}\\
\dTo^{\mathbf{p}_{\Gamma',{A\{f\}}}} & & \dTo_{\mathbf{p}_{\Gamma,{A}}}\\
\Gamma' & \rTo_f & \Gamma,
\end{diagram}
where $\mathbf{q}_{f,{A}}:=\langle f\mathbf{p}_{\Gamma',{A\{f\}}},\mathbf{v}_{\Gamma',{A\{f\}}}\rangle$. We leave this verification to the reader as an exercise. Alternatively, a proof for this fact in DTT, that will transfer to our setting in its entirety, can be found in \cite{hofmann1997syntax}.
\end{proof}
\begin{remark}
Note that $\Ccat$ is a display map category (or, less specifically, a full comprehension category) and, using the axiom of choice to make a choice of pullbacks, can be viewed as a (non-strict) indexed category with full and faithful comprehension, another, slightly weaker, commonly used notion of model of dependent types. We shall see that, in many ways, we can regard $\Ccat$ as the cartesian content of $\Dcat$.
\end{remark}
\begin{remark}We shall see that this functor will give us a unique candidate for $!$-types: $!:=FU$, where $F\dashv U$. We conclude that, in dDILL, the $!$-modality is uniquely determined by the indexing. This is worth noting, because, in propositional linear type theory, we might have many different candidates for $!$-types.

Moreover, it explains why we do not demand $U$ to be fully faithful in the case of linear types. Indeed, although we have a map $\Dcat(\Gamma)(A,B)\ra{U_\Gamma}\Ccat(\Gamma)(\mathbf{p}_{\Gamma,{A}},\mathbf{p}_{\Gamma,{B}})\cong \Dcat(\Gamma.{A})(I,B\{\mathbf{p}_{\Gamma,{A}}\})$, this is not generally an isomorphism. In fact, in the presence of $!$-types, we shall see that the right hand side is precisely isomorphic to $\Dcat(\Gamma)(!A,B)$ and the map is precomposition with dereliction.
\end{remark}

Next, we prove that we have a sound interpretation of dDILL in such categories.

\begin{theorem}[Soundness] A strict indexed symmetric monoidal category with comprehension $(\Bcat,\Dcat,\mathbf{p},\mathbf{v})$ defines a model $\widetilde{\mathbb{T}}^{(\Bcat,\Dcat,\mathbf{p},\mathbf{v})}$ of dDILL with $I$- and $\otimes$-types.
\end{theorem}
\begin{proof}
We define
\begin{mccorrection}
\begin{enumerate}
\item \textsf{Contexts}: $\mathsf{CCtxt}:=\mathsf{ob}(\Bcat)$
\item \textsf{Types}: $\mathsf{LType}(\Gamma):=\mathsf{ob}(\Dcat(\Gamma))$
\item $\mathsf{LCtxt}(\Gamma):=\mathsf{free-monoid}(\mathsf{LType}(\Gamma))$ (where we  write $[]$ and $++$ for the monoid operations)\\
\textsf{C-Emp1}: $\cdot_{\mathsf{LCtxt}(\Gamma)}:=[]_{\mathsf{LCtxt}(\Gamma)}$\\
\textsf{Lin-C-Ext}: $\Delta._{\mathsf{LCtxt}}A:=\Delta++ A$\\ 
\textsf{Terms}:  $\mathsf{LTerm}(\Gamma,\Delta,A):=\Dcat(\Gamma)(\bigotimes\Delta,A)$
\item \textsf{C-Emp2}: $\cdot_\mathsf{CCtxt}:=\cdot_{\Bcat}$
\item \textsf{Cart-C-Ext}: $\Gamma._\mathsf{CCtxt}{A}:=\Gamma._{\Bcat}{A}$.\\

\item \textsf{Cart-Weak}: The required morphisms are interpreted as follows. Suppose we are given $A,\Gamma'\in\mathsf{ob}(\Dcat(\Gamma))$. We  define a weakening functor
\begin{diagram}
\Dcat(\Gamma.\Gamma') & \rTo^{\Dcat(\langle f,a\rangle)} & \Dcat(\Gamma.{A}.\Gamma'\{\mathbf{p}_{\Gamma,{A}}\}),
\end{diagram}
where $f$ and $a$ are defined as follows.
\begin{diagram}
\Gamma.{A}.\Gamma'\{\mathbf{p}_{\Gamma,{A}}\} & \rTo^{f:=\mathbf{p}_{\Gamma.{A},\Gamma'\{\mathbf{p}_{\Gamma,{A}}\}};\mathbf{p}_{\Gamma,{A}}} & \Gamma
\end{diagram}
and 
\begin{diagram}
I & \rTo^{a=\mathbf{v}_{\Gamma,{A}.\Gamma'\{\mathbf{p}_{\Gamma,{A}}\}}} & \Gamma'\{f\}=\Gamma'\{\mathbf{p}_{\Gamma,{A}.\Gamma'\{\mathbf{p}_{\Gamma,{A}}\}}\} \;\in \Dcat(\Gamma.{A}.\Gamma'\{\mathbf{p}_{\Gamma,{A}}\}).
\end{diagram}
Note that this interpretation of weakening preserves $\der$ (by definition) and $\id$ (as it is a functor) and commutes with the three substitution operations (by functoriality of $-\{\mathbf{p}\}$ and by functoriality of $-\{-\}$ in the second argument).

\item \textsf{Cart-Idf}: $\mathsf{der}_{\Gamma,A,\Gamma'}\in\mathsf{LTerm}(\Gamma.{A}.\Gamma',\cdot,A)$ is defined as $$\mathbf{v}_{\Gamma,{A}}\{\mathbf{p}_{\Gamma.{A},\Gamma'}\}:I\ra{}A\{\mathbf{p}_{\Gamma.{A},\Gamma'};\mathbf{p}_{\Gamma,{A}}\} \in \Dcat(\Gamma.{A}.\Gamma')$$ Note that $\mathsf{der}_{\Gamma,A,\Gamma'}$ defines a morphism\\ $\Gamma.{A}.\Gamma'\ra{\mathsf{diag}_{\Gamma,{A},\Gamma'}}\Gamma.{A}.\Gamma'.{A}\{\mathbf{p}_{\Gamma.{A},\Gamma'};\mathbf{p}_{\Gamma,{A}}\}:=\langle \id_{\Gamma.{A}.\Gamma'},\mathsf{der}_{\Gamma,A,\Gamma'}\rangle $.\\ We shall later show that this in fact behaves as a diagonal morphism on ${A}$.

\item \textsf{Lin-Idf}: $\id_A\in \mathsf{LTerm}(\Gamma,A,A)$ is taken to be $\mathsf{id_A}\in\Dcat(\Gamma)(A,A)$. Note that this is indeed the neutral element for our semantic linear term substitution operation that we shall define shortly.

\item \textsf{Cart-Ty-Subst} and \textsf{Cart-Tm-Subst}: substitution along a term $\Gamma;\cdot \vdash a:A$, are interpreted by the functors $\Dcat(\langle \id_\Gamma,a\rangle)=-\{\langle \id_\Gamma,a\rangle\}$. Indeed, let $B\in\Dcat(\Gamma.{A}.\Gamma')$ and $a\in\Dcat(\Gamma)(I,A)$. Then, we define the context $\Gamma.\Gamma'\{\Gamma.{a}/x\}$  as  $\Gamma.(\Gamma'\{\langle \id_\Gamma,a\rangle\})$ and the type $B\{\Gamma.{a}.\Gamma'\}$ as $B\{\langle f,a'\rangle\}$, where
\begin{diagram}
\Gamma.\Gamma'\{\langle \id_\Gamma,a\rangle\}.{I} & \rTo^{\langle f,a'\rangle} & \Gamma.{A}.\Gamma'
\end{diagram}
is defined from
\begin{diagram}
\Gamma.\Gamma'\{\langle \id_\Gamma,a\rangle\} & \rTo^{\mathbf{p}_{\Gamma,\Gamma'\{\langle \id_\Gamma,a\rangle\}}}  &\Gamma\\
&\rdTo^f &\dTo_{\langle \id_\Gamma,a\rangle}\\
& &  \Gamma.{A}
\end{diagram} 
and
\begin{diagram}
I & \rTo^{a':=\mathbf{v}_{\Gamma,\Gamma'\{\langle \id_\Gamma,a\rangle\}}} & \Gamma'\{f\}=(\Gamma'\{\langle \id_\Gamma,a\rangle\})\{\mathbf{p}_{\Gamma,\Gamma'\{\langle \id_\Gamma,a\rangle\}}\}.
\end{diagram}
\item \textsf{Lin-Tm-Subst}: interpreted by composition in $\Dcat(\Gamma)$. To be precise, given $b\in \Dcat(\Gamma)((\bigotimes\Delta)\otimes A\otimes(\bigotimes \Delta'),B)$ and $a\in \Dcat(\Gamma)(\bigotimes\Delta'',A)$, we define $b[a/x]\in$ $\Dcat(\Gamma)((\bigotimes\Delta)\otimes (\bigotimes \Delta')\otimes(\bigotimes\Delta''),B)$ as $( \id_{\bigotimes\Delta}\otimes a\otimes \id_{\bigotimes\Delta'});\mathsf{braid}_{\bigotimes\Delta',\bigotimes\Delta''};b$.
\end{enumerate}
\end{mccorrection}
Note that \textsf{Cart-Ty-Subst} and \textsf{Cart-Tm-Subst} are interpreted by functors and therefore preserve identities and compositions and are  associative in their composition. \textsf{Lin-Tm-Subst} is interpreted by composition in the fibre categories, hence is also associative.

The fact that \textsf{Cart-Idf} and \textsf{Cart-Weak} define compatible diagonals and projections follows from the fact that $\Gamma.A.B\{\proj{\Gamma}{A}\}\ra{\proj{\Gamma.A}{B\{\proj{\Gamma}{A}\}}}\Gamma.A\ra{\proj{\Gamma}{A}}\Gamma$ defines the cartesian product of $\proj{\Gamma}{A}$ and $\proj{\Gamma}{B}$ in $\Bcat/\Gamma$.

Finally, the model clearly supports $I$- and $\otimes$-types. We interpret $I\in\mathsf{LType}(\Gamma)$ as the unit object in $\Dcat(\Gamma)$ while its term $*$ is interpreted as the identity morphism. Similarly, we interpret $\otimes$ by the monoidal product on the fibres: $*:=\id_I\in\Dcat(\Gamma)$, $\mathsf{let}\; t\;\mathsf{be}\; *\;\mathsf{in} \; a:=t\otimes a$, $a\otimes b$ is defined as the tensor product of morphisms in $\Dcat(\Gamma)$, and $\mathsf{let}\;t\;\mathsf{be}\; \id_A\otimes \id_B\;\mathsf{in}\;c:= (\id_{\Delta'}\otimes t);c$ (\mccorrect{leaving out} associatiators and unitors, here). The $\beta$- and $\eta$-rules are immediate.
\end{proof}
In fact, the converse is also true: we can build a category of this sort from the syntax of dDILL.
\begin{theorem}[Co-Soundness] A model $\widetilde{\mathbb{T}}$ of dDILL with $I$ and $\otimes$-types defines a strict indexed symmetric monoidal category with comprehension $(\Bcat^\mathbb{T},\Dcat^\mathbb{T},\mathbf{p}^\mathbb{T},\mathbf{v}^\mathbb{T})$.\end{theorem}
\begin{proof}
The main technical difficulty in this proof will be that our \mccorrect{syntactic category} has context morphisms as morphisms (corresponding to lists of terms of the type theory) while the type theory only talks about individual terms. This exact difficulty is also encountered when proving completeness of the categories with families semantics for ordinary DTT. It is sometimes fixed by (conservatively) extending the the type theory to also talk about context morphisms explicitly. See e.g. \cite{pitts1995categorical}.

\begin{enumerate}
\item[1.] We define $\mathsf{ob}(\Bcat^\mathbb{T}):=\mathsf{CCtxt}$, modulo $\alpha$-equivalence, and write $\Gamma.{A}$ for the equivalence class of $\Gamma,x:A$. The designated object $\cdot$ of $\Bcat^\mathbb{T}$ will be the (equivalence class of) $\cdot$ (from \textsf{C-Emp}), which will automatically become a terminal object because of our definition of a morphism of $\Bcat^\mathbb{T}$ (context morphism). Indeed, we define morphisms in $\Bcat^\mathbb{T}$, as follows, by induction.

We start out by defining $\Bcat^\mathbb{T}(\Gamma',\cdot):=\{\langle\rangle\}$ and for $\Gamma\in\mathsf{CCtxt}$ that are not of the form $\Gamma''.{A}$, define $\Bcat^\mathbb{T}(\Gamma',\Gamma)=\{\id_\Gamma\}$ if $\Gamma'=\Gamma$ and $\Bcat^\mathbb{T}(\Gamma',\Gamma)=\emptyset$ otherwise.

Then, by induction on the length $n$ of $\Gamma=x_1:A_1,\ldots,x_n:A_n$, we define 
$$\Bcat^\mathbb{T}(\Gamma',\Gamma.{A_{n+1}}):=\Sigma_{f\in\Bcat^\mathbb{T}(\Gamma',\Gamma)}\mathsf{LTerm}(\Gamma',\cdot,A_{n+1}[f/x]),$$
where $A_{n+1}[f/x]$ is defined, using \textsf{Cart-Ty-Subst}, to be the (syntactic operation of) parallel substitution (see \cite{hofmann1997syntax}, section 2.4) of the list $f_1,\ldots, f_n$ of linear terms $\Gamma';\cdot \vdash f_i:A_i[f_1/x_1,\ldots,f_{i-1}/x_{i-1}]$ that $f$ is made up out of, for the identifiers $x_1,\ldots, x_n$ in $\Gamma$.

Note that, in particular, according to \textsf{Cart-Idf}, $\mathsf{LTerm}({A_1}.\ldots .{A_n}.,\cdot,{A_i})$ contains a term $\mathsf{der}_{{A_1}.\ldots.{A_{i-1}},{A_i},{A_{i+1}}.\ldots.{A_n}}$, which allows us to define, inductively,
$$\mathbf{p}_{{A_1}.\ldots.{A_n}}^n:=\langle\rangle\in\Bcat^\mathbb{T}({A_1}.\ldots.{A_n},\cdot)$$
$$\mathbf{p}_{{A_1}.\ldots.{A_n}}^{n-i}:=$$
$$ \mathbf{p}_{{A_1}.\ldots.{A_n}}^{n-i+1},\mathsf{der}_{{A_1}.\ldots.{A_{i-1}},{A_i},{A_{i+1}}.\ldots.{A_n}}\in\Bcat^\mathbb{T}({A_1}.\ldots.{A_n},{A_1}.\ldots.{A_i})$$
In particular, we define identities in $\Bcat^\mathbb{T}$ from these: $\id_{{A_1}.\ldots.{A_n}}:=\mathbf{p}_{{A_1}.\ldots.{A_n}}^0$. We shall also use these \mccorrect{`}projections' in 3. to define  the comprehension schema. In all cases, projections, identities and diagonals defined using $\mathsf{der}$ behave as such via substitutions because we have demanded that the actions of \textsf{Cart-Idf} and \textsf{Cart-Weak} interact via the laws induced from the theory of cartesian products.

We define composition in $\Bcat^\mathbb{T}$ by induction. Let ${B_1}.\ldots.{B_m}=\Gamma'\ra{f=f_1,\ldots,f_n}\Gamma={A_1}.\ldots.{A_n}$ and $\Gamma''\ra{g=g_1,\ldots,g_m}\Gamma'$. Then, we define\mccorrect{, by induction, $g;():=()$ and}\\ $ g;(f_1,\ldots,f_{n-1},f_n):= g;(f_1,\ldots, f_{n-1}),f_n[g/x]$, where $f_n[g/x]$ denotes the parallel substitution of $g=g_1,\dots,g_m$ for the free identifiers $x_1,\ldots,x_m$ in $f_n$, using \textsf{Cart-Tm-Subst}. Note that associativity of composition comes from the associativity of substitution that is implicit in the syntax as well as the compatibility of substitution with weakening while the identity morphism we defined clearly acts as a neutral element for our composition.

\item[2.] Define $\mathsf{ob}(\Dcat^\mathbb{T}(\Gamma)):=\mathsf{LCtxt}(\Gamma)$ and $\Dcat^\mathbb{T}(\Gamma)(\Delta,\Delta'):=\mathsf{LTerm}(\Gamma,\Delta,\bigotimes \Delta')$. Composition is defined through \textsf{Lin-Tm-Subst} and $\otimes$-E. Identities are given by \textsf{Lin-Idf}. The monoidal unit is given by $\cdot\in\mathsf{LCtxt}(\Gamma)$, while the monoidal product $\otimes$ on objects is given by context concatenation. The monoidal product $\otimes$ on morphisms is given by $\otimes$-I. Note that the associators and unitors follow from the associative and unital laws for the commutative monoid of contexts together with $\otimes$-$\beta$ and $\otimes$-$\eta$ \mccorrect{and that the symmetry/braid comes from the commutativity of the monoid}. (Note that the rules for $\otimes$ give us an isomorphism between an arbitrary context $\Delta$ and the one-type-context $\bigotimes \Delta$, while the rules for $I$ do the same for $\cdot$ and $I$.)

We define $\Dcat^\mathbb{T}(f)$ on objects by parallel substitution and weakening in each type in a linear context, via \textsf{Cart-Ty-Subst} and \textsf{Cart-Weak}, and on morphisms by parallel substitution \mccorrect{and weakening}, via \textsf{Cart-Tm-Subst} and \textsf{Cart-Weak}. Note that functoriality is given by implicit properties of the syntax like associativity of substitution. Note that this defines a \mccorrect{strict} symmetric monoidal functor. We conclude that $\Dcat^\mathbb{T}$ is a functor $\Bcat^\mathbb{T}{}^{op}\ra{}\mathsf{SMCat}$.

\item[3.] We define following comprehension schema on $\Dcat^\mathbb{T}$. Suppose $\Gamma\in\Bcat^\mathbb{T}$ and $A\in\Dcat^\mathbb{T}(\Gamma)$.

Define $\Gamma.{A}\ra{\mathbf{p^\mathbb{T}}_{\Gamma,A}}\Gamma$ as $\mathbf{p}_{\Gamma.A}^1$ from 1. and $I\ra{\mathbf{v^\mathbb{T}}_{\Gamma,{A}}} A\{\mathbf{p}_{\Gamma,{A}}^\mathbb{T}\}$ (through \textsf{Cart-Idf}) as\\ $\mathsf{der}_{A}\in\mathsf{LTerm}(\Gamma.{A},\cdot,A)=\Dcat^\mathbb{T}(\Gamma.{A})(I,A\{\mathbf{p^\mathbb{T}}_{\Gamma,{A}}\})$.

Suppose we are given $\Gamma'\ra{f}\Gamma$ and $a\in\Dcat^\mathbb{T}(\Gamma')(I,A\{f\})=\mathsf{LTerm}(\Gamma',\cdot,A[f/c])$. Then, by definition of the morphisms in $\Bcat^\mathbb{T}$, there is a unique morphism\\ $\langle f,a\rangle:= f,a\in\Bcat^\mathbb{T}(\Gamma',\Gamma.{A}):=\Sigma_{f\in\Bcat^\mathbb{T}(\Gamma',\Gamma)}\mathsf{LTerm}(\Gamma',\cdot,A[f/x])$ such that\linebreak $\langle f,a\rangle ;\mathbf{p}^\mathbb{T}_{\Gamma,{A}}=f$ and $\mathbf{v}^\mathbb{T}_{\Gamma,{A}}\{\langle f,a\rangle \}=a$. The uniqueness follows from the fact that $ -;\mathbf{p}^\mathbb{T}_{\Gamma,{A}} $ and $\mathbf{v^\mathbb{T}}_{\Gamma,{A}}\{-\}$ are the two (dependent) projections of the $\Sigma$-type (in Set) that defines this homset. We can note this bijection is natural in the sense that $g;\langle f,a\rangle =\langle g;f,a\{g\}\rangle$ because of the associativity of the substitution \textsf{Cart-Tm-Subst} in the syntax.
\end{enumerate}
\end{proof}

\begin{theorem}[Completeness] The construction described in   \mccorrect{`}Co-Soundness' followed by the one described in \mccorrect{`}Soundness' is the identity (up to categorical equivalence\mccorrect{)}: i.e. strict indexed symmetric monoidal categories with comprehension provide a complete semantics for dDILL with $I$- and $\otimes$-types\footnote{It is easy to see that, similarly, indexed symmetric multicategories with comprehension form a complete semantics for dDILL, possibly without $I$- and $\otimes$-types.}.
\end{theorem}
\begin{proof}This is a trivial exercise.\end{proof}
\begin{theorem}[Failure of Co-Completeness] The construction described in \mccorrect{`}Soundness' followed by the one described in \mccorrect{`}Co-Soundness' may not be equivalent to the identity: i.e. Co-Completeness can fail (as for the categories with families semantics for DTT). Its fixed-points (up to equivalence) are precisely the models for which the comprehension is democratic.
\end{theorem}
\begin{proof}
Indeed, if we start with a strict indexed symmetric monoidal category with comprehension, construct the corresponding model $\widetilde{\mathbb{T}}$ and then construct its syntactic category, we effectively have thrown away all the non-trivial morphisms into objects that are not of the form $\Gamma.{A}$ or $\cdot$. The definition of a democratic comprehension is precisely that every object is of that form.

Of course, we can easily obtain a co-complete model theory by putting this extra restriction on our models. Alternatively -- this may be nicer from a categorical point of view --, we can take the obvious (see e.g. \cite{pitts1995categorical}) conservative extension of our syntax by also talking about context morphisms (corresponding to morphisms in our base category). In that case, we would obtain an actual internal language for strict indexed symmetric monoidal categories with comprehension. This also has the advantage that we can easily obtain an internal language for strict indexed monoidal categories by dropping the axioms Cart-C-Ext, Cart-C-Ext-Eq, Cart-Idf and Cart-Weak, which correspond to the comprehension schema. We have not chosen this route as it would mean that the syntax would not fit as well with what has been considered so far in the syntactic tradition.\end{proof}
\begin{corollary}[Relation to DTT and ILTT] \mccorrect{As we have seen, a} model $(\Bcat,\Dcat,\mathbf{p},\mathbf{v})$ of dDILL with $I$- and $\otimes$-types defines a model $\Ccat$ of DTT,  that should be thought of the cartesian content of the linear type theory. This will become even more clear through our treatment of $!$-types and in the examples we treat.

Moreover, it clearly defines a model of ILTT with $I$- and $\otimes$-types (i.e. a symmetric monoidal category) in every context.

Conversely, it is easily seen that every  model of DTT can be obtained this way (up to equivalence), by noting that it is in particular a model of dDILL and that every model of ILTT can be embedded in a model of dDILL. (As we shall see in section \ref{sec:dismod}, we can cofreely add type dependency on $\Set$.)\end{corollary}

\subsubsection*{Semantic Type Formers}
Next, we discuss the interpretation of various type formers in models of dDILL.
\begin{theorem}[Semantic type formers]\label{thm:semtype} For the other type formers, we have the following. A model of dDILL with $I$- and $\otimes$-types (a strict indexed symmetric monoidal category with comprehension)...
\begin{enumerate}
\item ...supports $\Sigma_!^\otimes$-types iff all the change of base functors $\Dcat(\mathbf{p}_{\Gamma,{A}})$ have left adjoints $\Sigma^\otimes_{!{A}}$ that satisfy the left Beck-Chevalley condition for $\mathbf{p}$-squares and that satisfy Frobenius reciprocity\footnote{Frobenius reciprocity expresses compatibility of $\Sigma_!^\otimes$ and $\otimes$, which is reasonable if we want a reading of $\Sigma_!^\otimes$ as a generalisation of $\otimes$. If one wants to drop Frobenius reciprocity in the semantics, it is easy to see that the equivalent in the syntax is setting $\Delta'=\cdot$ in the $\Sigma_!^\otimes$-E-rule. Therefore, Frobenius reciprocity automatically follows if we have $\multimap$-types.} in the sense that the canonical morphism $$\Sigma^\otimes_{!{A}}(\Delta'\{\mathbf{p}_{\Gamma,{A}}\}\otimes B)\ra{} \Delta'\otimes \Sigma^\otimes_{!{A}}B$$ is an isomorphism , for all $\Delta'\in\Dcat(\Gamma)$, $B\in\Dcat(\Gamma.{A})$ .
\item ...supports $\Pi_!^\multimap$-types iff all the change of base functors $\Dcat(\mathbf{p}_{\Gamma,{A}})$ have right adjoints $\Pi^\multimap_{!{A}}$ that satisfy the right Beck-Chevalley condition for $\mathbf{p}$-squares.
\item ...supports $\multimap$-types iff $\Dcat$ factors over the category $\mathsf{SMCCat}$\mccorrect{ of symmetric monoidal categories and (strict) symmetric monoidal functors.}
\item ...supports $\top$-types and $\&$-types iff $\Dcat$ factors over the category \mccorrect{$\mathsf{SMcCat}$} of cartesian categories with a symmetric monoidal structure and their \mccorrect{(strict)} homomorphisms.
\item ...supports $0$-types and $\oplus$-types iff $\Dcat$ factors over the category $\mathsf{dSMcCCat}$ of cocartesian categories with a distributive\footnote{Note that in the light of theorem \ref{thm:pisigmainf}, the demand of distributivity here is essentially the same phenomenon as the demand of Frobenius reciprocity for $\Sigma_!^\otimes$-types.} symmetric monoidal structure and their \mccorrect{(strict)} homomorphisms.
\item ...that supports $\multimap$-types\footnote{Actually, we only need this for the \mccorrect{`}if'. The \mccorrect{`}only if' always holds. To make the \mccorrect{`}if' work, as well, in absence of $\multimap$-types, we have to restrict $!$-E to the case where $\Delta'=\cdot$. Alternatively, we could note that the semantic condition that precisely corresponds to having $!$-types (even in absence of $\multimap$-types) is to have a natural isomorphism $\Dcat(\Gamma.A)(\Delta\{\proj{\Gamma}{A}\},B\{\proj{\Gamma}{A}\})\cong \Dcat(\Gamma)(!A\otimes \Delta,B)$ (which we immediately recognise as a specific case of $\Sigma_!^\otimes$-types).}, supports $!$-types iff all the comprehension functors $\Dcat(\Gamma)\ra{U_\Gamma}\Ccat(\Gamma)$ have a strong monoidal left adjoint $\Ccat(\Gamma)\ra{F_\Gamma}\Dcat(\Gamma)$ in the 2-category $\mathsf{SMCat}$ of symmetric monoidal categories, lax symmetric monoidal functors, and monoidal natural transformations\footnote{i.e. a symmetric lax monoidal left adjoint functor $F_\Gamma$ such that an inverse for its lax structure is given by the oplax structure on $F_\Gamma$ coming from the lax structure on $U_\Gamma$. Put differently, $F_\Gamma$ is a left adjoint functor to $U_\Gamma$ and is a strong monoidal functor in a way that is compatible with the lax structure on $U_\Gamma$.} and (compatibility with substitution) for all $\Gamma'\ra{f}\Gamma\in \Bcat$ we have that $F_\Gamma ;\Dcat(f)= \Ccat(f);F_{\Gamma'}$ (which makes $F_-$ into a morphism of indexed categories). Then the linear exponential comonad $!_\Gamma:=U_\Gamma;F_\Gamma:\Dcat(\Gamma)\ra{}\Dcat(\Gamma)$ will be our interpretation of the comodality $!$ in the context $\Gamma$.
\item ... supports $\Id_!^\otimes$-types iff for all $A\in\mathsf{ob}\;\Dcat(\Gamma)$, we have left adjoints $\Id^\otimes_{!A}\dashv -\{\mathsf{diag}_{\Gamma,{A}}\}$ that satisfy the left Beck-Chevalley condition for $\mathsf{diag}$-squares and Frobenius reciprocity in the sense that the canonical morphisms
$$\Id^\otimes_{!A}(B)\ra{}\Id^\otimes_{!A}(I)\otimes B\{\proj{\Gamma.A}{A\{\proj{\Gamma}{A}\}}\}$$
are isomorphisms.
\end{enumerate}\end{theorem}
\begin{proof}
\begin{enumerate}
\item Assume our model supports $\Sigma_!^\otimes$-types. We exhibit the claimed adjunction. The morphism from left to right is provided by $\Sigma_!^\otimes$-I. The morphism from right to left is provided by $\Sigma_!^\otimes$-E. $\Sigma_!^\otimes$-$\beta$ and $\Sigma_!^\otimes$-$\eta$ say exactly that these are mutually inverse. Naturality corresponds to the compatibility of $\Sigma_!^\otimes$-I and $\Sigma_!^\otimes$-E with substitution.
\begin{diagram}
c'&\rMapsto &(!\mathbf{v}_{\Gamma,{A},\cdot}\otimes\id_B);(c'\{\mathbf{p}_{\Gamma,{A}}\}) \\
\Dcat(\Gamma)(\Sigma^\otimes_{!{A}}B,C)
& \pile{\rTo^{}\\\cong \\ \lTo_{}} & \Dcat(\Gamma.{A})(B,C\{\mathbf{p}_{\Gamma,{A}}\})\\
\mathsf{let}\;z\;\mathsf{be}\;  !{x} \otimes y\;\mathsf{in}\;c &\lMapsto & c 
\end{diagram}

We show how the morphism from left to right arises from $\Sigma_!^\otimes$-I.\\
\\
\resizebox{\linewidth}{!}{
\AxiomC{}
\RightLabel{\textsf{Cart-Idf}}
\UnaryInfC{$\Gamma,x:A;\cdot \vdash x:A$}
\AxiomC{}
\RightLabel{\textsf{Lin-Idf}}
\UnaryInfC{$\Gamma,x:A;w:B\vdash w:B$}
\RightLabel{\textsf{$\Sigma_!^\otimes$-I}}
\BinaryInfC{$\Gamma,x:A;w:B\vdash  !{x} \otimes w:\Sigma^\otimes_{!(x:A)}B$}
\AxiomC{$\Gamma;z:\Sigma^\otimes_{!(x:A)}B\vdash c': C$}
\RightLabel{\textsf{Cart-Weak}}
\UnaryInfC{$\Gamma,x:A;z:\Sigma^\otimes_{!(x:A)}B\vdash c': C$}
\RightLabel{\textsf{Lin-Tm-Subst}}
\BinaryInfC{$\Gamma,x:A;w:B\vdash c'[ !{x}\otimes w/z]:C$}
\DisplayProof\hspace{20pt}\;}
\quad\\
\\
We show how the morphism from right to left is exactly $\Sigma_!^\otimes$-E (with $\Delta'=\cdot$, $\Delta= z:\Sigma^\otimes_{!(x:A)}B$, $t= z$).\\
\\
\resizebox{\linewidth}{!}{\AxiomC{$\Gamma;\cdot \vdash C\type$}
\AxiomC{}
\RightLabel{\textsf{Lin-Idf}}
\UnaryInfC{$\Gamma;z:\Sigma^\otimes_{!(x:A)}B \vdash z:\Sigma^\otimes_{!(x:A)}B$}
\AxiomC{$\Gamma,x:A;y:B\vdash c:C$}
\RightLabel{\textsf{$\Sigma_!^\otimes$-E}}
\TrinaryInfC{$\Gamma;z:\Sigma^\otimes_{!(x:A)}B\vdash \mathsf{let}\;z\;\mathsf{be}\;  !{x} \otimes y \;\mathsf{in}\;c:C$}
\DisplayProof\hspace{100pt}\;}
\\
\\
We show how Frobenius reciprocity can be proved in our type system (particularly relying on the form of the $\Sigma_!^\otimes$-E-rule\footnote{To be precise, we shall see Frobenius reciprocity is validated because we allow dependency on $\Delta'$ in the $\Sigma_!^\otimes$-E-rule. Conversely, it is easy to see we can prove Frobenius reciprocity holds in our model if we have (semantic) $\multimap$-types, as this allows us to remove the dependency on $\Delta'$ in $\Sigma_!^\otimes$-E.}).
\begin{claim*}[Frobenius reciprocity] The canonical morphism $$\Sigma^\otimes_{!{A}}(\Delta'\{\mathbf{p}_{\Gamma,{A}}\}\otimes B)\ra{f} \Delta'\otimes \Sigma^\otimes_{!{A}}B$$ is an isomorphism, for all $\Delta'\in\Dcat(\Gamma)$, $B\in\Dcat(\Gamma.{A})$.
\end{claim*}
\begin{proof} We first show how to construct the morphism $f$ we mean.\\
\\
\resizebox{\linewidth}{!}{
\AxiomC{}
\RightLabel{\textsf{Lin-Idf}}
\UnaryInfC{$\Gamma;x':\Sigma^\otimes_{!(x:A)}(\Delta'\otimes B)\vdash x':\Sigma^\otimes_{!(x:A)}(\Delta'\otimes B)$}
\AxiomC{}
\RightLabel{\textsf{Lin-Idf}}
\UnaryInfC{$\Gamma,x:A;z:\Delta'\vdash z:\Delta'$}
\AxiomC{}
\RightLabel{\textsf{Cart-Idf}}
\UnaryInfC{$\Gamma,x:A;\cdot \vdash x:A$}
\AxiomC{}
\RightLabel{\textsf{Lin-Idf}}
\UnaryInfC{$\Gamma;y:B\vdash y:B$}
\RightLabel{\textsf{$\Sigma_!^\otimes$-I}}
\BinaryInfC{$\Gamma,x:A;y:B\vdash  !{x} \otimes y:\Sigma^\otimes_{!(x:A)}B$}
\RightLabel{\textsf{$\otimes$-I}}
\BinaryInfC{$\Gamma,x:A;z:\Delta', y:B\vdash z\otimes  ! {x} \otimes y:  \Delta'\otimes \Sigma^\otimes_{!(x:A)}B$}
\RightLabel{\textsf{$\otimes$-E}}
\UnaryInfC{$\Gamma,x:A;w: \Delta'\otimes B\vdash \mathsf{let}\; w\;\mathsf{be}\;z\otimes y\;\mathsf{in}\;z\otimes  ! {x} \otimes y :\Delta'\otimes \Sigma^\otimes_{!(x:A)}B$}
\RightLabel{\textsf{$\Sigma_!^\otimes$-E}}
\BinaryInfC{$\Gamma;x':\Sigma^\otimes_{!(x:A)}(\Delta'\otimes B)\vdash f: \Delta'\otimes \Sigma^\otimes_{!(x:A)}B$}
\DisplayProof}\\
\\
We now construct its inverse. Call it $g$\footnote{Frobenius reciprocity really comes in where $\Sigma_!^\otimes$-E is used, because of the factor $\Delta'$ in the $\Sigma_!^\otimes$-E-rule.}.\\
\\
\resizebox{\linewidth}{!}{
\AxiomC{}
\RightLabel{\textsf{Lin-Idf}}
\UnaryInfC{$\Gamma;y_2:\Sigma^\otimes_{!(x:A)}B\vdash y_2:\Sigma^\otimes_{!(x:A)}B$}
\AxiomC{}
\RightLabel{\textsf{Cart-Idf}}
\UnaryInfC{$\Gamma,x:A;\cdot \vdash x:A$}
\AxiomC{}
\RightLabel{\textsf{Lin-Idf}}
\UnaryInfC{$\Gamma;y_1:\Delta'\vdash y_1:\Delta'$}
\AxiomC{}
\RightLabel{\textsf{Lin-Idf}}
\UnaryInfC{$\Gamma;y:B\vdash y:B$}
\RightLabel{\textsf{$\otimes$-I}}
\BinaryInfC{$\Gamma,x:A;y_1:\Delta',y:B\vdash y_1\otimes y:\Delta'\otimes B$}
\RightLabel{\textsf{$\Sigma_!^\otimes$-I}}
\BinaryInfC{$\Gamma,x:A;y_1:\Delta',y:B\vdash  !{x} \otimes  y_1\otimes y:\Sigma^\otimes_{!(x:A)}(\Delta'\otimes B)$}
\RightLabel{\textsf{$\Sigma_!^\otimes$-E}}
\BinaryInfC{$\Gamma;y_1:\Delta',y_2:\Sigma^\otimes_{!(x:A)}B\vdash \mathsf{let}\; y_2\;\mathsf{be}\; !{x} \otimes y\; \mathsf{in}\;  !{x} \otimes  y_1\otimes y:\Sigma^\otimes_{!(x:A)}(\Delta'\otimes B)$}
\RightLabel{\textsf{$\otimes$-E}}
\UnaryInfC{$\Gamma;y':\Delta'\otimes \Sigma^\otimes_{!(x:A)}B\vdash g:\Sigma^\otimes_{!(x:A)}(\Delta'\otimes B)$}
\DisplayProof}
\\
\\
We leave it to the reader to verify that these morphisms are mutually inverse in the sense that $$\Gamma;x':\Sigma^\otimes_{!(x:A)}(\Delta'\otimes B)\vdash g[f/y']= x':\Sigma^\otimes_{!(x:A)}(\Delta'\otimes B)$$ and $$\Gamma;y':\Delta'\otimes \Sigma^\otimes_{!(x:A)}B \vdash f[g/x']= y':\Delta'\otimes \Sigma^\otimes_{!(x:A)}B .$$\end{proof}
For the converse, we show how to obtain $\Sigma_!^\otimes$-I from our morphism from left to right:\\
\\
\resizebox{\linewidth}{!}{
\AxiomC{}
\RightLabel{\textsf{Lin-Idf}}
\UnaryInfC{$\Gamma;z:\Sigma^\otimes_{!(x:A)}B\vdash z:\Sigma^\otimes_{!(x:A)}B$}
\RightLabel{\textsf{``left to right''}}
\UnaryInfC{$\Gamma,x:A;w:B\vdash  !{x}\otimes w:\Sigma^\otimes_{!(x:A)}B$}
\AxiomC{$\Gamma;\cdot \vdash a:A$}
\RightLabel{\textsf{Cart-Tm-Subst}}
\BinaryInfC{$\Gamma;w:B\vdash  !{a}\otimes w:\Sigma^\otimes_{!(x:A)}B$}
\AxiomC{$\Gamma;\Delta\vdash b:B[{a}/x]$}
\RightLabel{\textsf{Lin-Tm-Subst}}
\BinaryInfC{$\Gamma;\Delta\vdash !{a}\otimes b:\Sigma^\otimes_{!(x:A)}B$}
\DisplayProof\hspace{10pt}\;}
\\
\\
We show how to obtain $\Sigma_!^\otimes$-E from our morphism from right to left, using Frobenius reciprocity.\\
\\
\resizebox{\linewidth}{!}{
\AxiomC{$\Gamma;\cdot\vdash C\type$}
\AxiomC{$\Gamma,x:A;y:\Delta', B\vdash c:C$}
\RightLabel{\textsf{$\otimes$-E}}
\UnaryInfC{$\Gamma,x:A;y:\Delta'\otimes B\vdash c:C$}
\RightLabel{\textsf{``right to left''}}
\BinaryInfC{$\Gamma;z:\Sigma^\otimes_{!(x:A)}(\Delta'\otimes B)\vdash \mathsf{let}\;z\;\mathsf{be}\;  !{x} \otimes y \;\mathsf{in}\;c:C$}
\RightLabel{\textsf{Frobenius reciprocity}}
\UnaryInfC{$\Gamma;z:(\Delta'\otimes \Sigma^\otimes_{!(x:A)}B)\vdash \mathsf{let}\;\frob{z}\;\mathsf{be}\;  !{x} \otimes y \;\mathsf{in}\;c:C$}
\RightLabel{\textsf{Lin-Tm-Subst,$\otimes$-I,2$\times$Lin-Idf}}
\UnaryInfC{$\Gamma;z_1:\Delta', z_2:\Sigma^\otimes_{!(x:A)}B\vdash \mathsf{let}\;\frob{z_1\otimes z_2}\;\mathsf{be}\;  !{x} \otimes y \;\mathsf{in}\;c:C$}
\AxiomC{$\Gamma;\Delta\vdash t: \Sigma^\otimes_{!(x:A)}B$}
\RightLabel{\textsf{Lin-Tm-Subst}}
\BinaryInfC{$\Gamma;z_1:\Delta',\Delta\vdash (\mathsf{let}\;\frob{z_1\otimes z_2}\;\mathsf{be}\;  !{x} \otimes y \;\mathsf{in}\;c)[t/z_2]:C$}
\DisplayProof}
\\
\\
As usual, the left Beck-Chevalley condition says precisely that $\Sigma_!^\otimes$-types commute with substitution, as dictated by the type theory.

\item Assume our model supports $\Pi_!^\multimap$-types. We exhibit the claimed adjunction. The morphism from left to right is provided by $\Pi_!^\multimap$-I -- in fact, it is exactly the I-rule -- and the one from right to left by $\Pi_!^\multimap$-E. $\Pi_!^\multimap$-$\beta$ and $\Pi_!^\multimap$-$\eta$ say exactly that these are mutually inverse. Naturality corresponds to the compatibility of $\Pi_!^\multimap$-I and $\Pi_!^\multimap$-E with substitution.
\begin{diagram}
b & \rMapsto & \lambda_{!(x:A)} b \\
\Dcat(\Gamma.{A})(\Delta\{\mathbf{p}_{\Gamma.{A}}\},B) & \pile{\rTo\\\cong \\ \lTo} & \Dcat(\Gamma)(\Delta,\Pi^\multimap_{!(x:A)} B)\\
f(!x) & \lMapsto & f.
\end{diagram}
We show how we obtain the definition of $f(!x)$ from $\Pi_!^\multimap$-E.\\
\\
\resizebox{\linewidth}{!}{
\AxiomC{}
\RightLabel{\textsf{Cart-Idf}}
\UnaryInfC{$\Gamma,x:A;\cdot \vdash x:A$}
\AxiomC{$\Gamma;\Delta \vdash f:\Pi^\multimap_{!(x:A)}B$}
\RightLabel{\textsf{Cart-Weak}}
\UnaryInfC{$\Gamma,x:A;\Delta \vdash f:(\Pi^\multimap_{!(x:A)}B)$}
\RightLabel{\textsf{$\Pi_!^\multimap$-E}}
\BinaryInfC{$\Gamma,x:A;\Delta\vdash f(!x):B$}
\DisplayProof
\hspace{145pt}\;}\\
\\
For the converse, we have to show that we can recover $\Pi_!^\multimap$-E from the definition of $f(!x)$.\\
\\
\resizebox{\linewidth}{!}{
\AxiomC{$\Gamma;\cdot\vdash a:A$}
\AxiomC{$\Gamma;\Delta \vdash f:\Pi^\multimap_{!(x:A)}B$}
\RightLabel{\textsf{Definition $f(!x)$}}
\UnaryInfC{$\Gamma,x:A;\Delta\vdash f(!x):B$}
\RightLabel{\textsf{Cart-Tm-Subst}}
\BinaryInfC{$\Gamma;\Delta\vdash f(!x)[{a}/x]:B[{a}/x]$}
\UnaryInfC{$\Gamma;\Delta\vdash f(!{a}):B[{a}/x]$}
\DisplayProof
\hspace{220pt}\;}\\
\\
This shows that individual $\Pi_!^\multimap$-types correspond to right adjoint functors to substitution along projections. The type theory dictates that $\Pi_!^\multimap$-types interact well with substitution. This corresponds to the right Beck-Chevalley condition, as usual.

\item From the categorical semantics of (non-dependent) linear type theory (see e.g. \cite{bierman1994intuitionistic} for a very complete account) we know that $\multimap$-types correspond to monoidal closure of the category of contexts. The extra feature in dependent linear type theory is that the syntax dictates that the type formers are compatible with substitution. This means that we also have to restrict the functors $\Dcat(f)$ to preserve the relevant categorical structure.

\item Idem.
\item Idem.
\item Assume that we have $!$-types. We define a left adjoint $F_\Gamma\dashv U_\Gamma$ as $F_\Gamma\mathbf{p}_{\Gamma,{A}}:=!A$ (this is easily seen to be well-defined up to isomorphism, so we can use AC for a definition on the nose) and, noting that every morphism $\mathbf{p}_{\Gamma,{A}}\ra{}\mathbf{p}_{\Gamma,{B}}$ in $\Bcat/\Gamma$ is of the form $\langle \mathbf{p}_{\Gamma,{A}},b\rangle$ for some unique $I\ra{b}B\{\mathbf{p}_{\Gamma,{A}}\}\in\Dcat(\Gamma.{A})$, we define $F_\Gamma$ as acting on $b$ as the map obtained from
\\
\\
\resizebox{\linewidth}{!}{
\AxiomC{$\Gamma,x:A;\cdot \vdash b:B$}
\RightLabel{\textsf{!-I}}
\UnaryInfC{$\Gamma,x:A;\cdot \vdash !b:!B$}
\AxiomC{}
\RightLabel{\textsf{Lin-Idf}}
\UnaryInfC{$\Gamma;y:!A\vdash y:!A$}
\RightLabel{\textsf{!-E}}
\BinaryInfC{$\Gamma;y:!A \vdash\mathsf{let}\; y\;\mathsf{be}\;!x\;\mathsf{in}\; !b:!B$}
\DisplayProof\hspace{250pt}\;}
\\
\\
which indeed gives us $F_\Gamma(\langle \mathbf{p}_{\Gamma,{A}},b\rangle)\in\Dcat(\Gamma)(!A,!B)$.

We exhibit the adjunction by the following isomorphism of hom-sets, where the morphism from left to right comes from $!$-I and the one from right to left comes from $!$-E.\\
\\
\resizebox{\linewidth}{!}{
\mbox{
\begin{diagram}
b & \rMapsto & b[!x/x']  \\
\Dcat(\Gamma)(F_\Gamma\mathbf{p}_{\Gamma,{A}},B)=\Dcat(\Gamma)(!A,B) &\pile{\rTo\\\cong\\ \lTo} &\Dcat(\Gamma.A)(I,B\{\proj{\Gamma}{A}\})\cong \Bcat/\Gamma(\mathbf{p}_{\Gamma,{A}},\mathbf{p}_{\Gamma,{B}})=\Ccat(\Gamma)(\mathbf{p}_{\Gamma,{A}},U_\Gamma B)\\
\mathsf{let}\; y\;\mathsf{be}\; !x\;\mathsf{in}\; b' & \lMapsto & b' 
\end{diagram}}}\\
\\
We show how to construct the morphism from left to right, using $!$-I.\\
\\
\resizebox{\linewidth}{!}{
\AxiomC{$\Gamma;x':!A\vdash b:B$}
\RightLabel{\textsf{Cart-Weak}}
\UnaryInfC{$\Gamma,x:A;x':!A\vdash b:B$}
\AxiomC{}
\RightLabel{\textsf{Cart-Idf}}
\UnaryInfC{$\Gamma,x:A;\cdot\vdash x:A$}
\RightLabel{\textsf{$!$-I}}
\UnaryInfC{$\Gamma,x:A;\cdot \vdash !x:!A$}
\RightLabel{\textsf{Lin-Tm-Subst}}
\BinaryInfC{$\Gamma,x:A;\cdot\vdash b[!x/x']:B$}
\DisplayProof\hspace{130pt}\;}\\
\\
We show to construct the morphism from right to left, using $!$-E. Suppose we're given $b'\in\Dcat(\Gamma.{A})(I,B\{\proj{\Gamma}{A}\})$. From this, we produce a morphism in $\Dcat(\Gamma)(!A,B)$ as follows.\\
\\
\resizebox{\linewidth}{!}{
\AxiomC{}
\RightLabel{\textsf{Lin-Idf}}
\UnaryInfC{$\Gamma;y:!A\vdash y:!A$}
\AxiomC{$\Gamma,x:A;\cdot \vdash b':B$}
\RightLabel{\textsf{!E}}
\BinaryInfC{$\Gamma;y:!A\vdash \mathsf{let}\; y\;\mathsf{be}\; !x\;\mathsf{in}\; b':B$}
\DisplayProof\hspace{230pt}\;}\\
\\
We leave it up to the reader to verify that these morphisms are mutually inverse, according to $!$-$\beta$ and $!$-$\eta$.

Note that $F_\Gamma$ is strong monoidal, as the rules for $!$ define a natural bijection between terms  $\Gamma;x':!A,y':!B\vdash t':C$ and $\Gamma,x:A,y:B;\cdot \vdash t:C$ if $\Gamma\vdash C\type$. In semantic terms, this gives a natural bijection
\begin{align*}
\Dcat(\Gamma)(!A\otimes !B,C)  &\cong \Dcat(\Gamma.A.B)(1,C\{\proj{\Gamma.A}{B};\proj{\Gamma}{A}\})\\
& \cong \Bcat/\Gamma(\proj{\Gamma.A}{B};\proj{\Gamma}{A}, \proj{\Gamma}{C})\\
&=  \Bcat/\Gamma(U_\Gamma A \times U_\Gamma B ,U_\Gamma C) \\
&\cong \Dcat(\Gamma)(F_\Gamma(U_\Gamma A \times U_\Gamma B),C) ,\end{align*} 
so strong monoidality follows by the Yoneda lemma. (A keen reader can verify that the oplax structure on $F_\Gamma$ corresponds with the lax structure on $U_\Gamma$.)

Conversely, suppose we have a strong monoidal left adjoint $F_\Gamma\dashv U_\Gamma$. We define, for $A\in\mathsf{ob}(\Dcat(\Gamma))$, $!A:=F_\Gamma U_\Gamma(A)$.

We verify that $!$-I can be derived from the homset morphism from left to right:\\
\\
\resizebox{\linewidth}{!}{
\AxiomC{}
\RightLabel{\textsf{Lin-Idf}}
\UnaryInfC{$\Gamma;x':!A\vdash x':!A$}
\RightLabel{\textsf{``left to right''}}
\UnaryInfC{$\Gamma,x:A;\cdot \vdash !x:!A$}
\AxiomC{$\Gamma;\cdot \vdash a:A$}
\RightLabel{\textsf{Cart-Tm-Subst}}
\BinaryInfC{$\Gamma;\cdot \vdash !x[{a}/x]:!A$}
\DisplayProof\hspace{280pt}\;}
\quad\\
\\
We verify that, in the presence of $\multimap$-types, $!$-E can be derived from the homset morphism from right to left:\\
\\
\resizebox{\linewidth}{!}{
\AxiomC{$\Gamma;\Delta\vdash t:!A$}
\AxiomC{}
\RightLabel{\textsf{Lin-Idf}}
\UnaryInfC{$\Gamma;w:\Delta'\vdash w:\Delta'$}
\AxiomC{$\Gamma,x:A;y:\Delta'\vdash b :B$}
\RightLabel{\textsf{$\multimap$-I}}
\UnaryInfC{$\Gamma,x:A;\cdot \vdash \lambda_{y:\Delta'}b :\Delta'\multimap B$}
\RightLabel{\textsf{``right to left''}}
\UnaryInfC{$\Gamma;z:!A\vdash  \mathsf{let}\; z\;\mathsf{be}\; !x\;\mathsf{in}\; \lambda_{y:\Delta'}b : \Delta'\multimap B$}
\RightLabel{\textsf{$\multimap$-E}}
\BinaryInfC{$\Gamma;z:!A,\Delta'\vdash  \mathsf{let}\; z\;\mathsf{be}\; !x\;\mathsf{in}\; b[w/y]:B$}
\RightLabel{\textsf{Lin-Tm-Subst}}
\BinaryInfC{$\Gamma;\Delta,\Delta'\vdash \mathsf{let}\; t\;\mathsf{be}\; !x\;\mathsf{in}\; b[w/y]: B$}
\DisplayProof\hspace{50pt}\;}
\quad\\
\\
Note that the $!$-$\beta$- and $!$-$\eta$-rules correspond precisely to the fact that our morphisms from left to right and from right to left define a homset isomorphism.

Finally, it is easily verified that the condition that $F_\Gamma;\Dcat(f)\cong \Dcat(f);F_{\Gamma'}$ corresponds exactly to the compatibility of $!$ with substitution.

\item Suppose we have $\Id^\otimes_{!A}\dashv -\{\mathsf{diag}_{\Gamma,A}\}$ (satisfying the appropriate Frobenius and Beck-Chevalley conditions). Then, we have a (natural) homset isomorphism
\begin{diagram}
\Dcat(\Gamma.A.A\{\proj{\Gamma}{A}\})(\Id^\otimes_{!A}(B),C) & \pile{\rTo\\ \cong \\ \lTo} & \Dcat(\Gamma.A)(B,C\{\mathsf{diag}_{\Gamma,A}\}).
\end{diagram}
The claim is that $\Id^\otimes_{!A}(I)$ satisfies the rules for the $\Id_!^\otimes$-type of $A$. Indeed, we have $\Id_!^\otimes$-I as follows.\\
\quad\\
\resizebox{\linewidth}{!}{
\AxiomC{}
\RightLabel{\textsf{Lin-Idf}}
\UnaryInfC{$\Gamma,x:A,x':A;w:\Id^\otimes_{!A}(I)(x,x')\vdash w:\Id^\otimes_{!A}(I)(x,x')$}
\RightLabel{\textsf{``left to right''}}
\UnaryInfC{$\Gamma,x:A;y:I\vdash \refl{!x}^y:\Id^\otimes_{!A}(I)(x,x)$}
\AxiomC{}
\RightLabel{\textsf{$I$-I}}
\UnaryInfC{$\Gamma,x:A;\cdot\vdash *:I$}
\RightLabel{\textsf{Lin-Tm-Subst}}
\BinaryInfC{$\Gamma,x:A;\cdot\vdash \refl{!x}:\Id^\otimes_{!A}(I)(x,x)$}
\AxiomC{$\Gamma;\cdot\vdash a:A$}
\LeftLabel{Cart-Tm-Subst}
\BinaryInfC{$\Gamma;\cdot\vdash \refl{!x}:\Id^\otimes_{!A}(I)(a,a)$}
\DisplayProof}
\quad\\
\\
We obtain $\Id_!^\otimes$-E as follows. Let $\Gamma,x:A,x':A;\cdot\vdash C\type$.\\
\quad\\
\resizebox{\linewidth}{!}{
\AxiomC{$\Gamma,x:A;B\vdash c:C[x/x']$}
\RightLabel{\textsf{``right to left''}}
\UnaryInfC{$\Gamma,x:A,x':A;\Id^\otimes_{!A}(B)\vdash c':C$}
\AxiomC{$\Gamma;\cdot\vdash a:A$}
\AxiomC{$\Gamma;\cdot\vdash a':A$}
\RightLabel{\textsf{Cart-Tm-Subst}}
\TrinaryInfC{$\Gamma;\Id^\otimes_{!A}(B)[a/x,a'/x']\vdash c'[a/x,a'/x']:C[a/x,a'/x']$}
\AxiomC{$\Gamma;B'\vdash p:\Id^\otimes_{!A}(I)[a/x,a'/x']$}
\LeftLabel{\textsf{Frobenius}}
\UnaryInfC{$\Gamma;B[a/x],B'\vdash p':\Id^\otimes_{!A}(B)[a/x,a'/x']$}
\LeftLabel{Lin-Tm-Subst}
\BinaryInfC{$\Gamma;B[a/x],B'\vdash \mathsf{let}\; (a,a',p)\;\mathsf{be}\;(z,z,\refl{!z})\;\mathsf{in}\; c:C[a/x,a'/x]$}
\DisplayProof}
\\
\\
Conversely, suppose we have $\Id_!^\otimes$-types. Then, define $\Id^\otimes_{!A}(B):=\Id^\otimes_{!A}\otimes B\{\proj{\Gamma.A}{A\{\proj{\Gamma}{A}\}}\}$, with the obvious extension on morphisms. (This immediately implies Frobenius reciprocity, clearly.) Then, we obtain the morphism ``left to right'' as follows.\\
\quad\\
\resizebox{\linewidth}{!}{
\AxiomC{$\Gamma,x:A,x':A;z:\Id^\otimes_{!A},y:B\vdash c:C$}
\RightLabel{\textsf{}}
\AxiomC{}
\RightLabel{\textsf{Cart-Idf}}
\UnaryInfC{$\Gamma,x:A;\cdot\vdash x:A$}
\RightLabel{\textsf{Cart-Tm-Subst}}
\BinaryInfC{$\Gamma,x:A;z:\Id^\otimes_{!A}[x/x'],y:B\vdash c[x/x']:C[x/x']$}
\AxiomC{}
\RightLabel{\textsf{Cart-Idf}}
\UnaryInfC{$\Gamma,x:A;\cdot\vdash x:A$}
\RightLabel{\textsf{$\Id_!^\otimes$-I}}
\UnaryInfC{$\Gamma,x:A;\cdot\vdash \refl{!x}:\Id^\otimes_{!A}(x,x)$}
\RightLabel{\textsf{Lin-Tm-Subst}}
\BinaryInfC{$\Gamma,x:A;y:B\vdash c': C[x/x']$}
\DisplayProof}
\quad\\
\\
The morphism ``right to left'' is obtained as follows.\\
\\
\resizebox{\linewidth}{!}{
\AxiomC{$\Gamma,x_0:A; y:B\vdash c:C[x_0/x_1]$}
\AxiomC{}
\RightLabel{\textsf{Lin-Idf}}
\UnaryInfC{$\Gamma,x_0:A,x_1:A;w:\Id^\otimes_{!A}\vdash w:\Id^\otimes_{!A}$}
\AxiomC{}
\RightLabel{\textsf{Cart-Idf}}
\UnaryInfC{$\Gamma,x_0:A,x_1:A;\cdot\vdash x_i:A$}
\RightLabel{\textsf{$\Id_!^\otimes$-E}}
\TrinaryInfC{$\Gamma,x_0:A,x_1:A; w:\Id^\otimes_{!A},y:B\vdash c':C$}
\DisplayProof}
\quad\\
\\
We leave it to the reader to verify that the $\Id_!^\otimes$-$\beta$- and $\Id_!^\otimes$-$\eta$-rules translate precisely into the ``right to left'' and ``left to right'' morphisms being inverse. As usual, the Beck-Chevalley condition corresponds to the compatibility of $\Id_!^\otimes$-types with substitution, while the Frobenius condition says that $\Id_{!A}^\otimes$-functors are entirely determined by the object $\Id_{!A}^\otimes(I)$.
\end{enumerate}
\end{proof}

The semantics of $!$ suggests an alternative definition for the notion of a comprehension: if we have $\Sigma_!^\otimes$-types in a strong sense, it is a derived notion!

\begin{theorem}[Lawvere Comprehension]\label{altcompr} Given a strict indexed monoidal category $(\Bcat,\Dcat)$ with left adjoints $\csigma{f}{}$ to $\Dcat(f)$ for arbitrary $\Gamma'\ra{f}\Gamma\in\Bcat$\mccorrect{,} satisfying the left Beck-Chevalley condition for all pullback squares, then we can define $\Bcat/\Gamma\ra{F_\Gamma}\Dcat(\Gamma)$ by
$$F_\Gamma(-):=\csigma{-}{I}.$$
In that case, $(\Bcat,\Dcat)$ has a comprehension schema iff $F_\Gamma$ has a right adjoint $U_\Gamma$ (which then automatically satisfies $ \Dcat(f);U_{\Gamma'}=U_\Gamma;\Dcat(f)$ for all $\Gamma'\ra{f}\Gamma\in\Bcat$). That is, our notion of comprehension generalises that of \cite{lawvere1970equality}.

In particular, if either condition is satisfied, it supports $!$-types iff $\Sigma^\otimes_{!}$ satisfies Frobenius reciprocity.
\end{theorem}
\begin{proof}Suppose that we have said right adjoints $U_\Gamma$. We construct a comprehension schema. 

This allows us to define $\mathbf{p}_{\Gamma,{A}}:=U_\Gamma(A)$ and note that we have natural isomorphisms\\
\begin{diagram} \Dcat(\Gamma')(I,A\{f\})& \rTo^\cong & \Dcat(\Gamma)(
\csigma{f}{} I_{\Gamma'},A)=\Dcat(\Gamma)(F_\Gamma f,A)& \rTo^\cong & \Bcat/\Gamma(f,U_\Gamma A) \\
a & \rMapsto & a_f& \rMapsto & \langle f,a\rangle, \end{diagram}
where the first natural isomorphism comes from the adjunction $\csigma{f}{}\dashv -\{f\}$ and the second one comes from the adjunction $F_\Gamma\dashv U_\Gamma$. This defines a comprehension for $\Dcat$.\\
\\
Conversely, suppose $\Dcat$ satisfies the comprehension schema. Then, we know, by theorem \ref{thm:comprfunc}, that we can define a comprehension functor $U_\Gamma$ such that $ \Dcat(f);U_{\Gamma'}= U_\Gamma;\Dcat(f)$. Then we have the following natural isomorphisms:
\begin{diagram} \Bcat/\Gamma(f,U_\Gamma A) & \rTo^\cong & \Dcat(\Gamma')(I,A\{f\})& \rTo^\cong & \Dcat(\Gamma)(\csigma{f}{} I_{\Gamma'},A)=\Dcat(\Gamma)(F_\Gamma f,A) \\
  \langle f,a\rangle & \rMapsto  & a & \rMapsto & a_f, \end{diagram}
where the first isomorphism is precisely the representation defined by our comprehension and the second isomorphism comes from the fact that $\csigma{f}{}\dashv -\{f\}$. We see that $F_\Gamma\vdash U_\Gamma$.

Finally, note that we have the following commutative triangle of natural isomorphisms
\begin{diagram}
\Dcat(\Gamma.A)(\Delta\{\proj{\Gamma}{A}\},B\{\proj{\Gamma}{A}\})  & \rTo^{\textnormal{$!$-types}}_\cong &  \Dcat(\Gamma)(!A\otimes \Delta,B)\\
 &\rdTo^\cong_{\textnormal{Definition $\Sigma_!^\otimes$}} &\dTo^\cong_{\textnormal{Frobenius}} \\
 & & \Dcat(\Gamma.A)(\Sigma_{!A}^\otimes \Delta\{\proj{\Gamma}{A}\},B).
\end{diagram}
Note that the Beck-Chevalley condition for \mccorrect{$\Sigma_F^\otimes$} takes care of the substitution condition for $!$-types. Therefore, the existence of $!$-types boils down to the top isomorphism. Meanwhile, the Frobenius condition is by the Yoneda lemma equivalent to the right isomorphism. Noting that the diagonal always holds if we have $\Sigma_!^\otimes$-types, it follows that we have $!$-types iff we have Frobenius reciprocity.
\end{proof}
\begin{theorem}[Type Formers in $\Ccat$]\label{thm:inttyp} $\Ccat$ supports $\Sigma$-types iff $\mathsf{ob}(\Ccat)$ is closed under compositions (as morphisms in $\Bcat$). It supports $\Id$-types iff $\mathsf{ob}(\Ccat)$ is closed under postcomposition with maps $\mathsf{diag}_{\Gamma,A}$. If $\Dcat$ supports $!$- and $\Pi_!^\multimap$-types, then $\Ccat$ supports $\Pi$-types. Moreover, we have that\\
\resizebox{\linewidth}{!}{
$\Sigma^\otimes_{!A}! B\cong F(\Sigma_{U A}UB) \qquad\qquad \Id^\otimes_{!A}(!B)\cong F\Id_{UA}(UB)\qquad\qquad U\Pi^\multimap_{!B}C\cong \Pi_{UB}UC.
$}\end{theorem}
\begin{proof}We write out the adjointness condition
\begin{align*}\Ccat(\Gamma)(\Sigma_{\proj{\Gamma}{B}}f,\proj{\Gamma}{D})&\stackrel{!}{\cong } \Ccat(\Gamma.B)(f,\proj{\Gamma}{D}\{\proj{\Gamma}{B}\})\\
&\cong \Ccat(\Gamma.B)(f,\mathbf{p}_{\Gamma,D\{\proj{\Gamma}{B}\}})\\
&\cong \Dcat(\Gamma.B.C)(I,D\{\proj{\Gamma}{B}\}\{f\})\\
&\cong \Dcat(\Gamma.B.C)(I,D\{ f;\proj{\Gamma}{B}\})\\
&\cong \Ccat(\Gamma)(f;\proj{\Gamma}{B},\proj{\Gamma}{D}).
\end{align*}
Now, the Yoneda lemma gives us that $\Sigma_{\proj{\Gamma}{B}}f= f;\proj{\Gamma}{B}$.\\\\
Similarly,
\begin{align*}
\Ccat(\Gamma.A.A)(\Id_{\proj{\Gamma}{A}}(f),\proj{\Gamma.A.A}{C})&\stackrel{!}{\cong} \Ccat(\Gamma.A)(f,\proj{\Gamma.A.A}{C}\{\mathsf{diag}_{\Gamma,A}\})\\
&\cong \Dcat(\Gamma.A.B)(I,C\{\mathsf{diag}_{\Gamma,A}\}\{f\})\\
&\cong \Dcat(\Gamma.A.B)(I,C\{f;\mathsf{diag}_{\Gamma,A} \})\\
&\cong \Ccat(\Gamma.A.A)(f;\mathsf{diag}_{\Gamma,A}, \proj{\Gamma.A.A}{C}),
\end{align*}
so $f;\mathsf{diag}_{\Gamma,A}$ models $\Id_{\proj{\Gamma}{A}}(f)$.\\
\\
Finally,
\begin{align*}
\Ccat(\Gamma)(U_\Gamma D,\Pi_{\proj{\Gamma}{B}}\proj{\Gamma.B}{C})&\stackrel{!}{\cong} \Ccat(\Gamma.B)((U_\Gamma D)\{\proj{\Gamma}{B}\},\proj{\Gamma.B}{C})\\
&\cong \Ccat(\Gamma.B)((U_\Gamma D)\{\proj{\Gamma}{B}\},U_{\Gamma.B}C)\\
&\cong \Dcat(\Gamma.B)(F_{\Gamma.B}((U_\Gamma D)\{\proj{\Gamma}{B}\}),C)\\
&\cong \Dcat(\Gamma.B)((F_{\Gamma}U_\Gamma D)\{\proj{\Gamma}{B}\},C)\\
&\cong \Dcat(\Gamma)(F_{\Gamma}U_\Gamma D,\Pi^\multimap_{!B}C)\\
&\cong \Ccat(\Gamma)(U_\Gamma D,U_{\Gamma}\Pi^\multimap_{!B}C).
\end{align*}
Again, using the Yoneda lemma, we conclude that $U_{\Gamma}\Pi^\multimap_{!B}C$ models $\Pi_{U_{\Gamma}B}U_{\Gamma.B}C$.\\
\\
In all cases, we have not worried about Beck-Chevalley (and Frobenius reciprocity for $\Sigma_!^\otimes$-types) as they are trivially seen to hold.\\
\\
Note that if $\Dcat$ has $\Sigma_!^\otimes$-types (and, therefore, $!$-types), then
\begin{align*}\Dcat(\Gamma)(F_\Gamma(\Sigma_{U_\Gamma A}U_{\Gamma.A}B),C)&\cong\Ccat(\Gamma)(\Sigma_{U_\Gamma A}U_{\Gamma.A} B,U_\Gamma C)\\
&\cong \Ccat(\Gamma.A)(U_{\Gamma.A} B,(U_\Gamma C)\{\proj{\Gamma}{A}\})\\
&\cong \Ccat(\Gamma.A)(U_{\Gamma.A} B,U_{\Gamma.A} (C\{\proj{\Gamma}{A}\}))\\
&\cong \Dcat(\Gamma.A)(! B,C\{\proj{\Gamma}{A}\})\\
&\cong \Dcat(\Gamma)(\Sigma^\otimes_{!A}! B,C).
\end{align*}
By the Yoneda lemma, conclude that $\Sigma^\otimes_{!A}! B\cong F_\Gamma (\Sigma_{U_\Gamma A}U_{\Gamma.A}B)$.\\
\\
Note that, in case $\Dcat$ admits $!$- and $\Id_!^\otimes$-types,
\begin{align*}\Dcat(\Gamma.A.A)(\Id^\otimes_{!A}(!B),C)&\cong \Dcat(\Gamma.A)(!B,C\{\mathsf{diag}_{\Gamma,A}\})\\
&\cong \Ccat(\Gamma.A)(U_{\Gamma.A} B,U_{\Gamma.A}(C\{\mathsf{diag}_{\Gamma,A}\}))\\
&\cong \Ccat(\Gamma.A)(U_{\Gamma.A} B,U_{\Gamma.A.A}(C)\{\mathsf{diag}_{\Gamma,A}\})\\
&\cong \Ccat(\Gamma.A.A)(( U_{\Gamma.A} B);\mathsf{diag}_{\Gamma,A},U_{\Gamma.A}(C))\\
&\cong \Ccat(\Gamma.A.A)(\Id_{U_\Gamma A}(U_{\Gamma.A} B),U_{\Gamma.A}(C))\\
&\cong \Dcat(\Gamma.A.A)(F_{\Gamma.A.A}\Id_{U_{\Gamma}A}(U_{\Gamma.A}B),C).
\end{align*}
We conclude that $\Id^\otimes_{!A}(!B)\cong F_{\Gamma.A.A}\Id_{U_\Gamma A}(U_{\Gamma.A}B)$ and in particular $\Id^\otimes_{!A}(I)\cong F_{\Gamma.A.A}\Id_{U_{\Gamma}A}(\id_{\Gamma.A})$. (The last statement is easily seen to also be valid in absence of $\top$-types.)
\end{proof}
\begin{remark}[Dependent Seely Isomorphisms?] Note that, in our setup, we have a version of the simply typed Seely isomorphisms in each fibre. Indeed, suppose $\Dcat$ supports $\top$-, $\&$\mccorrect{-}, and $!$-types. Then, $U_\Gamma(\top)=\id_\Gamma$ and $U_\Gamma(A\& B)=U_\Gamma(A)\times U_\Gamma(B)$, as $U_\Gamma$ has a left adjoint and therefore preserves products. Now, $F_\Gamma$ is strong monoidal and $!_\Gamma=F_\Gamma U_\Gamma$, so it follows that $!_\Gamma\top=I$ and $!_\Gamma(A\& B)=!_\Gamma A\otimes !_\Gamma B$.

Now, theorem \ref{thm:inttyp} suggests the possibility of similar Seely isomorphisms  for $\Sigma_!^\otimes$-types and $\Id_!^\otimes$-types. Indeed, $\Ccat$ supports $\Sigma$-types iff we have additive $\Sigma$-types in $\Dcat$ in the sense of objects $\Sigma_A^{\&} B$ such that
$$U\Sigma_A^{\&} B\cong \Sigma_{UA}UB\txt{and hence} !\Sigma_A^{\&} B\cong \Sigma^\otimes_{!A} !B.$$
In an ideal world, one would hope that $\Sigma_A^{\&} B$ generalise\mccorrect{s} $A\& B$ in a similar way as how $\Sigma^\otimes_{!A} B$ is a dependent generalisation of $!A\otimes B$. In fact, it is easily seen that such categorical $\Sigma^{\&}$-types precisely (soundly and completely) correspond with the syntactic rules of figure \ref{fig:additivesigma}, where we see a slight mismatch with $\&$-types in the sense that the introduction and elimination rules only apply for cartesian contexts (without linear assumptions), here.

\begin{figure}
\fbox{
\resizebox{\linewidth}{!}{
\begin{tabular}{ll}
\AxiomC{$\vdash \Gamma,x:A,y:B;\cdot\ctxt$}
\RightLabel{$\Sigma^{\&}$-\textsf{F}}
\UnaryInfC{$\Gamma\vdash \Sigma_{x:A}^{\&}B\type$}
\DisplayProof\hspace{70pt}\;
& \AxiomC{$ \Gamma;\cdot \vdash a:A$}
\AxiomC{$ \Gamma;\cdot \vdash b:B[a/x]$}
\RightLabel{$\Sigma^{\&}$-\textsf{I}}
\BinaryInfC{$\Gamma;\cdot\vdash \langle a,b\rangle: \Sigma_{x:A}^{\&}B$}
\DisplayProof\hspace{70pt}\;\\
&\\
\AxiomC{$\Gamma;\cdot\vdash t: \Sigma_{x:A}^{\&}B$}
\RightLabel{$\Sigma^{\&}$-\textsf{E1}}
\UnaryInfC{$\Gamma;\cdot\vdash \fst(t):A$}
\DisplayProof
&
\AxiomC{$\Gamma;\cdot\vdash t: \Sigma_{x:A}^{\&}B$}
\RightLabel{$\Sigma^{\&}$-\textsf{E2}}
\UnaryInfC{$\Gamma;\cdot\vdash \snd(t):B[\fst(t)/x]$}
\DisplayProof
\end{tabular}
}
}
\caption{\label{fig:additivesigma} Rules for additive $\Sigma$-types. We also demand the obvious $\beta$- and $\eta$-equations.}
\end{figure}

Similarly, we get a notion of additive $\Id$-types: $\Ccat$ supports $\Id$-types iff we have objects $\Id_A^{\&}(B)$ in $\Dcat$ such that
$$U\Id_A^{\&}(B)\cong \Id_{UA}(UB)\txt{and hence} !\Id_A^{\&}(B)\cong \Id_{!A}^\otimes(!B).$$
 Note that this suggests that, in the same way that $\Id^\otimes_{!A}(B)\cong \Id^\otimes_{!A}(I)\otimes B$ (a sense in which usual $\Id_!^\otimes$-types are multiplicative connectives), $\Id_A^{\&}(B)\cong \Id_A^{\&}(\top)\& B$. In fact, if we have $\top$- and $\&$-types, we only have to give $\Id_A^{\&}(\top)$ and can then {define} $\Id_A^{\&}(B):=\Id_A^{\&}(\top)\& B$ to obtain additive $\Id$-types in generality.

In the light of theorem \ref{thm:inttyp}, we obtain such additive $\Sigma$- and $\Id$-types in the fibre over $\Gamma$ if some $U_\Gamma$ is essentially surjective. In particular, we are in this situation if $F_\cdot \dashv U_\cdot$ is the usual co-Kleisli adjunction of $!_\cdot$, where $\Ccat(\cdot)\cong\Bcat$. This shows that if we are hoping to obtain a model of dDILL indexed over the co-Kleisli category, in the natural way, we need to support these additive connectives.

From experience, it seems like the natural models of dDILL do not generally support them\mccorrect{, meaning that co-Kleisli categories often fail to give models of dependent types}. Similarly, it is difficult to come up with an intuitive interpretation of the meaning of such connectives, in the sense of a resource interpretation.

To get some intuition of why such objects may be problematic, note that the usual resource interpretation $A\& B$ is as follows: we either have (a resource of type) $A$ or $B$. This means that we would expect a $\Sigma_A^{\&} B$-type, which should be a dependent generalisation of the ordinary $\&$-type, to have an additive reading too. However, $B$ represents a predicate on $A$, so, if we have an object $c$ of type $\Sigma^{\&}_AB$, we are in the situation that we can either produce an object $\fst c$ of type $A$ or an object $\snd c$ embodying a property $B$ of $\fst c$.

\begin{mccorrection}
This is like the Cheshire cat of Alice in Wonderland of figure \ref{fig:cat}: let $A$ be the type of cats and let $B$ be the predicate ``is grinning''. Then, $\fst c$ corresponds to the cat and $\snd c$ may be thought to embody having a grin (of a cat) without having the cat.
\end{mccorrection}

In section \ref{sec:depprojprod}, we shall see a similar problem in the operational semantics of terms of such types. Terms of linear types generally represent dynamic objects. We shall see that the term $\fst c$, like the Cheshire cat, can lose information (e.g. the cat disappears; the term $\fst c$, for instance, could be a computation which proceeds to make a non-deterministic choice or print to console) after which properties $\snd c$ which held true of $\fst c$ before the change no longer make sense (e.g. the cat is grinning; the program $\fst c$ is going to make a non-deterministic choice or prints $\mathtt{hello\; world}$ to console). 
\end{remark}
\begin{figure}[!tb]\centering \reflectbox{
\includegraphics[scale=.4]{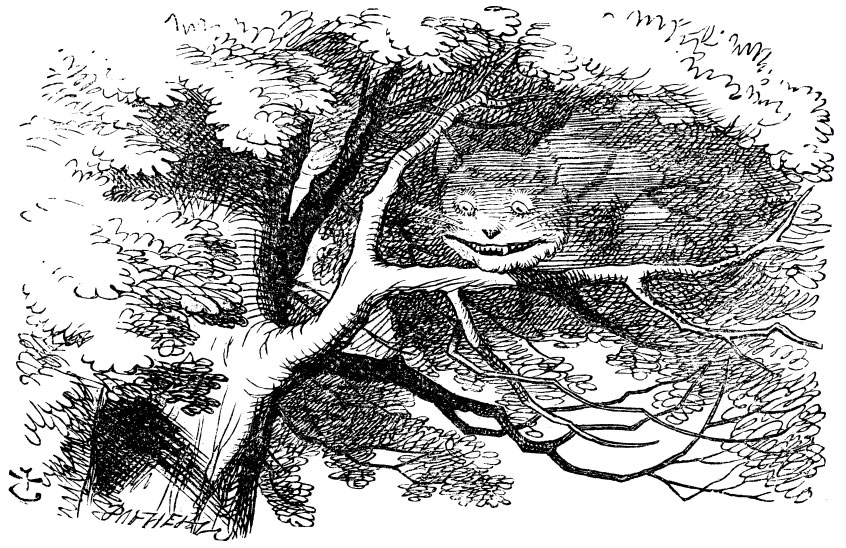}}
\caption{\label{fig:cat} We encourage the reader to compare the idea of additive $\Sigma$-types with Lewis Carroll's invention of the Cheshire cat. \\ \textit{``Well! I've often seen a cat without a grin,'' thought Alice; ``but a grin
without a cat! It's the most curious thing I ever saw in all my life.''.}~\cite{carroll1965annotated}}
\end{figure}
\section{dLNL Calculus}\label{sec:deplnl}
Independently from the author, Krishnaswami, Pradic and Benton developed a syntax for a dependently typed version of the LNL calculus in \cite{krishnaswami2015integrating}, which we refer to as the dLNL calculus. It is a system with both cartesian types and linear types both of which are allowed to depend on terms of cartesian types, but not linear types. The cartesian type formers they consider are $1$-, (strong) $\Sigma$-, $\Pi$- and extensional $\Id$-types as well as universes \mccorrect{(as cartesian types)} that code for both linear and cartesian types and, finally, $U$-types which map linear types to cartesian types. The linear type formers they consider are $I$-, $\otimes$-, $\multimap$-, $\top$-, $\&$-, $\csigma{-}{}$- and $\cpi{-}{}$-types and, finally, $F$-types which map cartesian types to linear types. In their work, they discuss an operational semantics but do not provide a denotational semantics. Therefore, we believe it might be useful to point out that our categorical framework can easily be adapted to model their dependent LNL calculus (minus universes, which can be given their usual awkward categorical semantics \cite{mendler1991predictive}). 

\begin{theorem}[Dependent LNL Calculus Semantics] A sound and complete categorical semantics for the universe-free fragment of the dependent LNL calculus of \cite{krishnaswami2015integrating} is given by the following:
\begin{itemize}
\item a model $\Bcat^{op}\ra{\Ccat}\Cat$ of pure cartesian dependent type theory in the sense of an indexed category (with an indexed terminal object) with full and faithful comprehension $(\mathbf{p},\mathbf{v},\langle -- ,- \rangle)$;
\item strong $\Sigma$-types in $\Ccat$;
\item strong \mccorrect{(extensional)} $\Id$-types in $\Ccat$;
\item $\Pi$-types in $\Ccat$;
\item an indexed symmetric monoidal closed category $\Bcat^{op}\ra{\Dcat}\mathsf{SMCCat}$;
\item indexed finite products $(\top,\&)$ in $\Dcat$;
\item a linear/non-linear indexed adjunction\footnote{Note that it is, in fact, enough to merely ask for an indexed functor $\Dcat\ra{U}\Ccat$, as we then automatically obtain a linear/non-linear adjunction by defining $F A:=\csigma{A}{I}$. } $F\dashv U:\Ccat\leftrightarrows \Dcat$;
\item $\csigma{-}{}$-types in $\Dcat$ in the sense of left adjoints $\csigma{A}{}\dashv \Dcat(\proj{\Gamma}{A})$, for all display maps $\Gamma.A\ra{\proj{\Gamma}{A}}\Gamma$, satisfying the left Beck-Chevalley condition for all $\mathbf{p}$-squares and Frobenius reciprocity in the sense that the canonical maps 
$$
\csigma{A}{(\Delta'\{\mathbf{p}_{\Gamma,{A}}\}\otimes B)}\ra{} \Delta'\otimes \csigma{{A}}{B}
$$
are isomorphisms;
\item $\cpi{-}{}$-types in $\Dcat$ in the sense of right adjoints $\Dcat(\proj{\Gamma}{A})\dashv \cpi{A}{}$ satisfying the right Beck-Chevalley condition for all $\mathbf{p}$-squares.
\end{itemize}
\end{theorem}
\begin{proof}
The proof is entirely analogous to that for the categorical semantics of dependently typed DILL.
\end{proof}
Again, we see that the main distinction between DILL and the LNL calculus is that the latter considers the cartesian category $\Ccat$ part of the structure while the DILL only assumes its existence. Note, in particular, that a model of the dependent LNL calculus in the sense described above defines a model of dependently typed DILL with $I,\otimes,\multimap,\top,\&,!,\Sigma_!^\otimes$ and $\Pi_!^\otimes$-types. Indeed, the comprehension on $\Ccat$ together with the adjunction $F\dashv U$ is easily seen to give a comprehension for $\Dcat$. However, we see that as $U$ may not be full and faithful, a full and faithful comprehension for $\Ccat$ may not give a full and faithful comprehension for $\Dcat$ (just like in our semantics for dependently typed DILL).

\section{Girard Translations}\label{sec:depgirard}
By contrast with the simply typed situation, I believe there are reasons to prefer an LNL calculus over a DILL-style system when working with dependent types. This is because the Girard translations fail for dDILL, meaning that it does not suffice to encode cartesian dependent type theory. Meanwhile, the LNL calculus is a proper extension of cartesian type theory, so, in particular, has more expressive power. Later, in chapter \ref{ch:3}, we shall see another reason to prefer a LNL-style calculus, motivated by separating proving and programming. The idea is that cartesian types are for pure proofs while linear types are assigned to (commutative) effectful programs.

Let us explain why the Girard translations of section \ref{sec:girardtrans} do not generalise well to dependent type theory in their conventional form. While we can define without any problems $(\Pi_A B)^f:=\Pi_{!A^f}^\multimap B^f$, we run into problems with the first Girard translation of $\Sigma$-types. One would expect the first Girard translation of $\Sigma_{A}B$ to take the form of $\Sigma$-types $\Sigma_{A^f}^{\&} B^f$. We can indeed define this, but we know such connectives are often problematic from a denotational and operational point of view. Similarly, we would expect $(\Id_A)^f:=\Id_{A^f}^{\&}$.

The second Girard translation is even more problematic. We would expect to define $(\Pi_AB)^s:=!\Pi_{A^s}^\multimap {B^s}$ for some linear dependent function type $\Pi_A^\multimap B$ which generalises $A\multimap B$. We encounter a similar problem when defining $(\Sigma_AB)^s$ and $(\Id_A)^s$ which we would expect to be $\Sigma_{A^s}^\otimes B^s$ and $\Id_{A^s}^\otimes$ for some dependent connectives $\Sigma^\otimes$ generalising $\otimes$ and a connective $\Id^\otimes$ which takes the multiplicative identity type of linear terms. Such $\Pi^\multimap$, $\Sigma^\otimes$ and $\Id^\otimes$ types are even more problematic, in a sense, than additive $\Sigma$-type\mccorrect{s}, as we simply cannot formulate natural deduction rules for such connectives in a system with types depending only on cartesian assumptions and not linear ones.

One can wonder if a satisfactory translation $(-)^t$ can be obtained by mixing the first and second Girard translations: using $(\Pi_{A}B)^t:=\Pi_{!A^t}B^t$, $(\Sigma_{A}B)^t:=\Sigma_{!A^t}^\otimes !B^t$ and $(\Id_A)^t:=\Id_{!A}^\otimes$. It is easily seen that this leads to violation of the $\eta$-rules for $\Sigma$-types (as we are effectively modelling pattern matching $\Sigma$-types in CBN) and $\Id$-types. It is at present unclear to us if such a translation still has value.

At the level of categorical semantics, the first Girard translation takes the form of the idea of modelling cartesian type theory in the co-Kleisli category $\Dcat_!$ for $!$ of a model $\Dcat$ of linear type theory. If we start with a model $\Bcat^{op}\ra{\Dcat}\Cat$ of dDILL with $!$-types and $\top$-types, it can easily be seen that the fibrewise co-Kleisli category $\Dcat_!$ is a model of cartesian dependent type theory (with full and faithful comprehension). Indeed, $\Dcat_!(\Gamma')(\top,A\{f\})\cong\Dcat(I,A\{f\})\cong\Bcat/\Gamma(f,\proj{\Gamma}{A})$ and
\begin{align*}
\Dcat_!(\Gamma')(A,B)&\cong \Dcat(\Gamma)(!A,B)\\
&\cong \Dcat(\Gamma.A)(I,B\{\proj{\Gamma}{A}\})\\
&\cong \Bcat/\Gamma(\proj{\Gamma}{A},\proj{\Gamma}{B}).
\end{align*}
$\Sigma$-types in $\Dcat_!$ are easily seen to correspond precisely to $\Sa$-types in $\Dcat$, which do not generally exist.

We can gain more insight into what is going on, by embedding the co-Kleisli category in the co-Eilenberg-Moore category, effectively closing it under certain equalisers (as every coalgebra has a presentation in terms of cofree coalgebras). Under that embedding, $A\& B$ is mapped to the cofree coalgebra $!A\otimes !B \cong !(A\&B)\ra{\delta_{A\& B}}!!(A\& B)\cong !(!A\otimes !B)$. Hence, we would expect $\Sa_A B$ to embed as a coalgebra $\Sigma_{!A}^\otimes !B\ra{}! \Sigma_{!A}^\otimes !B$. On closer inspection, it turns out that while (as in the simply typed case \cite{benton1995mixed}) we can always define a canonical coalgebra structure\footnote{Indeed, given a coalgebra $B\ra{k}!B$ in $\Dcat(\Gamma.A)$, we can define the coalgebra\\
\\
\resizebox{\linewidth}{!}{
\AxiomC{}
\UnaryInfC{$\Gamma,x:A,y:B;\cdot\vdash x:A$}
\AxiomC{}
\UnaryInfC{$\Gamma,x:A,y:B;\cdot\vdash y:B$}
\BinaryInfC{$\Gamma,x:A,y:B;\cdot\vdash !x\otimes y:\Sigma^\otimes_{!A}B$}
\UnaryInfC{$\Gamma,x:A,y:B;\cdot\vdash !(!x\otimes y):!\Sigma^\otimes_{!A}B$}
\UnaryInfC{$\Gamma,x:A;z:!B\vdash \lbi{z}{!y}{!(!x\otimes y)}:!\Sigma^\otimes_{!A}B$}
\AxiomC{}
\UnaryInfC{$\Gamma,x:A;w:B\vdash k:!B$}
\BinaryInfC{$\Gamma,x:A;w:B\vdash \lbi{k}{!y}{!(!x\otimes y)}:!\Sigma^\otimes_{!A}B$}
\UnaryInfC{$\Gamma;v:\Sigma_{!A}^\otimes B\vdash \lbi{v}{!x\otimes w}{\lbi{k}{!y}{!(!x\otimes y)}}:!\Sigma^\otimes_{!A}B$.}
\DisplayProof\hspace{120pt}\;
}\\
\\
We can, in particular, do this for the case that $k=\delta_C$.} on $\Sigma_{!A}^\otimes !B$, this coalgebra structure is not always a cofree one $\delta_{\Sa_AB}$.

However, while we can always define $\Sigma$-types in the co-Eilenberg-Moore category, we have no guarantee that $\Pi$-types exist, by contrast with the co-Kleisli category where $\Pi_{!A}^\multimap B$ is the $\Pi$-type of $A$ and $B$. Recalling that the co-Kleisli adjunction is the initial adjunction giving rise to $!$ and the co-Eilenberg-Moore one the terminal, we can hope to find a sweet spot $\Ccat$ in between $\Dcat_!$ and $\Dcat^!$ which is closed under both $\Sigma$- and $\Pi$-types. If no natural candidate for $\Ccat$ is available, we could, for instance,  try to inductively close $\Dcat_!$ under $\Sigma$-types in $\Dcat^!$ (we work with formal $\Sigma$-types of cofree coalgebras) or coinductively close $\Dcat^!$ under the $\Pi$-types of $\Dcat_!$ (we work with a category of exponentiable coalgebras). We choose to employ the former technique in the setting of game semantics. As we have an isomorphism of types $\Pi_{x:A}\Sigma_{y:B}C\cong \Sigma_{f:\Pi_{x:A}B}\Pi_{x:A}C[f(x)/y]$, we would expect these (co)inductively constructed categories to be  closed under both $\Sigma$- and $\Pi$-types. 

\section{Concrete Models}\label{sec:concretemodels}

\subsection{Some Discrete Models: Monoidal Families}
\label{sec:dismod}
We discuss a simple class of models in terms of families with values in a symmetric monoidal category. On a logical level, what the construction boils down to is starting with a model $\mathcal{V}$ of a linear propositional logic and taking the cofree linear predicate logic on $\Set$ with values in this propositional logic. This important example illustrates how $\Sigma_!^\otimes$- and $\Pi_!^\multimap$-types can represent infinitary additive disjunctions and conjunctions. The model is discrete in nature, however, and in that respect not representative for the type theory.

Suppose $\mathcal{V}$ is a symmetric monoidal category. We can then consider a strict $\Set$-indexed category, defined through the following enriched Yoneda embedding $\Fam(\mathcal{V}):={\mathcal{V}}^{-}:=\mathsf{SMCat}(-,\mathcal{V})$:
\begin{diagram}
\Set^{op} & \rTo^{\Fam(\mathcal{V})} &\mathsf{SMCat} & & & S\ra{f}S' & & \rMapsto & \mathcal{V}^S\stackrel{ f;-}{\longleftarrow} \mathcal{V}^{S'}.
\end{diagram}
Note that this definition naturally extends to a functor $\Fam$.
\begin{theorem}[Families Model dDILL] \label{thm:familiesmodelddill}The construction $\Fam$ adds type dependency on $\Set$ cofreely in the sense that it is right adjoint to the forgetful functor $\mathsf{ev}_1$ that evaluates a model of dDILL at the empty context to obtain a model of linear propositional type theory (where $\mathsf{SMCat}_{\mathsf{compr}}^{\Set^{op}}$ is the full subcategory of $\mathsf{SMCat}^{\Set^{op}}$ on the objects with comprehension):
\begin{diagram} 
\mathsf{SMCat} & \pile{\lTo^{\mathsf{ev}_1} \\ \bot \\ \rInto_{\Fam}} & \mathsf{SMCat}_{\mathsf{compr}}^{\Set^{op}}.
\end{diagram}
\end{theorem}
\begin{proof}$\Fam(\mathcal{V})$ admits a comprehension, by the following isomorphism
\begin{align*}
\Fam(\mathcal{V})(S)(I,B\{f\})&=\mathcal{V}^S(I, f;B)\\
&=\Pi_{s\in S}\mathcal{V}(I,B(f(s)))\\
&\cong \Set/S(S\ra{\id_S}S, \Sigma_{s\in S}\mathcal{V}(I,B(f(s)))\ra{\mathsf{fst}}S)\\
&\cong \Set/S'(S\ra{f}S',\Sigma_{s'\in S'}\mathcal{V}(I,B(s'))\ra{\mathsf{fst}}S')\\
&=\Set/S'(f,\mathbf{p}_{S',{B}}),
\end{align*}
where $\mathbf{p}_{S',{B}}:=\Sigma_{s'\in S'}\mathcal{V}(I,B(s'))\ra{\mathsf{fst}}S'$. ($\mathbf{v}_{S',{B}}$ is obtained as the image of $\id_{S'}\in \Set/S'$ under this isomorphism.)
To see that $\mathsf{ev}_1\dashv \Fam$, note that we have the following naturality diagrams for elements $1\ra{s}S$
\begin{diagram}
1& \quad \mathsf{ev}_1(\Dcat)=\Dcat(1) & \rTo^{\quad\phi_1\quad} & \mathcal{V}=\Fam(\mathcal{V})(1)\\
\dTo^s & \uTo^{-\{s\}} & & \uTo_{ s;-}\\
S & \Dcat(S) & \rTo_{\phi_S\quad} & \mathcal{V}^S=\Fam(\mathcal{V})(S)
\end{diagram}
and that all $1\ra{s}S$ are jointly surjective and therefore all $ s;-$ are jointly injective, meaning that a natural transformation $\phi\in\mathsf{SMCat}^{\Set^{op}}(\Dcat,\Fam(\mathcal{V}))$ is uniquely determined by $\phi_1\in\SMCat(\mathsf{ev}_1(\Dcat),\mathcal{V})$.
\end{proof}

We have the following results for type formers.

\begin{theorem}[Type Formers for Families]$\mathcal{V}$ has small coproducts that distribute over $\otimes$ iff $\Fam(\mathcal{V})$ supports $\Sigma_!^\otimes$-types. In that case, $\Fam(\mathcal{V})$ also supports $0$- and $\oplus$-types (which correspond precisely to finite distributive coproducts).

$\mathcal{V}$ has small products iff $\Fam(\mathcal{V})$ supports $\Pi_!^\multimap$-types. In that case, $\Fam(\mathcal{V})$ also supports $\top$- and $\&$-types (which correspond precisely to finite products).

$\Fam(\mathcal{V})$ supports $\multimap$-types iff $\mathcal{V}$ is monoidal closed. 

$\Fam(\mathcal{V})$ supports $!$-types iff $\mathcal{V}$ has small coproducts of $I$ that are preserved by $\otimes$ in the sense that the canonical morphism $\bigoplus_S(\Delta'\otimes I)\ra{}\Delta'\otimes \bigoplus_S I$ is an isomorphism for any $\Delta'\in\mathsf{ob}\;\mathcal{V}$ and $S\in\mathsf{ob}\;\Set$. In particular, if $\Fam(\mathcal{V})$ supports $\Sigma_!^\otimes$-types, then it also supports $!$-types.

$\Fam(\mathcal{V})$ supports $\Id_!^\otimes$-types if $\mathcal{V}$ has a distributive initial object. Supposing that $\mathcal{V}$ has a terminal object, the only if also holds.
\end{theorem}\begin{proof}The statement about $0$-, $\oplus$-, $\top$-, and $\&$-types should be clear from the previous sections, as products and coproducts in $\mathcal{V}^S$ are pointwise (and hence automatically preserved under substitution).

We  denote coproducts in $\mathcal{V}$ with $\bigoplus$. Then,
\begin{align*}
\Pi_{s'\in S'}\mathcal{V}(\bigoplus_{s\in f^{-1}(s')}A(s),B(s'))
&\cong \Pi_{s'\in S'}\Pi_{s\in f^{-1}(s')}\mathcal{V}(A(s),B(s'))\\
&\cong \Pi_{s\in \Sigma_{s'\in S'}f^{-1}(s')}\mathcal{V}(A(s),B(f(s)))\\
&\cong \Pi_{s\in S}\mathcal{V}(A(s),B(f(s))\\
&=\mathcal{V}^S(A, f;B).
\end{align*}
So, we see that we can define $\csigma{f}{}(A)(s'):=\bigoplus_{s\in f^{-1}(s')}A(s)$ to get a left adjoint $\csigma{f}\dashv -\{f\}$, if we have coproducts. (With the obvious definition on morphisms coming from the cocartesian monoidal structure on $\mathcal{V}$.) Conversely, we can clearly use $\csigma{f}{}$ to define any coproduct by using, for instance, an identity function for $f$ on the set we want to take a coproduct over and a family $A$ that denotes the objects we want to sum. The Beck-Chevalley condition is taken care of by the fact that our substitution morphisms are given by precomposition. Frobenius reciprocity precisely corresponds to distributivity of the coproducts over $\otimes$.

Similarly, if $\mathcal{V}$ has products, we  denote them with $\bigwith$ to suggest the connections with linear type theory. In that case, we can define $\cpi{f}{}(A)(s'):=\bigwith_{s\in f^{-1}(s')}A(s)$ to get a right adjoint $-\{f\}\dashv \cpi{f}{}$. (With the obvious definition on morphisms coming from the cartesian monoidal structure on $\mathcal{V}$.) Indeed,
\begin{align*}
\Pi_{s'\in S'}\mathcal{V}(B(s'),\bigwith_{s\in f^{-1}(s')}A(s))&\cong \Pi_{s'\in S'}\Pi_{s\in f^{-1}(s')}\mathcal{V}(B(s'),A(s))\\
&\cong \Pi_{s\in \Sigma_{s'\in S'}f^{-1}(s')}\mathcal{V}(B(f(s)),A(s))\\
&\cong \Pi_{s\in S}\mathcal{V}(B(f(s)),A(s))\\
&=\mathcal{V}^S( f;B,A).
\end{align*}
Again, in the same way as before, we can construct any product using $\cpi{f}{}$. The right Beck-Chevalley condition comes for free as our substitution morphisms are precomposition.

The claim about $\multimap$-types follows immediately from the previous section: $\Fam(\mathcal{V})$ supports $\multimap$-types iff all its fibres have a monoidal closed structure that is preserved by the substitution functors. Seeing that our monoidal structure is pointwise, the same will hold for any monoidal closed structure. Seeing that substitution is given by precomposition, the preservation requirement comes for free.

The characterisation of $!$-types is given by theorem \ref{thm:!fromsigma}, which tells us we can define $!A:=\Sigma_{F(\mathbf{p}_{S',{A}})}^\otimes I=s'\mapsto\bigoplus_{\mathcal{V}(I,A(s'))}I$ and conversely.

Finally, for $\Id_!^\otimes$-types, note that the adjointness condition $\Id^\otimes_{!A}\dashv -\{\mathsf{diag}_{\Gamma,A}\}$ boils down to the requirement (*)\\
\resizebox{\linewidth}{!}{
\begin{tabular}{l}\parbox{\linewidth}{
\begin{align*}\Pi_{s\in S}\Pi_{a\in A(s)}\mathcal{V}(B(s,a),C(s,a,a))&\cong\mathcal{V}^{\Sigma_{s\in S}A(s)}(B,C\{\mathsf{diag}_{S,A}\})\\
&\stackrel{!}{\cong} \mathcal{V}^{\Sigma_{s\in S}A(s)\times A(s)}(\Id^\otimes_{!A}(B),C)\\
&\cong \Pi_{s\in S}\Pi_{a\in A(s)}\Pi_{a'\in A(s)}\mathcal{V}(\Id^\otimes_{!A}(B)(s,a,a'),C(s,a,a')).
\end{align*}}\end{tabular}
}\\
We see that if we have an initial object $0\in\mathsf{ob}(\mathcal{V})$, we can define $$\Id^\otimes_{!A}(B)(s,a,a'):=\left\{\begin{array}{l}B(s,a)\textnormal{ if $a=a'$}\\ 0\textnormal{ else} \end{array}\right.$$
Distributivity of the initial object then gives us Frobenius reciprocity. For a partial converse, suppose we have a terminal object $\top\in \mathcal{V}$. Let $V\in\mathsf{ob}(\mathcal{V})$. Let $S:=\{*\}$, $A:=\{0,1\}$ and $C$ s.t. $C(0,0)=C(1,1)=C(0,1)=\top$ and $C(1,0)=V$. Then, (*) becomes the condition that $\{*\}\cong \mathcal{V}(\Id^\otimes_{!A}(B)(1,0),V)$. We conclude that $\Id^\otimes_{!A}(B)(1,0)$ is initial in $\mathcal{V}$.
\end{proof}

\begin{remark}Note that an obvious way to guarantee distributivity of coproducts over $\otimes$ is by demanding that $\mathcal{V}$ is monoidal closed.
\end{remark}
\begin{remark}It is easily seen that $\Sigma$-types in $\Ccat$, or additive $\Sigma$-types in $\Dcat=\Fam(\mathcal{V})$, boil down to having an object $\mathsf{or}_{s\in S}C(s)\in\mathsf{ob}(\mathcal{V})$ for a family $(C(s)\in \mathsf{ob}(\mathcal{V}))_{s\in S}$ such that $\Sigma_{s\in S}\mathcal{V}(I,C(s))\cong \mathcal{V}(I,\mathsf{or}_{s\in S} C(s))$. Similarly, $\Id$-types in $\Ccat$, or additive $\Id$-types in $\Dcat$, boil down to having objects $\mathsf{one},\mathsf{zero}\in\mathsf{ob}(\mathcal{V})$ such that $\mathcal{V}(I,\mathsf{one})\cong 1$ and $\mathcal{V}(I,\mathsf{zero})=0$.
\end{remark}
Two particularly simple concrete examples of $\mathcal{V}$ come to mind that can accommodate all type formers (except additive $\Sigma$- and $\Id$-types, which are easily seen not to be supported) and form a nice illustration: a category $\mathcal{V}=\mathsf{Vect}_F$ of vector spaces over a field $F$, with the tensor product, and the category $\mathcal{V}=\Set_*$ of pointed sets, with the smash product. All type formers get their obvious interpretation, but let us stop to think about $!$ for a second as it is a novelty of dDILL that it gets uniquely determined by the indexing, while in propositional linear type theory we might have several different choices. In the first example, $!$ boils down to the following: $(!B)(s')=\bigoplus_{ \mathsf{Vect}_F(F,B(s'))}F\cong   \bigoplus_{ B(s')}F$, i.e. taking the vector space freely spanned by all the vectors. In the second example, $(!B)(s')=\bigoplus_{\Set_*(2_*,B(s'))}2_*=\bigvee_{B(s')}2_*=B(s')+\{*\}$, i.e. $!$ freely adds a new basepoint.

We note the following consequence of theorem \ref{thm:familiesmodelddill}.
\begin{theorem}[DTT, DILL$\subsetneq$dDILL] dDILL is a proper generalisation of DTT and DILL: we have inclusions of the classes of models DTT,DILL$\subsetneq$dDILL.\end{theorem}
\begin{proof}Models of DTT with $1$- and $\times$-types, i.e. indexed cartesian monoidal categories with full and faithful comprehension, clearly, are a special case of our notion of model of dDILL. Moreover, in such cases, we easily see that $!A\cong A$. From their categorical descriptions, it is also clear that the other connectives of dDILL reduce to those of DTT. This proves the inclusion DDT$\subseteq$dDILL.

The dDILL models described above based on symmetric monoidal families are clearly more general than those of DTT, as we are dealing with a non-cartesian symmetric monoidal structure on the fibre categories. This proves that the inclusion is proper.

We have seen that the $\Fam$-construction realises the category of models of DILL as a reflective subcategory of the category of models of dDILL. Moreover, from various non-trivial models of DTT indexed over other categories than $\Set$ it is clear that this inclusion is proper as well.

Finally, we note that these inclusions still remain valid in the sub-algebraic setting where we do not have $I$- and $\otimes$-types. A simple variation of the argument using multicategories rather than monoidal categories does the trick.\end{proof}
Although this class of families models is important, it is clear that it only represents a very limited part of the generality of dDILL: not every model of dDILL is either a model of DTT or of DILL. Hence, we are in need of models that are less discrete in nature but still linear, if we are hoping to observe interesting new phenomena arising from the connectives of dDILL.

\subsection{Commutative Effects}\label{sec:ddillcommeffects}
As in the simply typed situation, commutative effects in dependent type theory give rise to linear types (under mild completeness and cocompleteness assumptions).
\begin{theorem}
Suppose we are given a model $\Ccat$ of pure dependent type  theory with $1,\Sigma,0,+,\Pi$-types which is equipped with an indexed commutative monad $T$, where $\Ccat$ additionally has equalisers and $\Ccat^T$ has reflexive coequalisers. Then, $\Ccat^T$ is a model of dDILL with $I,\otimes, \multimap, \top,\&,!,\Sigma^\otimes_!,\Pi_!^\multimap$-types.
\end{theorem}
\begin{proof}
The interpretation of $\top,\&,\Pi_!^\multimap,!$-types follows from  theorem \ref{thm:dcbpv-modelsfrommonads}. Meanwhile, an indexed variation of theorem \ref{thm:commtolinear} gives us the interpretation of $I,\otimes,\multimap$-types. Finally, theorem \ref{thm:dcbpvmodels} lets us interpret $\Sigma_!^\otimes$-types.
\end{proof}

\subsection{A Double Glueing Construction}\label{sec:doublegluing}
Of course, any model $\Ccat$ of cartesian type theory is a degenerate model $\Dcat=\Ccat$ of linear type theory in which the additive and multiplicative connectives coincide and where we can define $!$ to be the identity to obtain $\Ccat=\Dcat_!$. This shows us that every model of dependent type theory is trivially obtained through a co-Kleisli construction on a model of dDILL. This shows, in particular, in a rather boring way, that $\Sa$- and $\Ida$-types are consistent.

A more interesting construction is the following, which arises as a simple case of the double gluing construction of \cite{hyland2003glueing}, saying that every model of propositional intuitionistic logic arises from a model of classical linear logic, as a category of cofree $!$-coalgebras. Note that this is a properly linear model in the sense that its symmetric monoidal structure is not cartesian, i.e. there is a real difference between additive and multiplicative connectives. This follows from Joyal's lemma which says that cartesian $*$-autonomous categories are preorders \cite{lambek1988introduction}.
\begin{theorem}
Let $(\Ccat,1,\times,\Rightarrow)$ be a cartesian closed category. Then, $\Dcat:= \Ccat\times \Ccat^{op}$ can be given the structure of a $*$-autonomous category equipped with a linear exponential comonad $!$, such that $\Ccat\cong \Dcat_!$.
\end{theorem}
\begin{proof}
We see $\Dcat$ as a special case of the Chu-construction, where the pairing takes values in the terminal object $1$: we define the duality $(a,x)^*:=(x,a)$. We obtain the usual formula for the symmetric monoidal structure on $\Dcat$:
\begin{align*}
I&:=(1,1)\\
(a,x)\otimes(b,y)&:= (a\times b, (a\Rightarrow y)\times (b\Rightarrow x)).
\end{align*}
This allows us to define $$(a,x)\multimap (b,y):=((a,x) \otimes (b,y)^*)^*=((a\Rightarrow b)\times (y\Rightarrow x),a\times y).$$

Note that we have a linear/non-linear adjunction
\begin{diagram}
(a,1) & \lMapsto & a\\
\Dcat & \pile{\lTo^F\\\bot\\\rTo_U} & \Ccat\\
(a,x) & \rMapsto & a
\end{diagram}
meaning that we obtain a linear exponential comonad $!:=FU$ on $\Dcat$. Finally, note that we have an equivalence of categories $\Dcat_!\ra{}\Ccat$, $(a,x)\mapsto a$, being a full and faithful and essentially surjective functor.
\end{proof}
\begin{remark}
Note that $\Dcat$ in the previous theorem supports finite products (or, equivalently, finite coproducts), if and only if $\Ccat$ supports finite coproducts. Indeed, $\top:=(1,0)$ and $(a,x)\&(b,y):=(a\times b,x+y)$. $\Dcat$ having finite products (additive conjunctions) is a sufficient but not necessary condition for $\Dcat_!$ to have finite products. Indeed, in the above example, $\Dcat_!\cong \Ccat$ always has finite products. It is not clear what additive versions of the dependent connectives $\Sigma$ and $\Id$ should be, except in the weaker sense of objects in $\Dcat$ that in $\Dcat_!$ give a sound interpretation of ordinary cartesian $\Sigma$-types and $\Id$-types. In particular, it is not clear what a dependently typed generalisation should be of the binary coproduct, in the same sense that $\Sigma$- and $\Pi$-types, respectively, provide dependently typed generalisations of the binary product and the internal hom: the idea of ``having one of two types of objects, where the type of the second depends on the first'' sounds puzzling at best.
\end{remark}

This result extends to the dependently typed setting, as follows.
\begin{theorem}
Let $\Bcat^{op}\ra{\Ccat}\Cat$ be a strict indexed cartesian closed category with full and faithful comprehension (i.e. a model of cartesian dependent type theory). Write $\Dcat:=\Ccat\times \Ccat^{op}$, where we take the cartesian product of the fibre categories. Then, $\Dcat$ is a strict indexed $*$-autonomous category with comprehension (i.e. a model of classical linear dependent type theory). Moreover,
\begin{itemize}
\item $\Dcat$ supports $!$-types and we have an (indexed) equivalence $\Dcat_!\cong \Ccat$.
\item Therefore, $\Dcat $ supports $\Sa$- and $\Ida$-types, respectively, iff $\Ccat$ supports (strong) $\Sigma$- and $\Id$-types.
\item $\Dcat$ supports $\Sm_!$ and $\Pi^\multimap_!$-types iff $\Ccat$ supports both (weak) $\Sigma$- and $\Pi$-types.
\item $\Dcat$ supports (extensional, resp. intensional) $\Id^\otimes_!$-types iff $\Ccat$ supports (weak) (extensional, resp. intensional) $\Id$-types.
\end{itemize}
\end{theorem}
\begin{proof}
These are straightforward verifications. By analogy with $!(-)\otimes(-)$ and $!(-)\multimap (-)$, we can define $\Sm_{!(a,x)}(b,y):=(\Sigma_ab,\Pi_ay)$, $\Pi^\multimap_{!(a,x)}(b,y):=(\Pi_ab,\Sigma_a y)$. Note they are dual, in the sense that $(\Sigma_{!(a,x)}(b,y))^*=\Pi_{!(a,x)}(b,y)^*$. $\Idm_!$-types, we can interpret by $\Idm_{!(a,x)}:=(\Id_a,1)$. (More generally, we define the functor $\Idm_{!(a,x)}(b,y):=\Idm_{!(a,x)}\otimes (b,y)=(\Id_a\times b,(\Id_a)\Rightarrow y)$.) Indeed, this definition of $\Idm_{!(a,x)}$ follows from applying the Yoneda lemma to the following sequence of natural isomorphisms:
\begin{align*}
\Ccat\times \Ccat^{op}(\Gamma.a.a)(\Idm_{!(a,x)},(b,y))&\cong \Ccat\times \Ccat^{op}(\Gamma.a)((1,1),(b,y)\{\diag{\Gamma}{a}\})\\
&\cong \Ccat(\Gamma.a)(1,b\{\diag{\Gamma}{a}\})\\
&\cong\Ccat(\Gamma.a.a)(\Id_a,b)\\
&\cong\Ccat\times\Ccat^{op}(\Gamma.a.a)((\Id_a,1),(b,y)).
\end{align*}
\end{proof}

\subsection{Scott Domains and Strict Functions}\label{sec:lindepscott}
We can extend the model of section \ref{sec:scottdom} to a model of dDILL, following \cite{palmgren1990domain}. All constructions and proofs are exactly as in \cite{palmgren1990domain} with the only difference that in some cases we have to replace the word domain with predomain.

We take $\Bcat$ to be the category of Scott predomains and continuous functions. For a preorder-enriched category like $\Bcat$, we call a pair of morphisms $e:A\leftrightarrows B:p$ an embedding-projection pair if $e;p=\id_A$ and $p;e\leq \id_B$. We write $\Bcat_{ep}$ for the lluf subcategory of $\Bcat$ of the embedding-projection pairs.  We can make this into a model $\Ccat$ of DTT in the same way as we can for the category of Scott domains and continuous functions \cite{palmgren1990domain}: we define $\Ccat(A)$ to be the category of Scott predomains parametrised over $A$ (continuous families of predomains), which we define to be \mccorrect{directed colimit preserving} functors from $A$ into $\Bcat_{ep}$. Change of base is given by precomposition. This supports $1$-, $\Sigma$-, $\Sigma_{1\leq i\leq n}$- and (intensional) $\Id$-types. Briefly, $1$ is the one-point predomain, $\Sigma$-types are just the set-theoretic $\Sigma$-types equipped with the product order, $\Sigma_{1\leq i\leq n}$-types are given by disjoint unions and $\Id_A(x,y):=\{z\in A\;|\; z\leq x \;\wedge z\leq y\}$ with the induced order from $A$. For predomains $A$ and a predomain $B$ parametrised over $A$, we can define a poset (which generally will not be a predomain, for size reasons) $\Pi_{x\in A}B$ as the set of continuous\footnote{\begin{mccorrection}Here, I am referring to the appropriate generalisation of continuity to dependent functions (called $p$-continuity in \cite{palmgren1990domain}): functions $f:\Pi_{x:A}B$ such that, for all $x\in A$, for all compact elements $b\leq_{B(x)} f(x)$, there exists a compact element $a\leq_A x$ such that $b\leq_{B(x)} f(a)$.\end{mccorrection}} dependent functions from $A$ to $B$ under the pointwise order. This allows us to define $\Ccat(A)(B,C):=\Pi_{x\in A}\Pi_{y\in B(x)}C(x)$ on which we have the obvious identities, composition and change of base functions.

We can define $\Dcat(A)$ to be the category of Scott domains parametrised over $A$ (continuous families of Scott domains) with strict continuous (families of) functions as morphisms. We do this by extending the operations $\top,\&,I,\otimes$ and $ \multimap$ to $\Dcat(A)$ in a pointwise way and defining $\Dcat(A)(B,C):=\Pi_{x\in UA}U(B(x)\multimap C(x))$. This extends to give an indexed category of parametrised Scott domains indexed over Scott predomains. We note that we have an indexed adjunction $F\dashv U$ to $\Ccat$ (pointwise) where $F$ is strong monoidal. We see that we have a model of the dependently typed LNL calculus.

We can note that $-\Dcat\{\proj{A}{B}\}$ have both left and right adjoints $\csigma{B}{}$ and $\cpi{B}{}$ satisfying (Frobenius and) Beck-Chevalley conditions. Here, $\csigma{B}{C}(x):=\{\langle b,c \rangle\in \Sigma_{y\in B(x)}{UC}(x,y)\;|\; c\neq \bot \}\cup\{\bot\}$ and $\cpi{B}{C}(x):=\Pi_{y\in B(x)} UC(x,y)$. In fact, we have additive $\Sigma$- and $\Id$-types as well by noting that $\Sigma_{y\in UB} UC(y)$ and $\Id_{UB}$ are  (parametrised) Scott domains.  Moreover, we can define $\Id^\otimes_{FA}$ as $F \Id_{A}$.

\subsection{Coherence Spaces}\label{sec:depcohsp}
The usual coherence space model of linear type theory can be extended with a notion of dependent types, which gives us a non-trivial model of  dDILL (of classical linear dependent type theory). We  define it as a strict indexed symmetric monoidal cloed category with comprehension
\begin{diagram}
\Stable^{op} & \rTo^{\Dcat} & \mathsf{SMCCat}.
\end{diagram}
For our category of cartesian contexts we take the category $\Stable$ of Scott predomains with pullbacks and continuous stable functions. Note that we have a large Scott (pre)domain $\mathcal{U}$ with pullbacks (given by intersection) of small coherence spaces, using the following ordering on coherence spaces:
$$X\unlhd Y := X\subseteq Y \;\wedge\; \coh_X=\;\coh_Y\big |_{X\times X}.$$
$\mathcal{U}$ will play the r\^ole of a cartesian universe of linear types. For a predomain $D\in\ob(\Stable)$, we define $\Dcat(D)$ to be a category with set of objects stable functions from $D$ to $\mathcal{U}$: $\ob( \Dcat(D)):=\Stable(D,\mathcal{U})$. We shall define its morphisms shortly, but, first, we define a few operations on the objects: $I (=\bot),\otimes, \multimap,\&,0(=\top),$ and $\oplus$ are defined pointwise, as on coherence spaces. Given $G'\in\Stable(D,\mathcal{U})$, we define two coherence spaces with the same underlying set $$\csigma{x:D}G'(x):=\cpi{x:D}G'(x):=\{(x,u)\; |\; x\in FD, u\in G'(x)\}$$
but different coherence relations
$$(s,u)\coh_{\csigma{x:D}G'(x)} (t,v):=s\coh_{FD} t\; \wedge\; u\coh_{G'(s\vee t)} v,
$$
and
$$(s,u)\scoh_{\cpi{x:D}G'(x)} (t,v):=s\coh_{FD} t\; \Rightarrow\; u\scoh_{G'(s\vee t)} v.
$$
This defines two coherence spaces $\csigma{x:D}G'(x)$ and $\cpi{x:D} G'(x)$.

We define (as the type theory dictates\footnote{Indeed, $\cliques(\cpi{x:D}(G'\multimap G)(x))=\Dcat(\cdot)(I,\cpi{x:D}(G'\multimap G)(x))
\cong \Dcat(D)(I,G'\multimap G)
\cong \Dcat(D)(G',G)$.}
, if $\cpi{x:D}{}$ is to give the $\cpi{-}{}$-type) the morphisms in the fibres of $\Dcat$ as cliques in the appropriate $\cpi{-}{}$-type:
$$\Dcat(D)(G',G):=\cliques(\cpi{x:D}(G'\multimap G)(x)).$$

Composition and identities are defined pointwise (where it is left to the reader to verify that these are indeed cliques):
\begin{mccorrection}
$$G'\ra{\sigma}G\ra{\tau}H:=\{(x,f,h)\;|\; \left(\exists_{y,z\leq x}\exists_{g\in G(x)}(y,f,g)\in \sigma\wedge (z,g,h)\in\tau\right)\;\wedge $$
$$\forall_{x'\leq x}  \left(\exists_{y',z'\leq x'}\exists_{g\in G(x')}(y',f,g)\in \sigma\wedge (z',g,h)\in\tau\right)\Rightarrow x'=x\}$$
\end{mccorrection}
$$G'\ra{\id_{G'}}G':=\{(x,f,f)\;|\; f\in G'(x)\wedge \forall_{x'\leq x}f\in G'(x')\Rightarrow x'=x\}.$$

The reader can check that the pointwise operations $I$ and $\otimes$ make $\Dcat(D)$ into a symmetric monoidal category.
\begin{theorem}
\mccorrect{Stable continuous families of coherence spaces, indexed over Scott predomains with pullbacks and stable continuous functions} define a strict indexed symmetric monoidal category with comprehension\mccorrect{, hence a model of dDILL}.
\end{theorem}
\begin{proof}Note that $\Stable$ is a category with terminal object the one point domain and that $\Dcat(D)$ is a symmetric monoidal category. We define change of base in $\Dcat$ for morphisms $D'\ra{f}D$ in $\Stable$: $\Dcat(f)(G'):=f;G'$ and, for $\sigma\in\Dcat(D)(G',G)$, $\Dcat(f)(\sigma):=F(f);\sigma,$ where we see $F(f)$ as a clique in $FD'\multimap FD$. This gives a clique in $\cpi{x:D'}(G'\multimap G)( f(x))$. $\Dcat(f)$ is easily seen to \mccorrect{strictly preserve $I$ and $\otimes$} as precomposition is compatible with the \mccorrect{pointwise defined connectives}.

What remains to be done is define the comprehension. We define this as $U_D(G'):=\mathbf{p}_{D,U{G'}}:=\Sigma_{x:D}UG'(x)\ra{\fst}D$, where the $\Sigma$-type is taken in $\Stable$. That is, we take the set theoretic $\Sigma$-type, equip it with the product order and note that it gives a Scott predomain with pullbacks. The fact that it is a Scott predomain follows from section \ref{sec:lindepscott} as we are taking the $\Sigma$-type of a continuous family of Scott domains. To see that it has pullbacks, note that $UG'(x)$ and $UG'(y)$ are downward closed subsets of $UG'(z)$ if $x,y\leq z$\mccorrect{, such that $UG'(x\wedge y)=UG'(x)\cap UG'(y)$ (because of stability of $G'$)}. That means that we can take component-wise meets.

We define\\
\\
\resizebox{\linewidth}{!}{
$\mathbf{v}_{D,U{G'}}\in \Dcat(\Sigma_{x:D}UG'(x))(I,G'\{\mathbf{p}_{D,U{G'}}\})\cong\cliques(\cpi{(z,s):\Sigma_{x:D}UG'(x)}G'(z))\hspace{50pt}\;
$}\\
\\
as the ``trace'' of the (dependent) projection onto the second component:\\
\\
\resizebox{\linewidth}{!}{\mccorrect{
$\mathbf{v}_{D,U{G'}}:=\{((x,s),v)\; |\; (x,s)\in F(\Sigma_{x:D}UG'(x)),v\in s ,\forall_{(x',s')\leq (x,s)}v\in s'\Rightarrow (x',s')=(x,s)\},
$}}\\
\\
which is a clique.
\begin{claim*}$\mathbf{v}_{D,U{G'}}$ is a clique in $\cpi{(z,s):\Sigma_{x:D}UG'(x)}G'(z)$.
\end{claim*}
\begin{proof} Assume $((x,s),v)\neq ((x',s'),v')$, then
\begin{align*}&((x,s),v)\scoh_{\cpi{(z,s):\Sigma_{x:D}UG'(x)}G'(z)}((x',s'),v')\\
&\equiv (x,s)\coh_{F(\Sigma_{x:D}UG'(x))} (x',s')\Rightarrow v\scoh_{G'(x\vee x')} v'\\
&\equiv \exists_{(x'',s'')\in \Sigma_{x:D}UG'(x)}((x,s),(x',s')\leq_{\Sigma_{x:D}UG'(x)} (x'',s''))\Rightarrow v\scoh_{G'(x'')} v'.
\end{align*}
Assume $v=v'$. Then, the maximality condition on $(x,s)$ gives that $(x,s)=(x,s')$, contradicting our assumption that $((x,s),v)\neq((x',s'),v')$. Therefore, $v\neq v'$. 

Then, $(x,s),(x',s')\leq (x'',s'')$ implies that $v\in s\subseteq s''\supseteq s'\ni v'$, which in turn implies that $v\coh_{G'(x'')}v'$ and, as $v\neq v'$, we conclude that $v\scoh_{G'(x'')}v'$.
\end{proof}

Given $f\in \Stable(D',D)$ and \mccorrect{$\sigma \in \Dcat(D')(I,f;G')=\cliques(\Pi_{F(x:D')}^\multimap G'( f(x)))$}, we define $\langle f,\sigma\rangle \in \Stable(D',\Sigma_{x:D}UG'(x))$ as the function $(f,\fun(\sigma))$ with first component $f$ and second component $\fun(\sigma)$, where (writing $\Pi_{x:D'}UG'(f(x))$ for the set of dependent continuous stable functions from $D'$ to $f';UG$)
$$\fun(\sigma):=\{(x,\bigvee \{a\in UG'(f(x))|\exists y\leq x, (y,a)\in\sigma\}) \;|\; x\in D' \}\in \Pi_{x:D'}UG' (f(x)).
$$

We verify that $(\mathbf{p},\mathbf{v},\langle-,-\rangle)$ gives a representation, demonstrating the comprehension axiom. Clearly, $ \langle f,\sigma\rangle;\mathbf{p}_{D,U{G'}}=\langle f,\sigma\rangle;\fst=f$ and $\mathbf{v}_{D,U{G'}}\{\langle f,\sigma\rangle\}=F(\langle f,\sigma\rangle);\trace(\snd)=F(( f,\fun(\sigma) ));\trace(\snd)=\trace(( f,\fun(\sigma) );\snd)=\trace(\fun(\sigma))=\sigma$. Conversely, it is easily seen that $\langle f,\sigma\rangle $ is uniquely determined by these two equations. Indeed, suppose $t\in\Stable(D',\Sigma_{x:D}UG'(x))$ such that $f=t;\mathbf{p}_{D,UG'}=t;\fst$ and $\sigma=\mathbf{v}_{D,UG'}\{t\}:=Ft;\mathbf{v}_{D,UG'}= Ft;\trace(\snd)=\trace(t;\snd)$. Then, $\langle f,\sigma\rangle =\langle  t;\fst,\trace(t;\snd )\rangle=( t;\fst,\fun(\trace(t;\snd)))=(t;\fst,t;\snd)=t$.
\end{proof}
\begin{theorem}
The model supports $I-,\otimes-,\multimap-,\top-,\&-,0-,\oplus-,\Sigma_!^\otimes-,\Pi_!^\otimes-,$ $!-,$ and $\Id_!^\otimes$-types. Moreover, it is a model of classical linear dependent type theory.
\end{theorem}
\begin{proof}
For $I-,\otimes-,\multimap-,\top-,\&-,0-$ and $\oplus-$-types the verifications are trivial as the type formers are defined pointwise. It is clear that $I$ and $\otimes$ give a symmetric monoidal structure on $\Dcat(D)$. It then follows that $\multimap$ gives internal homs, from the facts that our operations are defined pointwise and that $\multimap$ gives internal homs in $\Coh$:
$\Dcat(D)(G'\otimes G,H)\cong \cliques(\Pi_{x:D}((G'\otimes G)\multimap H))(x))=\cliques(\Pi_{x:D}((G'(x)\otimes G(x))\multimap H(x))\cong \cliques(\Pi_{x:D}(G'(x)\multimap (G(x)\multimap H(x)))=\cliques(\Pi_{x:D}((G'\multimap (G\multimap H))(x))=\Dcat(D)(G',G\multimap H)$.

We have to show that $\top$ and $\&$ give finite products on $\Dcat(D)$. Let $G'\in\ob (\Dcat(D))$. Then, we have a unique $!_{G'}\in\Dcat(D)(G',\top)=\cliques(\Pi_{x:D}(G'\multimap \top)(x))=\cliques(\Pi_{x:D}\emptyset)=\cliques(\emptyset)=\{\emptyset\}$. Let $G',G\in\ob (\Dcat(D))$. Then, we have   projections \mccorrect{$(G'\& G\ra{\fst} G')=\{(x,f,f)\; |\;x \textnormal{ is minimal such that }f\in G'(x)\subseteq G'\& G(x) \}$ and $(G'\& G\ra{\snd} G)=\{(x,g,g)\; |x\textnormal{ is minimal such that }\;g\in G(x)\subseteq G'\& G(x) \}$}. Given $H\ra{f}G'$ and $H\ra{g}G$, we define $(H\ra{\langle f,g\rangle} G'\& G):= f\cup g$. This is a clique in $\Pi_{x:D}(H\multimap G'\& G)(x)$, as $(x,h,e)\scoh (x',h',e')=x\coh_{FD} x'\Rightarrow (h\coh_{H(x)} h'\Rightarrow e\scoh_{G'(x)\& G(x)} e')$. Now, we have three cases: if both $e$ and $e'$ are in $G'(x)$, the fact that $f$ is a clique takes care of the coherence and, similarly, if both $e$ and $e'$ are in $G(x)$, $g$ does this. Finally, if $e\in G'(x)$ and $e'\in G(x)$ (or vice versa), the definition of coherence in $G'(x)\& G(x)$ makes sure that $e\coh_{G'(x)\&  G(x)} e'$, so $(x,h,e)\coh (x',h',e')$.

We verify the rules for $\Sigma_!^\otimes$-types.
\begin{claim*} We have a left adjoint
$$\Dcat(\Sigma_{x:D}UG'(x))\ra{\csigma{UG'}{}}\Dcat(D) $$
to the change of base functor
$$\Dcat(D)\ra{-\{\mathbf{p}_{D,U{G'}}\}}\Dcat(\Sigma_{x:D}UG'(x)).$$
Moreover, this satisfies Frobenius reciprocity,
$$\csigma{U{G'}}(\mathbf{p}_{D,U{G'}};G\otimes H)\cong G\otimes \csigma{U{G'}}H,$$
and the left Beck-Chevalley condition.
\end{claim*}
\begin{proof}We define, on objects,
$$\csigma{U{G'}}(G)(x):= \csigma{s:UG'(x)}G(x,s)$$
and, on morphisms,
\begin{mccorrection}
\begin{align*}
&\csigma{U{G'}}(G\ra{\sigma}H):=\\
&\{(x,(s,g),(s,h))\in \Pi_{F(x\in D)}^\multimap\left(\Sigma_{F(s\in UG'(x))}^\otimes G(x,s)\right)\multimap\left( \Sigma_{F(s'\in UG'(x))}^\otimes H(x,s')\right) \\
&\; |\; ((x,s),(g,h))\in\sigma\}.
\end{align*}
\end{mccorrection}

We verify that this, indeed, defines a clique in
$$\cpi{x:D}(\csigma{s:UG'(x)}G(x,s))\multimap (\csigma{s:UG'(x)}H(x,s)):$$
\begin{align*}
&(x,(s,g),(s,h))\scoh(x',(s',g'),(s',h'))\\
& \equiv x\coh x'\Rightarrow ((s,g),(s,h))\scoh((s',g'),(s',h'))\\
&\equiv x\coh x'\Rightarrow ((s,g)\coh (s',g')\Rightarrow (s,h)\scoh(s',h'))\\
&\equiv x\coh x'\Rightarrow ( (s\coh s'\wedge g\coh g')\Rightarrow  (s,h)\scoh(s',h'))\\
&\equiv x\sincoh x' \vee s\sincoh s' \vee g\sincoh g'\vee (s,h)\scoh (s',h').
\end{align*}
Now, as $\sigma$ is a clique in $\cpi{(x,s):\Sigma_{x:D}UG'(x)}G(x,s)\multimap H(x,s)$, we have that
$$x\sincoh x'\vee s\sincoh s' \vee g\sincoh g'\vee h\scoh h'. $$
Suppose that not $x\sincoh x'\vee s\sincoh s' \vee g\sincoh g'$ (then, in particular, $s\coh s'$). We have to show that $h\scoh h'\Rightarrow (s,h)\scoh (s',h')$. This clearly holds as $s\coh s'$. We conclude that $\csigma{UG'}(\sigma)$ is a clique.

$\csigma{UG'}$ clearly respects identities and composition so we conclude it is a well-defined functor.

We verify that adjointness condition
$$\Dcat(\Sigma_{x:D}UG'(x))(G, \mathbf{p}_{D,UG'};H)\cong \Dcat(D)(\csigma{UG'} G,H). $$
The LHS is equal to 
$$\cliques(\cpi{(x,s):\Sigma_{x:D}UG'(x)}(G(x,s)\multimap H(x))),$$
while RHS is equal to
$$\cliques(\cpi{x:D}((\csigma{s:UG'(x)}(G(x,s)))\multimap H(x))).$$
Now,
\begin{align*}&\cpi{(x,s):\Sigma_{x:D}UG'(x)}(G(x,s)\multimap H(x))\\
&=\{((x,s),(g,h))\;| \; x\in FD, s\in UG'(x), g\in G(x,s), h\in H(x)\}\\
&\cong \{(x,((s,g),h)) \; |\; x\in FD, s\in UG'(x),g\in G(x,s), h\in H(x)\}\\
&=\cpi{x:D}((\csigma{s:UG'(x)}G(x,s))\multimap H(x)).
\end{align*}
Moreover, $((x,s),(g,h))$,$((x',s'),(g',h'))\in \cpi{(x,s):\Sigma_{x:D}UG'(x)}(G(x,s)\multimap H(x)))$ are related via $\scoh$ iff any of the following equivalent conditions hold
\begin{align*}
(x\coh x'\wedge s\coh s')\Rightarrow (g,h)\scoh (g',h')&\equiv x\sincoh x'\vee s\sincoh s'\vee (g\coh g'\Rightarrow h\scoh h')\\
&\equiv x\sincoh x'\vee s\sincoh s' \vee g\sincoh g'\vee h\scoh h',
\end{align*}
while $(x,((s,g),h)\scoh (x',((s',g'),h'))$ in $\cpi{x:D}(\csigma{s:UG'(x)}(G(x,s))\multimap H(x))$  iff any of the following equivalent conditions hold
\begin{align*}
x\coh x'\Rightarrow ((s,g),h)\scoh ((s',g'),h') & \equiv x\coh x'\Rightarrow ((s,g)\coh (s',g')\Rightarrow h\scoh h')\\
&\equiv x\coh x' \Rightarrow ((s\coh s'\wedge g\coh g')\Rightarrow h\scoh h')\\
&\equiv x\sincoh x'\vee s\sincoh s'\vee g\sincoh g'\vee h\scoh h'.
\end{align*}
We see that the conditions on both sides coincide.

We can therefore take the canonical bijection between both sets of vertices to induce an isomorphism of coherence spaces, hence a bijection of cliques.

Furthermore, it is immediately obvious from the definitions that Frobenius reciprocity holds:
\begin{align*}
\csigma{UG'}(G\{\mathbf{p}_{D,UG'}\}\otimes H)&=x\mapsto \csigma{s:UG'(x)}G(x)\otimes H(x,s)\\
&\cong x\mapsto G(x)\otimes \csigma{s:UG'(x)}H(x,s)\\
&=G\otimes\csigma{UG'}H.
\end{align*}
Finally, the Beck-Chevalley condition trivially holds, as the change of base functors act by precomposition.\end{proof}
We verify the rules for $\Pi$-types.
\begin{claim*} We have a right adjoint
$$\Dcat(\Sigma_{x:D}UG'(x))\ra{\cpi{UG'}{}}\Dcat(D) $$
to the change of base functor
$$\Dcat(D)\ra{-\{ \mathbf{p}_{D,U{G'}}\}}\Dcat(\Sigma_{x:D}UG'(x)),$$
satisfying the right Beck-Chevalley condition.\end{claim*}
\begin{proof}We define, on objects,
$$\cpi{U{G'}}(G)(x):= \cpi{s:UG'(x)}G(x,s)$$
and, on morphisms,
\begin{mccorrection}
\begin{align*}
&\cpi{U{G'}}(G\ra{\sigma}H):=\\
& \{(x,(s,g),(s,h))\in \Pi_{F(x\in D)}^\multimap\left(\Pi_{F(s\in UG'(x))}^\multimap UG(x,s)\right)\multimap \left(\Pi_{F(s'\in UG'(x))}^\multimap UH(x,s')\right)\\
& \; |\; ((x,s),(g,h))\in\sigma\}.
\end{align*}
\end{mccorrection}
We verify that this, indeed, defines a clique in $$\cpi{x:D}(\cpi{s:UG'(x)}G(x,s))\multimap (\cpi{s:UG'(x)}H(x,s)):$$
\begin{align*}
&(x,(s,g),(s,h))\scoh(x',(s',g'),(s',h')) \\& \equiv x\coh x'\Rightarrow ((s,g),(s,h))\scoh((s',g'),(s',h'))\\
&\equiv x\coh x'\Rightarrow ((s,g)\coh (s',g')\Rightarrow (s,h)\scoh(s',h'))\\
&\equiv x\coh x'\Rightarrow ( (s,g)\coh (s',g')\Rightarrow  (s\coh s'\Rightarrow h\scoh h'))\\
&\equiv x\sincoh x' \vee (s,g)\sincoh (s',g')\vee s\sincoh s'\vee h\scoh h'.
\end{align*}
Now, as $\sigma$ is a clique in $\cpi{(x,s):\Sigma_{x:D}UG'(x)}G(x,s)\multimap H(x,s)$, we have that
$$x\sincoh x'\vee s\sincoh s' \vee g\sincoh g'\vee h\scoh h'. $$
Suppose that not $x\sincoh x'\vee s\sincoh s' \vee h\scoh h'$ (then, in particular, $s\coh s'$). We have to show that $(s\coh s'\wedge g\sincoh g')\Rightarrow (s,g)\sincoh (s',g')$. For this, note that $g\sincoh g'$ implies that $g\neq g'$ hence $(s,g)\neq (s',g')$, so an equivalent thing to prove would be that $(s\coh s'\wedge g\sincoh g')\Rightarrow \neg ((s,g)\scoh (s',g'))$, which is $(s\coh s'\wedge \neg(g\coh g'))\Rightarrow \neg (s\coh s'\Rightarrow g\scoh g')$ by definition of $\scoh$ on $\cpi{-}{}$-types, which clearly holds. We conclude that $\cpi{UG'}(\sigma)$ is a clique.

$\cpi{UG'}{}$ clearly respects identities and composition so we conclude it is a well-defined functor.

We verify \mccorrect{the} adjointness condition
$$\Dcat(\Sigma_{x:D}UG'(x))( \mathbf{p}_{D,UG'};G,H)\cong \Dcat(D)(G,\cpi{UG'}H). $$
The LHS is equal to 
$$\cliques(\cpi{(x,s):\Sigma_{x:D}UG'(x)}(G(x)\multimap H(x,s))),$$
while RHS is equal to
$$\cliques(\cpi{x:D}((G(x))\multimap \cpi{s:UG'(x)}H(x,s))).$$

Now,
\begin{align*}&\cpi{(x,s):\Sigma_{x:D}UG'(x)}(G(x)\multimap H(x,s))\\
&=\{((x,s),(g,h))\;| \; x\in FD, s\in UG'(x), g\in G(x), h\in H(x,s)\}\\
&\cong \{(x,(g,(s,h))) \; |\; x\in FD, s\in UG'(x),g\in G(x), h\in H(x,s)\}\\
&=\cpi{x:D}((G(x))\multimap \cpi{s:UG'(x)}H(x,s)).
\end{align*}
Moreover, $((x,s),(g,h))$,$((x',s'),(g',h'))\in \cpi{(x,s):\Sigma_{x:D}UG'(x)}(G(x)\multimap H(x,s)))$ are related via $\scoh$ iff any of the following equivalent conditions hold
\begin{align*}
(x\coh x'\wedge s\coh s')\Rightarrow (g,h)\scoh (g',h')&\equiv x\sincoh x'\vee s\sincoh s'\vee (g\coh g'\Rightarrow h\scoh h')\\
&\equiv x\sincoh x'\vee s\sincoh s' \vee g\sincoh g'\vee h\scoh h',
\end{align*}
while in $\cpi{(x,s):\Sigma_{x:D}UG'(x)}(G(x)\multimap \cpi{s:UG'(x)}H(x,s))$ we have that $(x,(g,(s,h))\scoh (x',(g',(s',h')))$   iff any of the following equivalent conditions hold
\begin{align*}
x\coh x'\Rightarrow (g,(s,h))\scoh (g',(s',h')) & \equiv x\coh x'\Rightarrow (g\coh ,g')\Rightarrow (s,h)\scoh (s',h'))\\
&\equiv x\coh x' \Rightarrow (g\coh ,g')\Rightarrow (s\coh s' \Rightarrow h\scoh h')\\
&\equiv x\sincoh x'\vee s\sincoh s'\vee g\sincoh g'\vee h\scoh h'.
\end{align*}
We see that the conditions on both sides coincide. We can therefore take the canonical bijection between both sets of vertices to induce an isomorphism of coherence spaces, hence a bijection of cliques.

Finally, the Beck-Chevalley condition trivially holds, as the change of base functors act by precomposition.
\end{proof}

We verify the rules for $!$-types. Seeing that we already have $\csigma{-}{}$-types and $I$-types, we know we can construct $!$-types as $\csigma{-}I$. We  also give a direct proof, however, as it may provide more insight in the definition of the exponential.
\begin{claim*} The comprehension functors $U_D$ have a strong monoidal left adjoint $F_D$.\end{claim*}
\begin{proof}
We define $F_D(\mathbf{p}_{D,UG'}):=G';U;F=G';!$. Note that is well-defined as we can construct $ G';U$ from $\mathbf{p}_{D,UG'}$ as $x\mapsto \mathbf{p}_{D,UG'}^{-1}(x)$. Moreover, $ G';!$ is a type family, as, obviously, $\fincliques(\bigcup_i A_i)=\bigcup_i \fincliques(A_i)$\mccorrect{, for a directed family $(A_i)_i$,} and $\fincliques(A\cap B)=\fincliques(A)\cap \fincliques(B)$, where we write $\fincliques(A):=!A$ to emphasise that we are taking the coherence space of finite cliques.

This definition extends to \mccorrect{morphisms} between objects of the form $\mathbf{p}_{D, UG'}$. Indeed, we note that a morphism $\mathbf{p}_{D,UG'}\ra{f}\mathbf{p}_{D,UG}$ restricts to stable functions $UG'(x)\ra{f_x}UG(x)$ for all $x\in D$. We define $F_D(f)\in\cliques(\cpi{x:D}(!G'(x)\multimap !G(x)))$ as $\{(x,(s,t))\;|\; x\in D, (s,t)\in F(f_x)\; \forall_{x'\leq x} (s,t)\in f_{x'}\Rightarrow x'=x\}$. This is a clique as
\begin{align*}
(x,(s,t))\scoh (x',(s',t')) &= x\coh x'\Rightarrow (s,t)\scoh (s',t')\\
&=x\coh x'\Rightarrow (s,t)\neq (s',t')\txt{(as $f_x$ is a clique)}.
\end{align*}

We finally take the unique strong monoidal extension to obtain a (strong monoidal) functor from $\Ccat(D)$.

The condition that $F_D$ is compatible with change of base (precomposition) follows immediately because of the pointwise definition of $F_D$.
\end{proof}

Finally, we verify that we have $\Id^\otimes_!$-types.
\begin{claim*}Our model supports (intensional) $\Id^\otimes_!$-types.
\end{claim*}
\begin{proof}
We verify the formation, introduction, elimination and $\beta$-rules for $\Id^\otimes_!$-types.
\begin{enumerate}
\item[$\Id_!^\otimes$-\textsf{F}] Let $a,a':I\ra{}G'\in\Dcat(D)$. Then, we define a type family over $D$:
$$\Id_{!G'}^\otimes(a,a')(x):=\{b\in !G'(x)\; | \; b\leq \fun(a)(x),\fun(a')(x)\}\subseteq !G'(x),$$
with the induced coherence relation.

The fact that this is a continuous function $D\ra{\Id^\otimes_{!G'}(a,a')}\mathcal{U}$ is a direct consequence of its definition as a subfamily of a continuous family $!G'$ with a continuous bound ($\fun(a)(x)\wedge \fun(a)(x')$). The exact same argument gives stability:
\begin{align*}
&\Id^\otimes_{!G'}(a,a')(x_1)\wedge \Id^\otimes_{!G'}(a,a')(x_2)\\
&=\{b\in !G'(x_1)\; | \; b\leq \fun(a)(x_1),\fun(a')(x_1)\}\cap\\
&\;\;\;\;\; \{b\in !G'(x_2)\; | \; b\leq \fun(a)(x_2),\fun(a')(x_2)\}\\
&=\{b\in !G'(x_1)\cap !G'(x_2)\;|\; b\leq \fun(a)(x_1)\wedge\fun(a)(x_2)\wedge\\
&\qquad\qquad \qquad\qquad\quad\qquad\qquad \fun(a')(x_1)\wedge \fun(a')(x_2)\}\\
&=\{b\in !G'(x_1\wedge x_2)\;|\; b\leq \fun(a)(x_1\wedge x_2)\wedge \fun(a')(x_1\wedge x_2)\}\\
&= \Id^\otimes_{!G'}(a,a')(x_1\wedge x_2).
\end{align*}

\item[$\Id_!^\otimes$-\textsf{I}] For $a:I\ra{}G'\in\Dcat(D)$, we define, for $x\in D$,
$$\refl{a}:=\{(x,b)\; | \; x\in D, b\in \Id^\otimes_{!G'}(a,a)(x),\forall_{x'\leq x}b\in \Id^\otimes_{!G'}(a,a)(x')\Rightarrow x'=x\}.$$
This is easily verified to be a clique in $\cpi{x:D}\Id^\otimes_{G'}(a,a)(x)$, as $a(x)$ is a clique in $G'(x)$, and hence a morphism $I\ra{\refl{a}}\Id^\otimes_{!G'}(a,a)\in\Dcat(D)$.

\item[$\Id_!^\otimes$-\textsf{E}] Suppose we're given
\begin{itemize}
\item $G'\in \ob  \Dcat(D)$)
\item $C\in \ob \Dcat(\Sigma_{(y,x,x'):\Sigma_{y:D})UG'(y)\times UG'(y)}U\Id^\otimes_{!G'}(x,x')(y)$
\item $c\in \Dcat(\Sigma_{y:D}UG')(\Xi,C\{\langle \id_D,\id_{G'},\id_{G'},\refl{\id_{G'}}\rangle\})$
\item $a,a'\in \Dcat(D)(I,G')$
\item $p\in \Dcat(D)(I,\Id^\otimes_{G'}(a,a'))$.
\end{itemize}

We construct
$$(\mathrm{let}\; (a,a',p)\;\mathrm{be}\;(\id_{G'},\id_{G'},\refl{\id_{G'}})\;\mathrm{in}\; c)\in\Dcat(D)(\Xi,C\{\langle \id_C,a,a',p\rangle\}).$$
(as a dependent stable function $\in\Pi_{x:D}U(\Xi\multimap C\{a,a',p\})$) by defining\\
\\
\resizebox{\linewidth}{!}{
$\fun(\mathrm{let}\; (a,a',p)\;\mathrm{be}\;(\id_{G'},\id_{G'},\refl{\id_{G'}})\;\mathrm{in}\; c)(y)(\xi):=\fun(d)(y,\bigcup \fun(p)(y))(\xi).
$}

\item[$\Id_!^\otimes$-$\beta$] We calculate
\begin{align*}
&\fun(\mathrm{let}\; (a,a,\refl{a})\;\mathrm{be}\;(\id_{G'},\id_{G'},\refl{\id_{G'}})\;\mathrm{in}\; c)(y)(\xi)\\&=\fun(c)(y,\bigcup \fun(\refl{a})(y))(\xi)\\&=\fun(c)(y,\fun(a)(y))(\xi)\\
&=\fun(c\{\langle \id_D,a\rangle\})(y)(\xi).
\end{align*}
\end{enumerate}
\end{proof}
Finally, we note that $1=\bot$ is a dualising object: $(-)\multimap \bot=(-)^\bot$ is an involution, as this is the case pointwise. This means we have a model of classical linear dependent type theory.
\end{proof}

While we can, in fact, define additive $\Id$-types: $\Id^{\&}_{G'}(a,a')(x):=\fun(a)(x)\cap\fun(a')(x)\subseteq G'(x)$, the interpretation of $\Sigma^{\&}$ in the model turns out to be problematic.
\begin{theorem}[Absence of $\Sigma^{\&}$-Types] \label{thm:cohnosigma} The model does not support $\Sigma^{\&}$-types.
\end{theorem}
\begin{proof}
Let $A$ be the coherence space $I$. Then, $UA=\{\emptyset\leq \{0\}\}$. Let $B$ be the stable continuous family of coherence spaces indexed by $UA$ where $B(\emptyset):=\top$ and $B(\{0\}):=I$. In that case, we note that $\Sigma_{UA}UB=\{\langle \emptyset,\emptyset\rangle\leq \langle \{0\},\emptyset\rangle\leq \langle \{0\},\{0\}\rangle\} $. Now, we claim that there is no coherence space $\Sigma_A^{\&}B$ such that $U\Sigma_A^{\&}B\cong \Sigma_{UA}UB$. To see this, note that $UC $ always has strictly more elements than the sum of the number of edges and vertices in $C$. Seeing that $\Sigma_{UA}UB$ has 3 elements, that would leave only three possibilities for $\Sigma_A^{\&}B$: $\top$, $I$ and $I\oplus I$. However, we have $U\top=\{\emptyset\}$, $UI=\{\emptyset\leq \{0\}\}$ and $U(I\oplus I)=\{\emptyset\leq \{0\},\{1\}\}$, none of which is isomorphic to $\Sigma_{UA}UB$. We conclude that no suitable $\Sigma_A^{\&}B$ exists.
\end{proof}
We see that the category $\Coh_!$ of coherent qualitative domains and stable functions is too restrictive to admit the interpretation of $\Sigma$-types. To interpret those, we  have to pass to a larger category of domains, like the category $\Stable$ of all Scott predomains with pullbacks. There, however, we face the usual problem that we cannot interpret $\Pi$-types (or even simple function types; this was the raison d'\^etre for dI-domains). We ask the reader to compare this to our discussion in section \ref{sec:depgirard} about finding a sweet spot between the co-Kleisli category and co-Eilenberg-Moore category where we can interpret both $\Sigma$- and $\Pi$-types. As is shown in \cite{Bertrand:1994:NSF:174776.174781}, dI-domains are such a sweet spot. In future work, we plan to demonstrate how these arise as the co-Kleisli category of another model of linear logic, a certain finitary variation on the linear information systems of \cite{bucciarelli2010linear}.


\begin{savequote}[8cm]
It may be that all games are silly. But then, so are humans.
	\qauthor{--- Roib\'eard \'O Floinn}
\end{savequote}

\chapter{\label{ch:5}Games for Dependent Types}{\DTT} can be seen as the extension of the simple $\lambda$-calculus along the Curry-Howard correspondence from a proof calculus for (intuitionistic) propositional logic to one for predicate logic. It forms the basis of many proof assistants, like NuPRL, LEGO and Coq, and is increasingly being considered as an expressive type system for programming, as implemented in e.g. ATS, Cayenne, Epigram, Agda and Idris \cite{altenkirch2005dependent} and with even Haskell  approaching its expressive power with the addition of GADTs \cite{mcbride2002faking}.

\mccorrect{A} recent source of enthusiasm in this field is homotopy type theory (\HoTT), which refers to an interpretation of {\DTT} into abstract homotopy theory \cite{awodey2009homotopy} or, conversely, an extension of {\DTT} that is sufficient to reproduce significant results of homotopy theory \cite{hottbook}. In practice, the latter means {\DTT} with $\Sigma$-, $\Pi$-, $\Id$-types (corresponding to existential and universal quantifiers and identity predicates, respectively, through the Curry-Howard correspondence), a universe (roughly, a type of types) satisfying the \emph{univalence axiom}, and certain higher inductive types (playing the r\^ole of ground types whose towers of iterated identity types behave like the homotopy types of certain spaces). The univalence axiom is an extensionality principle which implies the axiom of function extensionality \cite{hottbook}.

Game semantics provides a unified framework for intensional, computational semantics of various type theories, ranging from pure logics~\cite{abramsky1994gamesll} to programming languages~\cite{hyland2000full,nickau1994hereditarily,abramsky2000full,abramsky2005game} with a variety of effects (e.g. non-local control \cite{laird1997full}, state \cite{abramsky1996linearity,abramsky1998fully,murawski2011game}, non-determinism \cite{harmer1999fully}, probability \cite{danos2002probabilistic}, dynamically generated local names \cite{abramsky2004nominal}) and evaluation strategies~\cite{abramsky1998cbvgames}.

A game semantics for {\DTT}  has,  surprisingly, so far been  absent, perhaps because of the naturally effectful character of game semantics. We hope to fill this gap in the present chapter.  
Our hope is that such a  semantics  will provide an alternative analysis of the implications of the subtle shades of intensionality that arise in the analysis of \DTT~\cite{streicher1993investigations,hofmann1997syntax}. 
Moreover, the game semantics of {\DTT} is based on very different, one might say orthogonal intuitions to those of the homotopical models: temporal  rather than spatial, and directly reflecting the structure of computational processes. One goal, to which we hope this work  will be a stepping stone, is a game semantics of \HoTT{}\mccorrect{ }doing justice to both the spatial and temporal aspects of identity types. Indeed, such an investigation might even lead to a computational interpretation of the univalence axiom which has long been missing, although a significant step in this direction was recently taken by the constructive cubical sets model of \HoTT~\cite{bezem2014model}.  Finally, a game semantics for {\DTT} should hopefully shed light on how dependent types can interact with effects.

We interpret dependent types as families of games indexed by strategies. We adapt the viewpoint of  the game semantics of system F of \cite{abramsky2005game}
to describe the $\Pi$-type, capturing the intuitive idea that the specialisation of a term at type $\Pi_{x:A} B$ to a specific instance $B[a/x]$  is the responsibility solely of the context that provides the argument $a$ of type $A$; in contrast, any valid term of $\Pi_{x:A} B$ has to operate within the constraints enforced by the context.  Our definition draws its power from the fact that in a game semantics, these constraints are enforced not only on completed computations, but also on incomplete ones that arise when a term interacts with its context.   The temporal character of game semantics results in a model with strikingly different properties from existing models like the domain semantics \cite{palmgren1990domain}.

In this chapter, we describe a game theoretic  model of {\DTT} with $1$-, $\Sigma$-, $\Pi$- and intensional $\Id$-types, where (lists of dependent) (call-by-name) AJM-games interpret types and (lists of) deterministic history-free well-bracketed winning strategies on games of dependent functions interpret terms. We next specialize to the semantic type hierarchy formed by the $1$-, $\Sigma$-, $\Pi$-, and $\Id$-constructions and substitution over finite dependent games. This gives a model of {\DTT} which additionally supports finite inductive type families. Our model has the following key properties.
\begin{itemize}
\item The place of the $\Id$-types in the intensionality spectrum (in either model) compares as follows with the domain semantics with totality and with \HoTT.\\
\\
\resizebox{\linewidth}{!}{
\begin{tabular}{l||c|c|c}
 & \;\; Domains\;\; & \;\;\HoTT\;\; &\;\; Games\;\;\\
 \hline
Failure of Equality Reflection  & \cmark & \cmark &  \cmark\\
Streicher \cite{streicher1993investigations} Intensionality Criteria $(I1)$ and $(I2)$\;& \cmark & \cmark & \cmark\\
Streicher Intensionality Criterion $(I3)$ & \xmark & \xmark & \cmark\\
Failure of Function Extensionality (\textsf{FunExt}) &\xmark & \xmark & \cmark\\
Failure of Uniqueness of Identity Proofs (\textsf{UIP})\;\; & \xmark & \cmark & \xmark
\end{tabular}}\\
\item We show that the smaller model faithfully interprets $\DTTGame$. Moreover, it is fully complete at the types $A$ which do not involve $\Id$ in their construction or which involve one strictly positive $\Id$-type as a subformula, if we add the \textsf{Ty-Ext} rule for types $x:A\vdash B\type$. Full completeness for the full type hierarchy remains to be investigated but seems plausible.   In contrast, the domain theoretic model of \cite{palmgren1990domain} is not (fully) complete or faithful.
\item It can be extended from a model of pure type theory to, additionally,  interpreting various effects when we drop some of the conditions on strategies. 
\end{itemize}

In section \ref{sec:depgame}, we introduce a notion of dependent game and dependently typed strategy, together with a semantic equivalent $\smiley(-)$ of $(-)^T$, the syntactic translation of section \ref{sec:trans} from $\DTTGame$ to $\STTGame$: a translation to simply typed game semantics. Although this almost gives a model of dependent type theory, we show that we cannot interpret $\Sigma$-types (or comprehension). Adding $\Sigma$-types formally, we next construct an interpretation of {\DTT} in sections \ref{sec:ctxt}, \ref{sec:semtype} and \ref{sec:grndtype}, in the form of a category with families with $\Sigma$-, $\Pi$- and $\Id$-types and finite inductive type families. Section \ref{sec:semtype} further characterises various intensionality properties of the $\Id$-types. Soundness and faithfulness of the interpretation of {\DTT} are finally proved in section \ref{sec:compl}, as the interpretation factors faithfully over the faithful sound games interpretation of \STT, as well as full completeness results which are obtained by a dependently typed modification of the definability proofs of \cite{abramsky2000full,Abramsky00axiomsfor}. Finally, in section \ref{sec:depgameeff} we lift the various conditions on strategies which ensure purity of the computations they model and we draw lessons on the interaction between dependent types and effects.

\begin{remark}[Related Publications]
This chapter is based on \cite{abramsky2015games,vakar2016gamsem}. We have changed the interpretation of $\Id$-types to make them compatible with effects. To give a uniform treatment for all classes of strategies, we have chosen to define a dependent game in the pure setting only on winning strategies. We have also slightly changed the equational theory $\DTTGame$ which lets us simplify the completeness proof considerably. We believe the current presentation to be both simpler and more robust with respect to extensions to broader classes of types and terms.  After we presented our game semantics for dependent types in \cite{abramsky2015games}, \cite{yamada2016game} provided an alternative game semantics for dependent type theory, while with very different motivations.  Where our work is motivated by precisely characterising effectful (CBN) type theory (e.g. through completeness results) with the purpose of understanding dependently typed effectful programming, \cite{yamada2016game}  seems to be interested exclusively in modelling pure type theory and providing a constructive foundation of mathematics. 
\end{remark}
\section{An Indexed Category of Dependent Games}\label{sec:depgame}
Section \ref{sec:backgame} sketched how $\Gamecat_!$ models simple cartesian type theory. In this chapter, we extend this to a model of dependent type theory. In this section, we first show how to equip $\Gamecat_!$ with a notion of dependent type and we show how this leads to an indexed ccc $\DGame_!$ of dependent games and dependently typed strategies.

We define a poset $\Gamecat_\trianglelefteq$ of games with
\begin{align*}
A\trianglelefteq B &:=(M_A= M_B)\;\wedge\;(\lambda_A=\lambda_B)\;\wedge\; (\J_A=\J_B)\;\wedge\;(P_A\subseteq P_B)\;\wedge\\
 & \qquad   (W_A=W_B\cap P_A^\infty)\;\wedge\;\forall_{s,t\in P_B}(s\approx_A t\quad \Leftrightarrow \quad s\in P_A\;\wedge \; s\approx_B t).
\end{align*}
Given a game $C$, we define the complete lattice $\Sub(C)$ as the poset of its $\trianglelefteq$-subgames.  We note that, for $A,B\in\Sub(C)$, $A\trianglelefteq B \Leftrightarrow P_A\subseteq P_B$. We make the following simple observation that we shall refer to later.
\begin{theorem}\label{thm:subfunctor} We have a functor 
$\Gamecat_!\ra{\Sub}\mathsf{CjsLat}$
to the category $\mathsf{CjsLat}$ of complete lattices and join-preserving functions.\end{theorem}
\begin{proof} An element of $\Sub(C)$ is precisely specified by a $\approx_C$-closed prefix-closed subset of $P_C$, so we can compute joins and meets simply by unions and intersections. Given $A\ra{f} B \in \Gamecat_!$ and $A'\trianglelefteq A$, we define
\begin{mccorrection}
$$\Sub(f)(A') := \{s\in P_B\;|\; \exists_{t\in f}s\leq t\upharpoonright_B\wedge \forall_i t\upharpoonright_{!A}\upharpoonright_i \in A'\}.$$\end{mccorrection} The result is clearly prefix-closed and closed under $\approx_B$, as $f$ is closed under $\approx_{A\Rightarrow B}$. $\Sub(f)$ clearly preserves unions.
\end{proof}
 This allows us to define a dependent game as follows, where $\tot{B}$ can be seen as the semantic counterpart to the syntactic translation $B^T$ of section \ref{sec:trans}.

\begin{definition}[Dependent game]
For a game $A$, we define the set $\ob(\DGame_!(A))$ of \emph{games with dependency on $A$} as the \mccorrect{set} of pairs of a game $\tot{B}$ (without dependency) and a function $\str(A)\ra{B}\Sub(\tot{B})$.\end{definition}

We note that $\ob(\DGame_!(I))$ is the set of pairs $(A(\bot),\tot{ {A}})$ where $A(\bot)\trianglelefteq \tot{A}$, in which $\ob(\Gamecat_!)$ embeds as the proper subset of diagonal elements $(A,A)$. As the definability results of section \ref{sec:compl} illustrate, we need the generality of $\ob(\DGame_!(I))$ to properly capture the notion of closed typed in {$\DTTGame$}. Therefore, we define, more generally, for a pair $(A,\tot{ A})\in\ob(\DGame_!(I))$, $\ob(\DGame_!(A(\bot),\tot{ A})):=\ob(\DGame_!(\tot{ A}))$. As an example, let us write $x:\mathsf{mm}\vdash \mathsf{dd-mm}(x)$ for the (finite inductive) type family encoding the calendar of the year 1984 in dd-mm format. For instance, $\mathsf{dd-mm}(\textnormal{02})$ has constructors 01-02,$\ldots$,29-02. In this case, we note that for the purposes of the type theory the closed type $\mathsf{dd-mm}(\textnormal{02})$ will behave differently  from the closed (inductive) type $\{\textnormal{01-02},\ldots,\textnormal{29-02}\}$. Indeed, when eliminating from the former, our case analysis contains (redundant) additional information on how to handle the all other days of the year as well. This example shows that for a substituted type like $\mathsf{dd-mm}(\textnormal{02})$ the type theory still remembers information about the whole type family $\mathsf{dd-mm}$ (like the constructors outside the particular fibre under consideration), hence our interpretation of closed types as pairs $A(\bot)\trianglelefteq \tot{A}$ of games rather than as single games. From now on, we write $A$ for the pair $(A(\bot),\tot{ A})\in \ob(\DGame_!(I))$ and, more generally and slightly ambiguously, $B$ for the pair $(B,\tot{ B})\in\ob(\DGame_!(A))$.

Writing $s\mapsto \overline{s}$ for the function from  $P_{!\tot{ A}}$ to the power set $\mathcal{P}P_{\tot{ A}}$, inductively defined on the empty play, Opponent moves and Player moves, respectively, as  $\epsilon\mapsto \emptyset,\;\; s(i,a)\mapsto~\overline{s}, \quad s(i,a)(i,b)\mapsto \overline{s(i,a)}\cup \{t\;|\; \exists_{s'\in\overline{s}}t\approx_{\tot{ A}} s'ab \}$,  we define the $\Pi$-game as follows.

\begin{definition}[$\Pi$-Game]
Given $A\in\ob(\DGame_!(I))$, $B\in\ob(\DGame_!(A ))$, we define $\Pi_{A}B\in\ob(\DGame_!(I ))$ with $\tot{ \Pi_{A}B}:=\tot{A}\Rightarrow\tot{B}$ and $(\Pi_{A}B)(\bot)$ carved out in $\tot{\Pi_{A}B}$ as follows\\
\resizebox{\linewidth}{!}{\parbox{1.02\linewidth}{
\begin{align*}
P_{(\Pi_{A} B)(\bot)}:=&\{\epsilon\}\;\bigcup \\
&\{sa \in P^\odd_{\tot{A}\Rightarrow \tot{B}}\;|\; s\in P_{(\Pi_{A} B)(\bot)}^\even\;\wedge\;  \exists_{\overline{sa\upharpoonright_{!\tot{ A}}}\subseteq \tau\in \str(A(\bot))} sa\in P_{A(\bot)\Rightarrow B(\tau)}\;\}\;\bigcup\\
&\{sab \in P^\even_{\tot{A}\Rightarrow \tot{B}}\;|\; sa \in P_{(\Pi_{A} B)(\bot)}^\odd\;\wedge\; \\
& \qquad\qquad\qquad\qquad  \quad\;\forall_{ \overline{sab\upharpoonright_{!\tot{ A}}}\subseteq \tau\in \str(A(\bot))}  sa\in P_{A(\bot)\Rightarrow B(\tau)}\Rightarrow sab\in P_{A(\bot)\Rightarrow B(\tau)}\;\}.
\end{align*}
}}
\end{definition}
We note that we can make $\DGame_!(A )$ into a ccc\footnote{Perhaps a more insightful way to think of this is as $\DGame_!(A)$ being obtained as a co-Kleisli category for a linear exponential comonad $!$ on a symmetric monoidal closed category $\DGame(A)$. Here, $\DGame(A)$ has the same objects as $\DGame_!(A)$ on which we define operations $I$, $\otimes$, $\multimap$ pointwise, while also performing the operation on $\tot{B}$, and $\tot{!B}:=!\tot{B}$ while $(!B)(\sigma):=\{s\in P_{!(B(\sigma))}\;|\; \exists_{\tau\in\str(B(\sigma))}\overline{s}\subseteq \tau\;\}$. We define $\DGame(A)(B,C):=\str(\Osat(\Pi_A(B\multimap C))$ with the obvious identity morphisms and composition. In fact, along similar lines, the games model of DTT that we present in this chapter can easily be modified to give a model of dDILL.} by defining $I$ and $\&$ pointwise on dependent games $B$, while also performing the operation on $\tot{ B}$, and by defining $P_{(B\Rightarrow C)(\sigma)}:=\{s\in P_{B(\sigma)\Rightarrow C(\sigma)}\;|\; \exists_{\tau\in\str(B(\sigma))}\overline{s\upharpoonright_{B(\sigma)}}\subseteq\tau\;\}$ and $\tot{ B\Rightarrow C}:=\tot{ B}\Rightarrow \tot{C}$. This lets us define $\DGame_!( A)( B,C):= \str(\Osat(\Pi_{A}(B\Rightarrow C)))$ with the obvious identity morphisms and composition, which we discuss later. Here, $\ob(\DGame_!(I))\ra{\Osat}\ob(\Gamecat_!)$, sends $(A(\bot),\tot{ A})$ to the game in which Opponent can play freely in $\tot{A}$ and Player has to respect the rules of the more restrictive game $A(\bot)$ as long as Opponent does:\\
\resizebox{\linewidth}{!}{\parbox{1.05\linewidth}{
\begin{align*}P_{\Osat(A(\bot),\tot{ A})}:=&\{\epsilon\}\;\bigcup\\
& \{sa \in P_{\tot{ A}}^{\odd}\;|\; s\in P^{\even}_{\Osat(A(\bot),\tot{ A})}\;\}\;\bigcup\;\\
&\{sab\in P^{\even}_{\tot{ A}}\;|\; sa\in P^{\odd}_{\Osat(A(\bot),\tot{ A})}\;\wedge\; (sa\in P_{A(\bot)}\Rightarrow sab\in P_{A(\bot)})\}.
\end{align*}}\hspace{20pt}\;}
\begin{remark}Note that, explicitly, the \emph{game of dependent functions from $A$ to $B$}, $\Osat(\Pi_{A} B)$, is carved out in $\tot{ A}\Rightarrow \tot{ B}$, as\\
\resizebox{\linewidth}{!}{\parbox{\linewidth}{
\begin{align*}
P_{\Osat(\Pi_{A} B)}:=&\{\epsilon\}\;\bigcup\\
&\{sa\in P^\odd_{\tot{A}\Rightarrow \tot{B}}\;|\; s\in P_{\Osat(\Pi_{A} B)}^\even \;\}\;\bigcup\\
&\{sab\in P^\even_{\tot{A}\Rightarrow \tot{B}}\;|\; sa \in P_{\Osat(\Pi_{A} B)}^\odd\;\wedge \;\\
& \qquad\qquad\qquad\qquad  \quad\;\forall_{ \overline{sab\upharpoonright_{!\tot{ A}}}\subseteq \tau\in \str(A(\bot))} sa\in P_{A(\bot)\Rightarrow B(\tau)}\Rightarrow sab\in P_{A(\bot)\Rightarrow B(\tau)}\;\}.
\end{align*}
}}\\
Indeed, this follows as $\overline{sab\upharpoonright_{!\tot{ A}}}=\overline{sa\upharpoonright_{!\tot{ A}}}$. An explicit proof is given for the more general claim of theorem \ref{thm:reppi}.

Recall that we would like $\tot{ -}$ to define a faithful functor to the world of simply typed games, being the semantic equivalent of $(-)^T$. It is for this reason that the game of dependent functions from $A$ to $B$ is saturated under all $O$-moves in $\tot{ A}\Rightarrow \tot{ B}$. We present $\Osat$ as a separate operation as this presentation will simplify the treatment of higher-order dependent functions in section \ref{sec:ctxt}.
\end{remark}

Following the mantra of game semantics for quantifiers \cite{abramsky2005game}, in $\Osat(\Pi_{A} B)$, Opponent can choose a strategy $\tau$ on $A(\bot)$ while Player has to play in a way that is compatible with all choices of $\tau$ that have not yet been excluded. Similarly to the approach taken in the game semantics for polymorphism \cite{abramsky2005game}, we do not specify all of $\tau$ in one go, as this would violate ``Scott's axiom'' of continuity of computation. Instead, $\tau$ is gradually revealed, explicitly so by playing in $!\tot{ A}$ and implicitly by playing in $\tot{ B}$. That is, unless Opponent behaves naughtily, in the sense that there is no strategy $\tau$ on $A(\bot)$ which is consistent with her behaviour \mccorrect{such that} $s\upharpoonright_{\tot{ B}}$ obeys the rules of $B(\tau)$. In case of such a naughty Opponent, any further play in $\tot{ A}\Rightarrow \tot{ B}$ is permitted.

\begin{remark}In particular, $\DGame_!(I)$ is a ccc which has $\Gamecat_!$ as a proper full subcategory. Note that the morphisms from $A$ to $B$ consist of the strategies on $\tot{ A}\Rightarrow \tot{ B}$ for which Player plays along the rules of $A(\bot)\Rightarrow  B(\bot)$ as long as Opponent does so and as long as there is a strategy on $A(\bot)$ which is consistent with her play.
\end{remark}

For a function $Y\ra{X}\Set$ to the class $\Set$ of sets, we define $\tot{X_*}:=(\bigcup_{y\in Y} X(y))_*$ and $X_*(y):=X(y)_*$.  For an example of non-constant type dependency, write  $\mathsf{days}(n):=\{m\;|\; \textnormal{there are $>m$ days in the year $n$}\}$\mccorrect{. Then }${\mathsf{days}_*}(n):={\mathsf{days}(n)_*}$ \mccorrect{is} a game depending on $\NInd$ (with ${\mathsf{days}_*}(n)={\mathbb{N}_{<365}{}_*}\textnormal{ or }{\mathbb{N}_{<366}{}_*}$ ). Note that this will not correspond to a finite inductive type family as the fibres of the type are not disjoint. Then, figure \ref{fig:depstr} gives four examples of valid dependently typed strategies.\begin{figure}[!tb]
\centering\resizebox{\linewidth}{!}{
$\begin{array}{c|c|c|c||c}
\begin{array}{ccc}
\NInd &\quad & {\mathsf{days}_*}\\
\hline
 & & *\\
 & & 364\\
 & & \\
 & & \\
 & & \\
 & & \\
\end{array}
\hspace{10pt}&\hspace{10pt}
\begin{array}{ccc}
!\NInd &\quad & {\mathsf{days}_*}\\
\hline
 & & *\\
(i,*) & & \\
(i,1984) & & \\
 & & 365 \\
 & & \\
 & & \\
\end{array}
\hspace{10pt}& \hspace{10pt}
\begin{array}{ccc}
!\NInd &\quad & {\mathsf{days}_*}\\
\hline
 & & *\\
(i,*) & & \\
(i,1985) & & \\
(i+1,*) & & \\
(i+1,1986) & & \\
 & & 365 \\
\end{array}
\hspace{10pt}& \hspace{10pt}
\begin{array}{ccccc}
!\NInd &\quad & !{\mathsf{days}_*} & \quad & {\mathsf{days}_*}\\
\hline
 & & & & *\\
 & &(i,*) &  & \\
 & &(i,m) &  & \\
 & & & & m\\
 & & & &\\
 & & & &
\end{array}
 \hspace{10pt}&
\begin{array}{c}
\vspace{4pt}\\
O\\
P\\
O\\
P\\
O\\
P
\end{array}
\end{array}
$}
\caption{\label{fig:depstr} Three plays in $\Osat(\Pi_{\NInd}{\mathsf{days}_*})$ and one in $\Osat(\Pi_{\NInd}({\mathsf{days}_*}\Rightarrow {\mathsf{days}_*}))$. The first as all years have $> 364$ days, the second as $1984$ was a leap year, the third as Player can play any move in $\tot{{\mathsf{days}}_*} ={\mathbb{N}_{<366}{}_*}$ after Opponent has not played along a (history-free) strategy on ${\mathbb{N}}_*$ and the fourth as Opponent makes the move $m$ first, after which Player can safely copy it. In the paired moves, Player chooses an (irrelevant)  index $i$. For an interpretation of the plays in $\Osat(\Pi_{\NInd}\mathsf{days}_*)$, imagine them as a dialogue between a departmental education manager (Opponent) and an academic (Player) where Player gets to choose for every year the date that she promises to have marked the students' end-of-year exam. A cheeky academic might try to suggest that she'll return the marked exams every year on the 366th day of the year without asking the manager for which year he wants to know the date. Clearly the manager should not accept this. This would correspond to a play $*365$, which is illegal as, by making the move $365$, Player would exclude certain fibres (the non-leap years), which is a privilege only Opponent has. }
\end{figure}
 The fourth example is especially important, as it generalises to a (derelicted) $B$-copycat on $\Osat(\Pi_{A}(B\Rightarrow B))$ for arbitrary $B$, denoted $\diagv{[A]}{[B]}$ in section \ref{sec:ctxt}. This~motivates why Opponent can narrow down the fibre of $B$ freely, while Player can only play without narrowing down the fibre further. To see that Player should not be able to narrow down the fibre of $B$, note that we do not want $f:=\{\epsilon, *365\}$ to define a strategy on $\Osat(\Pi_{\NInd}{\mathsf{days}_*})$, as $1983;f=\{\epsilon,*365\}\notin\str({\mathsf{days}_*(1983)})$.

\mccorrect{We now obtain the following result, whose proof we omit, as we shall prove a more general result in theorem} \ref{thm:cwf}.
\begin{theorem}We obtain a strict indexed ccc\footnote{That is, a functor from $\DGame_!(I)^{op}$ to the $1$-category $\mathsf{CCCat}$ of cartesian closed categories and \mccorrect{strict} cartesian closed functors.}
\begin{diagram}
\DGame_!(I)^{op} & \rTo^{(\DGame_!,-\{-\})} & \mathsf{CCCat}
\end{diagram}
 of dependent games, if we define
\begin{itemize}
\item fibrewise objects $\ob(\DGame_!(A)):=\{\str(\tot{ A})\ra{B}\Sub(\tot{ B})\;|\; \tot{ B}\in \ob(\Gamecat_!)\;\}$;
\item fibrewise hom-sets $\DGame_!(A)(B,C):=\str(\Osat(\Pi_{A}(B\Rightarrow C)))$;
\item fibrewise identities $\der_B:= \{s\in P^\even_{\Osat(\Pi_{A}(B\Rightarrow B))}\;|\;\forall_{s'\in P_{\Osat(\Pi_{A}(B\Rightarrow B))}^\even}s'\leq s\Rightarrow \exists_i  s'\upharpoonright_{!\tot{ B}}\upharpoonright_i\approx_{\tot{ B}} s '\upharpoonright_{\tot{ B}}\}$;
\item for $B\ra{\tau}C\ra{\tau'}D\in\DGame_!(A)$, we define the fibrewise composition $B\ra{\tau^\dagger;_A \tau'}D\in\DGame_!(A)$ as $\tau^\dagger;_A \tau':=\mathsf{diag}^\dagger_{A};(\tau^\dagger \otimes \tau') ; \mathsf{comp}_{\tot{ B},\tot{ C},\tot{ D}}$;
\item  given $f\in \Gamecat_!(A',A)$, we define the change of base functor $-\{f\}$: $B\{f\}\in \ob(\DGame_!(A'))$ where $B\{f\}(\sigma):=B(!(\sigma);f)$ and $\tot{ B\{f\}}:=\tot{ B}$ and $\tau\{f\}:=f^\dagger ; \tau$.
\end{itemize}
\end{theorem}

Seeing that $\DGame_!(I)$ additionally has a terminal object $I$ to interpret the empty context, we are well on our way to producing a  model of dependent type theory: we only need to interpret context extension. This takes the form of the \mccorrect{full and faithful} comprehension axiom for $\DGame_!$, which states that for each $A\in\ob(\DGame_!(I))$ and $B\in\ob(\DGame_!(A))$ the following presheaf is representable $$x\mapsto\DGame_!(\mathsf{dom}(x))(I,B\{x\}):(\DGame_!(I)/A)^{op}\ra{}\Set$$ 
\mccorrect{and that this induces a bijection $\DGame_!(A)(B,C)\cong \DGame_!(I)/A(\proj{A}{B},\proj{A}{C})$.}
Unfortunately, this fails, as $\DGame_!(I)$ does not yield a sound interpretation of dependent contexts. Essentially, the problem is that we do not have \emph{additive $\Sigma$-types}, appropriate generalisations $\Sigma_A^{\&} B$ of $\&$ to interpret dependent context extension in $\DGame_!(I)$ (\mccorrect{c.f theorem} \ref{thm:cohnosigma}). 

\begin{theorem}$\DGame_!$ does not satisfy the \mccorrect{full and faithful} comprehension axiom.
\end{theorem}
\begin{proof}
Let us write $\mathbb{B}:=\{\ttt,\fff\}$. Then, $\BInd$ is the usual flat game of Booleans. We can define a dependent game $\just_*$ over $\BInd$, where $\tot{ {\just}_*}:=\BInd$, ${\just}_*(\fff)={\{\fff\}}_*$ and  ${\just}_*(\ttt)={\{\ttt\}}_*$.

Then, note that the comprehension axiom (supposing that it holds) implies that, for any $C\in\ob(\DGame_!(\BInd))$,
\begin{mccorrection}
\begin{align*}
\str(\Osat(\Pi_{\BInd}({\just}_* \Rightarrow C)))&=\DGame_!(\BInd)({\just}_*,C)\\
&\cong \DGame_!(I)/\BInd(\proj{\BInd}{{\just}_*},\proj{\BInd}{C})\\
&\cong\DGame_!(\Sigma_{\BInd}^{\&}{\just}_*)(I,C\{\proj{\BInd}{{\just}_*}\})\\
&= \str(\Osat(\Pi_{\Sigma_{\BInd}^{\&}{\just}_*}C)),
\end{align*}
\end{mccorrection}
\mccorrect{where the second isomorphism is the full and faithfulness of the comprehension functor and the third isomorphism is the comprehension axiom (representability condition)} and where $\Sigma_{\BInd}^{\&}{\just}_*\ra{\proj{\BInd}{{\just}_*}}\BInd$ is the representing object above for $A=\BInd$ and $B={\just}_*$.

Now, taking $C(\tau)=I$ for all $\tau$ and $\tot{C}=D$ for some game $D$, implies that $\tot{ \Sigma_{\BInd}^{\&} {\just}_*}\cong\BInd\&\BInd$. Indeed, we have a natural bijection $\Gamecat_!(\tot{\Sigma_{\BInd}^{\&}{\just}_*},D)=\str(\tot{\Sigma_{\BInd}^{\&}{\just}_*}\Rightarrow D)=\str(\Osat(\Pi_{\Sigma_{\BInd}^{\&}{\just}_*}C))=\str(\Osat(\Pi_{\BInd}{\just}_*\Rightarrow C))=\str(\BInd\Rightarrow\BInd\Rightarrow D)=\Gamecat_!(\BInd,\BInd\Rightarrow D)\cong \Gamecat_!(\BInd\&\BInd,D)$, which according to the Yoneda lemma is induced by an isomorphism $\tot{ \Sigma_{\BInd}^{\&} {\just}_*}\cong\BInd\&\BInd$ in $\Gamecat_!$.

According to theorem \ref{thm:subfunctor} this induces an isomorphism $\Sub(\tot{ \Sigma_{\BInd}^{\&} {\just}_*})\cong\Sub(\BInd\&\BInd)$. Therefore, symmetry of $\just$ in $\ttt$ and $\fff$ implies that there are only nine options for $(\Sigma_{\BInd}^{\&} {\just}_*)(\bot)$: $I\& I$, $I\&{\emptyset}_*$, ${\emptyset}_*\& I$, ${\emptyset}_*\&{\emptyset}_*$,  $I\&\BInd$, $\BInd\& I$, ${\emptyset}_*\&\BInd$, $\BInd\& {\emptyset}_*$  and $\BInd\&\BInd$.

We take $C={\just}_*$ in the bijection implied by the comprehension axiom above, to obtain $\str(\Osat(\Pi_{\BInd}({\just}_* \Rightarrow {\just}_* )))\cong \str(\Osat(\Pi_{\Sigma_{\BInd}^{\&}{\just}_*}{\just}_* ))$. We see that none of the nine options is satisfactory. Indeed, $I\& I$, $I\&{\emptyset}_*$, ${\emptyset}_*\& I$, ${\emptyset}_*\&{\emptyset}_*$,  $\BInd\& I$ and $\BInd\& {\emptyset}_*$ would imply that the negation between the two copies of ${\just}_*$ is a member of the right hand side, but not the left hand side, which is a contradiction. Similarly, $I\& I$, $I\&{\emptyset}_*$, ${\emptyset}_*\& I$, ${\emptyset}_*\&{\emptyset}_*$,  $I\&\BInd$ and ${\emptyset}_*\&\BInd$ would imply that the negation between $\BInd$ and the second copy of ${\just}_*$ is a member of the right hand side, but not the left hand side, which is a contradiction. The last case of $\BInd\&\BInd$ also leads to a contradiction as it would restrict members of the right hand side to output $\ttt$ in response to having been supplied with arguments $\ttt$ and $\fff$ to the function upon request, while members of the left hand side would also be free to answer~$\fff$.
\end{proof}
This is a common problem we discussed in section \ref{sec:depgirard}. It also occurred for coherence space semantics, which is not surprising if we view game semantics as coherence space semantics extended in time. While we had a good candidate category $\Stable$ of $!$-coalgebras to extend $\Coh_!$, such an obvious candidate is not available for games. Section \ref{sec:depgirard} suggested that such a suitable category of $!$-coalgebras may be \emph{constructed} either by inductively closing the co-Kleisli category under $\Sigma$-types or by coinductively restricting the co-Eilenberg-Moore category to be closed under $\Pi$-types. As it is easier to get an explicit description of the former category,  we construct a category of \emph{context games} by formally closing $\Gamecat_!$ under a notion of $\Sigma$-type. It is on this category that we  base our model of dependent type theory.

\section{A Category with Families of Context Games}\label{sec:ctxt}

All is not lost, however. In fact, we have almost translated the syntax of dependently typed equational logic into the world of games and strategies. The remaining generalisation, necessitated by the lack of additive $\Sigma$-types, is to dependent games depending on multiple (mutually dependent) games. We can produce a categorical model of {$\DTTGame$} out of the resulting structure by applying a so-called \emph{category of contexts ($\Ctxt$) construction}, which is precisely how one builds a categorical model from the syntax of dependent type theory \cite{hofmann1997syntax,pitts1995categorical}. This construction can be seen as a way of making our indexed category satisfy the comprehension axiom, extending its base category by (inductively) adjoining (strong) $\Sigma$-types formally, analogous to the $\mathsf{Fam}$-construction of \cite{abramsky1998cbvgames} which adds formal coproducts. \mccorrect{We encourage the reader to view this closure in the light of section} \ref{sec:depgirard}: \mccorrect{as inductively closing the co-Kleisli category under $\Sigma$-types.}

The problem which needs to be addressed is how to interpret dependent types and dependent functions of more identifiers. This is done through a notion of context game and a generalisation of the $\Pi$-game construction from the previous section.
\begin{definition}[Context Game]
We inductively define a \emph{context game} to be a (finite) list $[X_i]_{1\leq i \leq n}$ where $X_{i}$ is a 
\emph{game with dependency on $[X_j]_{j< i}$}, i.e. a function $\str(\tot{ X_1})\times\cdots\times\str(\tot{ X_{i-1}})\cong \str(\tot{ X_1}\&\cdots\&\tot{ X_{i-1}})\ra{X_i}\Sub(\tot{ X_i})$ for some game $\tot{ X_i}$.
\end{definition}
To keep notation light, we sometimes abuse notation and write $A$ for the context game $[A]$ of length $1$.
\begin{definition}[Dependent $\Pi$-game]For a game $X_{n+1}$ depending on $[X_i]_{i\leq n}$, we define the game $\Pi_{X_n}X_{n+1}$ depending on $[X_i]_{i\leq n-1}$ by $\tot{ \Pi_{X_n}X_{n+1}}:=\tot{ X_n}\Rightarrow \tot{ X_{n+1}}$ from which $(\Pi_{X_n}X_{n+1})(\sigma_1,\ldots,\sigma_{n-1})$ is carved out as\newline\resizebox{\linewidth}{!}{ \parbox{\linewidth}{
\begin{align*}
P_{(\Pi_{X_n}X_{n+1})(\sigma_1,\ldots,\sigma_{n-1})}=&\{\epsilon\}\;\bigcup\\
&\{sa\;|\; s\in P_{(\Pi_{X_n}X_{n+1})(\sigma_1,\ldots,\sigma_{n-1})}^\even \;\wedge\; \\
&\qquad\;\;\exists_{\overline{sa\upharpoonright_{!\tot{ X_n}}}\subseteq\tau\in\str(X_n(\sigma_1,\ldots,\sigma_{n-1}))}sa\in P_{X_n(\sigma_1,\ldots,\sigma_{n-1})\Rightarrow X_{n+1}(\sigma_1,\ldots,\sigma_{n-1},\tau)}\;\}\;\bigcup\\
&\{sab\;|\; sa \in P_{(\Pi_{X_n}X_{n+1})(\sigma_1,\ldots,\sigma_{n-1})}^\odd\wedge \\
&\qquad\;\;\;\forall_{\overline{sab\upharpoonright_{!\tot{ X_n}}}\subseteq\tau\in\str(X_n(\sigma_1,\ldots,\sigma_{n-1}))}sa\in P_{X_n(\sigma_1,\ldots,\sigma_{n-1})\Rightarrow X_{n+1}(\sigma_1,\ldots,\sigma_{n-1},\tau)}\Rightarrow\\
&\qquad\;\;\; sab\in P_{X_n(\sigma_1,\ldots,\sigma_{n-1})\Rightarrow X_{n+1}(\sigma_1,\ldots,\sigma_{n-1},\tau)}\;\}.
\end{align*}}}
\end{definition}
The following explicit characterisation of \emph{the game $\Osat(\Pi_{X_1}\cdots\Pi_{X_n}X_{n+1})$ of dependent functions of multiple arguments} will be useful later. Indeed, its strategies will represent dependent functions from $[X_i]_{1\leq i\leq n}$ to $X_{n+1}$.
\begin{theorem}\label{thm:reppi}
Explicitly, $(\Pi_{X_k}\cdots\Pi_{X_n}X_{n+1})(\sigma_1,\ldots,\sigma_{k-1})$ can be inductively defined as the following subset of the plays of $\tot{X_k}\Rightarrow\cdots\Rightarrow \tot{X_{n+1}}$: \\
\resizebox{\linewidth}{!}{\parbox{\linewidth}{
\begin{align*}
&\{\epsilon\}\;\bigcup\\
&\{sa\;|\; s\in P_{(\Pi_{X_k}\cdots\Pi_{X_n}X_{n+1})(\sigma_1,\ldots,\sigma_{k-1})}^{\even}\;\wedge\;\\
&\qquad\;\;\; \exists_{\overline{sa\upharpoonright_{!\tot{ X_k}}}\subseteq\sigma_k\in\str(X_k(\sigma_1,\ldots,\sigma_{k-1}))}\cdots \exists_{\overline{sa\upharpoonright_{!\tot{ X_n}}}\subseteq\sigma_n\in\str(X_n(\sigma_1,\ldots,\sigma_{n-1}))}\\
&\qquad\;\;\; sa\in P_{X_k(\sigma_1,\ldots,\sigma_{k-1})\Rightarrow\cdots\Rightarrow  X_{n+1}(\sigma_1,\ldots,\sigma_{n})}\;\}\;\bigcup \\
&\{sab\;|\; sa \in P_{(\Pi_{X_k}\cdots\Pi_{X_n}X_{n+1})(\sigma_1,\ldots,\sigma_{k-1})}^\odd\wedge\\
&\qquad\;\;\; \forall_{\overline{sab\upharpoonright_{!\tot{ X_k}}}\subseteq \sigma_k\in\str(X_k(\sigma_1,\ldots,\sigma_{k-1}))}\cdots\forall_{\overline{sab\upharpoonright_{!\tot{ X_n}}}\subseteq \sigma_n\in \str(X_n(\sigma_1,\ldots,\sigma_{n-1}))} \\
&\qquad\;\;\; sa\in P_{X_k(\sigma_1,\ldots,\sigma_{k-1})\Rightarrow\cdots\Rightarrow  X_{n+1}(\sigma_1,\ldots,\sigma_{n})}\Rightarrow sab\in P_{X_k(\sigma_1,\ldots,\sigma_{k-1})\Rightarrow\cdots\Rightarrow  X_{n+1}(\sigma_1,\ldots,\sigma_{n})}\;\}.\qquad\;
\end{align*}}\hspace{30pt}\;}
As a consequence, the game of dependent functions $\Osat(\Pi_{X_1}\cdots\Pi_{X_n}X_{n+1})$ is carved out in $\tot{ X_1}\Rightarrow\cdots\Rightarrow\tot{ X_n}\Rightarrow \tot{ X_{n+1}}$ as the set of plays\newline
\resizebox{\linewidth}{!}{ \parbox{\linewidth}{
\begin{align*}
&\{\epsilon\}\;\bigcup\\
&\{sa\;|\; s\in P_{\Osat(\Pi_{X_1}\cdots\Pi_{X_n}X_{n+1})}^\even \;\}\;\bigcup\\
&\{sab\;|\; sa \in P_{\Osat(\Pi_{X_1}\cdots\Pi_{X_n}X_{n+1})}^\odd\wedge\\
&\qquad\;\;\; \forall_{\overline{sab\upharpoonright_{!\tot{ X_1}}}\subseteq \tau_1\in\str(X_1())}\cdots\forall_{\overline{sab\upharpoonright_{!\tot{ X_n}}}\subseteq \tau_n\in \str(X_n(\tau_1,\ldots,\tau_{n-1}))} \\
&\qquad\;\;\; sa\in P_{X_1()\Rightarrow\cdots \Rightarrow X_{n}(\tau_1,\ldots,\tau_{n-1})\Rightarrow X_{n+1}(\tau_1,\ldots,\tau_n)}\Rightarrow sab\in P_{X_1()\Rightarrow\cdots \Rightarrow X_{n}(\tau_1,\ldots,\tau_{n-1})\Rightarrow X_{n+1}(\tau_1,\ldots,\tau_n)}\;\}.
\end{align*}}}\newline
That is, the set of plays where Opponent can do whatever she pleases, while Player can only move without further determining the fibre of any of $X_1,\ldots, X_{n+1}$ as long as Opponent plays along compatible strategies $\sigma_1,\ldots,\sigma_n$ on $\tot{ X_1},\ldots,\tot{ X_n}$, in the sense that they extend to\\
\resizebox{\linewidth}{!}{$\langle \tau_1,\ldots,\tau_n\rangle \in \Sigma(\str(X_1),\ldots,\str(X_n)):=\{\langle \tau_1,\ldots,\tau_n\rangle\;|\; \tau_1\in\str(X_1())\wedge\;\cdots\wedge\tau_n\in\str(X_n(\tau_1,\ldots,\tau_{n-1}))\}$}\\ such that the current play obeys the rules of $X_1()\Rightarrow\cdots\Rightarrow X_n(\tau_1,\ldots,\tau_{n-1})\Rightarrow X_{n+1}(\tau_1,\ldots,\tau_n)$. 
\end{theorem}
\begin{proof}We first note that the second claim follows straightforwardly from the first. Clearly, the proposed description of Opponent moves in the second claim is correct as, by definition of $\Osat$, Opponent is free to move in $\tot{X_1}\Rightarrow\cdots\Rightarrow \tot{X_n}\Rightarrow \tot{X_{n+1}}$ in $\Osat(\Pi_{X_1}\cdots\Pi_{X_n}X_{n+1})$. For Player moves, note that $\overline{sab\upharpoonright_{!\tot{X_i}}}=\overline{sa\upharpoonright_{!\tot{X_i}}}$ for all $1\leq i \leq n$. Therefore, assuming the first claim holds, it follows that we are in one of two cases:
\begin{itemize}
\item Opponent has been naughty and has broken the rules of $\Pi_{X_1}\cdots\Pi_{X_n}X_{n+1}$. In this case, there are no $\overline{sab\upharpoonright_{!\tot{X_1}}}\subseteq\tau_1\in\str(X_1()),\ldots,\overline{sab\upharpoonright_{!\tot{X_n}}}\subseteq\tau_n\in\str(X_n(\tau_1,\ldots,\tau_{n-1}))$ such that $sa\in X_1()\Rightarrow\cdots\Rightarrow X_{n+1}(\tau_1,\ldots,\tau_n)$. In this case, Player is allowed to do whatever she wants according to our proposed description of the second claim as the hypotheses of the implication defining the incremental condition on Player moves are false. This matches, of course,  the definition of $\Osat(\Pi_{X_1}\cdots\Pi_{X_n}X_{n+1})$ from the description of $\Pi_{X_1}\cdots\Pi_{X_n}X_{n+1}$ of the first claim. 
\item Opponent has been nice and has followed the rules of $\Pi_{X_1}\cdots\Pi_{X_n}X_{n+1}$ (i.e. there are in fact such $\tau_1,\ldots,\tau_n$). In this case, Player has to keep obeying the rules of $\Pi_{X_1}\cdots\Pi_{X_n}X_{n+1}$ as well according to the definition of $\Osat(\Pi_{X_1}\cdots\Pi_{X_n}X_{n+1})$. This matches our proposed description of the second claim.
\end{itemize}
For the first claim, the idea is that Opponent has to play precisely such that there is \emph{some} compatible assignment of strategies $\sigma_k,\ldots,\sigma_n$ on $X_k,\ldots,X_n$ while Player has to play such that she does not exclude \emph{any} such compatible assignment of  strategies. Formally, we prove by induction that the proposed description of plays in $(\Pi_{X_k}\cdots\Pi_{X_n}X_{n+1})(\sigma_1,\ldots,\sigma_{k-1})$ coincides with its definition
\\
\\
\resizebox{\linewidth}{!}{
$
\begin{alignedat}{2}
&\{\epsilon\}\;\bigcup\\
&\{sa\;|\; s\in P_{(\Pi_{X_k}\cdots\Pi_{X_n}X_{n+1})(\sigma_1,\ldots,\sigma_{k-1})}^{\even}\; \wedge\;  \exists_{\overline{sa\upharpoonright_{!\tot{ X_k}}}\subseteq\sigma_k\in\str(X_k(\sigma_1,\ldots,\sigma_{k-1}))}  sa\upharpoonright_{!\tot{X_{k+1}},\ldots,\tot{X_{n+1}}}\in P_{(\Pi_{X_{k+1}}\cdots\Pi_{X_n}X_{n+1})(\sigma_1,\ldots,\sigma_{k})}\\
& \wedge\; \exists_{\overline{sa\upharpoonright_{!\tot{ X_{k+1}}}}\subseteq\sigma_{k+1}\in\str(X_{k+1}(\sigma_1,\ldots,\sigma_{k}))}  sa\upharpoonright_{!\tot{X_{k+2}},\ldots,\tot{X_{n+1}}}\in P_{(\Pi_{X_{k+2}}\cdots\Pi_{X_n}X_{n+1})(\sigma_1,\ldots,\sigma_{k+1})} \;\wedge\;\cdots \;\\
& \wedge\;  \exists_{\overline{sa\upharpoonright_{!\tot{ X_{n}}}}\subseteq\sigma_n\in\str(X_{k+1}(\sigma_1,\ldots,\sigma_{n-1}))} sa\upharpoonright_{\tot{X_{n+1}}}\in P_{X_{n+1}(\sigma_1,\ldots,\sigma_n)}\;\}\;\bigcup \\
&\{sab\;|\; sa \in P_{(\Pi_{X_k}\cdots\Pi_{X_n}X_{n+1})(\sigma_1,\ldots,\sigma_{k-1})}^\odd\wedge( \forall_{\overline{sab\upharpoonright_{!\tot{ X_k}}}\subseteq \sigma_k\in\str(X_k(\sigma_1,\ldots,\sigma_{k-1}))}sa\in P_{X_k(\sigma_1,\ldots,\sigma_{k-1})\Rightarrow (\Pi_{X_{k+1}}\cdots\Pi_{X_n}X_{n+1})(\sigma_1,\ldots,\sigma_k)}\;\\
&\Rightarrow \;
sab\upharpoonright_{!\tot{X_k}}\in P_{X_k(\sigma_1,\ldots,\sigma_{k-1})})\;\wedge\;(\forall_{\overline{sab\upharpoonright_{!\tot{ X_{k+1}}}}\subseteq \sigma_{k+1}\in\str(X_{k+1}(\sigma_1,\ldots,\sigma_{k}))} sa\upharpoonright_{!\tot{X_{k+1}},\ldots,\tot{X_{n+1}}}\in P_{X_{k+1}(\sigma_1,\ldots,\sigma_{k})\Rightarrow (\Pi_{X_{k+2}}\cdots\Pi_{X_n}X_{n+1})(\sigma_1,\ldots,\sigma_{k+1})}\;\\
&\Rightarrow \;
sab\upharpoonright_{!\tot{X_{k+1}}}\in P_{X_{k+1}(\sigma_1,\ldots,\sigma_{k})})\;\wedge\;
\cdots\;\wedge\;\ (\forall_{\overline{sab\upharpoonright_{!\tot{ X_n}}}\subseteq \sigma_n\in \str(X_n(\sigma_1,\ldots,\sigma_{n-1}))} \\
& sa\upharpoonright_{!\tot{X_{n}},\tot{X_{n+1}}}\in P_{(\Pi_{X_n}X_{n+1})(\sigma_1,\ldots,\sigma_{n-1})}\Rightarrow sab\upharpoonright_{!\tot{X_{n}}}\in P_{X_n(\sigma_1,\ldots,\sigma_{n-1})}\;\wedge\;sab\upharpoonright_{\tot{X_{n+1}}}\in P_{X_{n+1}(\sigma_1,\ldots,\sigma_{n})})\;\}.
\end{alignedat}
$}\\
\\
We note that the proposed description is valid for $\epsilon$. Let us suppose it is valid for $s$. Note that all conjuncts involving $s$ (rather than $sa$) in the incremental condition on Opponent moves then hold by induction. Rearranging the incremental condition on Opponent moves now gives us a description in which we have obtained the required incremental condition on Opponent moves, but not yet on Player moves -- in particular, our proposed description is now valid for $sa$:
\\
\resizebox{\linewidth}{!}{\parbox{\linewidth}{
\begin{align*}
&\{\epsilon\}\;\bigcup\\
&\{sa\;|\; s\in P_{(\Pi_{X_k}\cdots\Pi_{X_n}X_{n+1})(\sigma_1,\ldots,\sigma_{k-1})}^{\even}\;\wedge\; \exists_{\overline{sa\upharpoonright_{!\tot{ X_k}}}\subseteq\sigma_k\in\str(X_k(\sigma_1,\ldots,\sigma_{k-1}))}\cdots \exists_{\overline{sa\upharpoonright_{!\tot{ X_n}}}\subseteq\sigma_n\in\str(X_n(\sigma_1,\ldots,\sigma_{n-1}))}\\
& sa\in P_{X_k(\sigma_1,\ldots,\sigma_{k-1})\Rightarrow\cdots\Rightarrow  X_{n+1}(\sigma_1,\ldots,\sigma_{n})}\;\}\;\bigcup \\
&\{sab\;|\; sa \in P_{(\Pi_{X_k}\cdots\Pi_{X_n}X_{n+1})(\sigma_1,\ldots,\sigma_{k-1})}^\odd\wedge( \forall_{\overline{sab\upharpoonright_{!\tot{ X_k}}}\subseteq \sigma_k\in\str(X_k(\sigma_1,\ldots,\sigma_{k-1}))}sa\in P_{X_k(\sigma_1,\ldots,\sigma_{k-1})\Rightarrow (\Pi_{X_{k+1}}\cdots\Pi_{X_n}X_{n+1})(\sigma_1,\ldots,\sigma_k)}\;\\
&\Rightarrow \;
sab\upharpoonright_{!\tot{X_k}}\in P_{X_k(\sigma_1,\ldots,\sigma_{k-1})})\;\wedge\;(\forall_{\overline{sab\upharpoonright_{!\tot{ X_{k+1}}}}\subseteq \sigma_{k+1}\in\str(X_{k+1}(\sigma_1,\ldots,\sigma_{k}))} sa\upharpoonright_{!\tot{X_{k+1}},\ldots,\tot{X_{n+1}}}\in P_{X_{k+1}(\sigma_1,\ldots,\sigma_{k})\Rightarrow (\Pi_{X_{k+2}}\cdots\Pi_{X_n}X_{n+1})(\sigma_1,\ldots,\sigma_{k+1})}\;\\
&\Rightarrow \;
sab\upharpoonright_{!\tot{X_{k+1}}}\in P_{X_{k+1}(\sigma_1,\ldots,\sigma_{k})})\;\wedge\;
\cdots\;\wedge\;\ (\forall_{\overline{sab\upharpoonright_{!\tot{ X_n}}}\subseteq \sigma_n\in \str(X_n(\sigma_1,\ldots,\sigma_{n-1}))} \\
& sa\upharpoonright_{!\tot{X_{n}},\tot{X_{n+1}}}\in P_{(\Pi_{X_n}X_{n+1})(\sigma_1,\ldots,\sigma_{n-1})}\Rightarrow sab\upharpoonright_{!\tot{X_{n}}}\in P_{X_n(\sigma_1,\ldots,\sigma_{n-1})}\;\wedge\;sab\upharpoonright_{\tot{X_{n+1}}}\in P_{X_{n+1}(\sigma_1,\ldots,\sigma_{n})})\;\}.
\end{align*}}}\\
\\
Next, noting that our proposed description holds for $sa$, we note that the conjuncts $$sa\upharpoonright_{!\tot{X_m},\ldots,\tot{ X_{n+1}}}\in P_{X_m(\sigma_1,\ldots,\sigma_{m-1})\Rightarrow (\Pi_{X_{m+1}}\cdots\Pi_{X_n}X_{n+1})(\sigma_1,\ldots,\sigma_{m})}$$ in the incremental condition on $P$-moves can be replaced by the conditions\\
\resizebox{\linewidth}{!}{
$\exists_{\overline{sa\upharpoonright_{!\tot{ X_{m+1}}}}\subseteq\sigma_{m+1}\in\str(X_{m+1}(\sigma_1,\ldots,\sigma_{m}))}\cdots \exists_{\overline{sa\upharpoonright_{!\tot{ X_n}}}\subseteq\sigma_n\in\str(X_n(\sigma_1,\ldots,\sigma_{n-1}))} sa\upharpoonright_{!\tot{X_m},\ldots, \tot{X_{n+1}}}\in P_{X_m(\sigma_1,\ldots,\sigma_{m-1})\Rightarrow\cdots\Rightarrow  X_{n+1}(\sigma_1,\ldots,\sigma_{n})}.$}\\
\mccorrect{(Seeing that Opponent chooses the fibre, provided that our description holds.)}

Now, again noting that $\overline{sa\upharpoonright_{!\tot{X_l}}}=\overline{sab\upharpoonright_{!\tot{X_l}}}$, this means that all universal quantifiers in the incremental condition on $P$-moves range over a non-empty domain, so we might as well move them to the front of our formula, seeing that they do not bind any more identifiers:
\\
\resizebox{\linewidth}{!}{\parbox{\linewidth}{
\begin{align*}
&\{\epsilon\}\;\bigcup\\
&\{sa\;|\; s\in P_{(\Pi_{X_k}\cdots\Pi_{X_n}X_{n+1})(\sigma_1,\ldots,\sigma_{k-1})}^{\even}\;\wedge\; \exists_{\overline{sa\upharpoonright_{!\tot{ X_k}}}\subseteq\sigma_k\in\str(X_k(\sigma_1,\ldots,\sigma_{k-1}))}\cdots \exists_{\overline{sa\upharpoonright_{!\tot{ X_n}}}\subseteq\sigma_n\in\str(X_n(\sigma_1,\ldots,\sigma_{n-1}))}\\
& sa\in P_{X_k(\sigma_1,\ldots,\sigma_{k-1})\Rightarrow\cdots\Rightarrow  X_{n+1}(\sigma_1,\ldots,\sigma_{n})}\;\}\;\bigcup \\
&\{sab\;|\; sa \in P_{(\Pi_{X_k}\cdots\Pi_{X_n}X_{n+1})(\sigma_1,\ldots,\sigma_{k-1})}^\odd\wedge \forall_{\overline{sab\upharpoonright_{!\tot{ X_k}}}\subseteq \sigma_k\in\str(X_k(\sigma_1,\ldots,\sigma_{k-1}))}\cdots \forall_{\overline{sab\upharpoonright_{!\tot{ X_n}}}\subseteq \sigma_n\in \str(X_n(\sigma_1,\ldots,\sigma_{n-1}))} \\
&  sa\in P_{X_k(\sigma_1,\ldots,\sigma_{k-1})\Rightarrow\cdots\Rightarrow  X_{n+1}(\sigma_1,\ldots,\sigma_{n})}\Rightarrow
sab\upharpoonright_{!\tot{X_k}}\in P_{X_k(\sigma_1,\ldots,\sigma_{k-1})}\;\wedge\;
\cdots\;\wedge\;\ \\
&  sab\upharpoonright_{!\tot{X_{n}}}\in P_{X_n(\sigma_1,\ldots,\sigma_{n-1})}\;\wedge\;sab\upharpoonright_{\tot{X_{n+1}}}\in P_{X_{n+1}(\sigma_1,\ldots,\sigma_{n})}\;\},
\end{align*}}}\\
which clearly carves out the same plays in $P_{\tot{X_1}\Rightarrow\cdots\Rightarrow\tot{X_n}\Rightarrow \tot{X_{n+1}}}$ as  our proposed description.
\end{proof}
\begin{remark}[Logical Predicates/Realizability?] Note that these games of dependent functions lead to quite a non-trivial notion of dependently typed strategy. Indeed, we can send a game $B$ with dependency on $A$ to a function $\str(\tot{A})\ra{}\mathcal{P}\str(\tot{B})$ which assigns a set of consistent strategies $\sigma\mapsto \mathsf{cstr}(B):=\str({B(\sigma)})\subseteq\str(\tot{B})$. One might wonder if this description is enough to recover our model from and if the model can be recast into a realizability style model \cite{coquand1998realizability}. In particular, this would mean that we send a pair $(A(\bot),\tot{A})$ to the pair $(\tot{A}, \mathsf{cstr}(A)\subseteq \str(\tot{A}))$. In a realizability model, one would expect this class $\mathsf{cstr}(A)$ of consistent strategies to behave as a logical predicate. In particular, given $\just\{\ttt\}=(\{\ttt\}_*,\BInd)$, if $\mathsf{cstr}$ were a logical predicate, we would have that $\mathsf{cstr}((\just\{\ttt\}\Rightarrow\BInd)\Rightarrow\BInd)=\mathsf{cstr}((\BInd\Rightarrow\BInd)\Rightarrow\BInd)=\str((\BInd\Rightarrow\BInd)\Rightarrow\BInd)$ as $\mathsf{cstr}(\BInd)=\str(\BInd)$. However, we have that $\mathsf{cstr}((\just\{\ttt\}\Rightarrow\BInd)\Rightarrow\BInd):=\str(\Osat((\just_*(\ttt)\Rightarrow\BInd)\Rightarrow \BInd))\subsetneq \str((\BInd\Rightarrow\BInd)\Rightarrow\BInd)$. Indeed, a consistent strategy on $(\just\{\ttt\}\Rightarrow\BInd)\Rightarrow\BInd$ cannot play $\fff$ in $\just\{\ttt\}\trianglelefteq \BInd$. We see that our notion of consistent strategy does not behave as a logical predicate. We get a more non-trivial notion of higher-order dependent function. The extra requirement that Player is constrained by type dependency in positive occurring types just as she is in strictly positively occurring ones is important to get an exact match with the syntax of dependent type theory.
\end{remark}

For illustration, define  a game ${\mathsf{RA}_*}$ depending on the context game $[\NInd,{\mathsf{days}_*}]$ by 
\begin{align*}
\mathsf{RA}(n,m)&:=\{\textnormal{Rick Astley lyrics from songs released before day $m$ of year $n$}\}
\end{align*} Then, the two strategies of figure \ref{fig:rickastley} illustrate that a dependent function may query its arguments in unexpected order or may not query some at all.
\begin{figure}[!tb]
\centering\resizebox{\linewidth}{!}{
$
\begin{array}{c|c||c}
\begin{array}{ccc}
!\NInd & !{\mathsf{days}_*} &{\mathsf{RA}_*} 
\\
\hline
&&*\\
&(i,*)&\\
&(i,m>206) & \\
(j,*) && \\
(j,1987) &&\\
 & & \textnormal{Never Gonna Give You Up}
\end{array}
&
\begin{array}{ccc}
!\NInd & {!\mathsf{days}_*}& {\mathsf{RA}_*} 
\\
\hline
&&*\\
(i,*) && \\
(i,n>1987) &&\\
 & & \textnormal{Never Gonna Let You Down}\\
 & & \\
 &&
\end{array}
&
\begin{array}{c}
\vspace{5pt}\\
O\\
P\\
O\\
P\\
O\\
P
\end{array}
\end{array}
$
}
\caption{\label{fig:rickastley} Two examples of (partial) strategies on the game $\Osat(\Pi_{\NInd}\Pi_{{\mathsf{days}_*}}{\mathsf{RA}_*})$, defining dependent functions of two arguments. Note that these lyrics come from a song released on day 207 of the year 1987, so Player does not constrain the fibre anywhere.}
\end{figure}
To illustrate the subtle nature of higher-order dependent functions with an example, define the game ${\mathsf{holidays}}_*$ depending on the context game $[\NInd,{\mathsf{days}}_*]$ by
\begin{align*}
\mathsf{holidays}(n,m)&:=\{\textnormal{holidays that are celebrated on day $m$ of year $n$}\}
\end{align*}

Figure \ref{fig:pirates} illustrates how Player is in charge of providing certain arguments (the positive ones) of dependent games and can therefore choose the fibre in some cases. (Opponent controls the negative arguments to dependent games.) In the figure below, Player controls the arguments of type ${\mathsf{days}_*}$ and ${\mathsf{holidays}_*}$, while Opponent is in charge of the type of years $\NInd$. We stress again that although Player has to play in accordance with any choice of year that Opponent could make, the converse is not true: Opponent can do what she likes and does not have to respect Player's choices of day and holiday.

\begin{figure}[!tb]
\centering\resizebox{\linewidth}{!}{
$
\begin{array}{c|c||c}
\begin{array}{ccccc}
!\NInd & !!{\mathsf{days}_*} & !!{\mathsf{holidays}_*} & !\BInd &\BInd \\
\hline
&&&&*\\
&&&(0,*)&\\
&&(0,(0,*))&&\\
&&(0,(0,\textnormal{International Talk Like a Pirate Day})) &&\\
&&&(0,\fff)&\\
&&&&\ttt\\
&&&&\\
&&&&\\
&&&&\\
&&&&
\end{array}
&\begin{array}{ccccc}
!\NInd & !!{\mathsf{days}_*} & !!{\mathsf{holidays}_*} & !\BInd &\BInd \\
\hline
&&&&*\\
&&&(0,*)&\\
&&(0,(0,*))&&\\
&&(0,(0,\textnormal{Holi})) &&\\
&(0,(0,*))& &&\\
(0,*)&&&&\\
(0,2015) &&&&\\
&(0,(0,65))&&&\\
&&&(0,\ttt) &\\
&&&&\ttt
\end{array}
&
\begin{array}{c}
\vspace{5pt}\\
O\\
P\\
O\\
P\\
O\\
P\\
O\\
P\\
O\\
P
\end{array}
\end{array}
$
}
\caption{\label{fig:pirates} Two plays in $\Osat(\Pi_{\NInd}\Pi_{\Pi_{{\mathsf{days}_*}}\Pi_{{\mathsf{holidays}}_*}\BInd}\BInd)$. For an interpretation, imagine Player is a PhD-student who is trying to decide if he is going on holidays and ends up asking his supervisor (Opponent) if she's okay with him doing so. The first play can be read as the dialogue where the supervisor asks if the student is planning to take any holidays, the student asks if he's allowed to, the supervisor wants to know what the occasion is, the student admits that his best excuse for wanting time off is International Talk Like a Pirate Day, the supervisor tells the student that he can't have time off and, finally, the student tells his supervisor that he's taking time off anyway for this important occasion. Here, Player can choose the holiday 'International Talk Like a Pirate Day' as it is celebrated each year, meaning that Player does not restrict the year we may be talking about (which, as a negative argument, belongs to Opponent). Note that by choosing this particular holiday, Player automatically fixes the day the holiday falls on, which is fine as the subgame ${\mathsf{days}_*}$ occurs positively in the total game we are playing in, meaning that Player is in charge of determining the corresponding argument. The second play corresponds to a dialogue with a more sensible student who uses the more respectable excuse of celebrating Holi to get time off from work. Here, Player has to let Opponent determine the year first, before she can answer with a date for Holi, as the date of Holi on the Gregorian calendar varies (while it is celebrated every year).}
\end{figure}
We define a category $\Ctxt(\DGame_!)$ with objects context games $[A_i]_{1\leq i\leq n}$ and morphisms which are defined inductively as (dependent) lists $[\sigma_i]_{1\leq i\leq n}$ of  strategies on appropriate games of dependent functions. To keep notation light, we sometimes abuse notation and write $\sigma$ for the context morphism $[\sigma]$ of length $1$.

We show that $\Ctxt(\DGame_!)$ has the structure of a category with families (CwF) (see definition \ref{def:cwf}), a canonical notion of model of dependently typed equational logic. This gives a more concise presentation of the resulting strict indexed category with comprehension, where we also add formal $\Sigma$-types in the fibres.

\begin{theorem}\label{thm:cwf}We have a CwF $(\Ctxt(\DGame_{!}),\Ty,\Tm,\proj{}{},\diagv{}{}, -.-, \langle - ,-\rangle)$.
\end{theorem}
\begin{proof}We define the required structures. All equations follow straightforwardly from the definitions and the two claims stated.\quad\\

\noindent\fbox{$\ob(\mathcal{C})$, $\Ty$, $-.-$, $\cdot$}\vspace{5pt}\\
We define a category $\mathcal{C}:=\Ctxt(\DGame_!)$ with  context games as objects. We define $\Ty([X_i]_i)$ as the set of \emph{context games with dependency on $[X_i]_i$}: $[Y_j]_j\in\Ty([X_i]_i)$ iff $[X_i]_i.[Y_j]_j:=[X_1,\ldots,X_n,Y_1,\ldots,Y_m]$ is a context game, while $\cdot:=[]$ is the terminal object.

\vspace{5pt}\noindent\fbox{$\mathsf{mor}(\mathcal{C})$, $-\{-\}_{\Ty}$}\vspace{5pt}\\
\begin{mccorrection}
Next, let $[X_i]_{i\leq n},[Z_k]_{k\leq n'}\in\ob(\mathcal{C})$, $[Y_j]_{j\leq m}\in\Ty([X_i]_{i\leq n}$ and let $$\tot{Z_1}\&\cdots\&\tot{Z_{n'}}\ra{f}\tot{X_1}\&\cdots\&\tot{X_n}$$ be a morphism in $\Gamecat_!$. Then, we define $[Y_j]_{j\leq m}\{f\}\in\Ty([Z_k]_{k\leq n'})$ by
\begin{align*}
\tot{Y_j\{f\}}&:=\tot{Y_j}\\
Y_j\{f\}(\sigma_1,\ldots,\sigma_{n'},\tau_1,\ldots,\tau_{j-1})&:=Y_j(\langle \sigma_1,\ldots,\sigma_{n'}\rangle^\dagger;f, \tau_1,\ldots,\tau_{j-1}).
\end{align*}
This, in turn, lets us define $\mathsf{mor}(\mathcal{C})$:\\
\resizebox{\linewidth}{!}{\parbox{\linewidth}{
\begin{align*}
&\Ctxt(\DGame_!)([X_i]_{i\leq n},[Y_j]_{j\leq m}) :=\left\{\;[f_j]_{j\leq m}\;|\; f_j\in\str(\Osat(\Pi_{X_1}\ldots\Pi_{X_n}Y_j\{\langle f_1,\ldots, f_{j-1}\rangle\}))\;\right\},
\end{align*}
}}\\
noting that $f_j$ can, in particular, be interpreted as a morphism
$$
\tot{X_1}\&\cdots\&\tot{X_n}\ra{f_j} \tot{Y_j} \in \Gamecat_!.
$$
\end{mccorrection}

\vspace{5pt}\noindent\fbox{$\id$, $\proj{}{}$, $\Tm$, $\diagv{}{}$, $\langle\cdot,\cdot\rangle$, (Cons-Id)}\vspace{5pt}\\
The identities are defined as lists of derelicted copycats. Let us define a strategy $\der_{[X_j]_j,X_i}$ which plays the derelicted copycat on all of $\tot{ X_i}$: $\der_{[X_j]_j,X_i}:=\{s\in P_{\Osat(\Pi_{X_1}\ldots\Pi_{X_{n}}X_i)}\;|\; 
\forall_{s'\in P^{\even}_{\Osat(\Pi_{X_1}\ldots\Pi_{X_{n}}X_i)}}s'\leq s\Rightarrow \exists_k s\upharpoonright_{!\tot{X_i}}\upharpoonright_k\approx_{\tot{ X_i}} s\upharpoonright_{\tot{X_i}} \}$. We then define $\id_{[X_i]_i}:=[\der_{[X_j]_j,X_i}]_i$ and $\proj{[X_i]_i}{[Y_j]_j}:=[\der_{[X_i]_i.[Y_j]_j,X_k}]_k$. Let us define\\
\\
\resizebox{\linewidth}{!}{$\Tm([X_i]_{i\leq n},[Y_j]_{j\leq m}):=\left\{\;[f_j]_{j\leq m}\;\big|\;[\der_{[X_i]_i,X_1},\ldots,\der_{[X_i]_i,X_n},f_1,\ldots,f_m] \in\Ctxt(\DGame_!)([X_i]_i,[X_i]_i.[Y_j]_j)\;\right\}.$}\\
\\
Then, we can define $\diagv{[X_i]_i}{[Y_j]_j}:=[\der_{[X_i]_i.[Y_j]_j,Y_k}]_k$. Note that these are well-defined because of the following claim.

\begin{claim*}$\der_{[X_j]_j,X_i}\in\str(\Osat(\Pi_{X_1}\cdots\Pi_{X_n}X_i\{[\der_{[X_j]_j,X_k}]_{k\leq i-1}\}))$.
\end{claim*}
\begin{proof}Note that Opponent makes every move first \mccorrect{in $X_i$}, so Player can copy it freely without restricting the fibre of $X_i$ further.\end{proof}

We define  $\langle [f_j]_{j\leq m},[g_k]_{k\leq l}\rangle :=[f_1,\ldots,f_m,g_1,\ldots,g_l]$, after which (Cons-Id) follows trivially.

\vspace{5pt}\noindent\fbox{Composition, $-\{-\}_{\Tm}$, (Cons-Nat)}\vspace{5pt}\\
We define the composition of $[X_i]_{i\leq n}\ra{[f_j]_j}[Y_j]_{j\leq m}\ra{[g_k]_k}[Z_k]_{k}$ in $\Ctxt(\DGame_!)$ by
$$[f_j]_j;[g_k]_k:=[\langle f_1,\ldots,f_m\rangle^\dagger ; g_k]_{k},$$
using the usual (co-Kleisli) composition of  strategies on $\tot{ X_1}\Rightarrow \cdots\Rightarrow \tot{ X_n}\Rightarrow (\tot{ Y_1}\&\cdots\&\tot{ Y_m})$ and $\tot{ Y_1}\Rightarrow \cdots\Rightarrow \tot{ Y_m}\Rightarrow \tot{Z_k}$. We note that we can assign to this composition a more precise dependent function type.

\begin{claim*}
The composition $[f_j]_j;[g_k]_k$ above does, in fact, define a morphism in $\Ctxt(\DGame_!)([X_i]_{i\leq n},[Z_k]_k)$.
\end{claim*}
\begin{proof}We need to verify that $\langle f_1,\ldots,f_m\rangle^\dagger ; g_k$ defines a winning strategy on $\Osat(\Pi_{X_1}\cdots\Pi_{X_n}W\{[f_j]_j\})$, where we write $W:=Z_k\{[g_{k'}]_{k'<k}\}$. The winning part of the claim follows trivially from the usual fact that winning strategies compose. What is to be verified is the claim that $\langle f_1,\ldots,f_m\rangle^\dagger; g_k$ is a strategy on $\Osat(\Pi_{X_1}\cdots\Pi_{X_n}W\{[f_j]_j\})$.

Recall that, by assumption, $g_k$ is a strategy on $\Osat(\Pi_{Y_1}\cdots\Pi_{Y_m}W)$. Suppose $\langle f_1,\ldots, f_m\rangle^\dagger;g_k$ wants to respond with a move $b$ in some $X_i$ \mccorrect{or $W$} after a play $sa$. Recall that by theorem \ref{thm:reppi}, we need to verify that for all $\overline{sab\upharpoonright_{!\tot{ X_1}}}\subseteq \sigma_1'\in\str(X_1()),\ldots,\overline{sab\upharpoonright_{!\tot{ X_n}}}\subseteq \sigma_n'\in \str(X_n(\sigma_1',\ldots,\sigma_{n-1}'))$, we have that
$$  sab\upharpoonright_{!\tot{ X_1}}\in P_{!X_1()} \wedge \cdots \wedge sab\upharpoonright_{!\tot{ X_n}}\in P_{!X_n(\sigma_1',\ldots,\sigma_{n-1}')}\wedge sab\upharpoonright_{\tot{ W}}\in P_{W\{[f_j]_j\}(\sigma_1',\ldots,\sigma_{n}')},
$$
provided that already
$$  sa\upharpoonright_{!\tot{ X_1}}\in P_{!X_1()} \wedge \cdots \wedge sa\upharpoonright_{!\tot{ X_n}}\in P_{!X_n(\sigma_1',\ldots,\sigma_{n-1}')}\wedge sa\upharpoonright_{\tot{ W}}\in P_{W\{[f_j]_j\}(\sigma_1',\ldots,\sigma_{n}')}.
$$
Here, all but the last conjunct follow from the fact that the $f_k$ are strategies on $\Osat(\Pi_{X_1}\cdots\Pi_{X_{n}}Y_k\{[f_{k'}]_{k'<k}\})$. \mccorrect{(Seeing that $g_k$ will not break the dependency in $[Y_k]_k$ as a strategy on $\Osat(\Pi_{Y_1} \cdots \Pi_{Y_m}Z_k\{[g_{k'}]_{k'<k}\})$.)}

What remains to be checked, therefore, is that
$$
sab\upharpoonright_{\tot{ W}}\in P_{W\{[f_j]_j\}(\sigma_1',\ldots,\sigma_{n}')},$$
or equivalently,
$$sab\upharpoonright_{\tot{ W}}\in P_{W(\langle \sigma_1',\ldots,\sigma_{n}'\rangle^\dagger; f_1,\ldots, \langle \sigma_1',\ldots,\sigma_{n}'\rangle^\dagger; f_m)}.$$
This follows immediately from the fact that $g_k$ is a strategy on $\Osat(\Pi_{Y_1}\cdots \Pi_{Y_m}W)$ if we can show that $[\sigma_i']_{i\leq n};[f_j]_{j\leq m}\in \Ctxt(\DGame_!)([],[Y_j]_{j\leq m})$, which is a special case of our claim when $[X_i]_i=[]$.

That is, we need to demonstrate that $\langle \sigma_1',\ldots,\sigma_n'\rangle^\dagger;f_j$ defines a strategy on\linebreak $\Osat(Y_j\{[\sigma_i']_{i\leq n};[f_j]_j\})$. (Again, it follows trivially that it will be a winning strategy.) We verify that for any play $sab$ in $\langle \sigma_1',\ldots,\sigma_n'\rangle^\dagger || f_j$, where $b$ is a Player move in $\tot{ Y_j}$, we have in fact that it is a move in $\Osat(Y_j\{[\sigma_i']_{i\leq n};[f_j]_j\})$. This follows from the fact that $f_j$ is a strategy on $\Osat(\Pi_{X_1}\cdots \Pi_{X_n}Y_j\{[f_{j'}]_{j'<j}\})$, which according to theorem \ref{thm:reppi} means, in particular, that for all\newline $\overline{sab\upharpoonright_{!\tot{ X_1}}}\subseteq \sigma_1'\in\str(X_1()),\ldots\,\overline{sab\upharpoonright_{!\tot{ X_n}}}\subseteq \sigma_n'\in \str(X_n(\sigma_1',\ldots,\sigma_{n-1}'))$, we have that
$$  sab\upharpoonright_{!\tot{ X_1}}\in P_{!X_1()} \wedge \cdots \wedge sab\upharpoonright_{!\tot{ X_n}}\in P_{!X_n(\sigma_1',\ldots,\sigma_{n-1}')}\wedge sab\upharpoonright_{\tot{ Y_j}}\in P_{Y_j\{[f_{j'}]_{j'<j}\}(\sigma_1',\ldots,\sigma_{n}')},
$$
provided that already
$$  sa\upharpoonright_{!\tot{ X_1}}\in P_{!X_1()} \wedge \cdots \wedge sa\upharpoonright_{!\tot{ X_n}}\in P_{!X_n(\sigma_1',\ldots,\sigma_{n-1}')}\wedge sa\upharpoonright_{\tot{ Y_j}}\in P_{Y_j\{[f_{j'}]_{j'<j}\}(\sigma_1',\ldots,\sigma_{n}')}.
$$
The last conjunct is what we are looking for, or rather its reformulation $sab\upharpoonright_{\tot{ Y_j}}\in P_{Y_j\{[\sigma_i']_i;[f_j]_j\}}$.

\end{proof}

Note that for $[X_i]_i\ra{[f_j]_j}[Y_j]_j$ and $[Y_j]_j\ra{\langle [g_k]_k,[h_l]_l\rangle }[Z_k]_k.[W_l]_l$, (Cons-Nat) holds in the sense that
$$[f_j]_j;\langle [[g_k]_k,[h_l]_l\rangle=\langle [f_j]_j;[g_k]_k,[h_l]_l\{[f_j]_j\}\rangle ,
$$if we define $-\{-\}_{\Tm}$ by
$$[h_l]_l\{[f_j]_j\}:=[\langle f_1,\ldots, f_m\rangle^\dagger ; h_l]_l,
$$
which then automatically type checks because of (Cons-Nat).

\vspace{5pt}\noindent\fbox{Id. Law, Assoc., (Ty-Id), (Tm-Id), (Ty-Comp), (Tm-Comp), (Cons-L), (Cons-R)}\vspace{5pt}\\
All these identities are direct consequences of the identity and associativity laws of the usual composition of strategies in $\Gamecat_!$. 
\end{proof}

\begin{remark}Note that, in $\Ctxt(\DGame_!)$, $[A,B]\cong [A\& B]$ if $A$ and $B$ are games (without mutual dependency) and $[]\cong [I]$.
\end{remark}

\section{Semantic Type Formers $1$, $\Sigma$, $\Pi$ and $\Id$}\label{sec:semtype}
We show that our CwF supports $1$-, $\Sigma$-, $\Pi$-, and $\Id$-types. We leave the verification of all term equations (which are inherited from their simply typed equivalents) to section \ref{sec:compl}. As for type equations, we can note that all type formers are preserved by substitution. We characterise some of the properties of the $\Id$-types, marking their place in the intensionality spectrum.

\paragraph*{$1$-types}
$1$-types are interpreted by the context game of length $0$. $\langle\rangle$ is interpreted by the list of strategies of length $0$.

\paragraph*{$\Sigma$-Types}
$\Sigma$-types are (formally) defined by concatenation of lists. For  $[Z_k]_{k\leq l}\in\Ty([X_i]_{i\leq n}.[Y_j]_{j\leq m})$, we define a $\Sigma$-type $$\Sigma_{[Y_j]_j}[Z_k]_k:=[Y_j]_j.[Z_k]_k\in\Ty([X_i]_{i\leq n}).$$
We can interpret $\langle -,-\rangle$ by concatenation $[\sigma_1,\ldots,\sigma_m,\tau_1,\ldots,\tau_l]$ of lists of strategies $[\sigma_j]_j$ and $[\tau_k]_k$, while we interpret $\mathsf{fst}$ as $[\der_{[X_i]_i.[Y_j]_j.[Z_k]_k,Y_{j'}}]_{j'}$ and $\mathsf{snd}$ as $[\der_{[X_i]_i.[Y_j]_j.[Z_k]_k,Z_{k'}}]_{k'}$.

\paragraph*{$\Pi$-Types}
We have already seen $\Pi$-types $$\Pi_{[Y_j]_{j\leq m}}[Z]:=[\Pi_{Y_1}\cdots\Pi_{Y_m}Z]\in \Ty([X_i]_{i\leq n})$$ of dependent games $[Z]\in\Ty([X_i]_{i\leq n}.[Y_j]_{j\leq m})$. They let us define $\lambda$-abstraction and evaluation as on the usual simply typed function game $\tot{Y_1}\Rightarrow\cdots\Rightarrow\tot{Y_m}\Rightarrow\tot{Z}$. 
What remains to be defined are $\Pi$-types $\Pi_{[Y_j]_j}[Z_k]_k$ of general dependent context games $[Z_k]_k\in\Ty([X_i]_{i\leq n}.[Y_j]_{j\leq m})$. These can be reduced to the former, as $\Sigma_{f:\Pi_{x:A}B}\Pi_{x:A}C[f(x)/y]$ satisfies the rules for $\Pi_{x:A}\Sigma_{y:B}C $. Conclusion: our CwF supports $\Pi$-types.

\begin{corollary}This means that $\Ctxt(\DGame_!)$ is in particular a ccc.\end{corollary}

\paragraph*{$\Id$-Types}
We turn to identity types next, which are essentially defined as one would expect from their definition as an inductive family. Interestingly, due to the intensional nature of function types in game semantics, these identity types acquire a very intensional character as well, refuting $\mathsf{FunExt}$.

For $[Y_j]_j\in\Ty([X_i]_i)$, define $\Id_{[Y_j]_j}\in\Ty([X_i]_i.[Y_j]_j.[Y_{j'}]_{j'})$:
$$\Id_{[Y_j]_j}([\sigma_i]_i,[\tau_j]_j,[\tau_j']_{j}):=[\begin{array}{ll}
\{\reflind\}_* & \mathsf{if}\; [\tau_j]_j = [\tau'_j]_j\\
\emptyset_* &\mathsf{else}
\end{array}].$$
Here, $\tot{\Id_{[Y_j]_j}}:=\{\reflind\}_*$. Note that, by definition, the plays of $\Id_{[Y_j]_j}([\sigma_i]_i,[\tau_j]_j,[\tau_j']_{j}) $ are closed under all Opponent moves in $\tot{ \Id_{[Y_j]_j}}$. Note that this means that $\Osat(\Id_{[Y_j]_j}([\sigma_i]_i,[\tau_j]_j,[\tau'_j]_j))=\Id_{[Y_j]_j}([\sigma_i]_i,[\tau_j]_j,[\tau'_j]_j)$.

$\Id$-I is interpreted by the non-strict strategy
\begin{mccorrection} \begin{align*}\refl{[f_j]_j}:=[\{\epsilon,*,*\reflind\}]&\in\Tm([X_i]_i,\Id_{[Y_j]_j}\{\langle[\der_{X_i}]_i,[f_j]_j,[f_j]_j\rangle\})\\
&=\{ [\sigma] \; |\; \sigma\in \str\{\reflind\}_*\}.\end{align*}
\end{mccorrection}
For the (strong) $\Id$-E rule, suppose we are given
\begin{itemize}
\item $[Z_k]_k\in\Ty([X_i]_i.[Y_j]_j.[Y_j]_j.\Id_{[Y_j]_j})$;
\item $[f_k]_k\in\Tm([X_i]_i.[Y_j]_j,[Z_k]_k\{\langle \der_{[X_i]_i},\der_{[Y_j]_j},\der_{[Y_j]_j},\refl{\der_{[Y_j]_j}}\rangle\})$.
\end{itemize}
Then, we produce $$[f'_k]_k\in\Tm([X_i]_i.[Y^{(1)}_j]_j.[Y^{(2)}_j]_j.\Id_{[Y_j]_j},[Z_k]_k).$$
Here, $f'_k$ is the strategy which responds to the initial move in $\tot{Z_k}$ by opening $\tot{\Id_{[Y_j]_j}}$, encoding the initial move in the index, and if Opponent responds $\reflind$ continues playing $f_k$ using the left hand side copy of $Y_j$. (Hence, $[f'_k]_k$ does not ever visit $[Y_j^{(2)}]_j$.) Note that such $f'_k$ are well-defined strategies, as
\begin{itemize}
\item all fibres of $\Id$-types contain the intial $O$-move, allowing $f_k'$ to always play it;
\item the moment that Opponent plays $\reflind$, she excludes the fibres for which the two arguments of type $[Y_j]_j$ are not equal as we have a bijection
\begin{diagram}\Ctxt(\DGame_!)([],[X_i]_i.[Y_j]_j)& \cong & \Ctxt(\DGame_!)([],[X_i]_i.[Y_j]_j.[Y_j]_j.\Id_{[Y_j]_j})\\
\langle [\sigma_i]_i,[\tau_j]_j\rangle & \rMapsto & \langle [\sigma_i]_i,[\tau_j]_j,[\tau_j]_j,[\reflind]\rangle.\end{diagram}
Hence we can continue playing $f_k$ from that point.
\end{itemize}
Noting that Player always has a response in the initial protocol and next follows $f_k$, it follows that $f'_k$ are winning iff $f_k$ are.

\begin{remark} There is an alternative, more extensional, definition of $\Id$-types which is tempting and which gives $\Id$-types satisfying the principle of function extensionality. One reason we have chosen to work with these $\Id$-types instead is that the other definition does not generalise to a situation where we are working with non-winning or non-deterministic strategies.

The idea is to restrict the model to dependent games which send applicatively equivalent strategies\footnote{That is, strategies which cannot be distinguished through their interaction with applicative contexts of ground type.} to equal subgames. In that case, we can define the identity type 
$$\Id_{[Y_j]_j}([\sigma_i]_i,[\tau_j]_j,[\tau_j']_{j}):=[\Id_{Y_j}]_j([\sigma_i]_i,[\tau_j]_j,[\tau_j']_{j}):=[\tau_j^{app}\cap \tau'_j{}^{app}]_j.$$
Here, $\tot{ \Id_{Y_j}}:=\tot{ Y_j}$ and we write $\phi^{app}$ for the closure of a set of plays $\phi$ under applicative equivalence. Then, $\Id$-I is interpreted by $\refl{[f_j]_j}:=[f_j]_j\in\Tm([X_i]_i,\Id_{[Y_j]_j}\{\langle[\der_{X_i}]_i,[f_j]_j,[f_j]_j\rangle\})=\{[g_j]_j\;|\; g_j\in \str(\Osat(\Pi_{[X_i]_i}f_j^{app}\{[g_k]_{k<j}\}))\}$, where we interpret the \mccorrect{non-deterministic} strategy \mccorrect{$f_j^{app}\{[g_k]_{k<j}\}$} as a game (which contains all Opponent moves) depending on $[X_i]_i$, noting that for any $[\sigma_i]_i\in\str(\tot{X_1})\times\cdots\times\str(\tot{X_n})$ we have that \mccorrect{$f_j^{app}\{[g_k]_{k<j}\}\{[\sigma_i]_i\}$} defines a \mccorrect{non-deterministic} strategy on $\tot{Y_j}$ and hence a subgame of $\tot{Y_j}$. To interpret the $\Id-E$-rule, we define $f'_k$ as the strategy $f_k$ where we identify $Y_j$ with $\Id_{Y_j}$. (Hence, $[f'_k]_k$ does not ever visit $[Y_j^{(1)}]_j$ or $[Y_j^{(2)}]_j$.) Note that such $f'_k$ are well-defined strategies, as long as we are only imposing the type dependency condition for a class of maximal strategies (like winning deterministic strategies).\end{remark}

\begin{remark}[Interpretation of $\substsf$] Note that $\Gamma,x:A,x':A,p:\Id_{A}(x,x')\vdash \substsf(p,-):B\Rightarrow B[x'/x]$ gets interpreted as an initial protocol querying the identity type, followed by a simple copycat between the two copies of $\sem{B}$ after Opponent plays $\reflind$.
\end{remark}

In addition to being non-extensional (i.e. refuting the principle of equality reflection), the intensionality of these identity types can be characterised as follows.

\begin{theorem}\label{thm:idprops} Streicher's Criteria of Intensionality are \mccorrect{satisfied}, i.e.
\begin{enumerate}
\item[(I1)] there exist $\vdash A\;\type$ such that $x,y:A,z:\Id_A(x,y)\not\vdash x\equiv y:A$;
\item[(I2)] there exist $\vdash A\;\type$ and $x:A\vdash B\;\type$ such that  $x,y:A, z:\Id_A(x,y)\not\vdash B\equiv B[y/x]\;\type$;
\item[(I3)] for all $\vdash A\;\type$, $\vdash p:\Id_A(t,s)$ implies $\vdash t\equiv s:A$.
\end{enumerate}
\end{theorem}
\begin{proof}
\begin{enumerate}
\item[(I1)] Let us write \mccorrect{$\mathbf{p}_{[\BInd^{(i)}]}$} for $\mathbf{p}_{[\BInd^{(1)},\BInd^{(2)},\Id_{\BInd}],[\BInd^{(i)}]}$ and $\sem{-}$ for the interpretation functor from the syntax of {$\DTTGame$}$-$ into $\Ctxt(\DGame_!)$. (I1) relies on the interpretation of terms carrying intensionality. Take $\sem{A}:=[\BInd]$. Then, we have to show that $\mathbf{p}_{[\BInd^{(1)}]}\neq \mathbf{p}_{[\BInd^{(2)}]}\in\Tm([\BInd].[\BInd].\Id_{[\BInd]},[\BInd])$. We note that $\mathbf{p}_{[\BInd^{(1)}]}\{\langle [\bot],[\ttt],[\bot]\rangle\}=[\bot]$ while $\mathbf{p}_{[\BInd^{(2)}]}\{\langle [\bot],[\ttt],[\bot]\rangle\}=[\ttt]$, which shows that $(I1)$ holds.
\item[(I2)] This property replies on semantic types having intensional features. In this case, our source of intensionality is that dependent games contain redundant information on their value for inconsistent tuples of strategies.  For instance, take $\sem{A}:=[\BInd]$ and $\sem{B}:=(\fff\mapsto [I], \ttt\mapsto [\BInd])$. Then, we have to show that $\sem{B}\{\mathbf{p}_{[\B^{(1)}]}\} \neq \sem{B}\{\mathbf{p}_{[\BInd^{(2)}]}\}\in\mathsf{Ty}([\BInd].[\B].\Id_{[\B]})$. Now, $\sem{B}\{\mathbf{p}_{[\B^{(1)}]}\}( \fff,\ttt,\reflind) =\sem{B}(\fff)=[I]$ while $\sem{B}\{\mathbf{p}_{[\BInd^{(2)}]}\}(\fff,\ttt,\reflind)=\sem{B}(\ttt)=[\BInd]$, so we conclude that $(I2)$ holds.
\item[(I3)] Given $[\sigma_i]_i,[\tau_i]_i\in\Tm([],[X_i]_i)$ and\\
\\
\resizebox{\linewidth}{!}{
\parbox{\linewidth}{
\begin{align*}
[p_i]_i\in\Tm([],\Id_{[X_i]_i}([\sigma_i]_i,[\tau_i]_i))&:=\{[q]\;|\; q\in \str\left(\Osat\left(\begin{array}{ll}
\{\reflind\}_* & \mathsf{if}\; [\sigma_i]_i=[\tau_i]_i\\
\emptyset_* &\mathsf{else}
\end{array}\right)\right)\}\\
&\cong \left\{\begin{array}{ll}
\str(\{\reflind\}_*) & \mathsf{if}\; [\sigma_i]_i=[\tau_i]_i\\
\str(\emptyset_*) &\mathsf{else}
\end{array}\right.\\
&\cong \left\{\begin{array}{ll}
\{\reflind\} & \mathsf{if}\; [\sigma_i]_i=[\tau_i]_i\\
\emptyset &\mathsf{else}
\end{array}\right. ,
\end{align*}
}}\\
\\
it clearly follows that $[\sigma_i]_i=[\tau_i]_i$.
\end{enumerate}
\end{proof}
Similar proofs also suffice to establish (I1) and (I2) for the domain model of {$\DTTGame$}. (I3) relies on a crucial difference between the domain and games models: our identity types compare strategies in intension rather than in extension. For similar reasons, \textsf{FunExt} is seen to fail in the games model.

The principle $\mathsf{FunExt}$ of function extensionality intuitively states that, from the point of view of the $\Id$-types, functions are extensional objects: black boxes which merely send inputs to outputs without any internal temporal structure. It is refuted in our model.

\begin{theorem}\label{thm:id} $\mathsf{FunExt}$ is refuted: for $\vdash f,g:\Pi_{x:A}B$, we do not generally have $z:\Pi_{x:A}\Id_B(f(x),g(x))\vdash\mathsf{FunExt}_{f,g}:\Id_{\Pi_{x:A}B}(f,g). $
\end{theorem}
\begin{proof}For our counter example, we let $\sem{A}=\sem{B}=[\BInd]$.

Let $f$ be the usual strict strategy that outputs $\mathsf{tt}$ (and examines its argument once) and let $g$ by the strategy which always outputs $\ttt$ but first examines its input twice.  Noting that always $\sem{f}\{x\}=\sem{g}\{x\}$ for all strategies $x$ on $\BInd$, we have an inhabitant $\reflind\in\str(\Osat(\Pi_{\BInd}\{\reflind\}_*))=\Tm([\BInd],\Id_{[\BInd]}(\sem{f},\sem{g}))$. However, as not $\sem{f}=\sem{g}$, we do not have an inhabitant of $$\Tm([],[\Id_{\Pi_{[\BInd]}[\BInd]}](\sem{f},\sem{g}))=\Tm([],[\emptyset_*])=\str({\emptyset}_*)=\emptyset.$$
\end{proof}
On the other hand, it turns out that we do have the principle of uniqueness of identity proofs \textsf{UIP}, as the strict strategy which first examines the first copy of $\sem{\Id_A}$ and then the second, before replying $\reflind$ in $\sem{\Id_{\Id_{A}}}$. We choose this more complicated witness rather than a non-strict one as this will generalise to settings where we consider a broader class of strategies. This principle intuitively says that types have trivial (discrete) spatial structure, from the point of view of the $\Id$-types.
\begin{theorem}We have $x,y:A,p,q:\Id_A(x,y)\vdash \mathsf{UIP}_A: \Id_{\Id_A(x,y)}(p,q)$.
\end{theorem}
\begin{proof}Player can open a copy of $\sem{\Id_A}$ as the initial move $*$ is in each fibre. After Player opens a copy of $\sem{\Id_A}$, Opponent can only reply $\reflind$. Player can then play $\reflind$ in $\sem{\Id_{\Id_A}}$ as it is in all possible fibres.
\end{proof}

\section{Ground Types: Finite Dependent Games}\label{sec:grndtype}
In this section, we show how we can additionally give finite inductive type families an interpretation if we restrict to a full subcategory $\Ctxt(\DGame_!)^{\fin1\Sigma\Pi\Id}$ of $\Ctxt(\DGame_!)$. Indeed, we use the full subcategory $\Ctxt(\DGame_!)^{\fin1\Sigma\Pi\Id}$ on the hierarchy of context games generated by the semantic constructions  interpreting $1$-, $\Sigma$-, $\Pi$- and $\Id$-types and substitution, starting from finite dependent games (as defined below). These finite dependent games will play the r\^ole of semantic ground types to build a type hierarchy for which we prove completeness results in the next section.

We consider the interpretation of {$\DTTGame$}$-$ in $\Ctxt(\DGame_!)^{\fin1\Sigma\Pi\Id}$ and will denote the interpretation functor by $\sem{-}$.

\begin{theorem}[Finite Dependent Game] \label{thm:fin} A finite inductive type family $B:=(a_i\mapsto_i \{b_{i,j}\;|\;j\})(x)$ in context $x:A$, where $B[a_i/x]$ is generated by $\{b_{ij}\;|\; 1\leq j\leq m_i\}$, has an interpretation in $\Ctxt(\DGame_!)^{\fin1\Sigma\Pi\Id}$ as a \emph{finite dependent game}:
\begin{diagram}
\sem{B}:
\sem{a_i}& \rMapsto & {\{b_{i,j}\;|\; j\}_*} & &
\elsetext & \rMapsto & {\emptyset}_*
\end{diagram}
and
$$\tot{ \sem{B}}={\{b_{i,j}|i,j\}_*}.
$$
\end{theorem}
\begin{proof}To be explicit, we interpret the $\case^{p,q}$-constructs rather than the (equivalent) $\case$-constructs, as we shall be using the former later.

The interpretation of the I-rules is clear: $\sem{b_{i,j}}$ is the unique strategy on $\sem{B[a_i/x]}$ that replies to $*$ with the move $b_{i,j}$.

We inductively construct $\sem{\case^{p,q}_{B[a/x],C}(b,\{c_{i,j}\}_{i,j})}:\sem{\cdot}\ra{}\sem{\Pi_{A'}C[b/y]}$, with structural induction on $C$ (apart from the case of $C=1$, which is trivial). We consider the (more general) base case of arbitrary $\sem{C}$ that assign to each $\sigma\in\str(\tot{ \sem{A'}}\&\tot{ \sem{B}})$ a finite inductive game with initial move $*$. After that, the case constructs for more general $C$ are obtained from the commutative conversions \mccorrect{for $\Sigma$- and $\Pi$-types induced from those} of figure \ref{fig:stt}.  Note that substitutions and $\Id$-types are already dealt with because we have been considering the more general base case where some of the constructors of $\sem{C}$ can coincide, while \mccorrect{substitution} commutes with $\Pi$ and $\Sigma$.

Let us consider our base case. We define $\sem{\case^{p,q}_{B[a/x],C}(b,\{c_{i,j}\}_{i,j})}$ by noting that 
$$\sem{\case_{ {B^T}, {C^T}}}(\sem{b^T},\{\sem{c_{i,j}^T}(\der_{A'},\reflind,\reflind)\}_{i,j})$$
in fact defines a (winning) strategy on $\sem{\Pi_{A'}C[b/y]}$, where $(-)^T$ is the syntactic translation from section \ref{sec:DTT}.

We verify that this yields a strategy $\sem{\case^{p,q}_{B[a/x],C}(b,\{c_{i,j}\}_{i,j})}$ on $\sem{\Pi_{A'}C[b/y]}$ (which clearly automatically is winning, as usual, as we never restrict $O$-moves in games of dependent functions). Let us write $\sem{A'}=[X_i]_i$, so $\sem{\case^{p,q}_{B[a/x],C}(b,\{c_{i,j}\}_{i,j})}$ will be a strategy on $\Osat(\Pi_{X_1}\cdots\Pi_{X_n}\sem{C[b/y]})$. Let $sab\in\sem{\case^{p,q}_{B[a/x],C}(b,\{c_{i,j}\}_{i,j})}$. Then, we verify  that for all $\overline{sab\upharpoonright_{!\tot{ X_1}}}\subseteq \sigma_1'\in\str(X_1()),\ldots\,\overline{sab\upharpoonright_{!\tot{ X_n}}}\subseteq \sigma_n'\in \str(X_n(\sigma_1',\ldots,\sigma_{n-1}'))$, we have that $(*)$
$$  sab\upharpoonright_{!\tot{ X_1}}\in P_{!X_1()} \wedge \cdots \wedge sab\upharpoonright_{!\tot{ X_n}}\in P_{!X_n(\sigma_1',\ldots,\sigma_{n-1}')}\wedge sab\upharpoonright_{\tot{ \sem{C}}}\in P_{\sem{C[b/y]}(\sigma_1',\ldots,\sigma_{n}')},
$$
provided that already
$$ sa\upharpoonright_{!\tot{ X_1}}\in P_{!X_1()} \wedge \cdots \wedge sa\upharpoonright_{!\tot{ X_n}}\in P_{!X_n(\sigma_1',\ldots,\sigma_{n-1}')}\wedge sa\upharpoonright_{\tot{ \sem{C}}}\in P_{\sem{C[b/y]}(\sigma_1',\ldots,\sigma_{n}')}.$$
Because of the type $\sem{\Pi_{A'}B[a/x]}$ of $\sem{b}$, all Player moves of $\sem{\case^{p,q}_{B[a/x],C}(b,\{c_{i,j}\}_{i,j})}$ respect the type $\sem{\Pi_{A'}C[b/y]}$ at least until some $\sem{c_{i,j}}$ is called. Now, the crux is that $\sem{c_{i,j}}$ is only ever called after $\sem{b}$ has already replied with the move $b_{i,j}$. This means that for any $[\sigma'_i]_i$ we are considering, we have that $[\sigma'_i]_i;\sem{b}=\sem{b_{i,j}}$. Moreover, because of the type of $b$, we have that $\sem{b_{i,j}}=[\sigma'_i]_i;\sem{b}=\sem{b}\{[\sigma'_i]_i\}$ is a winning strategy on $\sem{B}\{\sem{a}\{[\sigma'_i]_i\}\}$, while $\sem{a}\{[\sigma'_i]_i\}$ is a winning strategy on $\sem{A}$. Therefore, because of the definition of $\sem{B}$, we conclude that $[\sigma_i']_i;\sem{a}=\sem{a_i}$, as the fibres of $B$ are disjoint. The upshot is that the semantic type  $\sem{\Pi_{x':A'}\Pi_{p_{i,j}:\Id_A(a_i,a)}\Pi_{q_{i,j}:\Id_{B[a/x]}(\substsf(p_{i,j},b_{i,j}),b)}C[b/y]}$ of $\sem{c_{i,j}}$ now gives us that the continuation of the play along  $\sem{c_{i,j}}(\der_{A'},\reflind,\reflind)$ still respects our condition $(*)$. 
\end{proof}

We have obtained the following.
\begin{corollary}\label{cor:snddttm}
{$\DTTGame$}$-$ has a sound interpretation   in $\Ctxt(\DGame_!)^{\fin1\Sigma\Pi\Id}$.
\end{corollary}
We turn to the issue of soundness of the interpretation of {$\DTTGame$} in the next section.

\section{Soundness, Faithfulness and Completeness}\label{sec:compl}
In this section, we show that the interpretation of $\DTTGame$ in $\Ctxt(\DGame_!)^{\fin1\Sigma\Pi\Id}$ is sound and faithful and, if we limit $\Id$-types to only occur strictly positively and at most once, that it is, additionally, fully complete. The proof of soundness and faithfulness follows from the fact that our game semantics for {$\DTTGame$} factors faithfully over the usual game semantics for simple type theory. The proof of definability proceeds in five steps:
\begin{enumerate}
\item interpreting a dependently typed strategy $f$ on a larger (simply typed) game;
\item for a strict $f$, performing the decomposition of \cite{abramsky2000full} in the simply typed world, as usual, to obtain simply typed strategies $g^j$ and $h_y$ that are called in the execution of $f$;
\item noting that these $g^j$ and $h_y$ can actually be assigned a more precise dependent type, the trick being that we accumulate appropriate negatively occurring $\Id$-types as the decomposition proceeds inductively;
\item observing that the iterated decomposition of strict strategies strictly decreases a positive integer norm and therefore eventually terminates after finitely many steps, producing only non-strict strategies;
\item for a non-strict $f$, noting that $f$ is directly definable using the constructors $b_{i,j}$ for finite type families and \textsf{Ty-Ext}.
\end{enumerate}

\subsection{Soundness and Faithfulness}
We first prove faithfulness of the interpretation of {$\DTTGame$} in our model.

\begin{theorem}[Soundness and Faithfulness]\label{thm:faith}
The interpretation $\sem{-}$ of {$\DTTGame$} in $\Ctxt(\DGame_!)^{\fin1\Sigma\Pi\Id}$ is sound and faithful. The rule \textsf{Ty-Ext} is sound for all type families $x:A\vdash B$ over a type $A$ for which definability holds.
\end{theorem}
\begin{proof}
We note that we have the following commutative diagram of (non-dashed) functors, where, in the light of corollary \ref{cor:snddttm}, soundness amounts to arguing that our interpretation of {$\DTTGame$}$-$ factors over {$\DTTGame$} (denoting the factorisation with the dashed functor)
\begin{diagram}{\DTTGame-} &  & &\\
& {\DTTGame} \rdOnto(1,1) \rdTo(1,3)_{(-)^T} \rdTo(3,1)^{\sem{-}} & \rDotsto_{\sem{-}} & \Ctxt(\DGame_!)^{\fin1\Sigma\Pi\Id}\\
& \dInto_{(-)^T} & & \dInto_{\tot{ -}}\\
& {\STTGame} & \rInto{\sem{-}}& \Gamecat_!^{\fin1\times\Rightarrow}.
\end{diagram}
Here, the top and bottom sides of the outer quadrangle, respectively are the interpretation functor of {$\DTTGame$}$-$ in our model, which exists according to corollary \ref{cor:snddttm}, and the usual interpretation of simple type theory with finite ground types (or, a total finitary {\PCF}, if you will) in the cartesian category of games and (winning) strategies of \cite{abramsky2000full}. Recall that the latter is (full and) faithful according to theorem \ref{thm:sttcompl}. The left side of the inner rectangle is the faithful (non-full) functor defined in section \ref{sec:trans}. Note that faithfulness of the interpretation of {$\DTTGame$} automatically follows from the faithfulness of these two functors, if we can prove soundness. Finally, the right side of either quadrangle is the semantic equivalent of this syntactic translation, which we define next.

We have an inductively defined translation $\Ctxt(\DGame_!)\ra{\tot{-}}\Gamecat_!$:
\begin{align*}
\tot{ [A_i]_{1\leq i\leq m}}&:=\bigwith_{1\leq i\leq m} \tot{ A_i}\\
\tot{[]}&:= I.
\end{align*}
Note that this also satisfies
\begin{mccorrection}
\begin{align*}
\tot{\Pi_{A_1}\cdots\Pi_{A_n}B}&=\tot{ A_1}\Rightarrow\cdots\Rightarrow \tot{ A_n}\Rightarrow\tot{ B}\\
\tot{\Id_C}&=\{\reflind\}_*.
\end{align*}
\end{mccorrection}

$\smiley$ automatically extends to a faithful (non-full) functor by interpreting the winning dependently typed strategies on $A$ as simply typed strategies on $\tot{ A}$, which are obviously also winning as we never restrict Opponent moves in our games of dependent functions. Faithfulness of this functor together with commutativity of the outer quadrangle gives us that the dashed arrow is a (unique) well-defined functor, i.e. we have a sound interpretation of {$\DTTGame$} in $\Ctxt(\DGame_!)^{\fin1\Sigma\Pi\Id}$.

We also note that \textsf{Ty-Ext} has a sound interpretation in our model for types $B$ depending on a type $A=[A_i]_i$ for which all morphisms are definable. Indeed, it follows that $\Pi_{A}B_1 = \Pi_{A}B_2$ if $\tot{B_1}=\tot{B_2}$ and for all $[]\ra{t} [A_i]_{1\leq i\leq n}$ we have that $B_1(t)=B_2(t)$. The reason is that the definition of $\Pi$-games only relies on $\tot{B}$ and the evaluation of $B_i$ on these consistent tuples $t$. If all such $t$ are definable, \textsf{Ty-Ext} follows.
\end{proof}

\subsection{Full Completeness}
Next, we first prove two technical lemmas, which encompass steps 2. and 3. and, respectively, step 4. in the definability proof. We use the notation $[\Id]_{[A_i]_i}([a_i]_i,[a_i']_i):=[\Id_{A_i}(a_i,a'_i)]_i$.

\begin{lemma}[Decomposition] \label{lem:decomp} Let us suppose we have a context game $[A_i]_{i\leq n}$ in $\Ctxt(\DGame_!)^{\fin1\Sigma\Pi\Id}$ with $A_i=\Pi_{B^{i,1}}\ldots\Pi_{B^{i,q_i}}{Y^i_*}=\Pi_{[{B^{i,j}}]_{j}}{Y^i_*}\{c^i\}$ where ${Y^i_*}$ is a finite inductive dependent game depending on the context game $[C_l^i]_l$ and $[A_k]_{k<i}.[B^{i,j}]_j\ra{c^i}[C_l^i]_l$. Let us say ${Y^i_*}$ has constructors $y$ in fibre ${Y^i_*}\{[c_y^i]_i\}$.\\
\\
Then, it follows that, when given a strategy $f$ that does not visit $[\Id]_{[D_k]_k}$, $$f\in\str(\Osat(\Pi_{[A_i]_i}\Pi_{[\Id]_{[D_k]_k}([d_k^0]_k,[d_k]_k)}{X_*})),$$
with $\str(\tot{A_1}\&\cdots\&\tot{A_n})\ra{{X}}\mathcal{P}(\tot{X})$ a function, where $\tot{X}$ is some finite set, and context morphisms $[d_k]_k,[d_k^0]_k : [A_i]_i\ra{}[D_k]_k$, we can decompose it (uniquely) as follows:
\begin{itemize}
\item if $f$ is non-strict, then $f=[A_i]_i\ra{}[]\ra{x}[\tot{ {X_*}}]$ for some $x\in \bigcup\funim( X)$ such that $x\in {X_*}([\tau_i]_i)$ for all $[]\ra{[\tau_i]_i}[A_i]_i$ such that $[d_k]_k\{[\tau_i]_i\}=[d_k^0]_k\{[\tau_i]_i\}$;
\item if $f$ is strict, then $f=\mathbf{C}_i(g^1,\ldots,g^{q_i},(h_y\;|\; y\in \bigcup \funim(Y^i)))$ where $\mathbf{C}_i$ embodies a $\case$-construct that we shall define in the proof,
\end{itemize}
where $$g^j\in\str(\Osat(\Pi_{[A_i]_i}\Pi_{[\Id]_{[D^k]_k}([d_k^0]_k,[d_k]_k)}B^{i,j}\{\langle [\der_{A_l}]_{l<i},[g^{j'}]_{j'<j}\rangle\}))$$ and $$h_y\in\str(\Osat(\Pi_{[A_i]_i}\Pi_{[\Id]_{[D^k]_k.[C_l^i]_l.[{Y^i_*}]}(\langle[d_k^0]_k,[c_y^i],[y]\rangle,\langle[d_k]_k,[\widetilde  c^i ],[\phi]\rangle)}{X_*})),$$ where $\widetilde c^i:=\langle [\der_{A_{i'}}]_{i'<i},[g_j]_j\rangle;c^i$ and $\phi:=\lambda_{[\tau_k]_k}\tau_i\{[g^j]_j\{[\tau_k]_k\}\}$ (and we write $\funim(Y^i)$ for the image of $Y^i$ and $\lambda_{[\tau_k]_k}$ for the obvious semantic $\lambda$-abstraction). Here, neither $g^j$ nor $h_y$ visits the $\Id$-type.
\end{lemma}

\begin{proof}
Note that we can consider $\tot{ f}$ as a strategy on $\tot{  A_1}\Rightarrow\cdots\Rightarrow \tot{A_n}\Rightarrow\tot{   {X_*}}$ as $f$ does not visit the $\Id$-type. The decomposition lemma \cite{abramsky2000full,Abramsky00axiomsfor} for the game semantics of (finitary) \PCF{}\mccorrect{ }now gives us three cases:
\begin{itemize}
\item $\tot{ f}=\bot$
\item $\tot{ f}=\bigwith_i \tot{ {A_i}}{}\ra{}I\ra{x}\tot{ {X_*}}$ for some $x\in \bigcup \funim(X)$;
\item $\tot{ f}= \mathbf{C}_i'(g'^1,\ldots,g'^{q_i},(h_y'\;|\; y\in \bigcup \funim(Y^i))),
$ for a (unique) $1\leq i\leq n$ and (unique) $g'^j \in\str( \tot{A_1}\Rightarrow\cdots\Rightarrow\tot{A_n}\Rightarrow \tot{ B^{i,j}})$ and $h_y '\in\str(\tot{A_1}\Rightarrow \cdots \Rightarrow\tot{A_n}\Rightarrow\tot{ {X_*}{}})$, where (writing $\pi^i$ for the derelicted projection to the $i$-th component, $\mathsf{ev}$ for the obvious evaluation morphism, and denoting the semantic case construct with $\sem{\case}$)
$$\mathbf{C}_i'(g'^1 ,\ldots,g'^{q_i} ,(h_y'\;|\; y\in \bigcup \funim(Y^i))):=$$
\resizebox{.95\textwidth}{!}{
\begin{diagram}
   &     &!\bigwith_i \tot{A_i}      &           &         &     \rTo^{\id_{!\bigwith_i \tot{A_i}}}  &                      & &!\bigwith_i{\tot{ A_i}} & &\\
!\bigwith_i{\tot{ A_i}} &\rTo^{\mathsf{diag}_{\bigwith_i{\tot{ A_i}}}^\dagger} &\bigotimes &         &!\bigwith_i{\tot{ A_i}}       &     \rTo^{\langle g'^1,\ldots,g'^{q_i}\rangle^\dagger}  &!\bigwith_{j} \tot{  B^{i,j}} &    & \bigotimes &\rTo^{\sem{\case}_{\tot{ {Y^i_*}},\tot{ {X_*}}}(-,[{h_y'}]_y)} &\tot{ {X_*}}.\\
   &     &!\bigwith_i{\tot{ A_i}}      &\rTo^{\mathsf{diag}_{\bigwith_i{\tot{ A_i}}}^\dagger}       & \bigotimes &    &          \bigotimes  & \rTo^{\mathsf{ev}}& \tot{ {Y^i_*}} & &\\
   &     &        &           &!\bigwith_i{\tot{ A_i}}       &      \rTo^{\pi^i} &\tot{ A_i}                   & & & &
\end{diagram}}
\end{itemize}
Note that the first case cannot occur as $f$ is winning.\\
\\
For the second case, due to the restriction on $P$-moves in $\Pi$-games and the interpretation of $\Id$-types, a non-strict $f$ needs to respond to $*$ with a move in $\bigcup\funim( X)$ such that $x\in {X_*}([\tau_i]_i)$ for all $[]\ra{[\tau_i]_i}[A_i]_i$ such that $[d_k]_k\{[\tau_i]_i\}=[d_k^0]_k\{[\tau_i]_i\}$.\\
\\
For the third case, note the following.
\begin{itemize}
\item $g'^j=\tot{ g^j}$ for (unique) $$g^j\in \str(\Osat(\Pi_{[{A}_i]_i}\Pi_{[\Id]_{[D^k]_k}([d_k^0]_k,[d_k]_k)}B^{i,j}\{\langle [\der_{A_l}]_{l<i},[g^{j'}]_{j'<j}\rangle\})).$$ This will follow once we show that $$((g^j)^\dagger)^\dagger\in \str(\Osat(\Pi_{[{A}_i]_i}\Pi_{[\Id]_{[D^k]_k}([d_k^0]_k,[d_k]_k)}!!B^{i,j}\{\langle [\der_{A_l}]_{l<i},[g^{j'}]_{j'<j}\rangle\})),.$$ The argument will proceed by complete induction on $j$. Assume the claim holds for $g^k$ with $k< j$. We  show it also holds for $g^j$.

We need to show that for $s^{j}=s'ab\in ((g^j)^\dagger)^\dagger$, for any $\overline{s'ab\upharpoonright_{!\tot{ A_1}}}\subseteq\tau_1\in\str(A_1()),\ldots,\overline{s'ab\upharpoonright_{!\tot{ A_n}}}\subseteq\tau_n\in\str(A_n(\tau_1,\ldots,\tau_{n-1}))$ s.t. $[
\tau_i]_i;[d^0_k]_k=[\tau_i]_i;[d_k]$, $s'a\in P_{A_1()\Rightarrow \cdots\Rightarrow A_n(\tau_1,\ldots,\tau_{n-1})\Rightarrow !!B^{i,j}(\tau_1,\ldots,\tau_{i-1},\langle \tau_1,\ldots,\tau_n\rangle;g^1,\ldots,\langle \tau_1,\ldots,\tau_n\rangle;g^{j-1})}$ implies that $s'ab\in P_{A_1()\Rightarrow \cdots\Rightarrow A_n(\tau_1,\ldots,\tau_{n-1})\Rightarrow !!B^{i,j}(\tau_1,\ldots,\tau_{i-1},\langle \tau_1,\ldots,\tau_n\rangle;g^1,\ldots,\langle \tau_1,\ldots,\tau_n\rangle;g^{j-1})}$.

Let us assume that the hypothesis of this implication is true. Now, note that $s^j\in ((g^{j})^\dagger)^\dagger$ extends to $tab=*_{{X_*}}(0,*)_{!{Y^i_*}}s^1\cdots s^{j-1}s^j\in f$ for any $s^k\in ((g^k)^\dagger)^\dagger$, for $1\leq k\leq j-1$. We can choose $s^k\in \langle \tau_1,\ldots,\tau_n\rangle||((g^k)^\dagger)^\dagger $ such that $\bigcup\overline{\overline{s^k\upharpoonright_{!!\tot{ B^{i,k}}}}}=\langle \tau_1,\ldots,\tau_n\rangle;g^k$. (We write $\overline{\overline{s}}$ to indicate we apply $\overline{(-)}$ first to $s$ and then again to each member of the resulting set of plays.) Note that we can do this as $\langle \tau_1,\ldots,\tau_n\rangle;g^k $ is finite as a partial function on moves. In fact, $$s^k\in P_{A_1()\Rightarrow \cdots\Rightarrow A_n(\tau_1,\ldots,\tau_{n-1})\Rightarrow !! B^{i,k}(\tau_1,\ldots,\tau_{i-1},\langle \tau_1,\ldots,\tau_n\rangle;g^1,\ldots,\langle \tau_1,\ldots,\tau_n\rangle;g^{k-1})},$$ as a consequence of our induction hypothesis.

Then, as $tab\in f$ is a play in $$\Osat(\Pi_{A_1}\cdots\Pi_{A_{i-1}}\Pi_{(\Pi_{B^{i,1}}\cdots \Pi_{B^{i,q_i}}{Y^i_*})}\Pi_{A_{i+1}}\cdots\Pi_{A_n}\Pi_{[\Id]_{[D^k]_k}([d_k^0]_k,[d_k]_k)}{X_*}),$$ we have that for all $$\overline{tab\upharpoonright_{!\tot{ A_1}}}\subseteq\tau_1'\in\str(A_1()),\ldots,\overline{tab\upharpoonright_{!\tot{ A_n}}}\subseteq\tau_n'\in\str(A_n(\tau_1',\ldots,\tau_{n-1}'))$$ s.t. $[
\tau_i']_i;[d^0_k]_k=[\tau_i']_i;[d_k]$, $ta\in P_{A_1()\Rightarrow \cdots\Rightarrow A_n(\tau_1',\ldots,\tau_{n-1}')\Rightarrow {X}_*(\tau_1',\ldots,\tau_n')}$ implies that also $tab\in P_{A_1()\Rightarrow \cdots\Rightarrow A_n(\tau_1',\ldots,\tau_{n-1}')\Rightarrow {X}_*(\tau_1',\ldots,\tau_n')}$. Note that by construction of $tab$, $[\tau_i]_i$ is one such $[\tau'_i]_i$ and is in fact the only one we are interested in, so we simply write $[\tau_i]_i$ for both. Note that the hypothesis of the implication under consideration actually holds by our assumptions about $s'a$ and $s^k$.  Therefore, its conclusion $tab\in P_{A_1()\Rightarrow \cdots\Rightarrow A_n(\tau_1,\ldots,\tau_{n-1})\Rightarrow {X}_*(\tau_1,\ldots,\tau_n)}$ follows.

Now, it follows immediately from the restriction on plays $tab$ in
$$\Osat(\Pi_{A_1}\cdots\Pi_{A_{i-1}}\Pi_{(\Pi_{B^{i,1}}\cdots \Pi_{B^{i,q_i}}{Y^i_*})}\Pi_{A_{i+1}}\cdots\Pi_{A_n}\Pi_{[\Id]_{[D^k]_k}([d_k^0]_k,[d_k]_k)}{X_*})$$
that if $((g^j)^\dagger)^\dagger$ makes the move $b$ in $!\tot{ A_l}$, then it also satisfies the rules of $\Osat(\Pi_{[{A}_i]_i}\Pi_{\Id_{[D^k]_k}([d_k^0]_k,[d_k]_k)}!!B^{i,j}\{\langle [\der_{A_l}]_{l<i},[g^{j'}]_{j'<j}\rangle\})$. The interesting case is when $b$ is a move in $!!B^{i,j}$. Let us presume that Opponent has not been naughty. \mccorrect{(Otherwise, anything goes.)} To deal with this case, we note that the restriction on plays $tab$ in $$\Osat(\Pi_{A_1}\cdots\Pi_{A_{i-1}}\Pi_{(\Pi_{B^{i,1}}\cdots \Pi_{B^{i,q_i}}{Y^i_*})}\Pi_{A_{i+1}}\cdots\Pi_{A_n}\Pi_{[\Id]_{[D^k]_k}([d_k^0]_k,[d_k]_k)}{X_*})$$ combined with the definition of $(\Pi_{B^{i,1}}\cdots\Pi_{B^{i,q_i}}{Y^j}_*)(\tau_1,\ldots,\tau_{i-1})$ \mccorrect{gives us} that there exist $\bigcup\overline{\overline{tab\upharpoonright_{!!\tot{ B^{i,1}}}}}\subseteq\sigma^1\in\str(B^{i,1}(\tau_1,\ldots,\tau_{i-1})),\ldots,\bigcup \overline{\overline{tab^\upharpoonright_{!!\tot{ B^{i,q_i}}}}}\subseteq\sigma^{q_i}\in\str(B^{i,q_i}(\tau_1,\ldots,\tau_{i-1},\sigma^1,\ldots,\sigma^{q_i-1}))$, such that $$tab\upharpoonright_{!\tot{ A_i}}\in P_{ !(B^{i,1}(\tau_1,\ldots,\tau_{i-1})\Rightarrow\cdots\Rightarrow B^{i,q_i}(\tau_1,\ldots,\tau_{i-1},\sigma^1,\ldots,\sigma^{q_i-1})\Rightarrow {Y^i}_*(\tau_1,\ldots,\tau_{i-1},\sigma^1,\ldots,\sigma^{q_i}) )}$$ so, in particular, $tab\upharpoonright_{!\tot{ A_i}}\upharpoonright_{!!\tot{ B^{i,{j}}}}\in P_{ !!B^{i,j}(\tau_1,\ldots,\tau_{i-1},\sigma^1,\ldots,\sigma^{j-1})}$.

To complete the argument, we note that by construction of $t$, we have that $\bigcup\overline{\overline{tab\upharpoonright_{!!\tot{ B^{i,k}}}}}=  \langle \tau_1,\ldots,\tau_n\rangle ;g^k$, for $1\leq k\leq j-1$. We conclude that $s'ab\upharpoonright_{!!B^{i,j}}\in  P_{!!B^{i,j}(\tau_1,\ldots,\tau_{i-1},\langle \tau_1,\ldots,\tau_n\rangle;g^1,\ldots,\langle \tau_1,\ldots,\tau_n\rangle;g^{j-1})}$.

\item $h'_y=\tot{ h_y}$ for (unique) $$h_y\in\str(\Osat(\Pi_{{A}_1}\cdots\Pi_{A_n}\Pi_{[\Id]_{[D^k]_k.[C_l^i]_l.[{Y^i_*}]}(\langle[d_k^0]_k,[c_y^i],[y]\rangle,\langle[d_k]_k,[\widetilde  c^i],[\phi]\rangle)}{X}_*)).$$ Indeed, note that $*s\in h_y$ iff $*(0,*)t(0,y)s\in f$ for some $*ty\in \phi\in \str(\Osat(\Pi_{{A}_1}\cdots\Pi_{A_n}\Pi_{[\Id]_{[D^k]_k}([d_k^0],[d_k]_k)}{Y^i_*}\{\langle [\der_{A_l}]_{l<i},[g^{j}]_{j}\rangle\}))$. It then follows that $*s\in \Osat(\Pi_{A_1}\cdots\Pi_{A_n}\Pi_{[\Id]_{[D^k]_k.[C_l^i]_l.[{Y^i_*}]}(\langle[d_k^0]_k,[c_y^i],[y]\rangle,\langle[d_k]_k,[\widetilde c^i],[\phi]\rangle)}{X}_*)$ by the following observation. Observing that $\phi\{[\overline{t\upharpoonright_{!A_i}}]_i\}=y$, note that, for winning $[\tau_i]_i\geq [\overline{*s\upharpoonright_{!A_i}}]_i$, we also have that $[\tau_i]_i\geq [\overline{*(0,*)t(0,y)s\upharpoonright_{!A_i}}]_i$ for some $*ty\in\phi$ iff $\phi\{[\tau_i]_i\}\geq y$ i.e. $\phi\{[\tau_i]_i\}=y$ as $y$ is a maximal strategy on $\tot{ {Y_*}}$. (Indeed, we can take $*ty\in \langle \tau_1,\ldots,\tau_n\rangle||\phi$.) It automatically then follows that also $\widetilde{c^i}\{[\tau_i]_i\}=c_y^i$\mccorrect{, by the type of $\langle [\widetilde{c^i}],[\phi]\rangle$ and the disjointness of fibres}.

\item We can now note that $f=\mathbf{C}_i(g^1,\ldots,g^{q_i},(h_y\;|\; y\in \bigcup \funim(Y^i)))$, where $\mathbf{C}_i$ is defined exactly as $\mathbf{C}_i'$ but using instead the dependently typed substitution and the dependently typed construct $\sem{\case^{p,q}}_{{Y^i_*}\{\langle [\der_{A_l}]_{l<i},[g^j]_{j\leq q_i}\rangle\},{X_*}}$. (That is, $\mathbf{C}_i$ and $\mathbf{C'}_i$ are the same, except for typing.) Note that $\mathsf{ev}(\pi^i,\langle g^1,\ldots,g^{q_i}\rangle)$ and $h_y$ feed into this case construct.

\item Finally, to see that $g^j$ and $h_y$ define winning strategies, we note that their infinite plays are Player-wins as they arise as labelled subtrees of $f$ which is winning. We need to verify that they are total. This also follows immediately from the totality of $f$ together with the fact that Opponent moves are, by definition, not restricted in ($O$-saturated) games of dependent functions. Indeed, if $s\in g^j$ and $sa$ is a valid extension of the play, then $*(0,*)sa$ is a valid extension of $*(0,*)s\in f$ to which $f$ hence $g^j$ has a response $b$. Similarly, if $*s\in h_y$ and $*sa$ is a valid extension of the play, then \mccorrect{there exists a suitable $t$ such that} $*(0,*)t(0,y)sa$ is a valid extension of $*(0,*)t(0,y)s\in f$ to which $f$ has a response $b$, being a total strategy. Therefore, $*sab\in h_y$.
\end{itemize}

\end{proof}

\begin{lemma}[Norm for dependent strategies] \label{lem:norm} Let $[A_i]_i\ra{[d_k]}[D_k]_k$ and ${X_*}$ as in the previous lemma. Let us write $E:=\Pi_{[A_i]_i}\Pi_{[\Id]_{[D^k]_k}([d_k^0],[d_k]_k)}{X_*}$. Then, we have a norm $||-||_E:\str(\Osat(E))\ra{}\mathbb{N}$ (we sometimes leave out the subscript $E$) for any such $E$  such that $f=\mathbf{C}_i(g^1,\ldots,g^{q_i},(h_y\;|\; y\in\bigcup\funim(Y^i)))$ implies that
$$||g^i||,||h_y||< ||f||.$$
\end{lemma}
\begin{proof}
We define a norm $||{-}||_{\tot{E}}:\str(\tot{E})\ra{}\mathbb{N}$ for games $\tot{E}$ of the $I\&\Rightarrow$-hierarchy over finite flat games and extend this to a norm on $\str(\Osat(E))$ by precomposition with the injection $\str(\Osat(E))\ra{\tot{-}}\str(\tot{E})$. The idea behind this norm is that winning strategies on games of the $I\&\Rightarrow$-hierarchy over finite flat games are finite objects in the sense that they only contain finitely many finite plays if we do not allow Opponent to open multiple threads of the same game -- remember that infinite plays in winning strategies are always due to Opponent opening an infinite number of threads of the same game.

Inductively, if $T$ is a type of {$\STTGame$} (i.e. formed from finite ground types $G$ by the grammar $T\; ::=\; G\;|\; \top\;|\; \&\;|\; \Rightarrow$), we define a type $LT$ of intuitionistic linear logic over finite types (i.e. formed from finite ground types $G$ by the grammar $LT \;::=\; G\;|\; !LT\;|\; LT \multimap LT \;|\; LT\otimes LT \;|\; LT\& LT\; |\; I \;|\; \top$, where we note that in our interpretation $\sem{\top}=\sem{I}$ and where we identify the cartesian type $A\Rightarrow B$ with the linear type $!A\multimap B$) by removing each positive occurrence of $!$ in $T$ or, equivalently, replacing each even-depth occurrence of $\Rightarrow$ with $\multimap$. Essentially, $\sem{LT}$ is obtained from the game $\sem{T}$ by not allowing Opponent to open more than one thread of any game. Note that we have a canonical winning strategy representing a generalised dereliction $\sem{T}\ra{\mathsf{gder}_{\sem{LT}}}\sem{LT}$ which is defined in the obvious way from dereliction maps on subtypes using the functoriality of $\top$, $\&$, $\Rightarrow$ and $\multimap$.

Now, if we can show that $W_{\sem{LT}}=\emptyset$, it follows that the norm $||\sigma||_{\sem{T}}:=\Sigma_{s\in \sigma;\mathsf{gder}_{\sem{LT}}/\approx_{\sem{LT}}}\mathsf{length}(s)$ is well-defined for $\sigma\in\str(\sem{T})$. (Here, we mean some skeleton for $\sigma;\mathsf{gder}_{\sem{LT}}$ when we write $\sigma;\mathsf{gder}_{\sem{LT}}/\approx_{\sem{LT}}$, ) Indeed, there are only finitely many Opponents for $\sem{LT}$ as Opponent can only make a choice between finitely many alternatives for each connective in formula $LT$, of which there are finitely many. Moreover, interactions with Player never become unboundedly long because $W_{\sem{LT}}=\emptyset$.

We show that $W_{\sem{LT}}=\emptyset$. Define classes of formulas $\allwin,\nowin$ by mutual induction as follows.  In this definition, we use $G$ to stand for any game all of whose maximal positions are of length 2, $F^A$ (respectively, $F^N$) and their subscripted versions to range over $\allwin$ (respectively, $\nowin$) games.  
\begin{align*}
\allwin  ::  & \ \ \ \  G \mid F^A_1 \lltensor F^A_2 \mid \  F^A_1 \llwith F^A_2 \mid !F^A \mid F^N \linimpl F^A \\
\nowin  :: & \ \ \ \  G  \mid F^N_1 \lltensor F^N_2 \mid \  F^N_1 \llwith F^N_2  \mid F^A \linimpl F^N.
\end{align*}
It follows from a simple inductive argument that 
\begin{itemize}
\item for all  $\allwin$ games $F^A$, $\win{F^A} = \Inf{F^A}$;
\item for all $\nowin$ formulas $F^N$, $\win{F^N} = \emptyset$.
\end{itemize}
Now, to conclude that $W_{\sem{LT}}=\emptyset$, we observe that $LT\in \nowin$, as all occurrences of $!$ are negative.

Finally, if $f=\mathbf{C}_i(g^1,\ldots,g^{q_i},(h_y\;|\; y\in\bigcup\funim(Y^i)))$, then, plays of $g^j$ and $h_y$ properly extend to plays of $f$ as discussed in the previous proof. Therefore, it follows that $||g^i||,||h_y||<||f||$.
\end{proof}

Now, we combine steps 1.-4. to reduce the definability of strict strategies to that of non-strict ones.

\begin{lemma}[Defining Strict Strategies from Non-Strict Ones] \label{lem:define} All morphisms in $\Ctxt(\DGame_!)^{\fin1\Sigma\Pi}$ are definable in {$\DTTGame$} if we assume that the non-strict ones are, where we write $\Ctxt(\DGame_!)^{\fin1\Sigma\Pi}$ for the full subcategory of $\Ctxt(\DGame_!)$ on the objects formed by the interpretation of types of $\DTTGame$ formed without $\Id$-constructors.
\end{lemma}
\begin{proof} Let $T$ be a type of {$\DTTGame$} with $\Pi$, $\Sigma$, $1$ and finite inductive type families and let $f\in\Ctxt(\DGame_!)([],\sem{T})$. If $T=\Sigma_{x_1:T^1}\ldots\Sigma_{x_{n-1}:T^{n-1}}T^n$ (including the case of $T=1$ if $n=0$), then, we know that both in the syntax and semantics $f$ decomposes as $\langle f_1,\ldots, f_n\rangle$. The interesting remaining case to deal with therefore is definability for $T=\Pi_{x:T'}S[q/x']$ where $x':Q\vdash S\;\type$ and $x:T'\vdash q:Q$, i.e. for ${T}=\Pi_{{x:T'}}{S[q/x']}$ where $S$ is a finite inductive type family. (In that case $\sem{S[q/x']}={X_*}$ has finite inductive games as fibres.)

From here, the argument to show that $f\in \Ctxt(\DGame_!)([],\sem{T})=\linebreak\str(\Osat(\sem{T}))$ is definable in {$\DTTGame$} will proceed by complete induction on $|| f||$, which  terminates according to lemma \ref{lem:norm}. For the sake of our inductive argument, let us consider the more general case of $f\in\str(\Osat(\sem{\Pi_{x:T'}\Pi_{[\Id]_{D}(d,d^0)}S[q/x']}))$ which does not visit the $\Id$-type.  Note that we may assume WLOG that $T'=\Sigma_{T^1}\ldots \Sigma_{T^{n-1}}T^n$ with ${T^i}=\Pi_{{T'^1}}\cdots\Pi_{{T'^{q_i}}}{U[v/x'']}$, where $x'':V\vdash U\type$ and $x_1:T^1,\ldots,x_{i-1}:T^{i-1},x_1':T^1{}',\ldots,x_{q_i}':T^{q_i} {}'\vdash v:V$ and where $\sem{U[v/x'']}={Y^i_*}$. This is where we invoke  lemma \ref{lem:decomp}.

If $f$ is strict, then $f$ can be expressed as\\
\\
\resizebox{\linewidth}{!}{$\begin{array}{c}\mathbf{C}_i(g^1,\ldots,g^{q_i},(h_y\;|\; y\in \bigcup\funim(Y^i)))\vspace{5pt}=\\
\sem{\lambda_{x:T'}\case^{p,q}_{U[v/x''][\mathsf{fst}(x)/x_1,\ldots,{(i-1)}-\mathsf{th}(x)/x_{i-1},G^1x/x_1',\ldots,G^{q_i}x/x_{q_i}'],S[q/x']}(x_i(G^1x)\cdots (G^{q_i}x),\{H_y x\}_y) },
\end{array}$}
\\\\ where by the induction hypothesis $g^i=\sem{G^i}$ and $h_y=\sem{H_y}$.

If $f$ is non-strict, it is definable by assumption.

We conclude that $f$ is definable in {$\DTTGame$}.
\end{proof}

Next, we complete the definability proof by showing how to define non-strict strategies from the syntax of $\DTTGame$ using the extensionality of types.
\begin{theorem}[Full Completeness at $\Id$-free type hierarchy] \label{thm:define} All morphisms in $\Ctxt(\DGame_!)^{\fin1\Sigma\Pi}$ are definable in $\DTTGame$, where we write $\Ctxt(\DGame_!)^{\fin1\Sigma\Pi}$ for the full subcategory of $\Ctxt(\DGame_!)$ on the objects formed by the interpretation of types of $\DTTGame$ formed without $\Id$-constructors.
\end{theorem}
\begin{proof}To show definability, by lemma \ref{lem:define}, all that remains to be done is demonstrate definability for non-strict $f$.

If $f$ is non-strict, we know from lemma \ref{lem:decomp} that $f$ answers with some move $a$ s.t. for all $\vdash t:T'$ and $\vdash  \overrightarrow{k}: \overrightarrow{\Id}_{ \overrightarrow{D}[t/x']}( \overrightarrow{d^0}[t/x'], \overrightarrow{d}[t/x'])$, $[]\ra{a}\sem{S[q[t/x]/x']}$. (Where, to keep notation light, we write $ \overrightarrow{D}$ for the list of types $D^1,\ldots, D^m$ and similarly for terms.) Now, as $S$ is a finite inductive type family, we know that $a=\sem{s_0}$ for some $\vdash s_0:S[q_0/x']$ where we write $q_0:=q[t/x]$ (noting that $q[t/x]$ is independent of $t$ as the fibres of $S$ are disjoint and the interpretations of constructors of finite inductive types is faithful).

Now, in particular, $S[q[t/x]/x']=S[q_0/x']$. Moreover, clearly, $S[q[t/x]/x']^T=S^T=S[q_0/x']^T$. Therefore, by \textsf{Ty-Ext} (which, by theorem \ref{thm:faith}, we know to hold for $S[q/x']$ and $S[q_0/x']$, as, by induction\footnote{Indeed, definability is only non-trivial for function types constructors. By induction, we have already established definability for $T'$. It then follows trivially for $\Sigma_{x':T'}\overrightarrow{\Id}_{ \overrightarrow{D}}( \overrightarrow{d^0}, \overrightarrow{d})$, as all closed witnesses of $\Sigma$- and $\Id$-types are canonical, both in syntax and semantics.}, we may assume that we have already established definability for $\Sigma_{x':T'}\overrightarrow{\Id}_{ \overrightarrow{D}}( \overrightarrow{d^0}, \overrightarrow{d})$), it follows that $\vdash \Pi_{x':T'}\Pi_{  \overrightarrow{p}:\overrightarrow{\Id}_{ \overrightarrow{D}}( \overrightarrow{d^0}, \overrightarrow{d})}S[q/x'] =\Pi_{x':T'}\Pi_{  \overrightarrow{p}:\Id_{ \overrightarrow{D}}( \overrightarrow{d^0}, \overrightarrow{d})}S[q_0/x'] $. Therefore, by \textsf{Ty-Conv}, we have $x':T',   \overrightarrow{p}:\Id_{ \overrightarrow{D}}( \overrightarrow{d^0}, \overrightarrow{d})\vdash s_0:S[q/x'] $ which is interpreted as $f$.
\end{proof}

We have obtained the following, combining theorems \ref{thm:faith} and \ref{thm:define}.
\begin{corollary}[Full and Faithful Completeness at $\Id$-free type hierarchy] All morphisms in $\Ctxt(\DGame_!)^{\fin1\Sigma\Pi}$ are faithfully definable in {$\DTTGame$}.
\end{corollary}

Next, we show that full (and faithful) completeness still holds if we allow one strictly positive occurrence\footnote{Recall that we say that a subformula $B$ occurs strictly positively in a type $A$ if it does not appear as the antecedent of any function types. In particular, in the case of {$\DTTGame$}, we say that $B$ occurs strictly positively in $A$ if it does not occur as the left hand side argument of a $\Pi$-type constructor.} of an $\Id$-type. This shows, in particular, that the notion of propositional identity coincides in syntax and semantics for open terms of the $\Id$-free type hierarchy.

\begin{theorem}[Full and Faithful Completeness for strictly positive $\Id$-types]\label{thm:complid} All morphisms in $\Ctxt(\DGame_!)([],\sem{\Pi_{x:A}\Id_B(f,g)})$ for $x:A\vdash f,g:B$ are faithfully definable in {$\DTTGame$}, if $\vdash A\;\type$ and $x:A\vdash B\;\type$ are types built without $\Id$-constructors.
\end{theorem}
\begin{proof}Faithfulness has already been argued in theorem \ref{thm:faith}.

Given an inhabitant ${p}$ of $\Ctxt(\DGame_!)(\sem{1},\sem{\Pi_{x:A}\Id_{B}(f,g)})$, for any $\cdot\vdash a:A$, evaluating $p$ at $\sem{a}$ gives ${p}\{\sem{a}\}\in \Ctxt(\DGame_!)(\sem{1},\sem{\Id_{B}(f[a/x],g[a/x])})$, which by (I3) of theorem \ref{thm:idprops} implies that $\sem{f[a/x]}=\sem{f}\{\sem{a}\}=\sem{g}\{\sem{a}\}=\sem{g[a/x]}$. Seeing that our model is faithful at the $1\Sigma\Pi$-hierarchy over finite inductive type families, we conclude that $\cdot \vdash f[a/x]\equiv g[a/x]:B$ for all $\cdot\vdash a:A$.

Noting that $\vdash \Id_B(f,g)^T=\{\reflind\}=\Id_B(f,f)^T$, we conclude, by \textsf{Ty-Ext}, that $\vdash \Pi_{x:A}\Id_B(f,g)=\Pi_{x:A}\Id_B(f,f)$. \textsf{Ty-Conv} now reduces full completeness for the former type to full completeness for the latter. We further note that we have an isomorphism of types $\vdash \Pi_{x:A}\Id(f,f)\cong \Pi_{x:A}\{\reflind\}$. For this last type, full completeness has already been established in theorem \ref{thm:define}. Our claim therefore follows.\end{proof}

\begin{remark}We would like to point out to the reader the phenomenon that (full) completeness at types involving positively occurring $\Id$-type constructors crucially relies on faithfulness of the model. This is illustrated here for the case of one strictly positively occurring $\Id$-type.  The question rises what the status is of full completeness results for general types of {$\DTTGame$}, in which $\Id$-types are also allowed to occur negatively and (non-strictly) positively. It seems very believable that our completeness proof could be adapted to this setting as well as the more general setting of arbitrary inductive families \cite{dybjer1991inductive,dybjer1994inductive} of which $\Id$-types are a special case. Indeed, the idea will just be to perform the decomposition in the simply typed translation, accumulating $\Id$-types as we progress in the decomposition, after which definability of non-strict strategies should follow by using constructors for inductive families together with \textsf{Ty-Ext}. We note, in particular, that $\lbi{p}{\refl{x}}{d}$ is entirely analogous to a simple $\mathsf{case}$ construct and $\refl{x}$ simply behaves as a constructor $\reflind$ for an inductive type. 
\end{remark}

\section{Dependent Games for Effects}\label{sec:depgameeff}
The whole previous development was carefully set up to be robust under failure of any combination of the conditions on strategies of being winning, history-free, well-bracketed or deterministic. All stated definitions would stay the same, where one should just interpret the word strategy and the set $\str(A)$ differently. All results would remain true with the exception of completeness results. In particular, we get sound faithful interpretations of $\DTTGame$. As we shall see, the completeness properties will be more subtle and require further study.

\begin{remark}[\mccorrect{Types as Homomorphisms?}]
The reader may wonder if we should not require that types are continuous or at the very least monotone. Indeed, the intuition may be that types should arise from strategies into some universe. However, we believe this point of view is not the most productive. Indeed, as we shall see in chapter \ref{ch:3}, types are best thought of as  functions on \mccorrect{values (like thunks of computations)} rather than some sort of homomorphism into which we can effectfully substitute computations. It should be noted that even if we work with a type universe, it is not clear that the codes (which strategies on the universe represent) correspond to types in a monotone way. However, universes $\mathcal{U}$ should be thought of as value types (as we have a family $\mathsf{El}$ depending on them), hence cannot be expected to exist in the same way in a type theory with unrestricted effects. For instance, we should not expect $\diverge_{\mathcal{U}}$ to code for a type.

As a concrete example, we note that, in Martin-L\"of's partial type theory \cite{palmgren1990domain}, essentially Martin-L\"of type theory with fixpoint combinators for all type families, types are not monotone. Indeed, $\Id_{\B\Rightarrow \B}(\lambda_{x:\B}\mathsf{case}(x,\ttt,{\diverge}),\lambda_{x:\B}\mathsf{case}(x,\ttt,\fff))$ has ${\diverge}$ as only inhabitant while $\Id_{\B\Rightarrow \B}(\lambda_{x:\B}\mathsf{case}(x,\ttt,{\diverge}\hspace{-3pt}),\lambda_{x:\B}\mathsf{case}(x,\ttt,{\diverge}\hspace{-3pt}))$ has two distinct inhabitants ${\diverge}$ and $\refl{\lambda_{x:\B}\mathsf{case}(x,\ttt,{\diverge})}$. At the same time, $\lambda_{x:\B}\mathsf{case}(x,\ttt,{\diverge})\leq \lambda_{x:\B}\mathsf{case}(x,\ttt,\fff)$. Type monotonicity is seen to be restored, for instance, if we impose the axiom of function extensionality on the $\Id$-types, but it is not a feature of bare intensional type theory with recursion.
\end{remark}

In the simply typed world, the idea is precisely (in the sense that we get full abstraction results) that dropping winning conditions allows us to interpret fixpoint combinators \cite{abramsky2000full}, dropping history-freeness allows us to interpret local references of ground type \cite{abramsky1996linearity}, dropping determinism allows us to interpret erratic non-deterministic choice primitives \cite{harmer1999fully} and switching from well-bracketed to weakly well-bracketed strategies allows us to interpret the universal control operator $\mathsf{call/cc}$ \cite{laird1997full}. We make some observations on the extent to which our context games with morphisms composed of the suitable strategies give a sound interpretation of various primitives for effects and we briefly discuss the status of definability results.

\subsection{Recursion}
We note that we can interpret fixpoint combinators in the world of partial (i.e. non-winning) strategies for all context games of length $1$ exactly as in \cite{abramsky2000full}. The only difficulty is posed by $\Sigma$-types. We can obtain an interpretation of fixpoint combinators from a general construction (see \cite{abramsky2000full}) if we can show that $\Ctxt(\Gamecat_!)$ is a rational category. This is not at all guaranteed if we impose no conditions on dependent games and define them simply as unconstrained functions on strategies. Indeed, while our homsets are always pointed partial orders, we cannot always take the colimit of chains $f^{(k)}$ of repeated function applications (starting from $\bot$).

To achieve this, it would be sufficient to demand that dependent games are continuous functions. However, this is easily seen to be inconsistent with our interpretation of $\Id$-types.
A weaker and still sufficient condition is to demand that for a dependent game $B$, for each infinite ascending chain $[\sigma_i^0]_i,[\sigma_i^1]_i,\ldots$ in $\str(\tot{A_1}\&\cdots\&\tot{A_n})$ and for every $s\in P_{\tot{B}}$, if there exists some integer $N$\footnote{\begin{mccorrection}
Indeed, given $f=\langle f_1,\ldots,f_n\rangle:\Sigma(A_1,\ldots,A_n)\ra{}\Sigma(A_1,\ldots,A_n)$, define the increasing $\omega$-chain $\sigma^N:=f^N(\bot)$. We want that $\bigcup_{N\in\N} f_k(\sigma^N)\in\str(A_k\{[f_{k'}]_{k'<k}\}(\bigcup_{N\in\N}\sigma^N))$. We have that $f_k(\sigma^N)\in\str(A_k\{[f_{k'}]_{k'<k}\}(\sigma^N))$, so we precisely need the extra condition that $B$ does not shirk in the limit $N\rightarrow \infty$.
\end{mccorrection}} such that $s\in P_{B(\sigma_1^k,\ldots,\sigma_n^k)}$ for all $k\geq N$, then $s\in P_{B(\bigcup_k \sigma_1^k,\ldots,\bigcup_k\sigma_n^k)}$. This condition is easily seen to be preserved by all our type formers. We obtain a model of $\DTTGame$ extended with fixpoint combinators at all types.

Further, we suspect\footnote{The details remain to be verified.} that for compact elements (noting that these have a finite norm \cite{abramsky2000full}), our definability proof of theorem \ref{thm:define}, can be adapted to this setting by making $h_y$ explore the newly accumulated $\Id$-type. The decomposition lemma then leaves us to define a strategy which visits all $\Id$-types from right to left and then gives a non-strict reply. This can then be defined using constructors, $\substsf$-operators (to visit all $\Id$-types) and \textsf{Ty-Ext}.

\subsection{Local Ground References}
Local references of ground type (\mccorrect{for instance, of integer type}) can be added to $\DTTGame$ in exactly the same way as they are to a simply typed language. Indeed, the new terms for handling state have types only involving (closed) inductive types $X$ and reference types $\Ref(X)$. This is described very clearly in \cite{abramsky1996linearity} where a slightly different definition of $!$ and $\multimap$ are used (threads are not distinguished by labelling moves with a thread number), meaning that justifiers are not uniquely determined by the type structure and have to be specified as part of the play. All results of \cite{abramsky1996linearity} transfer to our setting, as long as we take care to include the obvious thread labelling in our interpretation of the terms for manipulating state. In particular, we get a sound interpretation of $\DTTGame$ extended with local ground type references.

Definability\footnote{The usual definability proof relies crucially on having partial strategies. It is not clear if it can be made to work in the world of winning strategies.} (and with that full abstraction) in the simply typed world, depends on the result that there is a universal history-sensitive (well-bracketed deterministic) strategy $\cell_X$ on $!\Ref(X)$, for any ground type $X$, such that any history-sensitive strategy $\sigma$ on a game $A$ in the simply typed hierarchy over ground types factors as $\cell_X$ for a suitably large $X$ (\mccorrect{this needs to be countably infinite, at least,} to encode the whole history of the play in $A$) followed by a history-free strategy $\tau$ on $\Ref(X)\Rightarrow A$ (which keeps updating $\Ref(X)$ to hold the current history of the play in $A$ and then responds as $\sigma$ would). We note that we should impose a visibility condition on strategies (a condition to exclude the use of higher-order references, which follows automatically from the restriction on plays for our type hierarchy) if we want the same factorisation result for more general games $A$. The factorisation theorem (and with that the definability result) of \cite{abramsky1996linearity} does not have an obvious generalisation to all types of $\Ctxt(\Gamecat_!)$. Indeed, context games of length greater than $1$ ($\Sigma$-types) can pose problems. If we apply the simply typed factorisation construction to a consistent tuple $[\sigma_i]_i$ of strategies, there is no guarantee that the resulting tuple $[\tau_i]_i$ is still consistent. Indeed, we only know that its interaction with $\cell_X$ would be consistent.

\subsection{Finite Non-Determinism}
Lifting the determinacy condition for strategies gives us a model of dependent type theory with non-deterministic features. In fact, following \cite{harmer1999fully}, we can clearly interpret the non-deterministic choice primitive $\vdash \Or_A : A\Rightarrow A\Rightarrow A$ for any type $A$ which gets interpreted as a game (rather than proper context game), by using the (well-bracketed history-free winning) strategy which plays a copycat between both $A^{(3)}$ and, non-deterministically, both $A^{(1)}$ and $A^{(2)}$. It is not clear that we can interpret this primitive for context games (or types containing $\Sigma$-constructors), however. Indeed, to get the correct simply typed translation, we should define $\Or_{[A_j]}$ by $[f_j]_j\;\Or\; [g_j]_j:=[f_j\;\Or\; g_j]_j$. However, we can easily see that this is well-defined for all context morphisms if and only if the dependent games $A_j$ are monotone functions of their inputs. (Indeed, $f_j\;\Or\;g_j$ represents the union of $f_j$ and $g_j$.) This is a condition that is  incompatible with the current interpretation of $\Id$-types. Hence, we are faced with a choice: interpreting $\Id$-types or $\Or$-primitives at $\Sigma$-types.

Definability (and, with that, full abstraction with respect to ``may-observational-equivalence'') in the simply typed world depends on the result that there is a universal non-deterministic strategy $\oracle$ on $\NInd\Rightarrow\NInd$ (which is strict and responds to a number $n$, non-deterministically, with some number between $0$ and $n$) such that any finitely non-deterministic strategy $\sigma$ on $A$ factors as $\oracle$, followed by some deterministic strategy \mccorrect{$\mathsf{det}(\sigma)$ (it responds with the number of different moves that $\sigma$ could make, in the left most copy of $\NInd$, and makes the $n$-th move of $\sigma$ in $A$ in response to $n$; it is winning, well-bracketed and history-free if $\sigma$ was such)} on $(\NInd\Rightarrow\NInd)\Rightarrow A$. However, it is easily seen\footnote{
\begin{mccorrection}
Indeed, take $\mathsf{contradict}$ be a game depending on $\BInd$ such that $\mathsf{contradict}:\BInd\mapsto \BInd$ and $\mathsf{contradict}: \mathsf{else}\mapsto \emptyset_*$. We have a strategy $\sigma=\langle\BInd, \ttt\rangle$ on $\Sigma(\BInd,\mathsf{contradict})$. Then $(\lambda_x 1);\mathsf{det}(\sigma)=\langle \ttt,\ttt\rangle$, which is not a strategy on $\Sigma(\BInd, \mathsf{contradict})$.
\end{mccorrection}
} that the simply typed factorisation, when performed on arbitrary context morphisms $[]\ra{[\sigma_i]_i}[A_i]_i$ (rather than just individual strategies, or context morphisms of length $1$), can result in inconsistent lists of strategies $[\tau_i]_i$, not defining a context morphism $[\NInd\Rightarrow \NInd]\ra{} [A_i]_i$.

\subsection{Control Operators}
The interpretation of control operators in game semantics seems to be more fragile than that of the other effects we have considered (partiality, local ground references, non-determinism). That is, \cite{laird1997full} shows that $\Gamecat_!$, for weakly well-bracketed strategies, interprets (with a history-free winning deterministic strategy violating the well-bracketing condition) the universal control operator\footnote{This is also known as Peirce's law in logic, which is easily seen to be equivalent to the principle of double negation elimination (take $Y_*$ to be a false formula). This law (with its computation rules) is the defining feature of constructive classical logic.} $\mathsf{call/cc}_{X_*,Y_*}:((X_*\Rightarrow Y_*)\Rightarrow X_*)\Rightarrow X_*$ (which plays a copycat between $X_*^{(1)}$ and $X_*^{(3)}$ and between $X_*^{(2)}$ and $X_*^{(3)}$ and which opens $X_*^{(1)}$ in response to the initial question in $Y_*$) for flat games $X_*$ and $Y_*$. Next, it shows that any strategy $\sigma$ on a game $A$ \emph{in the simply typed hierarchy over ground types} (rather than a general game) can be factored as $\mathsf{call/cc}_{X_*,X_*};\tau$ for (appropriate $X_*$ and) a well-bracketed strategy $\tau$  on $(((X_*\Rightarrow Y_*)\Rightarrow X_*)\Rightarrow X_*)\Rightarrow A$.

We note that we can define 
\begin{align*}
\mathsf{call/cc}_{I, C}()&:=\langle\rangle\\\mathsf{call/cc}_{A\& B, C}(f,g)&:=\langle \mathsf{call/cc}_{A,C}(\lambda_\phi f(\lambda_x \phi(\fst(x))), \mathsf{call/cc}_{B,C}(\lambda_\phi f(\lambda_x \phi(\fst(x)))\rangle\\
\mathsf{call/cc}_{A\Rightarrow B, C}(f)&:=\mathsf{call/cc}_{B,C}(\lambda_\phi t(\lambda_y \phi(y(x))(x))).
\end{align*} That is, adding the control operator $\mathsf{call/cc}$ at ground types to an intuitionistic type theory with $1,\times,\Rightarrow$-types gives us the control operator at all types, making it into a constructive classical type theory. In particular, we can interpret these control operators in our game semantics.

Now, it is well-known that constructive classical dependent type theory is degenerate in the sense that it identifies all terms \mccorrect{(propositionally)} \cite{herbelin2005degeneracy}. The situation in our model is that the obvious candidate for $\mathsf{call/cc}$ on $\Sigma$-types (based on the simply typed translation) does not define a context morphism. Indeed, in particular, for the type $A=\Sigma_{x:\mathbb{B}}\Id_{\mathbb{B}}(\ttt,x)$ used by Herbelin to derive his degeneracy result, our candidate for $\mathsf{call}/\mathsf{cc}$ (for $C=\Id_\B(\ttt,\fff)$ an inconsistent proposition) does not type check. Indeed, such an appropriate term $\mathsf{call}/\mathsf{cc}_A $ of type $((A\Rightarrow \Id_{\mathbb{B}}(\ttt,\fff))\Rightarrow A)\Rightarrow A$ would decompose as $\langle\mathsf{call}/\mathsf{cc}_A^1,\mathsf{call}/\mathsf{cc}_A^2\rangle$, where $\vdash \mathsf{call}/\mathsf{cc}_A^1:((\Sigma_{x:\mathbb{B}}\Id_{\mathbb{B}}(\ttt,x)\Rightarrow \Id_{\mathbb{B}}(\ttt,\fff))\Rightarrow \Sigma_{x:\mathbb{B}}\Id_{\mathbb{B}}(\ttt,x))\Rightarrow\mathbb{B}$ and $\vdash \mathsf{call}/\mathsf{cc}_A^2:\Pi_{t:(\Sigma_{x:\mathbb{B}}\Id_{\mathbb{B}}(\ttt,x)\Rightarrow \Id_{\mathbb{B}}(\ttt,\fff))\Rightarrow \Sigma_{x:\mathbb{B}}\Id_{\mathbb{B}}(\ttt,x)}\Id_{\mathbb{B}}(\ttt,\mathsf{call}/\mathsf{cc}_A^1(t))$. The equivalent of the usual interpretation of $\mathsf{call}/\mathsf{cc}$ of \cite{laird1997full}, which plays a copycat back and forth between $A^{(3)}$ and $A^{(2)}$, which plays the initial move in $A^{(1)}$ if $\Id_{\mathbb{B}}(\ttt,\fff)$ is opened by Opponent and which copies back Opponent's response in $A^{(1)}$ to $A^{(3)}$, is seen not to yield a sound interpretation of $\mathsf{call}/\mathsf{cc}_A$ in our model. Indeed, $\sem{\mathsf{call}/\mathsf{cc}^2_A}$ violates the rules of the game $$\Osat(\sem{\Pi_{t:(\Sigma_{x:\mathbb{B}}\Id_{\mathbb{B}}(\ttt,x)\Rightarrow \Id_{\mathbb{B}}(\ttt,\fff))\Rightarrow \Sigma_{x:\mathbb{B}}\Id_{\mathbb{B}}(\ttt,x)}\Id_{\mathbb{B}}(\ttt,\mathsf{call}/\mathsf{cc}_A^1(t))}),$$ in particular, its play $* (0,*)(0,(0,*))(0,(0,(0,*)))(0,(0,(0,\reflind)))\reflind$ does so with its last move. Indeed, this Player move excludes the value $$\tau = \sem{\lambda_{k:A\Rightarrow\Id_\mathbb{B}(\ttt,\fff)}\langle \fff,k(\langle \ttt,\refl{\ttt}\rangle )\rangle}$$ for $\sem{t}$, while only Opponent is allowed to restrict the fibre.

Moreover, the factorisation construction of \cite{laird1997full} does not give valid well-bracketed context morphisms when applied naively in $\Ctxt(\Gamecat_!)$, so it is not clear how a definability result could be established for this model.

\subsection{Lessons for Combining Dependent Types and Effects}
The models of dependent type theory presented in this section can be seen as assigning dependent types to effectful programs under a CBN equational theory. We mean that in the following way. The strategies we have considered are known to correspond very closely to programs with recursion, local ground references, finite non-determinism and control operators with CBN evaluation \cite{abramsky2000full,
abramsky1996linearity,
harmer1999fully,
laird1999semantic}. Rather than the usual simple types, modelled in usual game semantics, we have considered more precise types for these programs here, modelled by dependent games. While what we arrive at  clearly represents game theoretic model of a CBN dependent type theory with some effects, it should be further examined how freely these effects are allowed to occur.

One thing that the semantics suggests is the interpretation of fixpoint combinators and non-deterministic choice may be difficult at (projection) $\Sigma$-types, unless types behave as suitable homomorphisms. This, however, seems to cause a tension with the interpretation of $\Id$-types. There is a real question, therefore, whether types should act as homomorphisms. Even stronger is the result that we saw that the interpretation of universal control operators at $\Sigma$-types can fail in this model and, generally, is known to cause degeneracy. In many cases, we see that the interpretation of effects at (projection) $\Sigma$-types can be problematic. That does not mean, however, that we should not include primitives for effects (like fixpoint combinators) at other types.

A more general, but related question that this semantics raises is whether one should aim for dependent type theory with unrestricted effects in the first place as this seems to lead to many technical challenges.  We address this question further in the next chapter. In hindsight, based on \mccorrect{lessons} learnt both in this chapter and in the other chapters of this thesis, we now believe the answer should be no. In most cases, it appears both safer and more useful to us to restrict the use of effects with the type system. A prime aim for future work should, therefore, in our opinion, be the development of a game semantics for dependently typed CBPV.
\begin{savequote}[8cm]
The impossible often has a kind of integrity which the merely improbable lacks.
  \qauthor{--- Douglas Adams}
\end{savequote}

\chapter{\label{ch:3}Dependently Typed Call-by-Push-Value (dCBPV)}
Dependent types \cite{hofmann1997syntax} are slowly being taken up by the functional programming community and are in the transition from a quirky academic hobby into a practical approach to building certified software. Purely functional dependently typed languages like Coq \cite{Coq:manual} and Agda \cite{norell2007towards} have existed for a long time. If the technology is to become more widely used in practice, however, it is crucial that dependent types can be smoothly combined with the wide range of effects that programmers make use of in their day to day work, like non-termination and recursion, mutable state, input and output, non-determinism, probability and non-local~control.

Although some languages exist which combine dependent types and effects, like Cayenne \cite{augustsson1998cayenne}, $\Pi\Sigma$ \cite{altenkirch2010pisigma}, Zombie \cite{casinghino2014combining}, Idris \cite{brady2013idris}, Dependent ML \cite{xi1999dependent} and F$\star$ \cite{swamy2015dependent}, there have always been some strict limitations. For instance, the first four only combine dependent types with unrestricted recursion (although Idris has good support for emulating other effects), Dependent ML constrains types to depend only on static natural numbers and F$\star$ does not allow types to depend on effectful terms at all (including non-termination). Somewhat different is Hoare Type Theory (HTT) \cite{nanevski2006polymorphism}, which defines a programming language for writing effectful programs as well as a separation logic encoded in a system of dependent types for reasoning about these programs. We note that the programming fragment is not merely an extension of the logical one, which would be the elegant solution suggested by the Curry-Howard~correspondence.

The sentiment of most papers discussing the marriage of these ideas seems to be that dependent types and effects form a difficult though not impossible combination. However, as far as we are aware, treatment has so far been on a case-by-case basis and no general theoretical analysis has been given which discusses, on a conceptual level, the possibilities, difficulties and impossibilities of combining general computational effects and dependent~types.

In a somewhat different vein, there has long been an interest in combining linearity and dependent types. This combination was first studied from the point of view of syntax in Cervesato and Pfenning's LLF \cite{cervesato1996linear}. To this, the author added a semantic perspective, as described in chapter \ref{ch:4}, which has  proved important e.g. in the development of the game semantics for dependent types of chapter \ref{ch:5}. One aspect that this abstract semantics as well as the study of particular models highlight is -- more so than in the simply typed case -- the added  insight and flexibility obtained by decomposing the $!$-comonad into an adjunction\footnote{Indeed, connectives seem to be most naturally formulated on either the linear or cartesian side: $\Sigma$- and $\Id$-constructors operate on cartesian types while $\Pi$-constructors operate on linear types.}. This corresponds to working with dependently typed version of Benton's LNL calculus \cite{benton1995mixed} rather than Barber and Plotkin's DILL \cite{barber1996dual}, as was done in \cite{krishnaswami2015integrating}.

Similarly, it has proved problematic to give a dependently typed version of Moggi's monadic metalanguage \cite{moggi1991notions}. We hope that this chapter illustrates that also in this case a decomposed adjunction perspective, like that of CBPV \cite{levy2012call}, is more flexible than a monadic perspective. (Recall that if we decompose both linear logic and the monadic metalanguage into an adjunction, we can see the former to be a restricted case of the latter which only describes (certain) commutative effects.) In particular, it turns out that the distinction that CBPV makes between \emph{dynamic computations} and \emph{static values} (including thunks of computations) is crucial. 

In this chapter, we show that the analysis of dDILL of chapter \ref{ch:4} generalises straightforwardly to general (non-commutative) effects to give a dependently typed CBPV calculus that we call dCBPV-, which allows types to depend on values (including thunks of computations) but which lacks a Kleisli extension (or sequencing) principle for dependent functions. This calculus is closely related to Harper and Licata's dependently typed polarized intuitionistic logic \cite{licata2009positively}. Its categorical semantics is obtained from that (see section \ref{sec:deplnl}) for the dependent LNL calculus, by relaxing a condition on the adjunction which would normally imply, among other things, the commutativity of the effects described (and by dropping the symmetric monoidal closed structure on $\Dcat$).

It straightforwardly generalises Levy's adjunction models for CBPV \cite{levy2005adjunction} (from locally indexed categories to more general comprehension categories \cite{jacobs1993comprehension}) and, in a way, simplifies Moggi's strong monad models for the monadic metalanguage \cite{moggi1991notions}, as was already anticipated by Plotkin in the late 80s: in a dependently typed setting the monad strength follows straightforwardly from the natural demand that its adjunction is compatible with substitution and, similarly, the distributivity of coproducts follows from their compatibility with substitution. In fact, we believe the categorical semantics of simply typed CBPV is most naturally understood as a special case of a that of dCBPV-. Particular examples of models are given by models of the dependent LNL calculus and by Eilenberg-Moore adjunctions for strict indexed monads on models of pure DTT. The small-step operational semantics for CBPV of \cite{levy2012call} transfers to dCBPV- without any difficulties with the expected subject reduction and (depending on the effects considered) strong normalization and determinacy results.

When formulating  candidate CBV and CBN translations of DTT into dCBPV-, it becomes apparent that the latter is only well-defined if we work with the weak (non-dependent) elimination rules for positive connectives, while the former is ill-defined altogether. To obtain a CBV translation and the CBN translation in all its glory, we have to add a principle of Kleisli extensions (or sequencing) for dependent functions to dCBPV-. Such a principle also seems appealing from the point of view of compositionality of the system. We call the resulting calculus dCBPV+, to which we can easily extend our categorical and operational semantics. Normalization and determinacy results for the operational semantics of the pure calculus remain the same. However, depending on the effects we consider, subject reduction may fail. We analyse on a case-by-case basis the principle of dependent Kleisli extensions in dCBPV- models of a range of effects. We see that it is not always valid, depending on the effects under consideration. These technical challenges make it questionable if the extra expressive power of dCBPV+ is worth the extra complications. Therefore, as an alternative, we discuss the possibility of extending dCBPV- with some extra connectives to a dependently typed eriched effect calculus (EEC) \cite{egger2009enriching}. This increases its expressive power in a slightly different way, but we argue that, similarly to dependent Kleisli extensions, it also restores compositionality, in a sense that we make precise.

On the one hand, we hope this analysis gives a helpful theoretical framework in which we can study various combinations of dependent types and effects from an algebraic, denotational and operational point of view. It gives a robust motivation for the equations we should expect to hold in both CBV and CBN versions of effectful DTT, through their translations into dCBPV, and it guides us in modelling dependent types in effectful settings like game semantics.

On the other, noting that not all effects correspond to sound logical principles, an~ expressive system like CBPV or a monadic language, with fine control over where effects occur,~is~an excellent combination with dependent types as it allows us to  use the language both for writing effectful programs and pure logical proofs about these programs. Similar to HTT in aim, but different in  implementation, we hope that dCBPV can be expanded in future to an elegant language, serving both for writing effectful programs and for reasoning about~them.

In section \ref{sec:deptypesandeffects}, we explain why the combination of effects and dependent types is not straightforward. Next, in sections \ref{sec:depcbpvwoklext} and \ref{sec:depcbpvklei}, we study the syntax, categorical semantics, a range of concrete models and operational semantics for, respectively, dCBPV- and dCBPV+. After that, we discuss dependent projection products (additive $\Sigma$-types) in section \ref{sec:depprojprod} and show that we encounter technical challenges, similar to those caused by dependent Kleisli extensions. In section \ref{sec:depklextbugorfeature}, we discuss the pros and cons of dependent Kleisli extensions, to introduce the dependently typed enriched effect calculus in section \ref{sec:moreconn} as a better behaved extension of dCBPV-. We end on a brief comparison with HTT in section \ref{sec:httcomparison}.

\begin{remark}[Related Publications] This chapter is largely based on \cite{vakar2015framework,vakar2016effectful}. In these preprints we incorrectly conjectured that subject reduction of dCBPV+ could be restored through appropriate subtyping conditions. While necessary, we have since realised that such subtyping conditions are likely not to be sufficient. Independently from our work \cite{vakar2015framework} on dependently typed CBPV, \cite{ahman2016depty} arrived at a syntax very similar to dCBPV- and an equivalent categorical semantics, presented in terms of fibrations rather than indexed categories. Where we additionally give a study of operational semantics, dCBPV+ and CBV and CBN translations, concrete models coming from indexed monads and dLNL and more connectives, \cite{ahman2016depty} gives a detailed exposition of algebraic effects in dCBPV-.
\end{remark}

\section{Dependent Types and Effects?}\label{sec:deptypesandeffects}
We believe that it is clear that the combination of dependent types and effects is important, being both of a fundamental theoretical interest and of a very practical interest in verifying real world code. Why is the combination not straightforward, however?

A first obstacle that we discussed in chapter \ref{ch:1} is that dependent types are largely useful for verification purposes. Therefore, it is important to be able to guarantee the logical consistency of a dependently typed language. At the same time, effects tend to introduce inconsistency. This is easily addressed by encapsulating effects in (strong) monads, to be thought of as logical modalities on the type system. This suggests we should pursue a dependently typed version of Moggi's monadic metalanguage rather than a dependent type theory with unrestricted effects.

A second conundrum that we face in formulating such a (monadic)  dependently typed language is related to sequencing of computations. A \emph{pure dependent type theory} involves the crucial dependent composition operation 
$$f:A\Rightarrow B, g:\Pi_{y:B} C\vdash f;g :\Pi_{x:A} C[f(x)/y],$$
which can be interpreted as saying that if we can prove that a predicate is universally true and we provide a witness, we can derive that the predicate holds for the witness. Similarly, a \emph{monadic effectful simple type theory} gets much of its power from the sequencing operation 
$$f:A\Rightarrow TB, g:B\Rightarrow TC \vdash f;g^* :A\Rightarrow TC,$$
which lets us first perform one effectful computation $f$ and then take its result as an input when we perform $g$ next. One would expect an effectful dependent type theory to combine both substitution operations. In particular, we would perhaps expect to be able to substitute effectful computations into dependent functions like
$$f:A\Rightarrow TB, g:\Pi_{y:B} TC\vdash f;g^* :\Pi_{x:A} TC[f(x)/y].$$
What could this mean? $TC$ is a type depending on $y:B$, while we are trying to substitute a value $f(x):TB$ in it. Semantically, this principle would correspond to having a Kleisli extension principle for dependent functions, which has so far -- as far as we are aware -- only been considered in \cite{hottbook} for a limited class of modalities $T$ and where $TC$ is a predicate on $TB$ which can be restricted to $B$ along $B\ra{\eta_B}TB$.

A third, closely related challenge is how we should interpret a type $B[M/x]$ into which we have substituted an effectful computation $M$. That is, what does type checking $N:B[M/x]$ constitute? We seem to have at least two choices. Do we first evaluate $M$ (as a dynamic computation) and substitute the result  $V$ into the type and type check $N:B[V/x]$? Or do we consider $B$ as expressing a property of the effectful computation $M$ and not just its outcome, meaning that we normalise the thunk of $M$ (as a static value) and type check \mccorrect{$N:B[\nf{(\thunk M)}/x]$}? We see that the dual r\^oles of $M$ as a dynamic computation and $\thunk M$ as a static value are crucial. This suggests that a CBPV perspective which distinguishes between values and computations is more suitable for understanding effectful dependent type theory than a monadic language without such an explicit distinction.

In formulating a dependently typed version of CBPV, we need to decide whether types can depend on identifiers of computation types or on those of value types (or on both). This precisely corresponds to the two interpretations of type checking $N:B[M/x]$ described above: type checking $N:B[M/{\nil}]$ or $N:B[\thunk M/x]$. We believe that type dependency on computations is problematic in a similar way that type dependency on identifiers of linear type is. (See chapter \ref{ch:4}.) Indeed, the dynamic nature of computations means that their evaluation might be non-deterministic or non-terminating or might even produce other side effects like write to disk. This means that type dependency on computations would mean throwing overboard the idea of types as providing static guarantees; it would erase the significance of the phase distinction in typed programming. One could argue that what we are left with could also be achieved by writing an untyped program with some suitably placed $\mathtt{Assert}$ statements.  Non-determinism (and reading state) would mean that type checking no longer provides guarantees. Indeed, our program may have passed the type check by chance as it happened to have chosen a safe trace. Non-termination would clearly make type checking undecidable. Other side effects like writing to disk could change type checking from a harmless procedure to something to think about carefully. We hope to have convinced the reader that we should focus on types depending solely on values.

Having addressed the first and third concerns though, the second concern still remains and we shall answer it in the course of this chapter. Indeed, part of the original motivation for CBPV was that both CBV and CBN versions of effectful type theories can be encoded in it. For that purpose, it is crucial to have sequencing operations for computations. Accordingly, we shall see that we need a dependent Kleisli extension like sequencing principle for dependently typed computations to obtain CBV and CBN translations from dependent type theory with unrestricted effects. We study both a calculus, dCBPV-, without and one, dCBPV+, with such a dependent effectful sequencing principle and discuss their virtues and vices.

\section{dCBPV without Dependent Kleisli Extensions (dCBPV-)}\label{sec:depcbpvwoklext}
	In this section, we show how the results of section \ref{sec:backcbpv} have an elegant dependently typed generalization, by allowing types to depend on values. We first consider a system in which we only allow sequencing $\toin{M}{x}{N}$ of a dependent function $N$ if the result type of $N$ does not depend on the identifier $x$ that the result of $M$ is bound to.
	
\subsection{Syntax}
The syntax of CBPV generalises straightforwardly to dependent types. As anticipated already by Levy \cite{levy2012call}, we only need to take care in the rule for $\toin{M}{x}{N}$. He suggested that the return type $\ct{B}$ of $N$ should not depend on $x$ in this rule. We shall apply this restriction as well for the moment. We call the resulting system \emph{dependently typed call-by-push-value without dependent Kleisli extensions}, or dCBPV-. We shall later revisit this assumption and study a system dCBPV+ in which we do allow such Kleisli extensions for dependent functions.

We distinguish between the following objects: contexts $\Gamma;\Delta$, where $\Gamma$ is a (cartesian) region consisting of identifier declarations of value types and $\Delta$ is a (linear) region for identifier declarations of computation type and where we write $\Gamma$ as a shorthand for $\Gamma;\cdot$, value types $A$, computation types $\ct{B}$, values $V$, computations $M$ and stacks $K$. The type theory talks about these objects using the judgements of figure \ref{fig:judgements}.

To derive these judgements, we have, to start with, rules, which we shall not list, which state that all judgemental equalities are equivalence relations and that all term, type and context constructors as well as substitutions respect judgemental equality. In similar vein, we have conversion rules which state that we may swap contexts and types for judgementally equal ones in all judgements. Additionally, we demand the obvious (admissible) substitution rules for both kinds of identifiers in all judgements as well as weakening rules for identifiers of value type (but not computation type). To form contexts, we have the rules of figure~\ref{fig:ctxtrules}.

\begin{figure}[!tb]
\fbox{\resizebox{\textwidth}{!}{
\begin{tabular}{ll}
\textbf{Judgement} & \textbf{Intended meaning}\vspace{2pt}\\
$\vdash \Gamma;\Delta \;\ctxt$ & $\Gamma;\Delta$ is a valid context\\
$\Gamma \vdash A\;\vtype$ &  $A$ is a value type in context $\Gamma$\\
$\Gamma \vdash \ct{B}\;\ctype$ &  $\ct{B}$ is a computation type in context $\Gamma$\\
$\Gamma\vdash V:A$ & $V$ is a value of type $A$ in context $\Gamma$\\
$\Gamma;\Delta\vdash K:\ct{B}$ & $K$ is a computation/stack of type $\ct{B}$ in context $\Gamma;\Delta$\\
$\vdash \Gamma;\Delta = \Gamma';\Delta'$\hspace{40pt} & $\Gamma;\Delta$ and $\Gamma';\Delta'$ are judgementally equal contexts\\
$\Gamma\vdash A=A'$ & $A$ and $A'$ are judgementally equal value types in context $\Gamma$\\
$\Gamma\vdash \ct{B}= \ct{B'}$ & $\ct{B}$ and $\ct{B'}$ are judgementally equal computation types in context $\Gamma$\\
$\Gamma\vdash V= V':A$ & $V$ and $V'$ are judgementally equal values of type $A$ in context $\Gamma$\\
$\Gamma;\Delta\vdash K= K':\ct{B}$ & $K$ and $K'$ are judgementally equal computations/stacks of type $\ct{B}$ in context $\Gamma;\Delta$
\end{tabular}}}
\normalsize
\caption{\label{fig:judgements} Judgements of dependently typed CBPV.}
\end{figure}
\begin{figure}[!tb]
\fbox{
\resizebox{\linewidth}{!}{
\begin{tabular}{ll}
\AxiomC{}
\UnaryInfC{$\cdot;\cdot \ctxt$}
\DisplayProof
& \\
& \\
\AxiomC{$\vdash\Gamma;\Delta\ctxt$}
\AxiomC{$\Gamma\vdash A\;\vtype$}
\BinaryInfC{$\vdash \Gamma,x:A;\Delta \ctxt$}
\DisplayProof\hspace{100pt}\;

&
\AxiomC{$\vdash \Gamma;\cdot\ctxt$}
\AxiomC{$\Gamma \vdash \ct{B}\ctype$}
\BinaryInfC{$\vdash \Gamma;\ct{B}\ctxt$}
\DisplayProof\hspace{100pt}\;
\end{tabular}
}
}
\caption{\label{fig:ctxtrules} Rules for forming contexts, where $x$ and $\nil$ are assumed to be fresh identifiers.}
\end{figure}
To form types, we have the rules of figure \ref{fig:depcbpvtypes}.
\begin{figure}[!tb]
\fbox{
\resizebox{\linewidth}{!}{
\begin{tabular}{ll}
\AxiomC{$\Gamma,x:A,\Gamma'\vdash A'\vtype$}
\AxiomC{$\Gamma\vdash V:A$}
\BinaryInfC{$\Gamma,\Gamma'[V/x]\vdash A'[V/x]\vtype$}
\DisplayProof\hspace{60pt}\;
&
\AxiomC{$\Gamma,x:A,\Gamma'\vdash \ct{B}\ctype$}
\AxiomC{$\Gamma\vdash V:A$}
\BinaryInfC{$\Gamma,\Gamma'[V/x]\vdash \ct{B}[V/x]\ctype$}
\DisplayProof\hspace{60pt}\;\\
&\\
\AxiomC{$\Gamma\vdash \ct{B}\ctype$}
\UnaryInfC{$\Gamma\vdash U\ct{B}\vtype$}
\DisplayProof
& 
\AxiomC{$\Gamma\vdash A\vtype$}
\UnaryInfC{$\Gamma\vdash FA\ctype$}
\DisplayProof
\\
& \\
\AxiomC{$\{\Gamma\vdash A_i\vtype\}_{1\leq i \leq n}$}
\UnaryInfC{$\Gamma\vdash \Sigma_{1\leq i\leq n}A_i\vtype$}
\DisplayProof

&
\AxiomC{$\{\Gamma\vdash \ct{B_i}\ctype\}_{1\leq i \leq n}$}
\UnaryInfC{$\Gamma\vdash \Pi_{1\leq i\leq n}\ct{B_i}\ctype$}
\DisplayProof\\
&\\
\AxiomC{$\Gamma,x:A\vdash A'\vtype$}
\UnaryInfC{$\Gamma\vdash \Sigma_{x:A}A'\vtype$}
\DisplayProof

&
\AxiomC{$\Gamma,x:A\vdash \ct{B}\ctype$}
\UnaryInfC{$\Gamma\vdash \cpi{x:A}{\ct{B}}\ctype$}
\DisplayProof\\
&\\
\AxiomC{$\vdash \Gamma\ctxt$}
\UnaryInfC{$\Gamma\vdash 1\vtype$}
\DisplayProof
&\\
&\\
\AxiomC{$\Gamma\vdash V:A$}
\AxiomC{$\Gamma\vdash V':A$}
\BinaryInfC{$\Gamma\vdash \Id_A(V,V')\vtype$}
\DisplayProof

\end{tabular}
}
}
\caption{\label{fig:depcbpvtypes} Rules for type formation.}
\end{figure}
For these types, we consider the values, computations and stacks formed using the rules of figure \ref{fig:vcdepterms}.
\begin{figure}[!tb]
\centering
\fbox{\resizebox{\linewidth}{!}{
\begin{tabular}{ll}
\AxiomC{}
\UnaryInfC{$\Gamma,x:A,\Gamma'\vdash x:A$}
\DisplayProof\hspace{50pt} & \AxiomC{$\Gamma\vdash V:A$}
\AxiomC{$\Gamma,x:A,\Gamma'\vdash^{} W:{A'}$}
\BinaryInfC{$\Gamma,\Gamma'[V/x]\vdash^{} \lbi{x}{V}{W} :{A'[V/x]}$}
\DisplayProof\\
&\\
& \AxiomC{$\Gamma\vdash V:A$}
\AxiomC{$\Gamma,x:A,\Gamma';\Delta\vdash^{} K:\ct{B}$}
\BinaryInfC{$\Gamma,\Gamma'[V/x];\Delta[V/x]\vdash^{} \lbi{x}{V}{K} :\ct{B}[V/x]$}
\DisplayProof
\\
&\\
\AxiomC{}
\UnaryInfC{$\Gamma;\nil:\ct{B}\vdash \nil:\ct{B}$}
\DisplayProof &
\AxiomC{$\Gamma;\Delta\vdash^{} K:\ct{B}$}
\AxiomC{$\Gamma;\nil:\ct{B}\vdash L:\ct{C}$}
\BinaryInfC{$\Gamma;\Delta\vdash^{} \lbi{\nil}{K}{L}:\ct{B}$}
\DisplayProof\\
&\\
\AxiomC{$\Gamma\vdash V:A$}
\UnaryInfC{$\Gamma;\cdot\vdash \return\; V:FA$}
\DisplayProof &
\AxiomC{$\Gamma;\Delta\vdash K:FA$}
\AxiomC{$\Gamma,x:A,\Gamma';\cdot\vdash N:\ct{B}$}
\AxiomC{$\vdash \Gamma,\Gamma';\nil:\ct{B}\ctxt$}
\TrinaryInfC{$\Gamma ,\Gamma';\Delta\vdash \toin{K}{x}{N}:\ct{B}$}
\DisplayProof\\
&\\
\AxiomC{$\Gamma;\cdot\vdash M:\ct{B}$}
\UnaryInfC{$\Gamma\vdash \thunk M:U\ct{B}$}
\DisplayProof
&
\AxiomC{$\Gamma\vdash V: U\ct{B}$}
\UnaryInfC{$\Gamma;\cdot\vdash \force V: \ct{B}$}
\DisplayProof\\
&\\
\AxiomC{$\Gamma\vdash V_i: A_i$}
\UnaryInfC{$\Gamma\vdash \langle i,V_i\rangle : \Sigma_{1\leq i\leq n}A_i$}
\DisplayProof
&
\AxiomC{$\Gamma\vdash V: \Sigma_{1\leq i\leq n}A_i$}
\AxiomC{$\{\Gamma,x:A_i\vdash W_i : A'[\langle i,x\rangle /z]\}_{1\leq i\leq n}$}
\BinaryInfC{$\Gamma\vdash \sipm{V}{i}{x}{W_i} : A'[V/z]$}
\DisplayProof\\
&\\

 &
\AxiomC{$\Gamma\vdash V: \Sigma_{1\leq i\leq n}A_i$}
\AxiomC{$\{\Gamma,x:A_i;\Delta[\langle i,x\rangle /z]\vdash K_i : \ct{B}[\langle i,x\rangle /z]\}_{1\leq i\leq n}$}
\BinaryInfC{$\Gamma;\Delta[V/z]\vdash \sipm{V}{i}{x}{K_i} : \ct{B}[V /z]$}
\DisplayProof
\\
&\\
\AxiomC{}
\UnaryInfC{$\Gamma\vdash\langle\rangle :1$}
\DisplayProof&
\AxiomC{$\Gamma\vdash V:1$}
\AxiomC{$\Gamma\vdash W:A'[\langle \rangle /z]$}
\BinaryInfC{$\Gamma\vdash \upm{V}{W}:A'[V /z]$}
\DisplayProof
\\
&\\
&\AxiomC{$\Gamma\vdash V:1$}
\AxiomC{$\Gamma;\Delta[\langle\rangle/z]\vdash K:\ct{B}[\langle\rangle /z]$}
\BinaryInfC{$\Gamma;\Delta[V/z]\vdash \upm{V}{K}:\ct{B}[V/z]$}
\DisplayProof\\
&\\
\AxiomC{$\Gamma\vdash V_1:A_1$}
\AxiomC{$\Gamma\vdash V_2:A_2[V_1/x]$}
\BinaryInfC{$\Gamma\vdash \langle V_1,V_2\rangle :\Sigma_{x:A_1} A_2$}
\DisplayProof
&
\AxiomC{$\Gamma\vdash V: \Sigma_{x:A_1} A_2$}
\AxiomC{$\Gamma,x:A_1,y:A_2\vdash W:A'[\langle x,y\rangle /z]$}
\BinaryInfC{$\Gamma\vdash \sipm{V}{x}{y}{W}:A'[V /z]$}
\DisplayProof\\
&\\
&
\AxiomC{$\Gamma\vdash V: \Sigma_{x:A_1} A_2$}
\AxiomC{$\Gamma,x:A_1,y:A_2;\Delta[\langle x,y\rangle/z]\vdash K:\ct{B}[\langle x,y\rangle /z]$}
\BinaryInfC{$\Gamma;\Delta[V/z]\vdash \sipm{V}{x}{y}{K}:\ct{B}[V /z]$}
\DisplayProof
\\
&\\
\AxiomC{$\Gamma\vdash V:A$}
\UnaryInfC{$\Gamma\vdash \refl{V}:\Id_A(V,V)$}
\DisplayProof &
\AxiomC{$\Gamma\vdash V:\Id_A(V_1,V_2)$}
\AxiomC{$ \Gamma,x:A\vdash W :A'[x/x',\refl{x}/p]$}
\BinaryInfC{$\Gamma\vdash \idpm{V}{x}{W}:A'[V_1/x,V_2/x',V/p]$}
\DisplayProof \\
&\\
&
\AxiomC{$\Gamma\vdash V:\Id_A(V_1,V_2)$}
\noLine
\UnaryInfC{$ \Gamma,x:A;\Delta[x,x/x',\refl{x}/p]\vdash K :\ct{B}[x/x',\refl{x}/p]$}
\UnaryInfC{$\Gamma;\Delta[V_1/x,V_2/x',V/p]\vdash \idpm{V}{x}{K}:\ct{B}[V_1/x,V_2/x',V/p]$}
\DisplayProof \\
&\\
\AxiomC{$\{\Gamma;\Delta\vdash K_i :\ct{B_i}\}_{1\leq i\leq n}$}
\UnaryInfC{$\Gamma;\Delta\vdash \lambda_i K_i : \Pi_{1\leq i\leq n}\ct{B_i}$}
\DisplayProof
&
\AxiomC{$\Gamma;\Delta\vdash K: \Pi_{1\leq i\leq n}\ct{B_i}$}
\UnaryInfC{$\Gamma;\Delta\vdash i\textquoteleft K : \ct{B_i}$}
\DisplayProof\\
&\\
\AxiomC{$\Gamma,x:A;\Delta\vdash K:\ct{B}$}
\UnaryInfC{$\Gamma;\Delta\vdash \lambda_xK:\cpi{x:A}{\ct{B}}$}
\DisplayProof
&
\AxiomC{$\Gamma\vdash V:A$}
\AxiomC{$\Gamma;\Delta\vdash K:\cpi{x:A}{\ct{B}}$}
\BinaryInfC{$\Gamma;\Delta\vdash V\textquoteleft K : \ct{B}[V/x]$}
\DisplayProof
\end{tabular}
}
}
\caption{\label{fig:vcdepterms} \mccorrect{Values and computations of dCBPV-.} To aid legibility, we have left implicit one of the obvious assumptions $\Gamma,z:A''\vdash A'\vtype$, $\Gamma,z:A''\vdash \ct{B}\ctype$, $\Gamma,x,x':A, p:\Id_A(x,x')\vdash A'\vtype$ and $\Gamma,x,x':A, p:\Id_A(x,x')\vdash \ct{B}\ctype$, in each of the rules for forming pattern matching eliminators $\mathsf{pm}\;V\;\mathsf{as}\;R\;\mathsf{in}\; S$ for values $V$ of type $A''$.}
\end{figure}

We generate judgemental equalities for values and computations through the rules of figure \ref{fig:vceqs} and \ref{fig:vcdepeqs}. Note that we are using extensional $\Id$-types, in the sense of $\Id$-types with an $\eta$-rule. This is only done for the aesthetics of the categorical semantics. They may not be suitable for an implementation, however, as they can (in the presence of $\Pi$-types, c.f. section \ref{sec:moreconn}) make type checking undecidable for the usual reasons \cite{hofmann1997syntax}. The syntax and semantics can just as easily be adapted to intensional $\Id$-types, which are the obvious choice for an implementation.

\begin{figure}[!tb]
\fbox{
\parbox{\linewidth}
{\resizebox{1.02\linewidth}{!}{
\begin{tabular}{ll}
\hspace{-6pt}$\idpm{\refl{V}}{x}{R} = R[V/x]$ \hspace{0pt}& $R[V_1/x,V_2/y,V/z] \stackrel{\#w}{=} \idpm{V}{w}{R[w/x,w/y,\refl{w}/z]}$
\end{tabular}
}}
}
\caption{\label{fig:vcdepeqs} Equations for terms involving reflexivity witnesses. Again, these rules should be read as equations of typed terms in context: they are assumed to hold if we can derive that both sides of the equation are terms of the same type in the same context.}
\end{figure}

\begin{figure}[!tb]
\fbox{
\resizebox{\linewidth}{!}{
\begin{tabular}{l|l||l|l}
\textbf{CBV type}  & \textbf{CBPV type} & \textbf{CBV term } & \textbf{CBPV term}\\
\hline
$\Gamma\vdash A[M/x]\type$ & $\vect{UF}\Gamma^v\vdash A^v[(\thunk M^v)^*/x]\vtype$ & $x_1:A_1,\ldots,x_m:A_m$ & $x_1:A_1^v,\ldots,x_m:A_m^v[\ldots \tr x_{i}/z_{i}\ldots] $\\
 & & $\vdash M:A$ & $;\cdot\vdash M^v:F(A^v[\tr x_1/z_1,\ldots,\tr x_n/z_n])$\\
 && $x$& $\return x$\\
  &&$\lbi{x}{M}{N}$ & $\toin{M^v}{x}{N^v}$ \\
$\Sigma_{1\leq i\leq n }A_i$ & $\Sigma_{1\leq i\leq n }A_i^v$& $\langle i,M\rangle $&$\toin{M^v}{x}{\return \langle i, x\rangle }$ \\
 &&$\sipm{M}{i}{x}{N_i}$& $\toin{M^v}{z}{(\sipm{z}{i}{x}{N_i^v})}$\\
$\Pi_{1\leq i\leq n}A_i $ & $U\Pi_{1\leq i \leq n} FA_i^v$ &$\lambda_iM_i$ &$\return \thunk (\lambda_i  M_i^v)$\\
&&$i\textquoteleft  N $&$\toin{N^v}{z}{(i\textquoteleft \force z)}$\\
$\Pi_{x:A} A'$ & $U(\cpi{x:A^v}{F A'^v[\tr x/z]})$ & $\lambda_x M$&$\return \thunk \lambda_x M^v$\\
&&$M\textquoteleft N$ &$\toin{M^v}{x}{(\toin{N^v}{z}{(x\textquoteleft \force z)})}$\\
$1$ & $1$ & $\langle\rangle$ & $ \return \langle\rangle$  \\
&&$\upm{M}{N}$&$\toin{M^v}{z}{(\upm{z}{N^v})}$\\
$\Sigma_{x:A}  A'$ & $\Sigma_{x:A^v} A'^v[\tr x/z]$ & $ \langle M, N\rangle $  & $\toin{M^v}{x}{(\toin{N^v}{y}{\return \langle x,y\rangle})}$\\
&&$\sipm{M}{x}{y}{N}$&$\toin{M^v}{z}{(\sipm{z}{x}{y}{N^v})}$\\
$\Id_A(M,N)$& {$\Id_{ UF A^v}( \thunk M^v$,} &$\refl{M}$& $\toin{M^v}{z}{\return \refl{\tr z}} $ \\
&{$ \thunk N^v)$}&$\idpm{M}{x}{N}$& $\toin{M^v}{z}{(\idpm{z}{y}{}}$\\
&&& $(\toin{\force y}{x}{N^v}))$
\end{tabular}}
}
\caption{\label{fig:depcbvtrans} A translation of dependently typed CBV into dCBPV. We write $\tr$ as an abbreviation for $\thunk\return$, $\vect{UF}\Gamma:=z_1:UF A_1,\ldots,z_n:UFA_n$ for a context $\Gamma=x_1:A_1,\ldots,x_n:A_n$ and $V^*$ for $\thunk(\toin{\force z_1}{x_1}{\ldots \toin{\force z_n }{x_n}{\force V}})$.}
\end{figure}
\begin{figure}[!tb]
\fbox{
\resizebox{\linewidth}{!}{
\begin{tabular}{l|l||l|l}
\textbf{CBN type}  & \textbf{CBPV type} & \textbf{CBN term } & \textbf{CBPV term }\\
\hline
${\Gamma}\vdash\ct{B}[M/x]\type$ &  $\vect{U}{\Gamma^n}\vdash\ct{B}^n[\thunk M/x]\ctype$& $x_1:\ct{B}_1,\ldots,x_m:\ct{B}_m\vdash M:\ct{B}$& $x_1:U\ct{B}_1^n,\ldots,x_m:U\ct{B}_m^n;\cdot\vdash M^n:\ct{B}^n$ \\
 && $x$& $\force x$\\
  & & $\lbi{x}{M}{N}$ & $\lbi{x}{(\thunk M^n)}{N^n}$ \\
$\Sigma_{1\leq i\leq n }\ct{B}_i$ & $F\Sigma_{1\leq i\leq n }U\ct{B}_i^n$& $\langle i,M\rangle $&$\return \langle i,\thunk M^n\rangle $ \\
 & &$\sipm{M}{i}{x}{N_i}$&$\toin{M^n}{z}{(\sipm{z}{i}{x}{N_i^n})}$ \\
$\Pi_{1\leq i\leq n}\ct{B}_i $ & $\Pi_{1\leq i \leq n} \ct{B}_i^n$ & $\lambda_iM_i$& $\lambda_iM_i^n$\\
 && $i\textquoteleft M$ & $i\textquoteleft M^n$\\
$\Pi_{x:\ct{B}} \ct{B'}$ & $\cpi{x:U\ct{B}^n}{\ct{B'}^n}$ & $\lambda_x M $& $\lambda_xM^n$\\
 &&$N\textquoteleft M$ & $(\thunk N^n) \textquoteleft M^n$ \\
$1$ & $F1$ & $\langle\rangle$ & $\return \langle\rangle$  \\
&&$\upm{M}{N}$&$\toin{M^n}{z}{(\upm{z}{N^n})}$\\
$\Sigma_{x:\ct{B}} \ct{B'}$ & $F(\Sigma_{x:U\ct{B}^n}  U\ct{B'}^n)$ & $\langle M, N\rangle $  & $\return \langle \thunk M^n,\thunk N^n\rangle$\\
&& $\sipm{M}{x}{y}{N}$& $\toin{M^n}{z}{(\sipm{z}{x}{y}{N^n})}$ \\
$\Id_{\ct{B}}(M,M')$&$F(\Id_{U\ct{B}}(\thunk M^n$ & $\refl{M}$& $\return \refl{\thunk M^n}$\\
&$,\thunk M'{}^n))$&$\idpm{M}{x}{N}$& $\toin{M^n}{z}{(\idpm{z}{x}{N^n})}$
\end{tabular}
}
}
\caption{\label{fig:depcbntrans} A translation of dependently typed CBN into dCBPV. We write $\vect{U}\Gamma:=x_1:U A_1,\ldots,x_n:UA_n$ for a context $\Gamma=x_1:A_1,\ldots,x_n:A_n$.}
\end{figure}

Figures \ref{fig:depcbvtrans} and \ref{fig:depcbntrans} indicate the natural candidate CBV and CBN translations of DTT into dCBPV, where we interpret $\Sigma$-types as having a pattern matching eliminator, as opposed to projection eliminators\footnote{To give the translations of projection $\Sigma$-types, we would need dependent connectives generalising $\Pi_{1\leq i\leq n}$ on computation types. These are the equivalents of additive $\Sigma$-types and they are similarly problematic. We discuss them in section \ref{sec:depprojprod}.}.

However, it turns out that without dependent Kleisli extensions, the CBV translation is not well-defined as it results in untypable terms. The CBN translation is, but only if we restrict to the weak (non-dependent) elimination rules for $\Sigma_{1\leq i\leq n}$-, $1$-, $\Sigma$- and $\Id$-types, meaning that the type we are eliminating into does not depend on the type being eliminated from. For an alternative to the CBV translation, we would expect the CBV translation to factorise as a translation into a dependently typed equivalent of Moggi's' monadic metalanguage, followed by a translation from this monadic language into dCBPV. It is, in fact, the former that is ill-defined if we do not have a principle of Kleisli extensions in our monadic language (or, correspondingly, in dCBPV). What we can define  is a translation from a dependently typed monadic language (without dependent Kleisli extensions) into dCBPV-. In this case, we can use the strong (dependent) elimination rules for all positive connectives.

By analogy with the simply typed scenario, it seems very likely that one would be able to state soundness and completeness results for these translations, if one used the canonical equational theories for CBV and CBN dependent type theory. As we are not aware of any such equational theories being described in literature, one could imagine \emph{defining} the CBV and CBN equational theory on dependent type theories through their translations into CBPV.

\subsection{Categorical Semantics}
We have now reached the point in the story that was our initial motivation to study dependently typed CBPV: its very natural categorical semantics. Note that we have the following elegant generalization of our reformulated notion of categorical model for simple CBPV of section \ref{sec:scbpvsem}.
\begin{definition}[dCBPV- Model] By a categorical model of dCBPV-, we shall mean the following data.
\begin{itemize}
\item an indexed category $\Bcat^{op}\ra{\Ccat}\Cat$ of \emph{values} with full and faithful democratic comprehension (including an indexed terminal object $1$);
\item an indexed category $\Bcat^{op}\ra{\Dcat}\Cat$ of \emph{computations and stacks};
\item strong $0,+$-types in $\Ccat$ such that, additionally, the following induced maps are bijections:
$$\Dcat(C.\Sigma_{1\leq i\leq n} C_i )(\ct{D},\ct{D'})\ra{}\Pi_{1\leq i \leq n}\Dcat(C.C_i)(\ct{D}\{\proj{C}{\langle i,\id_{C_i}\rangle }\},\ct{D'}\{\proj{C}{\langle i,\id_{C_i}\rangle }\});$$
\item an indexed adjunction \mbox{
\begin{diagram}
\Dcat & \pile{\lTo^F\\\bot\\\rTo_U} &\Ccat;
\end{diagram}}
\item $\cpi{-}{}$-types in $\Dcat$ in the sense of having right adjoint functors $-\Dcat(\proj{A}{B})\dashv \cpi{B}{}$ satisfying the right Beck-Chevalley condition for $\mathbf{p}$-squares;
\item finite indexed products $(\top,\&)$ in $\Dcat$;
\item strong $\Sigma$-types in $\Ccat$;
\item strong extensional $\Id$-types in $\Ccat$\footnote{In case we work with intensional $\Id$-types, we should add the additional condition, which corresponds to pattern matching for stacks, that says that the canonical map $\Dcat(A.A'.A'.\Id_{A'})(\ct{B},\ct{B'})\ra{}\Dcat(A.A')(\ct{B}\{\langle \diag{A}{A'},\refl{A'}\rangle\},\ct{B'}\{\langle \diag{A}{A'},\refl{A'}\rangle\})$ is a retraction. This map is automatically an isomorphism in our case of extensional $\Id$-types.}.
\end{itemize}
\end{definition}
Again, this semantics is sound and complete.
\begin{theorem}[dCBPV- Semantics]\label{thm:dcbpv-soundcomplete} We have a sound interpretation of dCBPV- in a dCBPV- model:\\
\\
\resizebox{\linewidth}{!}{$
\begin{array}{ll}
\sem{\cdot}  = \cdot &\sem{\Gamma;\cdot}= F1\\
\sem{\Gamma,x:A} = \sem{\Gamma}.\sem{A} & \sem{\Gamma;{\nil}:\ct{B}}=\sem{\ct{B}}\\
\sem{\Gamma\vdash A}=\Ccat(\sem{\Gamma})(1,\sem{A})&
\sem{\Gamma;\Delta\vdash \ct{C}}  = \Dcat(\sem{\Gamma})(\sem{\Delta},\sem{\ct{C}})\\
\sem{A[V/x]}  = \sem{A}\{\qu{\langle \id_{\sem{\Gamma}},\sem{V}\rangle}{\sem{\Gamma'}}\} &
\sem{\ct{B}[V/x]} = \sem{\ct{B}}\{\qu{\langle \id_{\sem{\Gamma}},\sem{V}\rangle}{\sem{\Gamma'}} \}\\
\sem{U\ct{B}} = U\sem{\ct{B}} & \sem{FA}=F\sem{A}
\\
\sem{\Sigma_{1\leq i\leq n}A_i}=(\cdot(\sem{A_1}+\sem{A_2})+\cdots)+\sem{A_n}) &
\sem{\Pi_{1\leq i\leq n}\ct{B}_i}=(\cdot(\sem{\ct{B}_1}\&\sem{\ct{B}_2})\&\cdots)\&\sem{\ct{B}_n})\\
\sem{\Sigma_{x:A} A'} = \Sigma_{\sem{A}} \sem{A'} & \sem{\cpi{x:A}{\ct{B}}} = \cpi{\sem{A}}{\sem{\ct{B}}}\\
\sem{1}=1 &\\
\sem{\Id_A(V,V')}  = \Id_{\sem{A}}\{\langle\langle  \id_{\sem{\Gamma}} , \sem{V}\rangle ,\sem{V'}\rangle \},
\end{array}$}
\\
\\
together with the obvious interpretation of terms. The interpretation in such categories is complete in the sense that an equality of terms of types holds in all interpretations iff it is provable in the syntax of dCBPV-. In fact, the interpretation defines a 1-1 correspondence between categorical models and syntactic theories in dCBPV- which satisfy mutual soundness and completeness results.
\end{theorem}
\begin{proof}[Proof (sketch)] The proof goes almost entirely along the lines of the soundness and completeness proofs for linear dependent type theory in chapter \ref{ch:4}. Nothing surprising happens in the soundness proof. For the completeness result, we build a syntactic category.\end{proof}
Performing the CBN translation in the semantics, this leads to an induced notion of model for CBN dependent type theory.
\begin{theorem}[Dependent CBN Semantics {1}] The (semantic equivalent of the) CBN translation of DTT with $\Sigma_{1\leq i\leq n}$-, $1$-, $\Sigma$-, $\Id$-, $\Pi_{1\leq i\leq n}$-, $\Pi$-types, where we use the weak (non-dependent) elimination rules for all positive connectives, into dCBPV-, lets us construct a categorical model of CBN dependent type theory with the connectives above out of any model of dCBPV- by taking the co-Kleisli (indexed) category for $!:=FU$. The interpretation of CBN dependent type theory is sound and complete for the equational theory induced from dCBPV-:\\
\resizebox{\linewidth}{!}{
$
\sem{\ct{B_1},\cdots,\ct{B_n}\vdash \ct{B}}=\Dcat(U\sem{\ct{B_1}}.\cdots.U\sem{\ct{B_n}})(F1,\sem{\ct{B}})\cong\Dcat_{!}(U\sem{\ct{B_1}}.\cdots.U\sem{\ct{B_n}})(\top,\sem{\ct{B}}).$
}
\end{theorem}
We note that this co-Kleisli category, our notion of a model of CBN dependent type theory, is very close to the usual notion of a model of pure DTT. (We have seen this in chapter \ref{ch:5}, in the context of CBN game semantics!) We note that even if we \mccorrect{start} with extensional $\Id$-types in dCBPV-, we may obtain intensional $\Id$-types in dependent CBN.
\begin{theorem}[Dependent CBN Categories] The co-Kleisli category $\Dcat_{!}$ is an indexed category with full and faithful (possibly undemocratic) comprehension with fibred finite products $\Pi_{1\leq i\leq n}$ as well as $\cpi{-}{}$-types. It supports weak $\Sigma_{1\leq i\leq n}$-, $\Sigma$- and $\Id$-types (non-dependent elimination rules, failure of the general $\eta$-rules).
\end{theorem}
\begin{proof}
$!$ being an indexed comonad, it follows that $\Dcat_!$ is an indexed category. $\Dcat_!$ satisfies the comprehension axiom in the sense that we have homset isomorphism
\begin{align*}
\Dcat_!(\Gamma')(\top,B\{f\})&=\Dcat(\Gamma')(FU\top,B\{f\})\\
&\cong\Dcat(\Gamma')(F1,B\{f\})\\
&\cong \Ccat(\Gamma')(1,U(B\{f\}))\\
&= \Ccat(\Gamma')(1,U(B)\{f\})\\
&\cong \Bcat/\Gamma(f,\proj{\Gamma}{UB}).
\end{align*}
As the comprehension functor $\Dcat_!(\Gamma)(B,B')\cong \Ccat(\Gamma)(UB,UB')\ra{}\Bcat/\Gamma(\proj{\Gamma}{UB},\proj{\Gamma}{UB'})$ is a special case of the comprehension functor for $\Ccat$, we know it to be full and faithful. Note that the comprehension may be undemocratic as $\Dcat_!(\cdot)$ is equivalent to the full subcategory of $\Ccat$ on the objects in the image of $U$, which may give a proper subcategory of $\Ccat(\cdot)\cong \Bcat$. 

We know from the simply typed case that fibre-wise products in $\Dcat$ give rise to products in $\Dcat_!$. These are stable under change of base\mccorrect{,} by assumption.

Note that $\Pi$-types directly follow as a special case of $\cpi{-}{}$-types in $\Dcat$:
\begin{align*}
\Dcat_!(\Gamma.UA)(B\{\proj{\Gamma}{UA}\},C)&=\Dcat(\Gamma.UA)(FU(B\{\proj{\Gamma}{UA}\}),C)\\
&=\Dcat(\Gamma.UA)((FUB)\{\proj{\Gamma}{UA}\},C)\\
&\cong\Dcat(\Gamma)(FUB,\cpi{UA}{C})\\
&=\Dcat_!(\Gamma)(B,\cpi{UA}{C}).
\end{align*}
For $\Sigma$-types, we note that we have maps back and forth, given by the unit and counit of the adjunction between $F$ and $U$ which satisfy a $\beta$-law given by one of the triangle identities for the adjunction:
\begin{align*}
\Dcat_!(\Gamma.UA)(B,C\{\proj{\Gamma}{UA}\})&=\Dcat(\Gamma.UA)(FU(B),C\{\proj{\Gamma}{UA}\})\\
&\cong\Ccat(\Gamma.UA)(UB,U(C\{\proj{\Gamma}{UA}\}))\\
&=\Ccat(\Gamma.UA)(UB,(UC)\{\proj{\Gamma}{UA}\})\\
&\cong\Ccat(\Gamma)(\Sigma_{UA}UB,UC)\\
&\cong\Dcat(\Gamma)(F\Sigma_{UA}UB,C)\\
&\leftrightarrows\Dcat(\Gamma)(FUF\Sigma_{UA}UB,C)\\
&=\Dcat_!(\Gamma)(F\Sigma_{UA}UB,C).
\end{align*}
The same argument gives us the corresponding statement for $\Sigma_{1\leq i\leq n}$- and $\Id$-types, using their definition as left adjoint functors.
\end{proof}
We postpone the categorical discussion of models for dependently typed CBV until we add dependent Kleisli extensions to dCBPV- in section \ref{sec:depcbpvklei}. For now, we would just like to point out that $\Ccat$ equipped with the indexed monad $T:=UF$ defines what should be regarded as a model of a dependently typed equivalent of Moggi's monadic metalanguage, without dependent Kleisli extensions.

\begin{theorem}[Dependent monadic metalanguage models]
Given a model $\Ccat\leftrightarrows \Dcat$ of dCBPV-, $T:=UF$ defines an indexed monad on $\Ccat$, which has a generalized notion of strength $\Sigma_ATB\ra{s_{A,B}}T\Sigma_AB$.
\end{theorem}
\begin{proof} As $F\dashv U$ is an indexed adjunction, $T$ is an indexed monad. We note that, starting from $\id_{\Sigma_AB}$, we can obtain a generalised notion of strength for $T$:
\begin{align*}
\Ccat(\Gamma)(\Sigma_A B ,\Sigma_A B) &\cong \Ccat(\Gamma.A)(B,(\Sigma_AB)\{\proj{\Gamma}{A}\})\\
&\ra{T} \Ccat(\Gamma.A)(TB,T(\Sigma_A B)\{\proj{\Gamma}{A}\})\\
&= \Ccat(\Gamma.A)(TB,(T\Sigma_A B)\{\proj{\Gamma}{A}\})\\
&\cong\Ccat(\Gamma)(\Sigma_ATB,T\Sigma_AB).
\end{align*}
In particular (for the case where $\Gamma=\cdot$, using full and faithful comprehension), we get $\Gamma.TA\ra{}T(\Gamma.A)\in \Bcat$.\end{proof}
\begin{remark}Note that we cannot\mccorrect{,} in general\mccorrect{,} define a costrength $\Sigma_{TA} B\ra{}T\Sigma_A B\{\proj{\Gamma}{\eta_A}\}$ or, therefore, a pairing $\Sigma_{TA}TB\ra{} T\Sigma_A B\{\proj{\Gamma}{\eta_A}\}$. This asymmetry does not occur in the simply typed setting. It can be mended by the addition of Kleisli extensions for dependent functions.\end{remark}

In the simply typed setting, one can factor the CBV translation from the $\lambda$-calculus into CBPV through the monadic metalanguage. While the translation from the dependently typed monadic metalanguage with dependent Kleisli extensions in dCBPV- works fine, we cannot define the obvious CBV translation from dependent type theory into the dependently typed monadic metalanguage, unless we have dependent Kleisli extensions.

\subsection{Some Basic Models}
We can first note that any model of pure dependent type theory is, by using the identity adjunction, in particular, a model of dependently typed CBPV\mccorrect{,} which shows consistency of the calculus.
\begin{theorem}
\mccorrect{dCBPV- is consistent both in the sense that not all terms are identified and in the sense that not all types are inhabited.}
\end{theorem}
More interestingly, any model of the dependent LNL calculus supporting the appropriate connectives (see chapter \ref{ch:4}) gives rise to a model of dependently typed CBPV without dependent Kleisli extensions, modelling commutative effects.
\begin{theorem}The notion of model given by section \ref{sec:deplnl} for the dLNL calculus of \cite{krishnaswami2015integrating} with the additional connective of finite disjunctions for cartesian types (indexed finite distributive coproducts in $\Ccat$) is precisely a dCBPV- model such that we have symmetric monoidal closed structures $(I,\otimes,\multimap)$ on the fibres of $\Dcat$, stable under change of base, ($\Dcat$ is an indexed symmetric monoidal closed category) s.t. $F$ consists of strong symmetric monoidal functors (sending nullary and binary products in $\Ccat$ to $I$ and $\otimes$ in $\Dcat$) and which supports $\csigma{-}{}$-types (see section \ref{sec:moreconn}).
\end{theorem}
As in the simply typed setting, models of pure DTT on which we have an indexed monad are again a source of examples of dCBPV- models. This shows that dCBPV- is compatible with a wide range of effects.
\begin{theorem}\label{thm:dcbpv-modelsfrommonads}
Let $\Bcat^{op}\ra{\Ccat}\Cat$ be a model of pure DTT (with all type formers discussed) on which we have an indexed monad $T$. Then, the indexed Eilenberg-Moore adjunction $\Ccat\leftrightarrows \Ccat^T$ gives a model of dCBPV-.
\end{theorem}
\begin{proof} A product of algebras is just the product of their carriers equipped with the obvious algebra structure. Indeed, it is a basic result in category theory that the forgetful functor from the Eilenberg-Moore category creates limits. Given an object $TB\ra{k}B$ of $\Ccat^T(\Gamma.A)$, we note that we also obtain a canonical $T$-algebra structure on $\Pi$-types of carriers (starting from the identity on $\Pi_A B$):
\begin{align*}
\Ccat(\Gamma)(\Pi_AB,\Pi_A B)&\cong \Ccat(\Gamma.A)((\Pi_AB)\{\proj{\Gamma}{A}\},B)\\
&\ra{T}  \Ccat(\Gamma.A)(T((\Pi_AB)\{\proj{\Gamma}{A}\}),TB)\\
&\cong  \Ccat(\Gamma.A)((T\Pi_AB)\{\proj{\Gamma}{A}\},TB)\\
&\ra{-;k}  \Ccat(\Gamma.A)((T\Pi_AB)\{\proj{\Gamma}{A}\},B)\\
&\cong \Ccat(\Gamma)(T\Pi_AB,\Pi_AB).
\end{align*}
We leave the verification of the $T$-algebra axioms to the reader. We define the result to be $\cpi{A}{k}$. Note that it is precisely defined so that, for an algebra $TC\ra{l}C$, the isomorphism $\Ccat(\Gamma.A)(C\{\proj{\Gamma}{A}\},B)\cong\Ccat(\Gamma)(C,\Pi_A B)$ restricts to $\Ccat^T(\Gamma.A)(l\{\proj{\Gamma}{A}\},k)\cong\Ccat^T(\Gamma)(l,\cpi{A}{k})$.
\end{proof}

A concrete example to which we can apply the previous theorem is obtained for any monad $T$ on $\Set$. Indeed, we can lift $T$ (point-wise) to an indexed monad on the usual families of sets model $\Fam(\Set)$ of pure DTT\footnote{Recall that $\Fam(\Set)$ is defined as the restriction to $\Set\subseteq \Cat$ of the ($\Cat$-enriched) hom-functor into $\Set$:
$\Set^{op}\subseteq \Cat^{op}\ra{\Cat(-,\Set)}\Cat$.}. In a different vein, given a model $\Ccat$ of pure DTT, the usual exception ($(-)+E$), global state ($S\Rightarrow(-\times S)$), reader ($S\Rightarrow (-)$), writer ($(-)\times M$) and continuation monads ($((-)\Rightarrow R)\Rightarrow R$), which we form using objects of $\Ccat(\cdot)$, and\mccorrect{, if we are dealing with a higher-order logic,} power set monad $\mathcal{P}(-)$ give rise to indexed monads, hence we obtain models of dCBPV-. More exotic examples are the many indexed monads that arise from homotopy type theory, like truncation modalities or cohesion (shape and sharp) modalities \cite{hottbook,shulman2015brouwer,schreiber2014quantum}. A caveat there is that the identity types in the model are intensional and that many equations are often only assumed up to propositional rather than judgemental equality.

\subsection{Operational Semantics and Effects}\label{sec:depop} 
We define an operational semantics for dCBPV-. It is a basic result in dependent type theory that the (parallel nested) $\beta$-reductions for values are strongly normalizing \cite{martin1998intuitionistic} (according to a variation on Tait's logical relations argument). Let us write again, $\nf{V}$ for the normal form of a value $V$ and $\nnf{V}$ to make explicit that $V$ is not in normal form. We again define a configuration to be a pair $M,K$ of a dCBPV- computation $\Gamma;\cdot\vdash M:\ct{B}$ and a stack $\Gamma;{\nil}:\ct{B}\vdash K:\ct{C}$. The CK-machine that evaluates our computations is again just that of figure \ref{fig:ckmachine} where we add the extra transitions and terminal configuration of figure \ref{fig:ckid}.
\begin{figure}[!tb]
\fbox{
\resizebox{\linewidth}{!}{
\begin{tabular}{l}
\textbf{Transitions}\\
\begin{tabular}{lllllll}
 $  \idpm{\refl{\nnf{V}}}{x}{}{M}      $ &,& $K$\hspace{40pt}\;& $\leadsto$ \hspace{40pt}\;& $\idpm{\refl{\nf{V}}}{x}{}{M}        $ &,& $K$ \\
 $  \idpm{\refl{\nf{V}}}{x}{}{M}      $ &,& $K$\hspace{40pt}\;& $\leadsto$ \hspace{40pt}\;& $  M[\nf{V}/x]        $ &,& $K$ \\
\end{tabular}\\
\\
\textbf{Terminal Configuration}\\
\begin{tabular}{lll}
$\idpm{\nf{V}^{x'}}{x}{M}$ &,& $K$
\end{tabular}
\end{tabular}
}
}
\caption{\label{fig:ckid} The additional transition and terminal configuration that specify the operational behaviour of identity witnesses.}
\end{figure}
As before, we can add the effects of figure \ref{fig:effects} together with their operational semantics of figures \ref{fig:opsemdivs} and \ref{fig:opsemprint} and equations of figure \ref{fig:effeqn}. We get the same determinism, strong normalization and subject reduction results as in the simply typed case.

\begin{theorem}[Determinism, Strong Normalization and Subject Reduction]\label{thm:subjreddcbpv-} Every transition respects the type of the configuration. No transition occurs precisely if we are in a terminal configuration. In absence of erratic choice, at most one transition applies to each configuration. In absence of divergence and recursion, every configuration reduces to a terminal configuration in a finite number of steps.
\end{theorem}
\begin{proof}[Proof] The important observation will be that types only depend on values. Therefore, the only real difference in this proof from the simply typed case are the rules involving the reduction of values. 

We recall from \cite{martin1998intuitionistic} that value types are closed under the untyped $\beta$-reductions for values. (This result applies as values form a conventional cartesian dependent type theory.) This implies that $\Gamma\vdash\nnf{V}=\nf{V}:A$ for any $\Gamma\vdash \nnf{V}:A$. Therefore, it follows that $\Gamma\vdash \ct{B}[\nnf{V}/x]= \ct{B}[\nf{V}/x]$ for any $\Gamma, x:A\vdash \ct{B}{\ctype}$. 
It follows that the rules involving value normalization also satisfy subject reduction.

As all transitions are defined on untyped terms, determinism and strong normalization results are no different from the simply typed case.\end{proof}

\begin{remark}[Type Checking]
While the operational semantics discussed here is very relevant as it describes the execution of a program of dCBPV-, one could argue that a type checker is as important an operational aspect to the implementation of a dependent type theory. We leave the description of a type checking algorithm to future work. We note that the core step in the implementation of a type checker is a normalization algorithm for directed versions (from left to right) of the equations for values of figures \ref{fig:vceqs} and \ref{fig:vcdepeqs} (with congruence laws) and perhaps some equations for values induced from computation equations of figure \ref{fig:effeqn} and from the specific equational theories for the effects under consideration, as this would give us a normalization procedure for types. One might be able to construct such an algorithm using normalization by evaluation by combining the techniques of \cite{abel2007normalization} and \cite{ahman2013normalization}. Our hope is that this will lead to a proof of decidable type checking of the system at least in absence of the $\eta$-law for $\Id$-types\mccorrect{.} We note that the complexity of a type checking algorithm can vary widely depending on which equations we include for specific effects. The idea is that one only includes a basic set of program equations as judgemental equalities to be able to decide type checking and one postulates other equations as propositional equalities, which can be used for manual or tactic-assisted reasoning about effectful programs.
\end{remark}

\section{dCBPV with Dependent Kleisli Extensions (dCBPV+)}\label{sec:depcbpvklei}
While the system dCBPV- is very clean in its syntax, operational semantics, categorical semantics and admits plenty of concrete models, it may be a bit of a disappointment to the reader who was expecting to see a proper combination of effects and dependent types, rather than a system that keeps both features side by side without them interacting meaningfully\footnote{As we shall discuss later, this is not entirely fair on dCBPV-, as it does allow us to form types (predicates) depending on thunks of effectful computations.}. In particular, one might find it unsatisfactory that the CBV translation from dependent type theory into dCBPV- fails and that the CBN translation only goes through to a limited extent.

To address these issues, we introduce a more expressive system in this section which we call dCBPV+ and which extends dCBPV- with Kleisli extensions for dependent functions. \mccorrect{We shall later discuss, in section} \ref{sec:depklextbugorfeature}, \mccorrect{whether such dependent Kleisli extensions are desirable.}

\subsection{Syntax}
We have seen the need to add dependent Kleisli extensions in the form of the rule shown in figure \ref{fig:depklext} if we want to obtain a dependently typed equivalent of the CBV translation into CBPV or if we want to model dependent elimination rules for the positive connectives in the CBN translation. We use the name dCBPV+ to explicitly refer to the resulting system of the rules of dCBPV- (figures \ref{fig:ctxtrules}, \ref{fig:depcbpvtypes}, \ref{fig:vcdepterms}, \ref{fig:vceqs} and \ref{fig:vcdepeqs}) and dependent Kleisli extensions (figure \ref{fig:depklext}).
\begin{figure}[!tb]\fbox{\resizebox{\linewidth}{!}{\begin{tabular}{l}
\AxiomC{$\Gamma,z:UF A,\Gamma'\vdash \ct{B}\ctype$}
\AxiomC{$\Gamma;\cdot\vdash M:FA$}
\AxiomC{$\Gamma,x:A,\Gamma'[\tr x/z];\cdot\vdash N:\ct{B}[\tr x/z]$}
\TrinaryInfC{$\Gamma,\Gamma'[\thunk M/z];\cdot\vdash \toin{M}{x}{N} : \ct{B}[\thunk M/z]$.}
\DisplayProof\hspace{40pt}\;\end{tabular}
}}
\caption{\label{fig:depklext} The rule for dependent Kleisli extensions in dCBPV. As before, we write $\tr$ as an abbreviation for $\thunk \return$.}
\end{figure}

We note that, in the presence of this extra rule, the translations of figures \ref{fig:depcbvtrans} and \ref{fig:depcbntrans} are finally well-defined. We would like to highlight the fact that a type $x_1:A_1,\ldots,x_n:A_n\vdash A\type$ gets translated into a type $z_1:UFA_1,\ldots, z_n:UFA_n\vdash A^v \vtype$ by the CBV translation. Briefly, this is necessitated by the CBV translation of substitution of terms in types. For example, to substitute a term $x:B\vdash M:A$ into $x:A\vdash C\type$ in the CBV translation, we have to be able to substitute $(x:B)^v;\cdot\vdash M^v:FA$ (or equivalently $(x:B)^v\vdash \thunk M^v:UFA$) into $(x:A)^v\vdash C^v\vtype$. This forces us to define the CBV translation $(x:A)^v$ of an identifier declaration in the context \mccorrect{of a type well-formedness judgement} as $z:UFA$ if we are to use the usual type substitution of CBPV (after taking the Kleisli extension of $\thunk M^v$).

We would like to say that the CBV and CBN translations are sound and complete. However, as no notion of a CBV or CBN equational theory has been formulated for dependent type theory, as far as we are aware, we take the equational theories induced by these translations as their definitions. Unsurprisingly, $\Sigma$-types behave equationally exactly like $\times$-types and $\cpi{-}{}$-types do as $\functype$-types. The interesting connective to study is the $\Id$-type. 
\begin{theorem}
Figures \ref{fig:depcbvtrans} and \ref{fig:depcbntrans} define CBV and CBN translations of dependent type theory with $\Sigma_{1\leq i\leq n}$-, $1$-, $\Sigma$-, $\Id$-, $\Pi_{1\leq i\leq n}$- and $\Pi$-types (with dependent elimination rules for all positive connectives) into dCBPV+. In fact, they allow us to transfer an arbitrary theory in CBV or CBN dependent type theory to one on dCBPV+ such that we again get well-defined CBV and CBN translations. As expected, CBN\mccorrect{ }$\Id$-types (even extensional ones) satisfy the $\beta$-law but may not satisfy the $\eta$-law. More surprising, perhaps, is that the same is true for CBV $\Id$-types. 
\end{theorem}
\begin{proof}It is easily seen that dependent Kleisli extensions make the translations of figures \ref{fig:depcbvtrans} and \ref{fig:depcbntrans} well-defined.

In the previous section, we have already seen the statement about CBN\mccorrect{ }$\Id$-types in case we use a non-dependent elimination rule. The case with a dependent elimination rules works similarly.

The interesting case here are the $\Id$-types in the CBV translation. For the $\beta$-rule, note that
\begin{align*}
&(\idpm{\refl{M}}{x}{N})^v=\\
&\toin{(\toin{(M^v)}{z}{\return\refl{\tr z}})}{\zeta}{\idpm{\zeta}{y}{\toin{(\force y)}{x}{N^v}}}=\\
&\toin{M^v}{z}{\idpm{\refl{\tr z}}{y}{\toin{(\force y)}{x}{N^v}}}=\\
&\toin{M^v}{z}{\toin{(\force \thunk \return z)}{x}{N^v}}=\\
&\toin{M^v}{z}{\toin{(\return z)}{x}{N^v}}=\\
&\toin{M^v}{x}{N^v}=\\
&(\lbi{x}{M}{N})^v.
\end{align*}
To see that the $\eta$-rule may fail, consider dCBPV+ with divergence. We note that in case the $\eta$-law held for $\Id$-types in CBV type theory, it would imply the following principle of reflection \cite{jacobs1999categorical}:\vspace{7pt}\\
\AxiomC{$x_1:A_1,\ldots,x_n:A_n\vdash P:\Id_A(M,N)$}
\UnaryInfC{$x_1:A_1,\ldots,x_n:A_n\vdash M=N:A,$}
\DisplayProof\vspace{7pt}\\
which, in dCBPV+ translates to the rule that
\vspace{7pt}\\
\AxiomC{$x_1:A_1^v,\ldots,x_n:A_n^v[\ldots\tr x_i/z_i\ldots];\cdot\vdash P:F\Id_{A^v}(\thunk M^v{}^*,\thunk N^v{}^*)$}
\UnaryInfC{$x_1:A_1^v,\ldots,x_n:A_n^v[\ldots\tr x_i/z_i\ldots];\cdot\vdash M^v=N^v:FA^v,$}
\DisplayProof\vspace{7pt}
\\
In particular, presence of divergence would make CBV type theory identify all terms in that case. In particular, this would mean that the terms $x:A^v, y:A^v;\cdot\vdash \return x:FA^v$ and $x:A^v,y:A^v;\cdot\vdash\return y:FA^v$ are judgementally equal in dCBPV+ with divergence, which they clearly are not. For a formal proof that they are not, we note that we can see this in the families of domains model given in section \ref{sec:dcbpvplusmod}. For a more syntactic intuition, we note that $\Id$-$\eta$ is less harmful in CBPV with effects than it is in CBV or CBN with effects due to the strict distinction between values and computations, as the obvious reflection rule it implies is the following which does not identify all terms in the presence of divergence, as it does not trivially let us satisfy the hypothesis of the rule, in the way it did in CBV type theory, given that divergence is a computation and not a value.
\vspace{7pt}
\\
\AxiomC{$\Gamma\vdash V:\Id_A(V_1,V_2)$}
\UnaryInfC{$\Gamma\vdash V_1= V_2:A.$}
\DisplayProof
\end{proof}

\subsection{Categorical Semantics}
To formulate a categorical semantics of dCBPV+, we need a dependently typed generalization of the notion of Kleisli triple. A similar notion of dependently typed Kleisli extension has been proposed before in \cite{hottbook} (section 7.7), be it for a more limited class of modalities. In practice, we shall see that, for a given indexed adjunction, dependent Kleisli extensions may not exist.

\begin{definition}[dCBPV+ Model] By a dCBPV+ model, we shall mean a dCBPV- model $F\dashv U:\Ccat\leftrightarrows \Dcat$ equipped with \emph{dependent Kleisli extensions}. That is, maps $$\Ccat(\Gamma.A.\Gamma'\{\proj{\Gamma}{\eta_A }\})(1,UB\{\qu{\proj{\Gamma}{\eta_A}}{\Gamma'}\})\ra{(-)^*}\Ccat(\Gamma.UFA.\Gamma')(1,UB),$$ where $\eta$ is the unit of the adjunction $F\dashv U$, such that the following laws hold for members of the same homset:
\begin{itemize}
\item unitality: $b^*\{\qu{\proj{\Gamma}{\eta_A}}{\Gamma'}\}=b$;
\item composition: $b^*\{\qu{\langle \id_\Gamma,a^*\rangle}{\Gamma'} \}=(b^*\{\qu{\langle \id_\Gamma,a\rangle}{\Gamma'} \})^*$;
\item agreement with the usual non-dependent Kleisli extension $(-)^\star$ for the adjunction $F\dashv U$:
\begin{diagram}
\Ccat(\Gamma)(A,UB) & \rTo^{\cong}& \Ccat(\Gamma.A)(1,UB\{\proj{\Gamma}{A}\})&\;=\;& \Ccat(\Gamma.A)(1,UB\{\proj{\Gamma}{UFA}\}\{\proj{\Gamma}{\eta_A}\})\\ 
\dTo^{(-)^\star} & && &\dTo^{(-)^*} \\
\Ccat(\Gamma)(UFA,UB)&&\rTo^{\cong}&&\Ccat(\Gamma.UFA)(1,UB\{\proj{\Gamma}{UFA}\}).
\end{diagram}
\end{itemize}
\end{definition}
\begin{remark}
Note that it is enough to just specify the dependent Kleisli extensions of the form
$$\Ccat(\Gamma.A.\Gamma'\{\proj{\Gamma}{\eta_A}\})(1,UFA'\{\qu{\proj{\Gamma}{\eta_A}}{\Gamma'}\})\ra{(-)^*}\Ccat(\Gamma.UFA.\Gamma')(1,UFA').$$
Then, we can define, more generally, $f^*:=\lambda_{x:\Gamma.UFA.\Gamma'}(f;\eta_{UB})^*\{x\}; U\epsilon_{B(x)}$, where $\eta$ is the counit of the adjunction $F\dashv U$.
\end{remark}
\begin{remark}[Dependent Costrength]
Note that dependent Kleisli extensions allow us, in particular, to define the \emph{dependent costrength}  $s_{A,B}'$ for the monad $T:=UF$ that we were missing (starting from the identity on $F\Sigma_A (B\{\proj{\Gamma}{\eta_A}\})$):
\begin{align*}
&\quad\Dcat(\Gamma)(F\Sigma_A (B\{\proj{\Gamma}{\eta_A} \}),F\Sigma_A (B\{\proj{\Gamma}{\eta_A} \}))\\
&\cong \Ccat(\Gamma)(\Sigma_A B\{\proj{\Gamma}{\eta_A} \},T \Sigma_A (B\{\proj{\Gamma}{\eta_A} \}))\\
&\cong \Ccat(\Gamma.A.B\{\proj{\Gamma}{\eta_A} \})(1,T\Sigma_A (B\{\proj{\Gamma}{\eta_A} \})\{\proj{\Gamma}{A}\}\{\proj{\Gamma.A}{B\{\proj{\Gamma}{\eta_A} \}}\})\\
&\cong \Ccat(\Gamma.A.B\{\proj{\Gamma}{\eta_A} \})(1,T\Sigma_A (B\{\proj{\Gamma}{\eta_A} \})\{\proj{\Gamma}{TA}\}\{\proj{\Gamma.TA}{B}\}\{\qu{\proj{\Gamma}{\eta_A}}{B}\})\\
&\ra{(-)^*} \Ccat(\Gamma.TA.B)(1,T\Sigma_A (B\{\proj{\Gamma}{\eta_A} \})\{\proj{\Gamma}{TA}\}\{\proj{\Gamma.TA}{B}\})\\
&\cong \Ccat(\Gamma)(\Sigma_{TA}B,T\Sigma_A (B\{\proj{\Gamma}{\eta_A}\})).
\end{align*}
As a consequence, we are able to define both a left and a right pairing (which will in general not coincide for non-commutative effects):
\begin{diagram}
\Sigma_{TA} {TB} & \rTo^{s'} & T\Sigma_{A} TB\{\proj{\Gamma}{\eta_A}\} &\rTo^{Ts} & T^2 \Sigma_AB\{\proj{\Gamma}{\eta_A}\} & \rTo^{\mu} & T\Sigma_AB\{\proj{\Gamma}{\eta_A}\}\\
\Sigma_{TA} {TB} & \rTo^{s} & T\Sigma_{TA} B&\rTo^{Ts'} & T^2 \Sigma_AB\{\proj{\Gamma}{\eta_A}\} & \rTo^{\mu} & T\Sigma_AB\{\proj{\Gamma}{\eta_A}\}.
\end{diagram}
\end{remark}
\begin{theorem}[dCBPV+ Semantics] We have a sound interpretation of dCBPV+ in a dCBPV+ model. The interpretation in such categories is complete in the sense that an equality of values or computations holds in all interpretations iff it is provable in the syntax of dCBPV+. In fact, the interpretation defines a 1-1 correspondence between categorical models and syntactic theories in dCBPV+ which satisfy mutual soundness and completeness results.
\end{theorem}
\begin{proof}
This follows from theorem \ref{thm:dcbpv-soundcomplete} together with the observation that we can interpret the rule of figure \ref{fig:depklext} by dependent Kleisli extensions combined with composition. Conversely, we can apply the rule of figure \ref{fig:depklext} with $\Gamma,x:UFA;\cdot \vdash \force x:FA$ for $M$ to derive the rule for dependent Kleisli extensions.
\end{proof}

\begin{theorem}[Dependent CBN Semantics {2}] The (semantic equivalent of the) CBN translation of DTT with $\Sigma_{1\leq i\leq n}$-, $1$-, $\Sigma$-, $\Id$-, $\Pi_{1\leq i\leq n}$-, $\Pi$-types, where we use the strong (dependent) elimination rules for all positive connectives, into dCBPV+, lets us construct a categorical model of CBN dependent type theory with the connectives above out of any model of dCBPV+ by taking the co-Kleisli category for $!=FU$. The interpretation of CBN dependent type theory is sound and complete for the equational theory induced from dCBPV+:\\
\resizebox{\linewidth}{!}{
$
\sem{\ct{B_1},\cdots,\ct{B_n}\vdash \ct{B}}=\Dcat(U\sem{\ct{B_1}}.\cdots.U\sem{\ct{B_n}})(F1,\sem{\ct{B}})\cong\Dcat_{!}(U\sem{\ct{B_1}}.\cdots.U\sem{\ct{B_n}})(\top,\sem{\ct{B}}).
$
}
\end{theorem}
\begin{theorem}[Dependent CBV Semantics] The (semantic equivalent of the) CBV translation of DTT with $\Sigma_{1\leq i\leq n}$-, $1$-, $\Sigma$-, $\Id$-, $\Pi_{1\leq i\leq n}$-, $\Pi$-types, where we use the strong (dependent) elimination rules for all positive connectives, into dCBPV+, lets us construct a categorical model of CBV dependent type theory with the connectives above out of any model of dCBPV+ by taking the Kleisli category for $T=UF$. The interpretation of CBN dependent type theory is sound and complete for the equational theory induced from dCBPV+:
\begin{align*}
\sem{A_1,\cdots,A_n\vdash A}&=\Dcat(\sem{A_1}. \cdots.\sem{A_n}\{\eta_{\sem{A_1},\ldots,\sem{A_{n}}}\})(F1,F\sem{A}\{\eta_{\sem{A_1},\ldots,\sem{A_{n}}} \})\\
&\cong\Ccat_T(\sem{A_1}. \cdots.\sem{A_n}\{\eta_{\sem{A_1},\ldots,\sem{A_{n}}}\})(1,\sem{A}\{\eta_{\sem{A_1},\ldots,\sem{A_{n}}} \}).
\end{align*}
Here, $\eta_{\sem{A_1},\ldots,\sem{A_{n}}}$ is inductively defined by $$\eta_{\sem{A_1},\ldots,\sem{A_{k}}}:=\qu{\eta_{\sem{A_1},\ldots,\sem{A_{k-1}}}}{\sem{A_k}};\proj{\sem{UFA_1}.\cdots.\sem{UFA_{k-1}}}{\eta_{\sem{A_k}}}.$$
\end{theorem}

\begin{remark}
We have finally arrived at a notion of a model for CBV dependent type theory. It seems much less straightforward than the corresponding notion of a model for CBN dependent type theory as a particular kind of model of pure dependent type theory in which the $\eta$-laws for positive connectives may fail. Then again, a similar phenomenon is already seen in the simply typed case.
\end{remark}

\subsection{Some Basic Models and Non-Models}\label{sec:dcbpvplusmod}
\begin{mccorrection}
As for dCBPV-, we can note that the identity adjunction on any model of pure DTT (in particular, the families of sets model) gives a model of dCBPV+, which demonstrates consistency.
\begin{theorem}[Consistency] dCBPV+ is consistent both in the sense that not all terms are identified and in the sense that not all types are inhabited.
\end{theorem}
 Indeed, the identity monad on any model of DTT trivially admits dependent Kleisli extensions.
\end{mccorrection}

However, as we shall see, it is not the case that any model of dCBPV- extends to a model of dCBPV+. In particular, not every indexed monad on a model of pure DTT admits dependent Kleisli extensions. As it turns out, the existence of dependent Kleisli extensions needs to be assessed on a case-by-case basis. As we shall see, in the case of various set-theoretic models, dependent Kleisli extensions naturally lead to certain subtyping conditions as a necessary requirement which can't always be satisfied.
Therefore, we treat some dCBPV- models for common effects and discuss the (im)possibility of dependent Kleisli extensions.

\subsubsection{A Non-Model and A Model: Writing}
We let $\Bcat$ be $\mathsf{Set}$ and $\Ccat$ be $\Fam(\mathsf{Set})$. Let $M$ be a non-trivial monoid, for instance a monoid of strings of ASCII characters. Then, we let $\Dcat$ be the Eilenberg-Moore category for the indexed monad $-\times M$. Now, we note that dependent Kleisli extensions do not have a sound interpretation in this model of dCBPV-. Indeed, it would amount to giving appropriate maps
\begin{diagram}
\Fam(\mathsf{Set})(\Gamma.A)(1,B\{\langle \id_\Gamma,\id_A,1_M\rangle \}\times M)&\rTo^{(-)^*}&\Fam(\mathsf{Set})(\Gamma.(A\times M))(1,B\times M)\\
\Pi_{\langle c,a\rangle \in \Gamma.A}B(c,a,1_M)\times M&\rTo^{(-)^*}&\Pi_{\langle c,a,m\rangle \in \Gamma.(A\times M)}B(c,a,m)\times M\\
f=\langle f_B,f_M\rangle & \rMapsto &\lambda_{c,a,m} \langle ?,f_M(c,a)*m\rangle .
\end{diagram}
We see that this is not always possible. For instance, let $\Gamma=1=A$ and let $B$ be a predicate that expresses that $m=1_M$ (a predicate which says that no printing happens). In that case, any $f^*$ cannot be a total function as it cannot send, for instance, $(*,*,\textnormal{\texttt{hello world}})$ anywhere.

We would like to stress that this does not show that dependent Kleisli extensions are incompatible with printing. Indeed, it only shows that this particular model of printing does not admit dependent Kleisli extensions. One could conceive of, for example, a model of printing where types depending on $TA$ are not allowed to refer to what is being printed, in which case we could define $f^*(c,a,m):=\langle f_B(c,a),f_M(c,a)*m\rangle $. More generally, such a definition could work if, for all $m\in M$, $B(c,a,1_M)\subseteq B(c,a,m)$.

A concrete instantiation of this idea can be given by considering the setoid model of dependent type theory instead \cite{streicher1993investigations}, which has as objects sets with an equivalence relation, as morphisms functions which send equivalent elements to equivalent elements and as dependent types equivalence respecting families. Note that any monoid $M$ in $\mathsf{Set}$ can be equipped with the codiscrete equivalence relation which identifies all elements to give a monoid internal to the category of setoids. This, in turn, defines an indexed monad $-\times M$ on the setoid model of type theory, which lets us model printing. Note that in this case, predicates cannot distinguish between functions with the same input output behaviour but different printing behaviour. The result is a model of dCBPV+ which models printing (which happens to have intensional $\Id$-types). \mccorrect{However, note that from the point of view of the identity types, $M$ only appears to have one element (although the judgemental equality can distinguish between the elements of $M$).}

\subsubsection{A Non-Model: Reading}
We let $\Bcat$ be $\mathsf{Set}$ and $\Ccat$ be $\Fam(\mathsf{Set})$. Let $S$ be some non-trivial set (that is, not $0$ or $1$), which we think of as a set of states for a storage cell. Then, we let $\Dcat$ be the Eilenberg-Moore category for the indexed monad $(-)^S$. Now, we note that dependent Kleisli extensions do not have a sound interpretation in this model of dCBPV-. Indeed, it would amount to giving appropriate maps
\begin{diagram}
\Fam(\Set)(\Gamma.A)(1,B\{\lambda_s\langle \id_\Gamma,\id_A\rangle \}^S)&\rTo^{(-)^*}&\Fam(\Set)(\Gamma.(A^S))(1,B^S)\\
\Pi_{\langle c,a\rangle \in \Gamma.A}B(s\mapsto \langle c,a\rangle)^S&\rTo^{(-)^*}&\Pi_{(s\mapsto\langle c,a_s\rangle )\in \Gamma.(A^S)}B(s\mapsto \langle c,a_s\rangle)^S\\
f & \rMapsto &\lambda_{s\mapsto \langle c,a_s\rangle}\lambda_{s'} ?.
\end{diagram}
We see that this is not always possible. For instance, let $\Gamma=1$ and $A=2$ and let $B$ be a predicate that expresses that $s\mapsto \langle *,a_s\rangle  $ is constant. In that case, any $f^*$ cannot be a total function as it cannot send a non-constant $s\mapsto \langle *,a_s\rangle$ anywhere.

If we want to define, as usual, $f^*(s\mapsto \langle c,a_s\rangle )(s'):=f(c,a_{s'})(s')$, we require that for all fixed $s'\in S$, $B(s\mapsto \langle c,a_{s'}\rangle)\subseteq B(s\mapsto \langle c,a_s\rangle )$, which is easily seen to be equivalent to $B$ being constant on $A^S$.

\subsubsection{A Non-Model: Global State}
Similarly, for global state, we let $\Bcat$ be $\Set$ and $\Ccat$ be $\Fam(\Set)$ and we take $T:= (-\times S)^S$, where $S$ is a non-trivial set, and let $\Dcat$ be the Eilenberg-Moore category for $T$. Then, dependent Kleisli extensions would amount to appropriate maps\\
\\
\begin{mccorrection}
\resizebox{\linewidth}{!}{
\mbox{\begin{diagram}
\Fam(\Set)(\Gamma.A)(1,(B\{\lambda_s\langle \id_\Gamma,\id_A\rangle \}\times S)^S)&\rTo^{(-)^*}&\Fam(\Set)(\Gamma.((A\times S)^S))(1,(B\times S)^S)\\
\Pi_{\langle c,a\rangle \in \Gamma.A}(B(s\mapsto \langle c,a,s\rangle)\times S)^S&\rTo^{(-)^*}&\Pi_{(s\mapsto\langle c,a_s,t_s\rangle )\in \Gamma.((A\times S)^S)}(B(s\mapsto \langle c,a_s,t_s\rangle)\times S)^S\\
f & \rMapsto &\lambda_{s\mapsto \langle c,a_s,t_s\rangle}\lambda_{s'} ?.
\end{diagram}}}\end{mccorrection}\quad\\
\\
Now, $B$ could express the property that $a_s=a$ ($a_s$ is independent of $s$) and $t_s=s$. In that case, no such dependent Kleisli extension exists.

One could imagine a different model of global state, however, in which, for every fixed $s'\in S$, $B(s\mapsto \langle c,a_{s'},s\rangle )\subseteq B(s\mapsto \langle c,a_s,t_s\rangle)$. In that case, one could define as one normally (for non-dependent Kleisli extensions) would $f^*(s\mapsto \langle c,a_s,t_s\rangle )(s'):=f(c,a_{s'})(s')$. At present, it is not clear to the author if non-trivial models along these lines exist.

\subsubsection{A Model: Exceptions or Divergence}
We consider a model for exceptions or divergence, where we use the monad $T=E+(-)$  on $\Fam(\Set)$, for some fixed set $E$ whose elements we think of as either exceptions or, perhaps, in the case of $E=1$, divergence. We let $\Bcat$ be $\Set$ and $\Ccat$ be $\Fam(\Set)$ and we take for $\Dcat$ the Eilenberg-Moore category for $T$. In this case, we in fact have maps
\begin{diagram}
\Fam(\Set)(\Gamma.A)(1,E+B\{\langle \id_\Gamma,\inr\rangle\})&\rTo^{(-)^*}&\Fam(\Set)(\Gamma.(E+A))(1,E+B)\\
\Pi_{\langle c,a\rangle \in \Gamma.A}E+B(c,\inr\;a)&\rTo^{(-)^*}&\Pi_{\langle c,t\rangle \in \Gamma.(E+A)}E+B(c,t)\\
f & \rMapsto & \lambda_{c}[\inl ,f(c,-)].
\end{diagram}
These are easily seen to give a sound interpretation of dependent Kleisli extensions. They indeed model the propagation of exceptions one would expect.

\subsubsection{A Dubious Model: Erratic Choice}
We consider a model for erratic choice, where we use the powerset monad $T=\p$  on $\Fam(\Set)$. We let $\Bcat$ be $\Set$ and $\Ccat$ be $\Fam(\Set)$ and we take for $\Dcat$ the Eilenberg-Moore category for $T$. Dependent Kleisli extensions would amount to appropriate maps
\begin{diagram}
\Fam(\Set)(\Gamma.A)(1,\p B\{\langle \id_\Gamma,x\mapsto \{x\} \rangle\})&\rTo^{(-)^*}&\Fam(\Set)(\Gamma.(\p A))(1,\p B)\\
\Pi_{\langle c,a\rangle \in \Gamma.A}\p B(c,\{a\})&\rTo^{(-)^*}&\Pi_{\langle c,t\rangle \in \Gamma.(\p A)}\p B(c,t)\\
f & \rMapsto & \lambda_{c,t}?.
\end{diagram}
We can, in principle, define $f^*(c,t):=(\bigcup_{a\in t} f(c,a))\cap B(c,t)$ to obtain a dependent Kleisli extension. However, this model might not correspond to the expected operational semantics. It would be preferable to consider, instead, a model of type theory $\Ccat$ in which it is always the case that $\bigcup_{a\in t} B(c,\{a\})\subseteq  B(c,t)$, in which case we can just define $f^*(c,t):=\bigcup_{a\in t}f(c,a)$ (cf. reader monad). At present, it is not clear to the author how a model with such properties can be constructed.

\subsubsection{A Puzzle: Control Operators}
We consider a dCBPV- model for control operators, where we use a continuation monad $T=R^{(R^-)}$  on $\Fam(\Set)$, for some non-trivial set $R$. We let $\Bcat$ be $\Set$ and $\Ccat$ be $\Fam(\Set)$ and we take for $\Dcat$ the Eilenberg-Moore category for $T$. Dependent Kleisli extensions would amount to appropriate maps
\begin{diagram}
\Fam(\Set)(\Gamma.A)(1, (R^{(R^{B\{\langle \id_\Gamma,x\mapsto \mathsf{ev}_x \rangle\}})}) &\rTo^{(-)^*}&\Fam(\Set)(\Gamma.( R^{(R^A)}))(1,R^{(R^B)})\\
\Pi_{\langle c,a\rangle \in \Gamma.A}(R^{(R^{B(c,\mathsf{ev}_a)})})&\rTo^{(-)^*}&\Pi_{\langle c,t\rangle \in \Gamma.(R^{(R^A)})}(R^{(R^{B(c,t)})})\\
f & \rMapsto & \lambda_{c,t}?.
\end{diagram}
In order to match the expected operational semantics, it is tempting to try to define, just as in the simply typed case, $f^*(c,t)(k):=t(\lambda_a f(c,a)(k))$. However, this is only well-defined if we have $\forall_{a\in A(c)}R^{B(c,t)}\subseteq R^{B(c,\mathsf{ev}_a)}$. In particular, we would have that \mccorrect{$B(c,\mathsf{ev}_a)=B(c,\mathsf{ev}_{a'})$} for all $a,a'\in A(c)$. This suggests a kind of incompatibility between control operators and dependent Kleisli extensions. We would like to further investigate the combination of dCBPV with control operators in future work, especially given the correspondence with classical logic. In the light of \cite{herbelin2005degeneracy}, we already know that the combination of dependent types and control operators can easily lead to degeneracy of the system (in the sense that all programs get equated \mccorrect{propositionally}).

\subsubsection{A Model: Recursion}
Note that the model of the dependent LNL calculus of section \ref{sec:lindepscott} in particular gives us a model of dCBPV-. The model clearly supports recursion, as we can define our usual fixpoint combinators. This model is easily seen further to support dependent Kleisli extensions: similar to our previous model of divergence, for a dependent function $f$, we define the Kleisli extension $f^*$ as sending the new bottom element to bottom and otherwise acting as $f$.

\subsection{Operational Semantics and Effects}
Using the CK-machine, we can again define an operational semantics for dCBPV+.

The definition of the operational semantics does not change in the presence of dependent Kleisli extensions and is exactly as that described in section \ref{sec:depop}. In particular, figures \ref{fig:ckmachine} and \ref{fig:ckid} define a CK-machine on which we evaluate the computations of pure dCBPV+. As before, we can add the effects of figure \ref{fig:effects} together with their operational semantics of figures \ref{fig:opsemdivs} and \ref{fig:opsemprint} and equations of figure \ref{fig:effeqn}.  We still have the same determinacy and strong normalization results as before, as the essentially untyped proofs remain valid.
\begin{theorem}[Determinacy, Strong Normalization] No transition occurs precisely if we are in a terminal configuration. In absence of erratic choice, at most one transition applies to each configuration. In absence of divergence and recursion, every configuration reduces to a terminal configuration in a finite number of steps.
\end{theorem}
However, the results of section \ref{sec:dcbpvplusmod} are reflected at the level of the operational semantics. While for some effects like divergence, exceptions and recursion, subject reduction can be established, certain subtyping conditions are necessary to obtain subject reduction in the presence of printing, global state and erratic choice. It is at present not clear if a these conditions are compatible with, for instance, $\cpi{-}{}$-types.
\begin{theorem}[Limited Subject Reduction]\label{thm:subjreddcbpv+} In absence of printing, global state and erratic choice, if the sequence of reductions of a well-typed computation $M$ passes through a well-typed configuration $M,K$ and \mccorrect{later} another configuration $M',K$, then the latter configuration is also well-typed and has the same type as the former. \end{theorem}
\begin{proof}[Proof] It is easy to see that all transitions preserve the type of a configuration as for dCBPV-, with the exception of the transitions for $\toin{M}{x}{N}$ and $\return \nf{V}$. Both involve a stack $\toin{[\cdot]}{x}{N}::K$ which is untypable when $x$ is free in the type of $N$. The crux, however, is that these transitions always occur in pairs and, in this case, two wrongs make a right. Say that we are evaluating a well-typed computation and that the former transition occurs from $\toin{M}{x}{N},K$ to \mccorrect{$M, \toin{[\cdot ]}{x}{N}::K$} where $\Gamma;\cdot\vdash \toin{M}{x}{N}:\ct{B}[\thunk M/z]$. After that, several other transitions may occur, but if we return to another configuration $M',K$, we know that the last transition that occurred was that from $\return \nf{V},\toin{[\cdot]}{x}{N}::K\leadsto N[\nf{V}/x],K$ (and so, $M'=N[\nf{V}/x]$).

Our claim is that also $\Gamma;\cdot\vdash N[\nf{V}/x]: \ct{B}[\thunk M/z]$ (if so, then the theorem follows). The important thing to notice is that, from inversion on $\Gamma;\cdot\vdash \toin{M}{x}{N}:\ct{B}[\thunk M/z]$, we have that $\Gamma,x:A;\cdot \vdash N:\ct{B}[\thunk \return x/z]$. Therefore, it follows that $\Gamma;\cdot\vdash N[V/x]:\ct{B}[\thunk\return V/z]$. Our claim follows by noting that $\return V= M$ as all reductions that could have been applied to $M$ are also equalities (seeing that the only effects we allow are recursion, divergence and errors, all of whose transitions are equations, as are $\beta$-reductions).\end{proof}
The proof above shows why subject reduction may fail in the presence of printing, global state and erratic choice: their transitions $M,K\leadsto M',K$ of figures \ref{fig:opsemdivs} and \ref{fig:opsemprint} on computations are not contained in the judgemental equalities we consider (see figure \ref{fig:effeqn}) in the sense that not $M=M'$. They represent real dynamics. In this sense, they differ from the other transitions we have considered. In fact, it is not reasonable to demand such an equality. In particular, in the case of reading global state and erratic choice, that would lead to all computations of the same type being equated.

\begin{remark}
On closer inspection, however, it seems that what we really needed to establish subject reduction was an inclusion of computation types, whenever $M,K,m,s\leadsto M',K,m',s'$,\\
\\
\AxiomC{$\Gamma;\Delta\vdash K:\ct{B}[\thunk M'/z]$}
\UnaryInfC{$\Gamma;\Delta\vdash K:\ct{B}[\thunk M/z].$}
\DisplayProof
\\\\
The idea is that the type of a computation becomes more determined in the computation progresses. We list concrete instantiations of this rule  in figure \ref{fig:dcbpvplussub}.
\begin{figure}[!tb]
\fbox{
\resizebox{\linewidth}{!}{
\begin{tabular}{ll}
\AxiomC{$\Gamma;\Delta\vdash K:\ct{B}[\thunk M/z]$}
\UnaryInfC{$\Gamma;\Delta\vdash K:\ct{B}[\thunk (\print{m}{M})/z]$}
\DisplayProof\hspace{40pt}\;
&
\AxiomC{$\Gamma;\Delta\vdash K:\ct{B}[\thunk M_{i'}/z]$}
\UnaryInfC{$\Gamma;\Delta\vdash K:\ct{B}[\thunk (\nondet{i}{M_i})/z]$}
\DisplayProof\hspace{40pt}\;\\
&\\
\AxiomC{$\Gamma;\Delta\vdash K:\ct{B}[\thunk M)/z]$}
\UnaryInfC{$\Gamma;\Delta\vdash K:\ct{B}[\thunk (\writecell{s}{M})/z]$}
\DisplayProof
&
\AxiomC{$\Gamma;\Delta\vdash K:\ct{B}[\thunk M_{s'}/z]$}
\UnaryInfC{$\Gamma;\Delta\vdash K:\ct{B}[\thunk (\readcell{s}{M_s})/z]$}
\DisplayProof
\end{tabular}
}
}
\caption{\label{fig:dcbpvplussub} Extra rules that are necessary in dCBPV+ to establish subject reduction in the presence of printing, global state and erratic choice.}
\end{figure}
It is clear that admissability of these rules is a necessary condition to  establish subject reduction property for dCBPV+ (compare this to the results in section \ref{sec:dcbpvplusmod}!). What is less clear, is if as adding them to the type system is sufficient, as this complicates the usual subject reduction proof, which relies on inversion on the typing rules.
\end{remark}

\section{Dependent Projection Products?}\label{sec:depprojprod}
It was somewhat surprising, perhaps, that while dependent pattern matching products arise so naturally in CBPV, dependent projection products seem less natural. The reader should compare this to the status of additive $\Sigma$-types, their cousins in linear logic, which often fail to be supported in natural models. In principle, we could include the system of rules of figure \ref{fig:addsigma} in dCBPV to replace $\Pi_{1\leq i\leq n}$-types.
\begin{figure}
[!tb]
\fbox{
\resizebox{\linewidth}{!}{
\begin{tabular}{ll}
\AxiomC{$\vdash \Gamma,z_1:U\ct{B_1},\ldots,z_n:U\ct{B_n}\ctxt$}
\UnaryInfC{$\Gamma\vdash \Pi_{1\leq i\leq n}^{z_1,\ldots,z_n} \ct{B_i}\ctype$}
\DisplayProof\hspace{30pt}\;
&
\AxiomC{$\{\Gamma;\cdot\vdash M_i:\ct{B_i}[\thunk M_1/z_1,\ldots,\thunk M_{i-1}/z_{i-1}]\}_{1\leq i\leq n}$}
\UnaryInfC{$\Gamma ;\cdot\vdash \lambda_i M_i : \Pi_{1\leq i\leq n}^{z_1,\ldots,z_n}\ct{B_i}$}
\DisplayProof\hspace{30pt}\;\\
&\\
&
\AxiomC{$\Gamma;\cdot\vdash M:\Pi_{1\leq i \leq n}^{z_1,\ldots,z_n} \ct{B_i}$}
\UnaryInfC{$\Gamma;\cdot\vdash i\textquoteleft M: \ct{B_i}[\thunk 1\textquoteleft M/z_1,\ldots,\thunk (i-1)\textquoteleft M/z_{i-1}]$}
\DisplayProof
\end{tabular}
}
}
\caption{\label{fig:addsigma} Rules for dependent projection products. We also demand the obvious $\beta$- and $\eta$-laws.}
\end{figure}
This allows us to define the appropriate CBV and CBN translations for dependent projection products in dCBPV, exactly as one defines the translation for simple projection products. This translation re-enforces the idea that the CBV translation of a type $x_1:A_1,\ldots,x_n:A_n\vdash A\type$ should be $z_1:UFA_1^v,\ldots,z_n:UFA_n^v\vdash A^v\vtype $. We note that we have CBV and CBN translations of dependent projection products (which have a dependent/strong elimination principle) even in dCBPV-. Moreover, we can use the usual operational semantics of computations of type $\Pi_{1\leq i\leq n}\ct{B_i}$ for these types.

Although we can formulate a sound and complete categorical semantics for dependent projection products (we demand strong $n$-ary $\Sigma$-types in $\Dcat$ in the sense of objects $\Pi_{1\leq i\leq n}^{dep}\ct{B_i}$ such that $\proj{\Gamma}{U\Pi_{1\leq i\leq n}^{dep}\ct{B_i}}=\proj{\Gamma.U\ct{B_1}.\cdots.U\ct{B_{n-1}}}{U\ct{B_n}};\ldots;\proj{\Gamma}{U\ct{B_1}}$), many models fail to support these connectives in practice. In particular, they are hard to obtain in models of linear logic, where they would give additive $\Sigma$-types in the sense of objects $\Sigma_A^{\&}B$ such that $!\Sigma_A^{\&}B\cong \Sigma_{!A}^\otimes !B$, and are difficult to give a satisfactory interpretation in models of the monadic metalanguage, where they would correspond to the construction of a $T$-algebra structure on $\Sigma_{Uk} {Ul}$, given $l\in \Ccat^T(\Gamma.Uk)$.

A related phenomenon is that subject reduction for dependent projection products can be problematic to establish (for the obvious operational semantics on untyped terms which is identical to that for $\Pi_{1\leq i\leq n}$-types). We encourage the reader to think of dependent projection products in a similar way to dependent Kleisli extensions, as their problems with subject reduction have a similar origin. That is, they lead to types depending on (thunks of) computations which might not be static objects during reduction (in the sense that some reductions might not be equalities in the presence of some effects). In that case, we are faced with a choice, either subject reduction fails or we have to make types into dynamic objects as well, meaning that they no longer provide the static guarantees which are their primary raison d'\^etre.

\begin{theorem}[Limited Subject Reduction]
Let us consider dCBPV- with dependent projection products. In absence of printing, global state and erratic choice, if the sequence of reductions of a well-typed computation $M$ passes through a well-typed configuration $M,K$ and \mccorrect{later} another configuration $M',K$, then the latter configuration is also well-typed and has the same type as the former. 
\end{theorem}
\begin{proof}
The proof is very similar to that of theorem \ref{thm:subjreddcbpv+}. Indeed, starting from a well-typed configuration $\Gamma;\cdot \vdash j\textquoteleft M:\ct{B}_j[\thunk 1\textquoteleft M/z_1,\ldots,\thunk (j-1)\textquoteleft M/z_{j-1}] ,\linebreak \Gamma;{\nil}:\ct{B}_j[\thunk 1\textquoteleft M/z_1,\ldots,\thunk (j-1)\textquoteleft M/z_{j-1}] \vdash K:\ct{C}$, we transition into $M, j::K$, where $j::K:= \lbi{\nil_1}{j\textquoteleft \nil_2}{K}$ is an untypable stack. Eventually, if have transitioned into a configuration $\lambda_i M_i, j::K$, we next transition into $M_j,K$, where we can derive by inversion on $\lambda_i M_i$ that $\Gamma;\cdot\vdash M_i:\ct{B}_i[\thunk M_1/z_1,\ldots, \thunk M_{i-1}/z_{i-1}]$. We  have now arrived at a well-typed configuration again if we can show that $\ct{B}_j[\thunk M_1/z_1,\ldots, \thunk M_{j-1}/z_{j-1}]=\ct{B}_j[\thunk 1\textquoteleft M/z_1,\ldots,\thunk (j-1)\textquoteleft M/z_{j-1}]$. This follows if we can show that $M_i = i\textquoteleft M$ for all $1\leq i\leq j-1$, which we know to be generally \mccorrect{true} in absence of of printing, global state and erratic choice and false otherwise. 
\end{proof}

One could add similar negative versions of the other positive connectives like identity types (which we have called additive identity types in the context of linear logic). Their categorical semantics would correspond to having computation type formers $R(\ct{B_1},\ldots,\ct{B_n})$ that $U$ maps to $R'(U\ct{B_1},\ldots, U\ct{B_n})$ where $R'$ is the corresponding positive type former. In the obvious operational semantics, destructors push to the stack and constructors pop the stack and substitute.

Let us briefly consider some specific models. We have already seen in section \ref{sec:lindepscott} that the domain model of dCBPV+ supports additive $\Sigma$-types as well. Similarly, the error monad admits a satisfactory definition of dependent projection products. We define the algebra structure $\Sigma_k l$, as expected, by $(\Sigma_k l)(e):=\langle k(e),l(k(e))(e)\rangle$.

Note that for the writer monad, we cannot use the expected generalisation of the product algebra structure on $\Sigma_{Uk} Ul$. Instead, we can use the trivial algebra structure. (Note that dependent projection products are only a generalisation of the product in a weak sense. In particular, they are far from unique.) That may not be what we are hoping for though, as the individual algebra structures $k$ and $l$ on $Uk$ and $Ul$ are ignored in the construction. (The product action may not respect the fibres of the $\Sigma$-type if $l$ is not invariant under $k$.) Similarly, for the reader and global state monads, we can equip $\Sigma_{Uk} Ul$ with an algebra structure by evaluating at an arbitrary state. Again, similarly, note that an algebra for the powerset monad is a join-semi-lattice. Therefore, assuming the axiom of choice (or rather, its equivalent, the well-ordering principle), we can define dependent projection products by equipping $\Sigma_{Uk} Ul$ with some well-order. However, this is not the algebra structure we would be expecting (the product order), as this may fail to be a join-semi-lattice (take, for instance, $k=\{0\leq 1\}=l(0)$ and $l(1)=\{0\}$; then $\langle 0,1\rangle$ and $\langle 1,0\rangle $ do not have a join).


We have already seen in section \ref{sec:depcohsp} that additive $\Sigma$-types are not always supported in models of linear dependent type theory.

\mccorrect{
As projection products are more natural than pattern matching products in CBN (or, at least,
more customary), we see that the CBN-translation into dCBPV runs into similar problems as the CBV. Where the latter requires us to extend dCBPV- with dependent Kleisli extensions, the former at least strongly suggests adding dependent projection products to dCBPV-}\footnote{
\mccorrect{
For instance, we require these if we want the co-Kleisli category for $!=FU$ to give a model of DTT indexed over $\Dcat(\cdot)_!$ (recall that this category is equivalent to full category of $\Bcat$ on the objects in image of $U$), rather than merely over $\Bcat$.
}
}. \mccorrect{Both these extensions lead to similar challenges with constructing concrete models and with establishing subject reduction.}

\section{Dependent Kleisli Extensions: a Bug or a Feature?}
\label{sec:depklextbugorfeature}
\subsection{Unrestricted Effects and Dependent Types?}
In sections \ref{sec:depcbpvwoklext} and \ref{sec:depcbpvklei}, we introduced two systems which combine dependent types with computational effects: dCBPV- and dCBPV+. Recall that the latter extends the former with a rule for Kleisli extensions of dependent functions.

One motivation for studying dCBPV+ is the possible criticism that can be made that in dCBPV- dependent types and effects sit side-by-side and do not interact meaningfully. (More about that later.) Another is the observation that we need dependent Kleisli extensions to obtain a well-defined CBV or CBN translation of dependent type theory with unrestricted effects (which are not encapsulated by the type system) into dCBPV. This may be interesting as real world languages like Agda and Idris include unrestricted recursion and only exclude non-terminating terms at a later stage through a termination check which is separate from type checking \cite{bove2009dependent,norell2007towards}. Moreover, models of dependent type theory in categories of domains and games\mccorrect{ }naturally model unrestricted effects. Therefore, we believed it to be important, from a theoretical point of view at least, to study the system dCBPV+.

It should be clear to the reader, however, that the expressive power of dCBPV+ comes at a cost of simplicity and, in particular, in many cases, subject reduction. The question is if dCBPV- or dCBPV+ is more suitable for practical implementations. We would like to argue that dCBPV+ does not add much practical value over dCBPV- as a programming language.

Indeed, let us return to the primary practical motivation for wanting to combine dependent types and effects: having a single elegant language in which we can both write practical software and perform its verification. For these purposes, as argued in section \ref{sec:effmodalformulae}, it is crucial that we constrain where effects are allowed to occur using the type system, for instance using modalities, as effects usually do not correspond to sound logical principles so should be excluded from proofs. For this reason, dependent type theory with unrestricted effects (and with it the corresponding CBV and CBN translations) is not what we are most interested in. Rather, a modal and ideally adjunction language like dCBPV is closer to what we are looking for.

\subsection{Fundamentalist vs Pragmatic Dependent Types}\label{sec:fundamentalist}
Observe that, in practice, dependent types tend to be used in two closely related but slightly different styles\footnote{These can be seen to be closely related to the two traditional schools of thought about types: typing \`a la Church and \`a la Curry \cite{barendregt2013lambda}. Also closely related to the fundamentalist school is the correctness-by-construction programming methodology advocated by Dijkstra and others \cite{kourie2012correctness}.}. On the one hand we have a style of programming where we build up the program immediately from dependently typed building blocks $c:C$, where $C$ may be formed using inductive families and other dependently typed constructs, by writing the code and proving its properties simultaneously, fundamentalist dependently typed programming, if you will. Some examples for  $C$ include a type of lists of a fixed length, sorted lists, heaps or binary search trees, red-black trees, suitably balanced trees and a type of $\lambda$-terms up to $\beta\eta$-equality. On the other hand, we can write simply typed programs $a:A$ first, where $A$ is a datatype formed from simple inductive types and simple connectives, like mere lists or binary trees, and only later prove the required properties $a': A'[a/x]$ (where $A' $ is a proof-relevant predicate, like the BST property, formed using inductive families and other dependently typed constructions). This is a more pragmatic stance where dependent types are simply seen as a tool for expressing appropriate program properties that we want to verify.

The latter style seems to be more popular in practice and more suitable for the creation of large modular code bases. The reason for this is that we often only decide on the properties we need to verify after we have already written code and, in fact, we often need to verify different properties of the same code in different contexts. It should be noted that both points of view are equivalent, that the distinction is mostly a matter of style and that both styles can be combined well. 

To illustrate the distinction, let us consider lists of length $n$. We can either view them directly as single inductive family $x:\N\vdash \mathsf{ListOfLen}(x)\vtype$ (corresponding to the former style) or as a predicate $x:\N, y:\mathsf{List}\vdash \mathsf{has-length}(y,x)\vtype$ from which we form $x:\N\vdash \Sigma_{y:\mathsf{List}}\mathsf{has-length}(y,x)\vtype$ (corresponding to the latter). If we want to write a program that for every $x:\N$ returns a list of length $x$, we write, in the fundamentalist view, directly, proof-carrying code $\vdash f:\Pi_{x:\N}\mathsf{ListOfLen}(x)$, which can be though of as both an algorithm producing a list and a proof that that list has length $x$ in one. Meanwhile, in the more pragmatic view of post hoc verification using dependent types, we first write the algorithm $\vdash g:\N\Rightarrow\mathsf{List}$ and then a separate proof $\vdash  p:\Pi_{x:\N} \mathsf{has-length}(g(x),x)$ about $g$.

The latter point of view generalises without problems to dCBPV-. Indeed, by keeping the (simply typed and effectful) algorithm separate from the (dependently typed and pure) proof, all we need is a sequencing operation from simply typed effectful computations and a regular composition operation for pure dependent functions, both of which are available in dCBPV-. To be precise, we write an effectful simply typed algorithm $\vdash g:\N\functype F\mathsf{List}$ and a separate pure dependently typed proof $\vdash  p:\mathsf{my-favourite-property}(\thunk g)$ about $g$. Here, we would like to point out that $z:U(\N\functype F\mathsf{List})\vdash \mathsf{my-favourite-property}(z)\vtype$. Therefore, to make dCBPV- into a practical system for verification of effectful programs, it is crucial that we extend it with mechanisms for defining interesting (value) types depending on types of thunks of effectful computations. In the case that we are working with printing with values in some monoid $M$ internal to the type theory, a simple example of a property to express using a type family could be that the program does not print anything or that its return value has a specific property. In particular, any type depending on $A\times M$ should give rise to a type depending on $UFA$. The design of good mechanisms for defining types depending on types of thunks of effectful computations is planned to be a central theme in our future work.

The former point of view, however, is more difficult to generalise without dependent Kleisli extensions, it would seem at first sight. Indeed, if our basic building blocks are dependent effectful functions $\vdash f:\cpi{x:\N}F\mathsf{ListOfLen}(x)$, we want to be able to compose them with each other, at the very least. In particular, we want to be able to precompose $f$ with some effectful function $\vdash h:\N\functype F\N $. To do this, however, we precisely need a sequencing principle for dependent effectful functions, or a principle of dependent Kleisli extensions. As it turns out, compositionality in this paradigm can be restored to a satisfying extent by considering dCBPV- with $\csigma{-}{}$-types, a much less intrusive and more well-behaved extension than dependent Kleisli extensions. We discuss this in section \ref{sec:moreconn}.

On the whole, we are inclined to view dependent Kleisli extensions as technical devices that were important to study for theoretical reasons, but which may not be very suitable for practical implementations of dCBPV. The extra complexity they introduce into the implementation of a type checker for may not be justified. Therefore, in the rest of this chapter, we  focus on dCBPV- and  add extra type formers to it to increase its expressive power.
 
\section{Dependent Enriched Effect Calculus and More Connectives}\label{sec:moreconn}
In this section, we show how to increase the power of dCBPV- by extending it with $\Pi$-, $\csigma{-}{}$- and $\homtype$-types. First we motivate why we might want to include them in our calculus. Next, we show that they are unproblematic from the points of view of categorical semantics, concrete models and operational semantics. 

Levy did not include function type formers for value types in his CBPV as he was mainly interested in (the CBV and CBN translations for) effectful programs, for which they are unnecessary. We, however, are also interested in pure proofs of universally quantified formulas. For those purposes, value function types are of crucial importance. This leads us to consider $\Pi$-types. 

As discussed in section \ref{sec:fundamentalist},  it is not as important as one might think to be able to substitute effectful computations into dependent functions. However, it might sometimes still be practically convenient. We would like to suggest that $\csigma{-}{}$-types give an alternative, more lightweight method of achieving this compared to dependent Kleisli extensions.

Recall that, given a dependent function $\Gamma,x:A\vdash M:B$ in pure type theory, we can transform it into a simple function $\Gamma,x:A\vdash \langle x,M\rangle:\Sigma_{x:A}B$ by viewing it as a section of $\Gamma,z:\Sigma_{x:A}B\vdash \fst(z):A$. Precomposition with $\Gamma,y:C\vdash N:A$ then gives $\Gamma,y:C\vdash \langle N,M[N/x]\rangle:\Sigma_{x:A}B$. We can employ a similar trick to get around the effectful composition of certain dependent functions in bare dCBPV- already. Indeed, we can represent any effectful dependent function $\Gamma,x:A;\cdot\vdash M:FA'$ as an effectful simple function $\Gamma,x:A;\cdot\vdash \toin{M}{z}{\return\langle x,z\rangle }:F\Sigma_{x:A} A'$. In this representation, we can use usual simple sequencing of effectful computations to achieve effectful precomposition: given $\Gamma,y:C;\cdot\vdash N:FA$, we have the effectful composition $\Gamma,y:C;\cdot\vdash \toin{N}{z}{\toin{M}{z}{\return\langle x,z\rangle }}:F\Sigma_{x:A} A'$ without using dependent Kleisli extensions.

We are in trouble, however, if $M$ is of the more general form $\Gamma,x:A;\cdot\vdash M:\ct{B}$. In order to repeat the trick above, we introduce the type $\csigma{x:A}{\ct{B}} $ to generalise $F\Sigma_{x:A}A'\cong \csigma{x:A}{FA'} $. This lets us define a simply typed effectful function $\Gamma,x:A;\cdot\vdash \return x\otimes M:\csigma{x:A}\ct{B}$ out of $M$ and therefore a precomposition $\Gamma,y:C;\cdot\vdash \toin{N}{x}{\return x\otimes M}:\csigma{x:A} \ct{B}$. The problem with sequencing an effectful computation $N$ into a dependent function $M$ was, essentially, that we do not know what fibre of the return type $\ct{B}$ the result would land in. Indeed, $N$ may, for instance, exhibit non-determinism or use state. $\csigma{-}{}$ solves this problem by bundling all fibres together and saying that we are not interested in the particular fibre it lands is, as long as there is one.

Finally, to reason about effectful programs and their evaluation, it can be very useful to include a type not just of arbitrary functions, but also a type of homomorphisms or stacks. While it is well-known that the sets of homomorphisms for a commutative monad $T$ on a cartesian closed category admit a natural $T$-algebra structure themselves \cite{kock1971closed} (leading us to models of linear logic), it should be familiar from the theory of monoids that such an algebra structure might not be available for non-commutative monads \cite{foltz1980algebraic}. This shows that, in general, for non-commutative effects, we cannot expect the type of homomorphisms/stacks from $\ct{B}$ to $\ct{C}$ to be a computation type itself. Luckily, we can often interpret it as a value type $\ct{B}\homtype\ct{C}$.

We include type and term forming rules for $\Pi$-, $\csigma{-}{}$- and $\homtype$-types in figure \ref{fig:extratypesrules}, their equations in figure \ref{fig:extratypesequations} and their operational semantics in figure \ref{fig:ckextraconn}. We then see that the results on the categorical semantics, concrete models and operational semantics of dCBPV- smoothly extend to these connectives.

\begin{figure}[!tb]
\fbox{
\resizebox{\linewidth}{!}{
\begin{tabular}{ll}
\AxiomC{$\Gamma,x:A\vdash A'\vtype$}
\UnaryInfC{$\Gamma\vdash \Pi_{x:A}A'\vtype$}\DisplayProof
&\\
&\\
\AxiomC{$\Gamma,x:A\vdash V:A'$}
\UnaryInfC{$\Gamma\vdash \lambda_x V:\Pi_{x:A}A'$}
\DisplayProof
&
\AxiomC{$\Gamma\vdash V:A$}
\AxiomC{$\Gamma;\Delta\vdash W:\Pi_{x:A}A'$}
\BinaryInfC{$\Gamma;\Delta\vdash V\textquoteleft W : A'[V/x]$}
\DisplayProof
\\
&\\
\AxiomC{$\Gamma,x:A\vdash \ct{B}\ctype$}
\UnaryInfC{$\Gamma\vdash \csigma{x:A}{\ct{B}}\ctype$}
\DisplayProof
&
\\
&\\
\AxiomC{$\Gamma\vdash V:A$}
\AxiomC{$\Gamma;\Delta\vdash K:\ct{B}[V/x]$}
\BinaryInfC{$\Gamma;\Delta\vdash \return V\otimes K:\csigma{x:A}{\ct{B}}$}
\DisplayProof\hspace{10pt}\;
&
\AxiomC{$\Gamma,x:A;{\nil}:\ct{B}\vdash K:\ct{C}$}
\AxiomC{$\Gamma \vdash \ct{C}\ctype$}
\AxiomC{$\Gamma;\Delta\vdash L:\csigma{x:A}{\ct{B}}$}
\TrinaryInfC{$\Gamma;\Delta\vdash \toin{L}{\return x\otimes {\nil}}{K}:\ct{C}$}
\DisplayProof\hspace{10pt}\;
\\
&\\
\AxiomC{$\Gamma\vdash \ct{B}\ctype$}
\AxiomC{$\Gamma\vdash \ct{C}\ctype$}
\BinaryInfC{$\Gamma\vdash \ct{B}\homtype \ct{C}\vtype$}
\DisplayProof
&\\
&\\
\AxiomC{$\Gamma;{\nil}:\ct{B}\vdash K:\ct{C}$}
\UnaryInfC{$\Gamma\vdash\lambda_{\nil}K:\ct{B}\homtype\ct{C}$}
\DisplayProof
&
\AxiomC{$\Gamma\vdash V: \ct{B}\homtype\ct{C}$}
\AxiomC{$\Gamma;\Delta\vdash K:\ct{B} $}
\BinaryInfC{$\Gamma;\Delta\vdash K\textquoteleft V:\ct{C}$}
\DisplayProof
\end{tabular}
}
}
\caption{\label{fig:extratypesrules} Rules for forming $\Pi$-, $\csigma{-}{}$- and $\homtype$-types and their terms.}
\end{figure}

\begin{figure}[!tb]
\fbox{
\resizebox{\linewidth}{!}{
\begin{tabular}{ll}
$V\textquoteleft \lambda_x W = W[V/x]$ & $V \stackrel{\#x}{=} \lambda_x x\textquoteleft V$\\
$\toin{\return V \otimes K}{\return x \otimes {\nil}}{L}  = L[V/x,K/{\nil}]$ \;& $ K[L/{\nil_1}] \stackrel{\# x,{\nil_2}}{=} \toin{K}{\return x\otimes{\nil_2}}{L[\return x \otimes {\nil_2}/{\nil_1}]} $\\
$K\textquoteleft \lambda_{\nil}L=L[K/{\nil}]$ & $K \stackrel{\#{\nil}}{=} \lambda_{\nil}{\nil}\textquoteleft K$
\end{tabular}
}
}
\caption{\label{fig:extratypesequations}  Equations we impose for the terms of $\Pi$-, $\csigma{-}{}$- and $\homtype$-types.}
\end{figure}
\begin{figure}[!tb]
\fbox{
\resizebox{\linewidth}{!}{
\begin{tabular}{l}
\textbf{Transitions}\\
\begin{tabular}{lllllll}
 $ \toin{M}{\return x \otimes {\nil}}{L}    $ &,& $K$\hspace{40pt}& $\leadsto$ \hspace{40pt}& $M        $ &,& $\toin{[\cdot]}{\return x \otimes {\nil}}{L}::K$ \\
 $ \return \nnf{V} \otimes M    $ &,& $K$\hspace{40pt}& $\leadsto$ \hspace{40pt}& $  \return \nf{V}\otimes M      $ &,& $K$ \\
  $ \return \nf{V} \otimes M    $ &,& $\toin{[\cdot]}{\return x \otimes {\nil}}{L}::K$\hspace{40pt}& $\leadsto$ \hspace{40pt}& $ L[\nf{V}/x,M /{\nil}]     $ &,& $K$ \\
$ M\textquoteleft \nnf{V} $ &,& $K$\hspace{40pt}& $\leadsto$ \hspace{40pt}& $M\textquoteleft \nf{V}       $ &,& $K$ \\
$ M\textquoteleft \lambda_{\nil} L $ &,& $K$\hspace{40pt}& $\leadsto$ \hspace{40pt}& $L[M/{\nil}]        $ &,& $K$ 
\end{tabular}\\
\\
\textbf{Terminal Configurations}\\
\begin{tabular}{lll}
$\return \nf{V}\otimes M$ &,& ${\nil}$\\
$M\textquoteleft\return \nf{V}^{x'}$ &,& $K$
\end{tabular}
\end{tabular}
}
}
\caption{\label{fig:ckextraconn} The additional transitions and terminal configurations that specify the operational behaviour of terms of $\Pi$-, $\csigma{-}{}$- and $\homtype$-types. Here, we use the abbreviation $\toin{[\cdot]}{\return x \otimes {\nil}}{L}::K$ for $\lbi{\nil_1}{\toin{\nil_2}{\return x\otimes {\nil_3}}{L}}{K}$. Note that the transitions for $\Pi$-types simply are contained in the ($\beta$) normalization rules of values.}
\end{figure}

\begin{theorem}[Categorical Semantics] A dCBPV- model $F\dashv U: \Ccat\leftrightarrows \Dcat$ supports
\begin{itemize}
\item $\Pi$-types iff we have $\Pi$-types in $\Ccat$;
\item $\csigma{-}{}$-types iff we have $\csigma{-}{}$-types in $\Dcat$ in the sense of having left adjoint functors $\csigma{A'}{}\dashv \Dcat(\proj{A}{A'})$ satisfying the left Beck-Chevalley condition for $\mathbf{p}$-squares;
\item $\homtype$-types iff we have objects $B\homtype C$ in $\Ccat$ \mccorrect{that are stable under change of base in the sense that $(B\homtype C)\{f\}\cong B(\{f\} \homtype C\{f\}$} such that we have natural bijections
$$\Dcat(\Gamma)(B,C)\cong \Ccat(\Gamma)(1,B\homtype C).$$
\end{itemize}
This semantics is both sound and complete in the usual sense of categorical semantics, leading to a 1-1 correspondence between models and theories supporting the appropriate connectives.
\end{theorem}
\begin{proof}
\begin{itemize}
\item This is a standard result in the semantics of dependent type theory \cite{jacobs1999categorical}, seeing that the value judgements form an ordinary (cartesian) dependent type theory.
\item This is precisely analogous to the situation in linear dependent type theory of chapter \ref{ch:4}. The introduction rule, by definition, corresponds to a natural transformation
$$\Sigma_{V\in \Ccat(\Gamma)(1,A)}\Dcat(\Gamma)(\Delta,B\{\langle \id_\Gamma, V\rangle\})\ra{}\Dcat(\Gamma)(\Delta,\csigma{A}{B}),$$ 
which can be equivalently represented, by taking $V=\diagv{\Gamma}{A}$ for the first argument and $\id_B$ for the second, as another natural transformation
$$
\Dcat(\Gamma.A)(\Delta,B)\ra{} \Dcat(\Gamma.A)(\Delta,\csigma{A}{B} \{\proj{\Gamma}{A}\}).
$$
This, by naturality in $\Delta$ is precisely determined by the image of $\id_B$ which is an element of
$$\Dcat(\Gamma.A)(B,\csigma{A}{B}\{\proj{\Gamma}{A}\}).$$
By a simple variation on the Yoneda lemma, we see that this is the same as specifying a natural transformation
$$
\Dcat(\Gamma)(\csigma{A}{B},C)\ra{}\Dcat(\Gamma.A)(B,C\{\proj{\Gamma}{A}\}).
$$
(This is one of the two defining natural transformations of the adjunction.)

The elemination rule corresponds by definition to a natural transformation $$\Dcat(\Gamma.A)(B,C\{\proj{\Gamma}{A}\})\times \Dcat(\Gamma)(\Delta,\csigma{A}{B})\ra{}\Dcat(\Gamma)(\Delta,C),$$
 which by naturality in $\Delta$ is equivalent to a natural transformation $$\Dcat(\Gamma.A)(B,C\{\proj{\Gamma}{A}\})\ra{}\Dcat(\Gamma)(\csigma{A}{B},C)$$ (where we have specialised to $\Delta= \csigma{A}{B}$ and have substituted $\id_{\csigma{A}{B}}$ for the second argument). (This is the other defining natural transformation of the adjunction.) The $\beta$- and $\eta$-rules precisely state that both defining natural transformations of the adjunction are inverse. As usual, the Beck-Chevalley condition corresponds to the compatibility of $\cpi{-}{}$-types with substitution.
\item Note that the introduction rule\mccorrect{,} by definition\mccorrect{,} corresponds precisely with the natural transformation from left to right in the categorical semantics. The elimination rule by definition corresponds to a natural transformation $$\Ccat(\Gamma)(1,B\homtype C)\times \Dcat(\Gamma)(\Delta,B)\ra{}\Dcat(\Gamma)(\Delta,C),$$ which by naturality in $\Delta$ is equivalent to a natural transformation $$\Ccat(\Gamma)(1,B\homtype C)\ra{}\Dcat(\Gamma)(B,C)$$ (where we have specialised to $\Delta= B$ and have substituted $\id_B$ for the second argument). The $\beta$- and $\eta$-laws precisely translate to these functions being inverse. Naturality of the bijections corresponds to compatibility of term formers with substitution. \mccorrect{Compatibility of the syntactic type formers with substitution corresponds with stability under change of base in the semantics.}
\end{itemize}
\end{proof}

Let us provide some context for thinking about $\csigma{-}{}$- and $\homtype$-types. As observed by Benton and Wadler \cite{benton1996linear}, linear logic can be seen as the term calculus of stacks for certain commutative effects. The question remained, if more general, possibly non-commutative effects would give rise to a certain kind of generalized, possibly non-commutative linear logic. In particular, the question was if one could define a monoidal-like structure on stacks in a general model of CBPV which generalizes the tensor of linear logic and similarly for the lollipop. A partial positive answer to this was given by the Enriched Effect Calculus (EEC) \cite{egger2009enriching}, telling us that any model of simple CBPV fully and faithfully embeds into a model where we have a binary operation $F(-)\otimes -$ (conventionally, somewhat misleadingly, written $!(-)\otimes -$) which takes a value type and a computation type and produces a computation type and a binary operation $- \homtype -$ (conventionally written $-\multimap -$) which takes two computation types to a value type. Our notation is chosen to be suggestive as these operations do not generalize the plain linear logic operations $\otimes$ and $\multimap$ but rather the composite connectives $F(-)\otimes (-)$ and $U(-\multimap -)$ that one can define in the LNL calculus \cite{benton1995mixed}.

Independently, linear dependent type theory forces a similar operation on us if we wish to extend $-\otimes -$ to a dependent connective \cite{vakar2014syntax}. Because types are only allowed to depend on cartesian assumptions and not linear ones, the best we can do is a multiplicative $\Sigma$-type $\csigma{-}{-}$. Seemingly for two very different reasons, the connective $F(-)\otimes -$ seems to be a preferred over $-\otimes -$, if one wants to generalize. We believe this is not a coincidence as the semantics of simply typed CBPV already forces various notions \mccorrect{from} dependent type theory on us.

In analogy with linear logic, we have the following isomorphisms of types, motivating some of our use of notation.
\begin{theorem}[Type Isomorphisms]
We have type isomorphisms
$$\begin{array}{lll}
 U\cpi{x:A}{\ct{B}} \cong \Pi_{x:A}U\ct{B} \quad\; & FA\homtype \ct{B} \cong U(A\functype\ct{B})\quad\; & F\Sigma_{x:A} A' \cong \csigma{x:A}{FA'}\\
 &  &\csigma{x:1}{\ct{B}}\cong\ct{B}\\
 && \csigma{x:A}{F1}\cong FA.
\end{array}$$
\end{theorem}
\begin{proof}
These are straightforward consequences of the universal properties in the categorical semantics of the various connectives involved, together with their compatibility with substitution.
\end{proof}

Let us say a few words about the interpretation of these connectives on some concrete classes of models.
\begin{theorem}[Concrete Models] \label{thm:dcbpvmodels} We have the following results on interpreting these connectives in concrete models.
\begin{itemize}
\item An indexed Eilenberg-Moore $\Ccat^T$ category for an indexed monad $T$ on a model $\Bcat^{op}\ra{\Ccat}\Cat$ of pure dependent type theory with $1$-, $\times$-, $0$-, $+$-, $\Sigma$-, $\Id$- and $\Pi$-types gives a model of dCBPV- with $\Pi$-types. If $\Ccat$ has indexed equalisers, we can interpret $\homtype$-types and if $\Ccat^T$ has indexed reflexive coequalisers, then we can interpret $\csigma{-}{}$-types. All these conditions are satisfied for $\Ccat = \Fam(\Set)$ and $T$ any finitary indexed monad like one of the usual reader, writer, state or exceptions  monads.
\item A model for the dependently typed LNL calculus with sum types, in the style of section \ref{sec:deplnl} gives a model of dCBPV- with $\Pi$, $\csigma{-}{}$- and $\homtype$-types. The dependent LNL calculus model of continuous families of predomains and domains is a specific example of this (see section \ref{sec:lindepscott}).
\end{itemize}
\end{theorem}
\begin{proof}
\begin{itemize}
\item The interpretation of $\Pi$-types is obvious.

Let us write $F\dashv U$ for the Eilenberg-Moore adjunction inducing $T$. As $T=UF$ is an indexed monad, recall that we have a dependent strength $\Sigma_A UFUk\ra{s_{A,Uk}}UF\Sigma_A Uk$. We note that we have a reflexive fork
\begin{diagram}
F(\Sigma_A Uk) & \rTo^{F\Sigma_{A}\eta_{Uk} } & F(\Sigma_A UFUk) & \pile{\rTo^{F\Sigma_{ A}k}\\\rTo_{F(s_{A,Uk}); \epsilon_{F\Sigma_A Uk}}} & F(\Sigma_A Uk).
\end{diagram}
Now, we can define $\csigma{A}{k}$ as the coequaliser of the reflexive pair. 
Note that a morphism $k\ra{\phi}l$ gives a natural transformation between the coequaliser diagrams for $\csigma{A}{k}$ and $\csigma{A}{l}$, or equivalently, a morphism $\csigma{A}{k}\ra{\csigma{A}{\phi}}\csigma{A}{l}$. This is easily seen to make $\csigma{A}{-}$ into a functor. Let us convey to the reader how we arrived at this definition: noting that $F(\Sigma_{A}A')\cong \csigma{A}{FA'}$ if we can prove that $\csigma{A}{-}\dashv -\{\proj{\Gamma}{A}\}$, we are defining $\csigma{A}k$ above as the coequaliser of
\begin{mccorrection}
\begin{diagram}
 \csigma{A}{FUFUk} & \pile{\rTo^{\csigma{A}{Fk}}\\\rTo_{\csigma{A}{\epsilon_{FUk}}}} & \csigma{A}{FUk},
\end{diagram}
\end{mccorrection}
showing that we are simply computing a $\Bcat$-indexed variation of Linton's construction of $\Set$-indexed coproducts of algebras \cite{linton1969coequalizers}. We now verify that indeed $\csigma{A}{-}\dashv \{\proj{\Gamma}{A}\}$. We can easily\footnote{These bijections can, in order, be motivated by the universal property of the coequaliser defining $\csigma{A}{k}$, the homset bijection of $F\dashv U$, naturality of $\epsilon$ and the observation that $\epsilon_l=l$, the homset bijection of $F\dashv U$, the homset bijection of $\Sigma_A\dashv -\{\proj{\Gamma}{A}\}$, the naturality of $s$ and functoriality of $T$, a triangle identity for $\Sigma_A\dashv -\{\proj{\Gamma}{A}\}$, the naturality of $\snd$, the homset bijection of $\Sigma_A\dashv -\{\proj{\Gamma}{A}\}$, and, finally, the definition of a homorphism from $k$ to $l$.} see that we have natural bijections between the following morphisms 
\\
\\
\resizebox{\linewidth}{!}{
$
\begin{array}{lll}
\phi\in \Ccat(\Gamma)^T(\csigma{A}{k},l)\qquad \;& & \\
\phi'\in \Ccat(\Gamma)^T(F(\Sigma_A Uk),l) &\textnormal{s.t.} & F(\Sigma_A k);\phi' = F(s_{A,Uk});\epsilon_{F\Sigma_A Uk};\phi' \qquad\;\\
\psi\in \Ccat(\Gamma)(\Sigma_A Uk , Ul)  &\textnormal{s.t.} & F(\Sigma_Ak);F(\psi);\epsilon_l = F(s_{A,Uk});\epsilon_{F\Sigma_A Uk};F(\psi);\epsilon_l \\
\psi\in \Ccat(\Gamma)(\Sigma_A Uk , Ul)  &\textnormal{s.t.} & F(\Sigma_Ak;\psi);\epsilon_l = F(s_{A,Uk};T\psi;l);\epsilon_l\\
\psi\in \Ccat(\Gamma)(\Sigma_A Uk , Ul)  &\textnormal{s.t.} & \Sigma_Ak;\psi = s_{A,Uk};T\psi;l \\
\psi'\in \Ccat(\Gamma.A)( Uk , Ul\{\proj{\Gamma}{A}\})  &\textnormal{s.t.} &\Sigma_A k; \Sigma_A\psi';\snd = s_{A,Uk};T(\Sigma_A\psi';\snd);l \\
\psi'\in \Ccat(\Gamma.A)( Uk , Ul\{\proj{\Gamma}{A}\})  &\textnormal{s.t.} &\Sigma_A (k; \psi');\snd = \Sigma_AT\psi';s_{A,Ul\{\proj{\Gamma}{A}\}};T\snd;l \\
\psi'\in \Ccat(\Gamma.A)( Uk , Ul\{\proj{\Gamma}{A}\})  &\textnormal{s.t.} &\Sigma_A (k; \psi');\snd = \Sigma_AT\psi';\snd;l \\
\psi'\in \Ccat(\Gamma.A)( Uk , Ul\{\proj{\Gamma}{A}\})  &\textnormal{s.t.} &\Sigma_A (k; \psi');\snd = \Sigma_A(T(\psi');l\{\proj{\Gamma}{A}\});\snd \\
\psi'\in \Ccat(\Gamma.A)( Uk , Ul\{\proj{\Gamma}{A}\})  &\textnormal{s.t.} & k;\psi' = T(\psi');(l\{\proj{\Gamma}{A}\}) \\
\psi''\in \Ccat(\Gamma.A)^T( k , l\{\proj{\Gamma}{A}\})  & & 
\end{array}
$
}\\
\\
Note that a sufficient condition to have reflexive coequalisers in $\Ccat^T$ is to have them in $\Ccat$ and to have $T$ preserve them. For a cartesian closed category $\Ccat$, a broad class of monads that preserve reflexive coequalisers are those arising from a finitary algebraic theory \cite{johnstone2002sketches} (section D5.3). 

Note that $\homtype$-types can be constructed exactly as in the proof of theorem \ref{thm:commtolinear}. We cannot usually construct an appropriate algebra structure on $k\homtype l$ (unless $T$ is a commutative monad).

\item We interpret $B\homtype C$ as $ U(B\multimap C)$. The rest should be obvious.
\end{itemize}
\end{proof}
As an example, consider the writing monad $-\times M$ on $\Set$ (which, as we have seen, does not admit dependent Kleisli extensions). We can note that its Eilenberg-Moore category is equivalent to the indexed category $\Fam(\Set^M)$ of families of $M$-modules (also known as $M$-sets or sets with an $M$-action). Being a presheaf category (if we consider $M$ as a one-object category), this is a topos and, in particular, we can construct $\csigma{-}{}$-types. If we calculate the coequaliser above (as colimits are computed pointwise), we find that $\csigma{A}{k}(s)$ has as carrier the quotient of $(\Sigma_{A}Uk)(s)\times M$ by the relation $(a,b,m\cdot m')\sim (a, b\cdot m,m')$ and the algebra structure that the quotient induces starting from the free one. Similarly, we have a (non-symmetric) indexed premonoidal  structure  $\otimes$ on $\Fam(\Set^M)$, where the carrier of $(k\otimes l)(s)$ is obtained by the quotient of $(Uk \times Ul)(s)\times M$ by the transitive closure of $(a,b,m\cdot m'\cdot m'')\sim (a\cdot m,b\cdot m', m'')$ and the algebra structure is obtained from the free one under the quotient. We can easily compute that $k\homtype l$ consists of the set of equivariant functions from $k$ to $l$. Note that $k\homtype l$ does not admit an appropriate algebra structure by \cite{foltz1980algebraic}\mccorrect{.} In this example, obviously, we have $\Pi$-types as usual sets of dependent functions.

Next, we consider the operational behaviour of terms of these new types. It turns out to be entirely well-behaved.
\begin{theorem}[Determinism, Strong Normalization, Subject Reduction] $\Pi$-, $\csigma{-}{}$- and $\homtype$-types do not alter any of the determinism, strong normalization or subject reduction results for dCBPV- of theorem \ref{thm:subjreddcbpv-}.
\end{theorem}
\begin{proof}
For $\Pi$-types, we note that we can still rely on the subject reduction and strong normalization proofs for $\beta$-reductions in pure dependent type theory of \cite{martin1998intuitionistic}.

Determinism and strong normalization of reductions for $\csigma{-}{}$- and $\homtype$-types is no different than for the other type formers. We verify subject reduction. In both cases it is clear from the subject reduction results for pure type theory that the transitions involving normalization of values satisfy subject reduction, so we focus on the remaining two transitions.

Let us assume that we start with a well-typed configuration $\Gamma;\cdot \vdash\return \nf{V}\otimes M:\csigma{A}{\ct{B}}\quad , \quad \Gamma;{\nil}:\csigma{A}{\ct{B}}\vdash \toin{[\cdot]}{\return x\otimes{\nil}}{L :: K}:\ct{C}$. Then, on the one hand, by inversion on the introduction rule for $\csigma{-}{}$-types, we have that $\Gamma\vdash \nf{V}:A$ and $\Gamma;\cdot\vdash M:\ct{B}[\nf{V}/x]$ (where $\Gamma,x:A\vdash \ct{B}\ctype$). On the other hand, noting that  $\toin{[\cdot]}{\return x\otimes{\nil}}{L :: K}$ is an abbreviation for $\lbi{\nil_1}{\toin{\nil_2}{\return x\otimes{\nil_3}}{L}}{K}$ by inversion on the rules for $\lbi{\nil_1}{}{}$ and the elimination rule for $\csigma{-}{}$-types, we have that $\Gamma,x:A;{\nil}:\ct{B}\vdash L:\ct{D}$ for some $\Gamma\vdash \ct{D}\ctype$ and that $
\Gamma;{\nil}:\ct{D}\vdash K:\ct{C}$. Therefore, because of the \mccorrect{substitution} property, we have that $\Gamma;\cdot\vdash L[V/x,M/{\nil}]:\ct{D}$. Hence, $L[V/x,M/{\nil}], K$ is a well-typed configuration.

Let us assume that we start with a well-typed configuration $\Gamma;\cdot \vdash M\textquoteleft \lambda_{\nil}L :\ct{B}\quad ,\quad \Gamma;{\nil}:\ct{B}\vdash K:\ct{C}$. Then, inversion on the elimination and introduction rules for $\homtype$-types gives us that $\Gamma;\cdot  M :\ct{D}$ and that $\Gamma;{\nil}:\ct{D}\vdash L:\ct{B}$ for some $\Gamma\vdash \ct{D}\ctype$. Therefore, the substitution property gives us that $\Gamma;\cdot \vdash L[M/{\nil}] :\ct{B}$, which means that $L[M/{\nil}],K$ is a well-typed configuration. 
\end{proof}

\begin{remark}[$\Id_{F(-)}^\otimes$-types] We could have included rules for $\Id^\otimes$-types, similarly to chapter \ref{ch:4}, connectives such that $F\Id_A \cong \Id_{FA}^\otimes$. While such connectives seem interesting from the point of view of linear logic, their use in CBPV is less clear. We are really interested in pure proofs of equality, rather than effectful ones (as, for instance, divergence can trivially inhabit any type $\Id_{FA}^\otimes$), so the use of $\Id_{F(-)}^\otimes$-types from the point of view of program verification is unclear. 
\end{remark}

\begin{remark}[Universes] \label{rmk:universes}
We have so far not considered higher-order quantification, which can be expressed in dependent type theory through universes, or types whose terms are (codes for) types. Universes (\`a la Tarski) arise as a special case of induction-recursion, a generalisation of more traditional, weaker induction schemes \cite{dybjer2000general}. As a rule of thumb, inductive-recursive families are more like an inductive than a coinductive construction, hence one would expect them to arise most naturally as value types. In the particular case of universes, we would expect separate universes $\mathcal{U}_v$ and $\mathcal{U}_c$ (both of which are value types) to classify value and computation types respectively. Like in pure type theory, one could include rules like\vspace{7pt}
\\
\AxiomC{$\vdash\Gamma\ctxt$}
\UnaryInfC{$\Gamma\vdash 1:\mathcal{U}_v$}
\DisplayProof\vspace{7pt}
\\
to build values of the universes which code for types and rules like\vspace{7pt}
\\
\AxiomC{}
\UnaryInfC{$x:\mathcal{U}_v\vdash \mathsf{El}_v(x)\vtype$}
\DisplayProof\vspace{7pt}
\\
\nopagebreak to make types out of codes.

Very recently, \cite{pedrot:hal-01441829} has further pursued a system resembling dCBPV- extended with universes.
\end{remark}

\section{Comparison with HTT}\label{sec:httcomparison}
We make a few observations on the relationship between dCBPV and an existing successful framework for certified effectful programming which is also based on dependent type theory: Hoare Type Theory (HTT) \cite{nanevski2006polymorphism} (implemented in Ynot \cite{nanevski2008ynot}).

Regarding the motivation behind both systems, HTT seems to have been developed from the start with the practical syntactic goal in mind of a language for verifying effectful programs. By contrast, dCBPV arose almost entirely from semantic considerations. In particular, dCBPV was motivated by the study of models of DTT which naturally model effects, like its domain semantics \cite{palmgren1990domain} and game semantics \cite{abramsky2015games}, the question if dDILL \cite{vakar2015syntax} could be interpreted as a DTT with commutative effects and the existing categorical semantics of CBPV which strongly suggests a dependently typed generalization \cite{levy2005adjunction}.

Regarding their implementation, HTT expresses a property $\phi$ of an  effectful program $V$ of type $A$ by saying that $V$ inhabits a type $T_\phi A$, where $T_\phi$ are monads which are indexed by formulae $\phi$ formed using an (external) separation logic. dCBPV sticks closer to the Curry-Howard correspondence in its formulation of properties $\phi$ of an effectful program $M$ of type $FA$: they are types $\phi$ depending on thunks of type $UFA$ and into which we can, in particular, substitute $\thunk M$ to see if we can construct an inhabitant witnessing the truth of $\phi(\thunk M)$.

\begin{savequote}[8cm]
I may not have gone where I intended to go, but I think I have ended up where I intended to be.
  \qauthor{--- Douglas Adams}
\end{savequote}

\chapter{\label{ch:6}Conclusions and Future Work}

\section{Conclusions}
In this thesis, we have examined the relationship between programming languages and formal logic, specifically the combination of computational effects with dependent types. We did this by analysing dependent types from three separate but related points of view on effects:
\begin{enumerate}
\item linear logic, which, as we argued, represents a type system for  computations exhibiting commutative effects;
\item game semantics, a setting to provide, in a unified way, models with strong completeness properties for a wide range of effectful type theories;
\item CBPV, an elegant framework for representing type theories with a wide range of effects with a fine-grained evaluation strategy that encompasses both traditional CBN and CBV.
\end{enumerate}
We believe these three perspectives often complement and sometimes reinforce each other.

Firstly, we constructed a dependently typed version of DILL, and showed it admits an elegant categorical semantics as well as a wide range of concrete models. Secondly, We constructed a CBN game semantics for dependent type theory, validated it by showing that it exhibits the usual completeness properties one expects of a game semantics and showed that it can be generalised to model effectful dependent type theory by relaxing the conditions on strategies. Thirdly, we studied a dependently typed version of CBPV and showed it has a simple categorical semantics, admits classes of models arising from both linear dependent type theory and indexed monads and that it has a well-behaved operational semantics. 

We learned that one principal source of the tension between type dependency and effects is the phenomenon that effectful computations are dynamic (in the sense that their evaluation can break equality) while we use types to provide static guarantees about programs. If types depend on effectful computations, therefore, they are at risk of losing their static nature.

Working in an adjunction language like CBPV (or the LNL calculus, for commutative effects) which distinguishes between (dynamic) computations and their thunks (which are static values), we can use types depending on thunks of effectful programs to express complex properties that we might want to verify for effectful programs.  Analogously, in linear logic, while types depending on linear terms are problematic, (linear) types depending on  cartesian terms are entirely harmless.

If we impose this restriction, one consequence is that we do not have CBV and CBN translations of type theory with unrestricted effects into CBPV (or dependently typed Girard translations, in the case of linear logic). Our view is that this is not at all a problem.  Indeed, dependently typed languages with unrestricted effects are of limited value, anyway, as effects render the language inconsistent as a logic, while the prime purpose of dependent types is to prove properties about programs. Moreover, the substitutio\mccorrect{n} of effectful computations in types introduces various technical challenges, as witnessed by our effectful game semantics for dependent type theory. One effect that could possibly be interesting to include in a dependent type theory in unrestricted fashion is that of non-local control operators because of its close relation to the classical principle of double negation elimination. However, it is already known that a constructive classical dependent type theory (i.e. a dependent type theory with $\mathsf{call/cc}$ at each type) is degenerate in the sense that it equates all programs \cite{herbelin2005degeneracy} \mccorrect{(propositionally)}. Therefore, any language that combines dependent types and effects in a meaningful way needs to have a mechanism for controlling the occurrence of effects. We hope to have demonstrated that modalities on the type system are an excellent tool for this purpose, in particular half-modalities (or adjunctions).

If one wants to obtain full CBV and CBN translations for dependent type theory with unrestricted effects, we showed that one needs to include Kleisli extensions for dependent functions in dCBPV. We have seen that these are not always supported in concrete models and can lead to problems with subject reduction in the operational semantics\footnote{We would like to point out that  dependent types are indeed combined with certain unrestricted effects, like recursion, in practice: Agda and Idris support unrestricted recursion and perform a separate optional termination check. It is no coincidence that recursion is precisely one of the effects for which dependent Kleisli extensions are well-defined. Indeed, its computations are not really dynamic in the sense that their evaluation respects equality. Therefore, they are much easier to combine with type dependency.}. Especially given that a similar effect can be achieved with the entirely unproblematic $\csigma{-}{}$-types, we believe such dependent Kleisli extensions are not a desirable feature of a dependently typed effectful language. Similar technical challenges arise with dependent projection products (or their linear logic equivalent, additive $\Sigma$-types). Indeed, the fact that these connectives were not naturally supported in categories of games and strategies is one of the prime reasons that a game semantics for dependent type theory had so far been absent and, more generally, that models of dependent type theory in computational settings of categories of cofree $!$-coalgebras had been missing.

Summarising, dependent types and computational effects form a delicate though not impossible combination. We hope to have demonstrated that robust systems can be achieved, as long as one is prepared to restrict type dependency to static values and exclude dependency on dynamic computations.

\section{Future Work}
We describe some interesting directions for future research, suggested by the work presented in this thesis.

\subsection{Linear Dependent Functions}
McBride \cite{mcbride2016got} presented a type system in which linear types depend on linear assumptions and with a type $\Pi_{x:A}^\multimap B$ of dependent functions from $A$ to $B$ that use $x$ exactly once (and in which types are allowed to refer to identifiers arbitrarily often). His solution relies on an unorthodox new view on linear logic in which we do not have separate classes of cartesian and linear types, but only one kind of type, achieving the linearity through annotation of identifier declarations with a count. We would like to analyse, using semantic methods, how this system relates to our work.

\subsection{Stable Homotopy as Effectful Homotopy?}
An indexed category of spectra up to homotopy, indexed over topological spaces, has been studied in e.g. \cite{may2006parametrized,ponto2012duality}, as a setting for stable homotopy theory. We can interpret this as a model of the dLNL calculus. It has been shown to admit $I$-, $\otimes$-, $\multimap$-, and $\Sigma_F^\otimes$-types. The natural candidate for a comprehension adjunction, here, is that between the infinite suspension spectrum and the infinite loop space: $F \dashv U \;\; = \;\; \Sigma^\infty\dashv \Omega^\infty$. What is particularly fascinating is that the corresponding monad $T= \Omega^\infty\Sigma^\infty$ seems to be very closely related (if not identical) to an important homotopical construction known as the Goodwillie exponential \cite{arone1995goodwillie}. This raises the  question whether one could phrase stable homotopy theory as an (commutative) effectful version of homotopy type theory and, if so, what the computational interpretation of the Goodwillie calculus in terms of effects should be.

\subsection{Dependently Typed Quantum Programming?}
Another fascinating possibility is that of models related to quantum mechanics. Non-dependent linear type theory has found very interesting interpretations in quantum computation (see e.g. \cite{AbrDun:CQLv2:2004}). The question rises if the extension to dependent linear types has a natural counterpart in physics. In \cite{schreiber2014quantization}, it was recently sketched how linear dependent types can serve as a language to talk about quantum field theory and quantisation. On a related note, one could well imagine using an extension of the dLNL calculus as a type system for a language in which we both have (cartesian) types for classical data and (linear) types for quantum data which may depend on the former. Such a precise type system may be useful for catching  bugs in quantum programs.

\subsection{Extending CBN Game Semantics for Dependent Types}
We see a few interesting directions for continuing the work we started in chapter \ref{ch:5}. One obvious continuation would be to try to extend the (full and faithful) completeness proof to the complete type hierarchy of $\DTTGame$. A next step is to study the interpretation of more general inductive families \cite{dybjer1994inductive,clairambault2009least} and inductive-recursive definitions (of which type universes are an obviously interesting example) \cite{dybjer2000general}. Such a study of universes should also lead to a more intensional notion of a dependent game as a kind of strategy.

On a different note, it would be desirable to find an alternative, less technical presentation of a suitable category of $!$-coalgebras extending $\Ctxt(\Gamecat_!)$ which also models dependent type theory (cf. section \ref{sec:depgirard}).

Finally, note that our games model of dependent types has identity types that are intensional in orthogonal ways compared to the homotopy semantics \cite{awodey2009homotopy}. In our case, all non-trivial propositional identities concern a kind of homotopy in the time direction (applicative equivalence of functions), rather than in the space direction, in the sense that our ground types are discrete  and we  accumulate non-trivial propositional identities if we ascend the function hierarchy. By contrast, in the homotopy semantics of dependent types, all non-trivial propositional identities exist on ground types and we do not acquire any non-trivial identifications of functions (beyond their pointwise identity). We propose to pursue a notion of what one might call a category of \emph{homotopy games}, which should factor over the pullback of the \emph{spatial and temporal extensional collapse} of the two models,\\
\\
\resizebox{\linewidth}{!}{\mbox{
\begin{diagram}
\mathsf{HtpyGame}&  &  &\\
& \;\;\infty-\mathsf{Gpd}\times_{\Set}\mathsf{Game}\SEpbk \rdOnto(3,1) \rdOnto(1,3) \rdDotsonto(1,1) & \rOnto & \mathsf{Game}\hspace{-20pt}\; &A\\
&  \dOnto & & \dOnto^{\parbox{150pt}{\textnormal{collapsing time-like identity,}\\ \textnormal{a.k.a. extensional collapse}}} & \hspace{-30pt}\dMapsto \\
& \infty-\mathsf{Gpd} &\rOnto_{\textnormal{\hspace{45pt} collapsing space-like identity, a.k.a. $0$-truncation}\hspace{-25pt}} & \Set \hspace{-20pt}\;& \hspace{-20pt}\str(A)/\mathsf{app.equiv.}\\
& X & \rMapsto & ||X||_0 .&
\end{diagram}}\hspace{33pt}\;}\\
\\
That is, we are looking for a setting to model DTT which combines the possibility of non-trivial propositional identity on ground types of the ($\infty$-)groupoid model of DTT\mccorrect{ }\cite{awodey2009homotopy} with the failure of function extensionality of the game semantics. We hope this would not only result in a satisfactory game semantics for quotient types and higher inductive types, but would also give deeper insights into the subtle shades of intensionality that arise in dependent type theory, by cleanly separating out the time-like and space-like aspects of propositional identity.   

\subsection{Game Semantics for dCBPV}
The practical challenge of constructing a CBN game semantics for dependent type theory was important for us in developing our understanding of the interaction between effects and type dependency. Moreover, we hope that it can be a useful addition to the large family of game semantics for various logics and programming languages.

However, in hindsight, we believe that dependently typed languages with unrestricted effects such as those modelled by our CBN game semantics are not the most interesting combination of dependent types and effects (though also not uninteresting). What would be really interesting is to construct a game semantics for dCBPV in the style of \cite{levy2005adjunction,levy2012call}. We hope it could both be simpler are more practically relevant. It should be emphasized though that the experience of working with the CBN game semantics of dependent types was necessary for us to reach this insight.

\subsection{Certified Real-World Programming in dCBPV-}
Cervesato and Pfenning pioneered the use of systems combining dependent types and linearity to reason about effectful computations in \cite{cervesato1996linear}. We hope that our system dCBPV- can be a step forward for this purpose, through its generalisation to non-commutative effects and the extra expressive power obtained by distinguishing between values and computations, where the value judgement can be seen as a pure logic that be used for reasoning about (thunks of) effectful computations, defined with the computation judgement. It is particularly salient that this distinction allows us to use $\Id$-types, which were painfully absent from Cervesato and Pfenning's system. 

In particular, we hope that dCBPV- can serve as an alternative to Hoare Type Theory that sticks closer to the elegance of the Curry-Howard correspondence and that it can be extended to a practical language for both writing and verifying real world effectful code. For that purpose, the next step is to add mechanisms for forming types that express delicate properties of thunks of computations exhibiting specific effects, like properties of the state before and after a computation is run. \cite{ahman2017handling}\mccorrect{ recently proposed using effect handlers, which have a semantics as monad algebras $TA\ra{k}A$, as a way of lifting predicates on $A$ to ones on $TA$. We believe this idea sounds very promising and deserves to be explored further.} Another important step would be the implementation of a type checker for the resulting system.


\startappendices
\begin{savequote}[8cm]
Huh?!
  \qauthor{--- Sylvester Stallone}
\end{savequote}

\chapter{Summary for a General Audience}
Over the past decades, we --  both as individuals and as a society as a whole -- have very rapidly become reliant on computer systems, to the point that we trust them with our private data (e.g. mobile phone communications and medical records), our critical resources (e.g.  bank transactions), the smooth running of society (e.g. elections,  financial markets, and classified government documents) and even our lives (e.g. auto-pilots in air planes, self-driving cars, medical devices and missile detection systems on which decisions whether or not to engage in nuclear war are based). While many people will agree that computers have changed our lives for the better, there have been enough incidents to make us question how much trust we should put into computer systems for critical applications: e.g.\footnote{Links to news stories on many fascinating software bugs can be found on \cite{huckle2015collection}.} false alarms in both US (in 1980) and Soviet (in 1983) missile detection systems which could have easily led to nuclear war, Therac-25 medical radiation therapy devices administering deadly doses of radiation, the destruction of the Ariane 5 rocket (costing \$370 million dollar) and regularly uncovered cryptography bugs like Heartbleed which enable hackers to get into critical computer systems (like those of the Democratic National Convention, in the context of the 2016 US election). Moreover, software bugs are unbelievably expensive, annually costing an estimated \$312 billion\footnote{To give a sense of scale, according to a 2015 U.N. report, it would cost around \$267 billion  annually to bring the roughly 800 million people world-wide living in  extreme poverty  up to the World Bank's poverty line immediately \cite{unendpoverty}.}, with software developers spending on average half their time debugging \cite{britton2013reversible}.

These bugs are often not introduced into software due to negligence on behalf of the programmer. Rather, we learn from experience that bugs are almost unavoidable when writing software. The human mind is simply rather ill-suited for writing watertight computer software. Many bugs can be found through testing, but, depending on the application, that may not be enough: real world software, particularly concurrent software, can have a space of possible executions that is simply too large to explore through testing. For instance, new critical bugs are found in web-protocols every week despite extensive testing. For the most critical of applications, only a formal machine-checked proof is enough to guarantee the correctness of a piece of code.

Currently, while very suitable programming languages for writing computer software exist as well as good logical frameworks for writing proofs (formal arguments demonstrating the truth of a proposition), there is no single satisfactory system in which we can both write production code and write and check a proof about the correctness of this code. We believe an important reason for the absence of such a system is a lack of fundamental understanding of how real-world programming languages relate to logical frameworks. The aim of this thesis is to improve on the state of the art of such an understanding and, as a consequence, to work towards the dream of having a single elegant language for writing provably correct production code.  

It is clear that programs written in very simple programming languages, so-called purely functional languages, are effectively the same thing as mathematical proofs. This idea is called the Curry-Howard correspondence. Similar to how we can organise proofs by the proposition (e.g. $A\txt{and} B\txt{implies} C\txt{or}  D$) whose validity they demonstrate, we can classify computer programs according to their so-called type (e.g. the type of programs that take two inputs, one of type $A$, one of type $B$ and produce an output which will either be of type $C$ or type $D$). That is, under the Curry-Howard correspondence, types correspond to propositions in the same way that purely functional programs correspond to proofs. Types can be thought of as expressing properties of programs. Some examples of types are
the type of booleans,
the type of integers,
the type whose elements consist of a pair of a boolean and a string,
the type whose elements are either a boolean or an integer,
the type of programs that take two integers as input and produce five booleans and
the type of programs that take a program from booleans to integers as input and produce an integer as output. 

Real world programming languages and useful logical frameworks are not the same thing, however! On the one hand, there are more good computer programs than acceptable proofs. For instance (this is just one of many examples), a computer program can loop indefinitely (like an operating system), but a circular argument is unacceptable as a proof.  These extra programs are very useful in practical software development. On the other hand, the most useful logics include more propositions than there are types in real world programming languages. Indeed, while these programming languages include so-called simple types, which correspond to propositions formed by the logical connectives ``and'', ``or'' and ``implies'', dependent types are missing, corresponding to the crucial propositions of the form ``$P(x)$ holds for all $x$'' or ``$P(x)$ holds for some $x$''. (An example would be the statement ``all Buddhists are happy''.)  

This thesis examines how this gap between programming languages and logical frameworks can be bridged and how a single language can be designed that can serve for writing both real world code as well as formal, machine-checked proofs that this code has the properties that one desires. Before tackling, head on, the question of how effectful programs (programs that do not correspond to proofs) can be given dependent types, we first study the closely related topics of how so-called linear logic can be extended with dependent types and how a game theoretic interpretation can be given to logical frameworks with dependent types.

Linear logic is a logic in which we keep track of how often each assumption is used in a proof: assumptions cannot be copied or discarded freely. Linear logic proofs are closely related to effectful programs. The intuition is that effectful programs could, for instance, make a random choice, meaning that two executing copies of the same program may later cease to be equal.

A formal logic can be equivalently phrased in terms of game theory by interpreting a proposition as a turn-based two-player game (think of it as the game of formal debates about the proposition, analogous to Socratic dialogues) and a proof of that proposition as a certain kind of winning strategy for that game (if we can win any debate about a proposition, it must be true and vice versa).  The charm of game semantics is that we can weaken the conditions we put on the strategies we consider to obtain various effectful programs. For instance, partial strategies (in which we do not always have a response to everything our opponent says in a debate, meaning that we do not always win) correspond to programs which may loop for ever and never return an output.

This thesis first presents a dependently typed linear logic and game theoretic interpretation for dependent types. This helps us build an understanding that is useful to, next, present an elegant language which can both serve as an effectful programming language for writing software and as a pure logic to prove properties about the software we write in it. We hope that this work, on the one hand, contributes to a better\mccorrect{ }understanding of the foundations of the disciplines of mathematics and computer science and their relationship and, on the other, ultimately will help us to work towards a world in which one can safely rely on critical computer systems, both as individuals and as a society.


\begin{savequote}[8cm]
\end{savequote}

\setlength{\baselineskip}{0pt} 

{\renewcommand*\MakeUppercase[1]{#1}%
\printbibliography[heading=bibintoc,title={\bibtitle}]}

\end{document}